\documentclass[reprint, amsmath,amssymb, aps]{revtex4-2}

\usepackage{amsmath,amssymb,amsthm}
\usepackage{graphicx}
\usepackage{dcolumn}
\usepackage{bm}
\usepackage{hyperref}
\usepackage{xcolor}
\usepackage{bbm}
\usepackage{braket}
\usepackage{algorithm}
\usepackage{algpseudocode}
\usepackage{booktabs}
\usepackage{multirow}
\usepackage{tikz}
\usetikzlibrary{quantikz}
\usepackage{tabularx}

\newtheorem{theorem}{Theorem}
\newtheorem{lemma}[theorem]{Lemma}
\newtheorem{corollary}[theorem]{Corollary}
\newtheorem{proposition}[theorem]{Proposition}
\newtheorem{definition}[theorem]{Definition}
\newtheorem{conjecture}[theorem]{Conjecture}
\newtheorem{remark}[theorem]{Remark}
\newtheorem{fact}[theorem]{Fact}

\newcommand{\R}{\mathbb{R}}

\newcommand{\calL}{\mathcal{L}}
\newcommand{\calG}{\mathcal{G}}
\newcommand{\calF}{\mathcal{F}}
\newcommand{\calN}{\mathcal{N}}
\newcommand{\calS}{\mathcal{S}}
\newcommand{\btheta}{\boldsymbol{\theta}}
\newcommand{\bphi}{\boldsymbol{\phi}}

\newcommand{\bvarphi}{\boldsymbol{\varphi}}

\newcommand{\bs}{\mathbf{s}}

\newcommand{\bx}{\mathbf{x}}

\newcommand{\E}{\mathbb{E}}
\newcommand{\Var}{\mathrm{Var}}
\newcommand{\tr}{\mathrm{tr}}
\newcommand{\poly}{\mathrm{poly}}

\begin{document}

\title{Scalable Quantum Machine Learning:\\ 
Trainability, Expressivity and Efficiency}

\author{Iordanis Kerenidis}
\email{iordanis@quantumsignals.ai}
\affiliation{IRIF, CNRS, Universit\'e Paris Cit\'e, France}
\affiliation{Quantum Signals SAS, Paris, France}

\date{18 October 2026}

\begin{abstract}
Designing scalable parameterized quantum circuits for machine learning has
remained elusive due to three fundamental obstacles: barren plateaus that
prevent gradient-based training, the absence of provable guarantees that the
learned function class is classically hard, and prohibitive circuit
evaluations per gradient step. We propose a solution based on the
\emph{unitary brick-wall} circuit: a $k$-particle fermionic architecture for
nearest-neighbor hardware, combining Reconfigurable Beam Splitter gates with
interleaved single-qubit phase gates and a non-Gaussian magic-state encoding,
where the particle number $k$ is a tunable resource dial that trades
classical simulation hardness against training cost.

\noindent
\textbf{Trainable.} The brick-wall has dynamical Lie algebra $\mathfrak{u}(n)$
and is surjective onto $U(n)$ via Givens rotations, enabling Haar
initialization. Two-body correlator readouts achieve gradient variance
$\Theta(k^3/n^5)$, which remains polynomial in $n$ throughout the regime
$n - 2k = \Omega(n)$.

\noindent
\textbf{Expressive.} The classical hardness of the pipeline is controlled by
the particle number $k$: best-known classical algorithms for sampling run in time
$2^{\Theta(k)}$~\cite{heo2026fast,reardonsmith2024classical},
worst-case \#P-hardness holds from $k = n^{\epsilon}$, and the full
average-case machinery of Fermion Sampling~\cite{oszmaniec2022fermion}
applies at $k = \Theta(n)$. The sampler is input-dependent, so classical
simulation costs are paid per query: at our operating point $k = 100$ one
query needs $\sim\!10^{15}$ operations, which no commercial machine serves at
query rates, and at $k = 120$ a dedicated exascale machine sustains fewer
than one query per second.

\noindent
\textbf{Efficient.} A multi-layer parallel parameter-shift rule computes all
$O(n^2)$ gradients from $4kn$ distinct circuit evaluations per gradient step --- a
factor $n/k$ reduction over the $4n^2$ evaluations required by the
standard parameter-shift rule, growing linearly with $n$ at fixed $k$.

\noindent
The \emph{unitary butterfly} variant targets all-to-all hardware, with
$2\log n$ layers of parallel gates and $\tfrac{3}{2}n\log n$ parameters. It achieves similar expressivity and
hardness guarantees to the brick-wall, and a gradient step needs $4k\log n$
circuit evaluations, the same factor-$n/k$ reduction over the
$4n\log n$ naive PSR cost. Its trainability is established at two levels: the
absence of exponential barren plateaus is \emph{unconditional}, whereas the
sharp $\Theta(k^3/n^5)$ rate holds under a conjectured scalar comparison with the two-particle Haar second moment.
\end{abstract}


\maketitle

\section{\label{sec:intro} Introduction}

Quantum machine learning has developed along two largely distinct trajectories---statistical learning theory and kernel methods on one hand, and deep learning and neural networks on the other---corresponding to two different visions of how quantum computation can augment artificial intelligence.

The first draws on quantum algorithms for linear systems, matrix operations, and singular value transformations~\cite{harrow2009quantum,kerenidis2017quantum,gilyen2019quantum}, covering a broad range of tasks including among others nearest-neighbor search~\cite{lloyd2013quantum}, principal component analysis~\cite{lloyd2014quantum}, clustering~\cite{kerenidis2019qmeans}, convolutional networks~\cite{kerenidis2020cnn}, and recommendation systems~\cite{kerenidis2017quantum}. The hardware requirements of these methods are fault-tolerant quantum computers with QRAM-encoded data and circuits of polynomial depth, and they achieve important theoretical speedups over the classical algorithms that are used in practice and operate on explicit matrix representations. The precise extent of their computational advantage was significantly revised by the dequantization work of Tang and collaborators~\cite{tang2019recommendation,tang2021pca,chia2022svt}, which showed that when classical algorithms are given analogous $\ell^2$-norm sampling access to the input data, the quantum speedup is polynomial for a broad class of linear algebraic tasks. The conclusion is doubly nuanced. First, the dequantization results do not eliminate quantum advantage: they show that under $\ell^2$-norm sampling access, the quantum speedup reduces from exponential to polynomial, and in some cases large polynomial. Second, the dequantized classical algorithms, while polynomially competitive with quantum under sampling access assumptions, appear in practice to be slower than the fastest classical algorithms that operate directly on explicit matrix representations---suggesting that the sampling access model, though theoretically interesting, may not correspond to practical classical computation. The net result is that quantum linear algebra may still offer genuine advantages over practical classical methods, but establishing this rigorously awaits large-scale, fast, fault-tolerant quantum hardware and efficient QRAM implementations. In fact, very recently, Zhao et al.~\cite{zhao2026exponential} proved unconditional exponential quantum advantage in classification and dimension reduction via quantum oracle sketching, establishing machine learning on classical data as a concrete domain of quantum advantage in the linear algebra paradigm. This paradigm represents the quantum analogue of traditional, kernel-based machine learning.

The second area, which we call \emph{quantum deep learning} (QDL), draws on parameterized quantum circuits (PQCs) trained by gradient descent~\cite{mitarai2018quantum,schuld2019quantum,benedetti2019parameterized,perez2020data}, with the circuit playing the role of a neural network layer or feature extractor~\cite{biamonte2017quantum,cerezo2021variational,farhi2018classification}. These methods are aligned with the deep learning revolution that has vastly enlarged the applicability and impact of machine learning: instead of hand-crafted features or explicit kernel methods, the circuit learns representations directly from data. This approach has been explored in supervised learning~\cite{havlicek2019supervised,abbas2021power,landman2022quantum} and as policy and value function approximators in reinforcement learning agents~\cite{dunjko2016quantum,jerbi2021parametrized,cherrat2023hedging}. Their hardware requirements are far more modest than the first paradigm---shallow circuits of linear or even logarithmic depth can suffice, the optimization loop runs classically, and the quantum device is called only for circuit evaluation---making QDL a natural target for near-term and early fault-tolerant quantum hardware~\cite{preskill2018nisq}. Analogously to classical deep learning, where the value of a given architecture or training method is established empirically rather than by theoretical proofs of superiority, the value of QDL will ultimately be established on real hardware and real tasks: the key question is whether quantum representations can access function classes that lie outside the reach of classical models of comparable parameter count, enriching model capacity in ways that matter in practice. This is the paradigm this paper addresses.

\subsection{\label{sec:stateofart} The State of Quantum Deep Learning}

The past decade has produced a rich body of work on quantum deep learning, advancing both the design of trainable architectures and our theoretical understanding of their limitations.

\paragraph{Quantum neural network architectures.}
Early proposals for quantum neural networks sought direct quantum analogues of classical layers: quantum perceptrons, circuits inspired by classical convolutional architectures, and variational models~\cite{farhi2018classification,schuld2014quest,benedetti2019parameterized}. A key insight that emerged was that PQCs can be interpreted as nonlinear feature maps into quantum Hilbert space~\cite{schuld2019quantum,havlicek2019supervised}, connecting quantum neural networks to the classical theory of kernel methods and opening a systematic framework for understanding their expressive power.

A central challenge identified early was the issue of \emph{barren plateaus}: for generic, randomly initialized circuits, gradients vanish exponentially with system size, rendering training infeasible~\cite{mcclean2018barren}, with the severity depending on the locality of the cost function~\cite{cerezo2021cost}. This motivated a sustained search for architectures with provable trainability. Quantum convolutional neural networks (QCNNs)~\cite{cong2019qcnn}, which interleave convolutional and pooling layers in a hierarchical structure, were shown to be free of exponential barren plateaus~\cite{pesah2021absence}. Hierarchical tensor-network classifiers~\cite{grant2018hierarchical} and equivariant quantum neural networks~\cite{meyer2023exploiting,larocca2022group} extended these ideas to symmetry-constrained architectures. The unifying theoretical framework---that gradient variance is controlled by the dimension of the dynamical Lie algebra (DLA)---was established by Fontana et al.~\cite{fontana2024adjoint} and Larocca et al.~\cite{larocca2021diagnosing}, providing a precise, quantitative account of why structured circuits train. A troubling corollary, however, is that circuits with small DLA---which avoid barren plateaus---can often be efficiently classically simulated~\cite{cerezo2025does}. This has been confirmed concretely for the QCNNs of~\cite{cong2019qcnn}: Bermejo et al.~\cite{bermejo2024qcnn} showed that they are effectively classically simulable because, when randomly initialized, they can only access information encoded in low-bodyness observables, and their prior empirical success is explained by benchmarking on datasets that are ``locally easy'' in precisely this sense. Similar simulability results hold for certain permutation-equivariant architectures~\cite{anschuetz2023symmetric}, though the general equivariant case remains only partially understood. These results sharpen the fundamental tension between trainability and genuine quantum advantage: the same structural constraints that prevent barren plateaus tend to limit the circuit to a classically accessible function class.

\paragraph{Hamming-weight-preserving quantum neural networks.}
A parallel line of work introduced \emph{orthogonal quantum neural networks} based on Hamming-weight-preserving fermionic gates, which implement orthogonal matrix multiplications as natural quantum operations~\cite{kerenidis2022quantum,landman2022quantum}. Depending on the hardware connectivity, different topologies were proposed: brick-wall and pyramid circuits for nearest-neighbor connectivity with linear depth and quadratic number of parameters, and butterfly circuits arranged in the pattern of the Cooley-Tukey FFT for all-to-all connectivity, operating in $O(\log n)$ depth with $O(n\log n)$ parameters~\cite{kerenidis2022quantum}. These circuits appeared in subsequent work spanning quantum vision transformers~\cite{cherrat2024quantum}, quantum deep hedging~\cite{cherrat2023hedging}, financial forecasting~\cite{thakkar2024forecasting}, and clinical data imputation~\cite{kazdaghli2023data}, and were tested on real quantum hardware for up to 32 qubits~\cite{johri2021nearest,mathur2026scalable}. The underlying subspace structure has also been studied as a primitive in its own right~\cite{kerenidis2022subspace}. It is worth noting that in the special case of Hamming weight $k=1$, these circuits are equivalent to linear optical networks and are efficiently classically simulable~\cite{jozsa2008matchgates}. Far from being a limitation, this simulability was productive: it revealed that the RBS gate provides a natural and complete parametrization of the orthogonal group via Givens rotations, yielding a fully classical algorithm for training orthogonal neural networks with exact orthogonality preservation and $O(1)$ gradient cost per parameter---a combination not achieved by prior classical methods based on Cayley transforms, Householder reflections, or Stiefel manifold optimization~\cite{kerenidis2022quantum}.

The interesting regime for quantum advantage requires larger values of Hamming weight $k$. The RBS circuit itself is always a passive FLO transformation and remains efficiently classically simulable regardless of $k$~\cite{jozsa2008matchgates}. Classical hardness arises not from the circuit but from the \emph{input state}: when the $k$-particle input is non-Gaussian, it can be provably hard to simulate classically when combined with a passive FLO circuit~\cite{oszmaniec2022fermion}, comparable to Boson Sampling hardness. The theoretical properties of these circuits were subsequently analyzed from two complementary angles. Monbroussou et al.~\cite{monbroussou2025trainability} showed that the gradient variance of the Hamming-weight preserving circuits (DLA $\mathfrak{so}(n)$) scales as $\Theta(1/\binom{n}{k})$. For constant $k$ this is polynomial in $n$, but for values of $k$ that grow with $n$ (for example $k=n/2$) this variance is exponentially small, constituting an actual barren plateau. Independently, the adjoint representation framework of Fontana et al.~\cite{fontana2024adjoint} established that for a Lie Algebra Supported Ansatz (LASA)---where the measurement observable $iM$ lies in the DLA---gradient variance scales as $\Theta(1/\dim(\mathfrak{g}))$ for any particle number $k$. For DLA $\mathfrak{u}(n)$ this gives $\Theta(1/n^2)$, a polynomial bound independent of $k$. Note that the occupation-number observable $n_j = c_j^\dagger c_j$ (measuring whether mode $j$ is occupied) is the natural readout of a computational basis measurement; $i n_j$ lies in $\mathfrak{u}(n)$ but \emph{not} in $\mathfrak{so}(n)$. This is the LASA condition for generic input states. For our specific magic-state encoding, however, every mode of the spread input carries equal occupation $k/n$ (Section~\ref{sec:encoding}), and losses built from statistics of fixed body number are classically computable via $r$-RDM propagation (Proposition~\ref{prop:onebody_const}); the trainable readout of interest is the two-body correlator $n_i n_j$, for which the relevant gradient variance is $\Theta(k^3/n^5)$ (Section~\ref{sec:bp}).

\paragraph{The classical-simulation boundary for non-Gaussian fermionic inputs.}
The complexity of classically simulating non-Gaussian fermionic inputs evolved under passive FLO has been refined considerably in recent work, and the resulting picture directly shapes our architectural choices. Sampling-based hardness (multiplicative precision on individual probabilities) for \emph{fully-packed} magic state inputs, i.e.\ at extensive particle number $k = \Theta(n)$, follows from the Fermion Sampling argument of Oszmaniec et al.~\cite{oszmaniec2022fermion}; whether hardness survives with fewer magic blocks was posed there as an open problem, a question that directly shapes the choice of $k$ in our construction (Section~\ref{sec:k_choice}). Dias and K\"onig~\cite{dias2024classical} and Reardon-Smith et al.~\cite{reardonsmith2024classical} gave classical algorithms for arbitrary non-Gaussian inputs whose cost scales as $O(\chi^2 \cdot \mathrm{poly}(n))$, where $\chi$ is the FLO extent of the input. More recently, Oh, Oszmaniec, Reardon-Smith, and Zimbor\'as~\cite{oh2026classical} exhibited a structurally tighter result for the specific class of \emph{paired} non-Gaussian inputs (form-rank-$2$ states such as $1/\sqrt{2}(\ket{0011}+\ket{1100})$): an additive-error classical estimator for transition amplitudes, overlaps, and arbitrary-weight number correlators, with sample complexity $O(\epsilon^{-2}\log\delta^{-1})$ matching the operational shot complexity of quantum hardware. Their algebraic core is a \emph{mixed-Pfaffian} reduction that compresses an exponentially large sum of determinants into a single coefficient of a multivariate Pfaffian polynomial. This places paired non-Gaussian inputs in the additive-error-tractable regime, where the FLO extent provides only an upper bound that the Pfaffian shortcut beats. Most recently, Heo et al.~\cite{heo2026fast} gave an \emph{exact} sampling algorithm for inputs that are products of four-mode magic blocks---precisely the input class used here---at cost $O(B^4 + 2^B)$ per sample, where $B$ is the number of blocks and $\chi = 2^B$; both the square on the extent and the polynomial prefactor disappear. This is the sharpest sampling cost currently known for our pipeline, and it is the figure against which we calibrate the particle number in Section~\ref{sec:k_choice}.

\paragraph{Training quantum neural networks.}
The standard training algorithm for PQCs is the parameter-shift rule~\cite{mitarai2018quantum,schuld2019gradients}, which computes the exact
gradient of each parameter using two circuit evaluations per parameter. For a
circuit with $T$ trainable parameters, this requires $2T$ circuit evaluations
per gradient step---growing linearly with the number of parameters and becoming
prohibitive for large circuits. Stochastic approximations can reduce the per-step cost to $O(1)$ circuit evaluations but introduce bias and increased variance~\cite{spall1992multivariate,gacon2021simultaneous}. Recent stochastic block coordinate
descent~\cite{wierichs2022general} selectively updates parameter subsets but
does not reduce the fundamental per-parameter cost. A complementary line of
work casts stochastic approximations, random coordinate descent, and the parameter-shift rule as
special cases of a single forward-gradient estimator parametrized by the
number of random directions probed per step, with an adaptive optimiser
choosing this number from the measurement statistics
themselves~\cite{coyle2026adaptive}. An alternative approach
is to reduce the effective number of parameters by mixing over ensembles of
circuits---the density QML framework~\cite{coyle2025density} showed that this
can improve both trainability and per-step cost, but at the price of reduced
expressivity, trading the capacity to represent complex functions for easier
optimization. Note as well that when there is a single trainable layer where
all gates act on disjoint qubit pairs, their generators commute, and in
principle all gradients within that layer could be estimated simultaneously
from a single set of circuit evaluations rather than
sequentially~\cite{mathur2026scalable}. From a practical point of
view, for circuits with $T = O(n^2)$ parameters, namely a linear depth parametrized quantum circuit, the naive PSR requires $O(n^2)$ circuit evaluations per
gradient step, making on-hardware training at scale impractical even for 100 qubits, regardless of
which of these strategies is employed.

\paragraph{Data loading.}
A key distinction from the quantum linear algebra paradigm is that quantum neural networks process one data point at a time: there is no requirement to load an entire dataset into a quantum state or to maintain a QRAM over the full data matrix. The data loading problem here is therefore a per-sample question---how to efficiently encode a single classical vector $\bx \in \R^n$ into a quantum state that is expressive enough to support learning, without incurring circuit depth that negates the advantage of the shallow trainable circuit that follows. The simplest approach, \emph{angle encoding}, uses one rotation gate per feature and requires $O(1)$ depth, but encodes data in rotation angles rather than amplitudes and does not preserve inner product structure~\cite{schuld2019quantum}. \emph{Amplitude encoding} represents $n$ features as the amplitudes of the unary basis of an $n$-qubit state, preserving $\ell_2$ distances, and can be implemented with $n-1$ RBS gates in a ladder formation on nearest-neighbor architectures or with a $O(\log n)$ depth RBS gate tree on all-to-all connectivity architectures~\cite{kerenidis2022quantum}. For $k$-Hamming weight states, B\"{a}rtschi and Eidenbenz~\cite{bartschi2022short} gave short-depth circuits for Dicke states and Farias et al.~\cite{farias2025quantum} gave optimal RBS-based encoders for $k$-particle states. The right encoding strategy for QDL remains an important question, since the encoding choice simultaneously impacts the trainability of the circuit and the expressivity of the learned function class~\cite{schuld2021effect}.

\subsection{\label{sec:potential} The Potential of Quantum Deep Learning}

The progress surveyed above has significantly advanced our understanding of what quantum neural networks can and cannot do: we know when they train, when they do not, when they are classically simulable, and what structural properties govern each. Yet despite this theoretical clarity, the fundamental question---what is the genuine potential of quantum deep learning, and what should we concretely expect from it?---remains wide open. Answering it requires more than additional technical results. It requires first being precise about what we mean by potential, and then identifying the specific challenges that must be overcome to realize it.

We argue that the right question is not whether quantum circuits can replace classical neural networks, or whether they offer provable worst-case speedups on standard benchmark tasks. Analogously to classical deep learning---where the value of convolutional networks or transformers was not established by theoretical proofs of superiority but by empirical demonstration of expressivity on real tasks---the right question is whether quantum circuits can represent useful functions or policies that lie outside the reach of classical models of comparable parameter count, enriching the model class in ways that matter in practice.

Two conditions must therefore be met for quantum deep learning to realize its potential. The first is \emph{expressivity}: the quantum model must access function classes that are provably beyond the reach of classical models of comparable size. This is a question of quantum complexity theory, and for the right architecture, precise lower bounds on the classical simulation cost can be established. The second condition is \emph{practical utility}: that the functions accessible to the quantum model are not merely hard for classical computers but actually useful---that they correspond to inductive biases or representations that help on real tasks. This second condition cannot be established by complexity-theoretic arguments alone; it requires running quantum models on real data and real hardware. And this is precisely where scalability becomes the binding constraint: there is no point in demonstrating quantum expressivity on toy problems if the architecture cannot be trained efficiently at scale. The path to understanding the practical potential of quantum neural networks therefore runs directly through the problem of building scalable, trainable, and expressive architectures---which is the problem this paper addresses.

We identify five concrete challenges that any scalable quantum framework must address simultaneously, and we show that two complementary architectural choices---the unitary brick-wall for nearest-neighbor hardware and the unitary butterfly for all-to-all connectivity---and an efficient algorithm for computing the gradients---the multi-layer parallel parameter-shift rule---resolve all five. The challenges are: efficient data loading~(C1); efficient parametrized quantum circuits as quantum neural layers~(C2), where the optimal depth depends on hardware 
connectivity; provable avoidance of exponential barren plateaus~(C3); efficient gradient computation with sub-linear circuit evaluations per step~(C4); and genuine classical hardness of the deployed model output~(C5). The key insight unifying the solution is that particle-number preservation---the defining property of RBS gates---simultaneously enables efficient gradient computation, prevents exponential barren plateaus, and places the circuit in the Fermion Sampling framework for classical hardness with the correct data loading procedure.

A sixth challenge---whether the function classes accessible to quantum models are not merely hard for classical computers but genuinely useful on real tasks---cannot be resolved by theory alone. It requires running scalable quantum models on real data and real hardware, at sizes where the quantum expressivity advantage is meaningful. Closing C1--C5 is the prerequisite for even posing this question seriously; the structure of the framework does, however, narrow the field of candidate applications considerably, and we identify them in Section~\ref{sec:discussion}.

This paper builds on a long series of work on Hamming-weight-preserving quantum neural networks. The prior program established the architecture, tested it on hardware, and demonstrated its utility across applications in finance and healthcare---but left two gaps open: the barren plateau at large particle number $k$, and the prohibitive gradient computation cost at scale. 
We close both simultaneously by augmenting the RBS brick-wall and butterfly circuits with single-qubit $R_z$ gates (which lifts the DLA from $\mathfrak{so}(n)$ to $\mathfrak{u}(n)$ and resolves the barren plateau issue) and by proving the multi-layer parallel parameter-shift rule (which achieves a factor-$O(n/k)$ reduction in circuit evaluations per gradient step for each architecture: from $O(n^2)$ to $O(kn)$ for the brick-wall, and from $O(n\log n)$ to $O(k\log n)$ for the butterfly). The first contribution ensures that the lifted architecture has polynomial gradient variance throughout the regime $n - 2k = \Omega(n)$, including at extensive particle number $k = cn$ with $c < 1/2$, where the sampling-hardness guarantees of Section~\ref{sec:hardness} apply---so challenges C3 (trainability) and C5 (classical hardness) are met simultaneously, the former unconditionally for the brick-wall. (All trainability statements are for the regime $n - 2k = \Omega(n)$; at exactly half filling the one-body signal vanishes identically and only an $O(1/n^4)$ upper bound on the rate is proved.)

The second contribution makes training cheap in the practically relevant regime $k \ll n$. In addition, we organize the classical-hardness guarantees of the framework into an explicit ladder in the particle number $k$, calibrated against the best-known classical simulation algorithms~\cite{reardonsmith2024classical,oh2026classical}, and identify the concrete operating point at which quantum advantage is possible while training remains efficient.

It is worth stating at the outset how challenges C3 and C5 coexist, since they are met on \emph{different objects} of the same circuit. Trainability is established for the gradient channel: losses that depend on the measurement samples through two-body statistics, for which the gradient variance is $\Theta(k^3/n^5)$. Classical hardness is established for the deployed output: the sampling task, whose best-known classical cost is $2^{\Omega(k)}\poly(n)$. These are genuinely distinct objects---knowledge of the two-body moments does not permit sampling, since exponentially many distributions on Hamming-weight-$k$ strings share them---and the separation is what allows both conditions to hold at once for a single architecture, with the two-body correlator serving as the trainable anchor and the sampler as the deployed model. Section~\ref{sec:where} develops this into a taxonomy of pipelines, including which gradients may be computed classically and which require the quantum device inside the training loop.

Our main technical contribution is the \textbf{multi-layer parallel
parameter-shift rule} (Theorem~\ref{thm:parallel_psr_rbs}): all gradients
of the QNN can be computed simultaneously from a number of circuit
evaluations that is a factor $n/k$ smaller than the naive
parameter-shift rule, using random sign vectors and the commuting structure
of the gates within each layer. The estimator is unbiased (exact in
expectation), requires no assumption on the observable beyond particle-number preservation, and applies identically to both architectures. The naive PSR costs $4$ circuit evaluations per RBS
parameter and $2$ per $R_z$ parameter; the parallel PSR computes all
gradients in a single layer from $2k$ evaluations for the RBS parameters
and $2k$ for the $R_z$ parameters, independently of $n$. The resulting per-gradient-step costs and reductions are:
\begin{itemize}
\item \textbf{Brick-wall} ($T = 3n^2/2$ parameters, $n$ columns, each
a full-width $R_z$ layer followed by that column's $n/2$ RBS gates): naive
PSR requires $4n^2$ circuit evaluations; parallel PSR requires
$4kn$---a factor of $n/k$ reduction counted in distinct circuit configurations.
  \item \textbf{Butterfly} ($T = \tfrac{3}{2}n\log n$ parameters, $\log n$
  layers): naive PSR requires $4n\log n$ circuit evaluations; parallel PSR
  requires $4k\log n$---a factor of $n/k$ reduction.
\end{itemize}
For $n = 1024$ at the sampling operating point $k = 100$
(Section~\ref{sec:k_choice}), the butterfly requires $4{,}000$ circuit
evaluations per gradient step (vs.\ $40{,}960$ for the naive PSR on
the butterfly, a $10.2\times$ reduction), and the brick-wall requires
$409{,}600$ evaluations (vs.\ $4{,}194{,}304$ for the naive PSR
on the brick-wall, also $10.2\times$). Crucially, the reduction factor
$n/k$ grows linearly with $n$ at fixed hardness level: it reaches
$41.0\times$ at $n = 4096$, and the butterfly's absolute cost,
$4k \log n$, is nearly independent of $n$. Throughout, the naive baseline is
computed for the \emph{same} circuit, and we quote the absolute per-step cost
alongside the ratio (Corollary~\ref{cor:full_gradient}). As $n$ grows the
training cost
scales only as $O(k\log n)$
for the butterfly and $O(kn)$ for the brick-wall---polylogarithmic and
linear respectively---while the Hilbert space dimension grows as
$\binom{n}{k}$. The extra factor of $n$ in the
brick-wall reflects the price of nearest-neighbor connectivity: $O(n)$
depth instead of $O(\log n)$, but with no long-range gates required. In
both cases the training cost scales far below the expressivity of the model.

\paragraph{Paper structure.}
Section~\ref{sec:encoding} presents the fermionic magic-state data encoding. Section~\ref{sec:butterfly} introduces the unitary brick-wall and unitary 
butterfly architectures and their mathematical properties, including the 
depth-connectivity tradeoff for each. Section~\ref{sec:bp} establishes the trainability and absence-of-barren-plateaus result. Section~\ref{sec:psr} proves the multi-layer parallel parameter-shift rule. Section~\ref{sec:pipeline} presents the complete pipeline: a single sample-based readout interface shared by all task families, from distributional losses to classical heads---linear, multilayer, or transformer---consuming the sampled Fock strings, together with an account of which gradients admit a variance guarantee and which require the quantum device in the training loop. Section~\ref{sec:hardness} proves classical hardness of the learned function class for sampling-based readouts. Section~\ref{sec:discussion} discusses implications and open problems.

\section{\label{sec:encoding} Data Loading: Magic-State Encoding}

The data loading happens in three steps: 
\begin{enumerate}
\item[A.] Magic State Preparation: we prepare a non-Gaussian magic state on a set of active blocks;
\item[B.] Spreading Operation: we spread the particles across all $n$ modes using a fixed Hadamard
spreading circuit; 
\item[C.] Data Encoding: we encode classical data through a layer of single-qubit phase gates. 
\end{enumerate}

The encoding pipeline is the same for both architectures; the particle number
$k$ is calibrated against classical simulation cost in
Section~\ref{sec:k_choice}, while the only connectivity-dependent element is
the spreading circuit in step~B, which has depth $O(\log n)$ on all-to-all
hardware (Hadamard butterfly) and $O(n)$ on nearest-neighbor hardware (Hadamard
brick-wall).

We now detail each step.

\subsection{Magic State Preparation}

The magic state is the minimal non-Gaussian construction adequate for the
Fermion Sampling hardness structure of
Section~\ref{sec:hardness}.\label{sec:psi4_encoding}
Divide $n$ qubits into $k/2$ active blocks of $4$ qubits and $n - 2k$ idle
qubits initialized to $\ket{0}$. Prepare:
\begin{equation}
  \ket{\Psi_{\mathrm{in}}^{(4)}}
  = \ket{\Psi_4}^{\otimes k/2}
  \otimes \ket{0}^{\otimes (n - 2k)},
  \qquad
  \ket{\Psi_4}
  = \frac{\ket{1100} + \ket{0011}}{\sqrt{2}}.
\end{equation}
Each block is prepared independently using a circuit of depth $4$ and $3$
entangling gates~\cite{oszmaniec2022fermion}; all $k/2$ blocks run in parallel,
so the total depth is $4$. The state has exactly $k$ particles and is
non-Gaussian: a superposition of two distinct Fock states differing in all 
modes, which cannot be written as a single Slater
determinant~\cite{oszmaniec2022fermion}. Each block is a 2-form (sum of two
pair-creation operators), so the overall input lies in the form-rank-2
manifold. Two complexity-theoretic consequences follow, developed in
Section~\ref{sec:hardness}: the mixed-discriminant structure of its output
amplitudes underlies the Fermion Sampling hardness argument at sufficiently
large particle number~\cite{oszmaniec2022fermion}, while its mixed-Pfaffian
structure makes few-body expectation values classically estimable to additive
error in polynomial time at \emph{any} $k$~\cite{oh2026classical}.

\begin{remark}[Sparse inputs: hardness scales with the number of magic blocks]
\label{rem:sparse_input}
The original Fermion Sampling hardness argument of Oszmaniec et
al.~\cite{oszmaniec2022fermion} considers fully-packed magic inputs
$\ket{\Psi_4}^{\otimes n/4}$ with $k = n/2$ particles across all $n$ modes.
Our construction uses only $B = k/2$ magic blocks and $n - 2k$ idle modes
initialized to vacuum. The mixed-discriminant \emph{structure} carries through
directly to the sparse case: the idle modes simply do not appear in the rows
and columns of the single-particle unitary indexed by the support of the magic
state, and, exactly as in the fully-packed case
\cite[Lemma~11]{oszmaniec2022fermion}, the output amplitude is a sum of $2^B$
determinants of $k \times k$ submatrices of the single-particle unitary $W$,
\begin{equation}
  \braket{x' | V | \Psi_{\mathrm{in}}^{(4)}}
  = 2^{-B/2} \sum_{S \in \calS} \det\!\big(W_{x',S}\big),
  \qquad |\calS| = 2^{B}.
  \label{eq:det_sum}
\end{equation}
What does \emph{not} carry through automatically is hardness, because
Eq.~\eqref{eq:det_sum} is simultaneously a classical algorithm of cost
$O(2^B \cdot k^3)$ per amplitude. The hard-instance size of the
Ivanov--Gurvits reduction to the permanent~\cite{ivanov2013permanents,gurvits2005complexity} is the block count $B$, so
worst-case \#P-hardness requires $B = n^{\Omega(1)}$. Moreover, the FLO extent
of the input is $\chi = 2^B$, and all best-known simulation methods---exact
amplitudes, few-body marginals, and, via the chain rule over marginals, exact
sampling---run in time $2^{\Theta(B)}\,\poly(n)$
\cite{reardonsmith2024classical}. In particular, for $B = O(\log n)$ the
entire pipeline, \emph{including sampling}, is classically simulable in
polynomial time. Whether hardness survives with fewer magic blocks than the
fully-packed input was posed as an open problem
in~\cite[Sec.~IV]{oszmaniec2022fermion}; the extent-based algorithms
of~\cite{reardonsmith2024classical} resolve its hardness half negatively for
logarithmically many blocks, while the regime of block number between
$\omega(\log n)$ and $o(n)$ remains open at the level of average-case
hardness. Section~\ref{sec:hardness} organizes our hardness claims along this
ladder.
\end{remark}

\subsection{Spreading Operation}
\label{sec:common_stages}

The remaining encoding stages are the Hadamard spreading and the
data-dependent phase layer. Throughout, we write
$\ket{\Psi_{\mathrm{in}}}$ for the magic state
$\ket{\Psi_{\mathrm{in}}^{(4)}}$.

For the spreading operation, we
apply $U^{\mathrm{had}}$---a passive FLO circuit with all FBS angles fixed at
$\pi/4$, implementing the Walsh--Hadamard transform in the single-particle
sector. On all-to-all connectivity this is the Hadamard butterfly of
$\log n$ layers (defined formally in Section~\ref{sec:butterfly}, with
Remark~\ref{rem:fbs_implementation} for its qubit-circuit depth); on nearest-neighbor
connectivity it is the Hadamard brick-wall of depth $n$ (defined formally in
Section~\ref{sec:butterfly}). In both cases the circuit is non-trainable and
fixed once and for all. This delocalizes the $k$ particles across all $n$ modes:
\begin{equation}
  \ket{\phi_0} = U^{\mathrm{had}} \ket{\Psi_{\mathrm{in}}}
  = \sum_{S \subseteq [n],\, |S| = k} \alpha_S \ket{S},
\end{equation}
where every qubit $j$ has single-particle occupation
$\rho_1^{\mathrm{out}}[j,j] = k/n$ (Lemma~\ref{lem:spreading}), and particle
number is preserved exactly. Non-Gaussianity is preserved: $U^{\mathrm{had}}$
is a passive FLO transformation, and passive FLO maps non-Gaussian states to
non-Gaussian states~\cite{hebenstreit2019all}. Crucially, passive FLO also
preserves form-rank, so the form-rank-2 structure of the
input is preserved by spreading.

\begin{lemma}[Hadamard spreading equalizes single-particle occupations]
\label{lem:spreading}
Let $\ket{\Psi_{\mathrm{in}}} = \ket{\Psi_{\mathrm{in}}^{(4)}}$ be the magic
state of Section~\ref{sec:psi4_encoding}, with $k$ particles distributed
across $2k$ active modes (the remaining $n-2k$ modes initialized to $\ket{0}$).
After applying the Hadamard spreading circuit $U^{\mathrm{had}}$ on all $n$
modes, the single-particle reduced density matrix of
$\ket{\phi_0} = U^{\mathrm{had}}\ket{\Psi_{\mathrm{in}}}$ satisfies
\begin{equation}
\rho_1^{\mathrm{out}} \;=\; \tfrac12\,\Pi,
\qquad
\rho_1^{\mathrm{out}}[j,j] \;=\; \frac{k}{n} \;>\; 0
\quad \text{for every } j ,
\end{equation}
where $\Pi$ is an orthogonal projector of rank $2k$. In particular, no qubit is
a vacuum mode after spreading: every mode carries single-particle weight
$k/n$, so the data-encoding phase layer imprints a nontrivial phase on every
feature. The same holds for $\rho_1(\bx)$ after the diagonal data layer, with
$\Pi$ replaced by a unitary conjugate of itself.
\end{lemma}

\begin{proof}
The $n-2k$ idle qubits, initialised to $\ket{0}$, carry no particles before
the spreading circuit. Let $c_j^\dagger$ and $c_j$ denote the fermionic
creation and annihilation operators for mode $j$ (under the Jordan--Wigner
correspondence, $c_j^\dagger$ maps the vacuum on qubit $j$ to an occupied
mode), so that $n_j = c_j^\dagger c_j$ is the occupation-number operator. For
any $k$-particle passive FLO unitary with single-particle matrix $W$, the
single-particle reduced density matrix transforms as
$\rho_1^{\mathrm{out}} = W\,\rho_1^{\mathrm{in}}\,W^\dagger$, where
$\rho_1^{\mathrm{in}}[j,j'] = \bra{\Psi_{\mathrm{in}}} c_j^\dagger
c_{j'} \ket{\Psi_{\mathrm{in}}}$ is the single-particle reduced density matrix
of the input.

\smallskip
\textbf{Diagonal entries.} Within each
block the two Fock-state contributions partition the block's modes into two
disjoint groups of equal amplitude: block $g$ occupies modes
$\{4g, 4g{+}1, 4g{+}2, 4g{+}3\}$ and superposes $\{4g, 4g{+}1\}$ with
$\{4g{+}2, 4g{+}3\}$ at equal amplitude $1/\sqrt{2}$.
Each active mode $j \in \{0,\dots,2k-1\}$ belongs to exactly one
block and appears in exactly one of that block's two Fock-state choices, so it
is occupied in exactly half of the input Fock states, giving
$\rho_1^{\mathrm{in}}[j,j] = \tfrac{1}{2}$ for $j \in \{0,\dots,2k-1\}$ and $0$
otherwise, with $\mathrm{tr}(\rho_1^{\mathrm{in}}) = (2k)\cdot\tfrac{1}{2} = k$.

\smallskip
\textbf{Off-diagonal entries vanish.} For $j \neq j'$, $\rho_1^{\mathrm{in}}[j,j']
= \bra{\Psi_{\mathrm{in}}} c_j^\dagger c_{j'} \ket{\Psi_{\mathrm{in}}}$. If $j$
and $j'$ lie in different blocks, $c_j^\dagger c_{j'}$ changes the particle
number within each of the two blocks; since $\ket{\Psi_{\mathrm{in}}}$ is a
tensor product of blocks each carrying a definite particle number ($2$ per
block), the result is orthogonal to
$\ket{\Psi_{\mathrm{in}}}$ and the overlap is zero. If $j$ and $j'$ lie in the
same block, $c_j^\dagger c_{j'}$ moves a single particle within that block;
acting on either Fock component (e.g.\ $\ket{1100}$ or $\ket{0011}$ for a paired
block) it produces a basis state of the same Hamming weight orthogonal to
both components---converting one component into the other would require moving
\emph{two} particles---so the overlap is again zero. Hence
$\rho_1^{\mathrm{in}} = \tfrac12 \Pi_A$ with $\Pi_A$ the projector onto the
$2k$ active modes, and $\rho_1^{\mathrm{out}} = \tfrac12 W\Pi_A W^\dagger$ is
$\tfrac12$ times a rank-$2k$ projector.

\smallskip
\textbf{Walsh--Hadamard transform.} Both the Hadamard butterfly and the
Hadamard brick-wall implement the normalised Walsh--Hadamard transform on $n$
modes in the single-particle sector,
$H[i,j] = (-1)^{\langle i, j \rangle}/\sqrt{n}$, where
$\langle i, j \rangle = \bigoplus_b i_b j_b$ is the bitwise inner product of
the binary expansions of $i$ and $j$, so $|H[i,j]|^2 = 1/n$ for every pair
$i,j$.

\smallskip
\textbf{Output occupations.} Because $\rho_1^{\mathrm{in}}$ is diagonal,
$\rho_1^{\mathrm{out}}[i,i] = [H\rho_1^{\mathrm{in}}H^\dagger]_{ii} =
\sum_{j} |H[i,j]|^2\, \rho_1^{\mathrm{in}}[j,j]$. Using $|H[i,j]|^2 = 1/n$ for
every pair $(i,j)$,
\begin{equation}
\rho_1^{\mathrm{out}}[i,i]
= \frac{1}{n}\sum_{j=0}^{2k-1}\rho_1^{\mathrm{in}}[j,j]
= \frac{1}{n}\,\mathrm{tr}(\rho_1^{\mathrm{in}})
= \frac{k}{n}
\end{equation}
for every output mode $i$, as claimed. The data layer $D(\bvarphi(\bx))$ is a
diagonal unitary, so $\rho_1(\bx) = D\rho_1^{\mathrm{out}}D^\dagger$ is again
$\tfrac12$ times a rank-$2k$ projector with the same diagonal.
\end{proof}

\subsection{Data Encoding Layer}
Apply $D(\bvarphi) = \bigotimes_{i=1}^{n} R_z(\varphi_i)$ where $\varphi_i = f(x_i)$,
e.g.\ $\varphi_i = \pi x_i / \max_i |x_i|$. Since $R_z(\varphi)\ket{1} =
e^{-i\varphi/2} \ket{1}$ and $R_z(\varphi)\ket{0} = e^{+i\varphi/2} \ket{0}$:
\begin{equation}
  \ket{\psi(\bx)} = D(\bvarphi(\bx)) \ket{\phi_0}
  = \sum_{|S| = k} \alpha_S\, e^{-i \sum_{j \in S} \varphi_j} \ket{S}.
  \label{eq:encoded_state}
\end{equation}
All $n$ features enter through independent phases. $R_z$ is diagonal in the
computational basis---it preserves Hamming weight, non-Gaussianity, and
form-rank. Note that this differs from standard angle encoding, where $R_z$
gates are applied directly to $\ket{0}^{\otimes n}$ to produce a product state;
here the $R_z$ layer imprints data as phases on the pre-spread entangled state
$\ket{\phi_0}$, so the data interacts with all $\binom{n}{k}$ superposition
components simultaneously rather than independently qubit by qubit. Unlike
amplitude encoding, this scheme does not preserve $\ell_2$ inner products
between data points: the encoded state defines an implicit kernel that is
translation-invariant in feature-phase differences. This is by design---the
goal of this encoding is expressivity and classical hardness, not metric
preservation.

The total encoding depth is $\log_2 n + 5$
on all-to-all connectivity (depth-4 magic
preparation $+$ $\log n$ Hadamard butterfly $+$ 1 data layer), and $n + 5$
on nearest-neighbor connectivity (depth-4 magic preparation $+$ $n$ Hadamard
brick-wall columns $+$ 1 data layer). All magic
blocks are prepared in parallel in Step~A, so the magic preparation depth does not scale with $k$.

\subsection{The Right Choice of $k$}
\label{sec:k_choice}

The particle number $k$ is a free parameter of the construction. A single quantity
controls the classical simulation cost of the full pipeline: the number of
non-Gaussian blocks, $B = k/2$, through the
FLO extent $\chi = 2^B$ (Section~\ref{sec:hardness}). Exact amplitudes cost
$O(2^B\,\poly(n))$ classically, and exact sampling $O(B^4 + 2^B)$ under
best-known algorithms~\cite{heo2026fast}. Few-body expectation
values are additively tractable at every $k$ for this
encoding~\cite{oh2026classical}, so $k$ calibrates the hardness of
\emph{sampling} specifically. The training cost of
the multi-layer parallel parameter-shift rule of Section~\ref{sec:psr} scales
in the opposite direction: $4kn$ circuit evaluations per gradient step for
the brick-wall and $4k\log n$ for the butterfly, a factor $n/k$ below the
naive parameter-shift rule for both. Larger $k$ therefore buys stronger classical
hardness at linearly growing training cost. These tradeoffs apply to both the
unitary brick-wall and the unitary butterfly, which share the same encoding
pipeline and differ only in the trainable and spreading circuits that follow.
We first state the asymptotic regimes that organize the hardness results of
Section~\ref{sec:hardness}, and then fix the concrete operating values used
throughout the paper.

\paragraph{$k = O(\log n)$: classically simulable.}
Here $\chi = 2^B = \poly(n)$, and the full pipeline---amplitudes, few-body
expectation values, and sampling---is classically simulable in polynomial
time (Remark~\ref{rem:sparse_input}). No quantum advantage is possible in
this regime. It is nonetheless operationally valuable: the same architecture,
run at small $k$, can be simulated exactly, which permits end-to-end
classical validation of the training pipeline and of hardware behavior before
scaling $k$ up. Training is cheapest here ($O(kn)$ evaluations per step for
the brick-wall, $O(k\log n)$ for the butterfly).

\paragraph{$k = \omega(\log n)$: superpolynomial hardness under best-known
algorithms.} Once the block count exceeds logarithmic, $\chi$ is
superpolynomial and best-known classical algorithms cost
$2^{\Omega(k)}\,\poly(n)$. At the representative point $k = \log^2 n$ this is
$n^{\Omega(\log n)}$, quasi-polynomial hardness for sampling
(the additive-error estimator
of~\cite{oh2026classical} does not enable sampling, which requires relative
precision on exponentially small marginals). No polynomial-hierarchy-based
statement is available in this regime---the worst case itself is computable
in quasi-polynomial time---but training remains genuinely cheap.

\paragraph{$k = n^{\epsilon}$: worst-case \#P-hardness.} For any constant
$\epsilon > 0$, the Ivanov--Gurvits reduction embeds permanents of
$n^{\Omega(1)}$-size matrices, and there exist parameter settings of the
brick-wall (which is surjective onto $U(n)$, Section~\ref{sec:butterfly})
whose output probabilities are \#P-hard to compute exactly
(Section~\ref{sec:hardness}; for the butterfly, whose fixed-depth
parametrization covers only a submanifold of $U(n)$, see the reachability
remark there). Average-case hardness and anticoncentration for such sparse
inputs are open---this is precisely the ``fermion sampling with less magic''
problem of~\cite{oszmaniec2022fermion}---and we state them as a conjecture.
Best-known simulation cost is $2^{\Omega(n^{\epsilon})}$, and training costs
$O(kn) = O(n^{1+\epsilon})$ (brick-wall) or $O(n^{\epsilon}\log n)$
(butterfly) evaluations per step.

\paragraph{$k = \Theta(n)$: the extensive limit.} At $k = cn$ with a constant
filling fraction $c$ the Hilbert space dimension $\binom{n}{k}$ is
exponential in $n$, and at $c = 1/2$ the construction matches the
original Fermion Sampling setup of~\cite{oszmaniec2022fermion}; the full
machinery---anticoncentration, worst-to-average-case reduction, and
conditional hardness of approximate sampling---applies to the
Haar-initialized brick-wall (Section~\ref{sec:hardness}). Trainability is at
its strongest here as well: for $c < 1/2$ the gradient variance is
$\Theta(1/n^2)$, the largest order it attains anywhere in the range
(Section~\ref{sec:bp}). What is lost is only
\emph{cheap} training: at $k = \Theta(n)$ both the naive and the parallel
parameter-shift rule cost $\Theta(n^2)$ circuit evaluations per gradient
step, the standard cost for a circuit with $\Theta(n^2)$ parameters. This regime is
thus the strongest-hardness theoretical limit; the operating point below
trades a calibrated amount of that hardness for training cost that falls
with $n$.

\medskip
\paragraph{Operating point: two rungs.} Asymptotic labels hide constants, so
for concrete deployments we calibrate $k$ directly against classical compute.
Best-known classical simulation of sampling from the pipeline costs
$O(B^4 + 2^B)$ operations per sample~\cite{heo2026fast}, and the
sampler is \emph{input-dependent}: the data enter through the diagonal layer
$D(\bvarphi(\bx))$, so each input induces a different output distribution and
a classical competitor pays this cost again for every fresh input, with
nothing to amortize. A machine of $R$ operations per second therefore sustains
$R/2^B$ classical queries per second, and two rungs follow.

\emph{Serving the model.} At
\begin{equation}
  B = 50 \text{ blocks:}\quad k = 100,\ \ket{\Psi_4}^{\otimes 50},
  \label{eq:operating_point}
\end{equation}
one classical query costs $2^{50} \approx 1.1 \times 10^{15}$ operations:
$19$ minutes on a $10^{12}$-operation/s server and still $1.1$~s on a
dedicated petascale cluster, against $10^{-4}$~s on the device. No commercial
machine serves this at query rates. The rung requires $n \ge 216$ modes, and
we use it for all concrete numbers.

\emph{Reproducing the model.} Distillation needs a corpus rather than a single
query. At $B = 50$ a $10^{18}$-flop/s machine still sustains $888$ queries per
second, so this rung is not met at the deployment point; at $B = 60$,
$k = 120$ the same machine sustains $0.87$ queries per second against $10^{4}$
on the device, so a corpus of any given size costs that many seconds of
dedicated exascale time. This rung requires $n \ge 258$ modes, hence $n = 512$
for the butterfly, whose layer count forces $n$ to be a power of two.

On the quantum side a sample is one circuit execution, $10^{-4}$~s at a
$10$~kHz repetition rate, independent of $B$; the cost appears only in
training, where the parallel parameter-shift rule of Section~\ref{sec:psr}
needs $4kn$ evaluations per step---$10.2$~s at $(n,k) = (256,100)$ and
$24.6$~s at $(512,120)$, so $10^{5}$ steps take twelve and twenty-eight days.
The reduction factor $n/k$ grows linearly with $n$ at fixed hardness:
$2.6\times$ at $n = 256$, $5.1\times$ at $n = 512$, $10.2\times$ at
$n = 1024$ and $41.0\times$ at $n = 4096$. The calibration is pinned to \emph{best-known}
algorithms~\cite{heo2026fast,reardonsmith2024classical,oh2026classical}; any
future
improvement is absorbed by increasing $B$, at training cost linear in $k$.

\medskip
Throughout the remainder of the paper we carry $k$ as a free
parameter in all statements, instantiate asymptotic claims on the ladder
above, and use $B = 50$ (Eq.~\eqref{eq:operating_point}) for all concrete
numbers. Complete parameter values at representative problem sizes are
collected in Section~\ref{sec:concrete_numbers}.

\begin{remark}[Alternative data loading via Dicke states]
\label{rem:dicke}
A natural alternative non-Gaussian input is the Dicke state
$\ket{D_k^n}$---the uniform superposition of all $\binom{n}{k}$
Hamming-weight-$k$ Fock states---which has form-rank exactly $k$. Its 1-RDM
is $(k/n)I_n$ by permutation symmetry, so the uniform-occupation property of
Lemma~\ref{lem:spreading} holds \emph{without} a separate spreading step. It
does not, however, supply classical hardness: $\ket{D_k^n}$ is the
degree-$k$ component of a single Gaussian state combined with a
roots-of-unity filter, so its FLO extent is $O(n)$ and the pipeline built on
it is classically simulable at every $k$. The block-product construction is
the source of hardness, and is in any case cheaper in preparation depth:
$O(k\log(n/k))$ on all-to-all connectivity~\cite{bartschi2022short,yuan2025efficient}
(polylogarithmic, but deeper than the constant-per-block magic-state preparation)
and $O(n)$ on nearest-neighbor connectivity. The block-product construction is
therefore preferable when depth is the binding constraint. We note that the
trainability analysis of Section~\ref{sec:bp} does not transfer to this
input: it is driven by $\tau = \|[g,\rho_1]\|_F^2$, which vanishes
identically when $\rho_1 \propto I$, so both the disconnected and the
contracted contributions of Proposition~\ref{prop:norm_eval} are zero and
the gradient variance is determined entirely by the connected two-body
correlations of the Dicke state, which we do not analyze here.
\end{remark}

\section{\label{sec:butterfly} The Unitary Brick-Wall and Butterfly Architectures}

\subsection{Reconfigurable Beam Splitter Gates}

The basic gate primitive is the Reconfigurable Beam Splitter (RBS) gate,
introduced in the context of quantum machine learning
in~\cite{kerenidis2022quantum,landman2022quantum} as the natural building
block for implementing orthogonal transformations on quantum hardware.
Acting on qubits $i$ and $j$, it is defined as:
\begin{equation*}
  \mathrm{RBS}(\theta) = \exp(-i\theta G_{ij}), \quad
  G_{ij} = i(\ket{01}\bra{10} - \ket{10}\bra{01})_{ij}
\end{equation*}
with matrix representation in the $\{\ket{00},\ket{01},\ket{10},\ket{11}\}$
basis:
\begin{equation}
  \mathrm{RBS}(\theta) = \begin{pmatrix} 1 & 0 & 0 & 0 \\ 0 & \cos\theta &
  \sin\theta & 0 \\ 0 & -\sin\theta & \cos\theta & 0 \\ 0 & 0 & 0 & 1
  \end{pmatrix}.
\end{equation}
The generator $G_{ij}$ has eigenvalues $\{-1, 0, 0, +1\}$ and acts only
on the $\{\ket{01},\ket{10}\}$ subspace. Crucially, RBS gates are
\emph{particle-number preserving}: they map $\ket{01} \leftrightarrow \ket{10}$
and fix $\ket{00}$ and $\ket{11}$.

Under the Jordan--Wigner transformation a two-mode beam splitter on modes
$i < j$ carries the parity of the modes lying strictly between them. The
\emph{fermionic beam splitter} $\mathrm{FBS}_{ij}(\theta)$
of~\cite{kerenidis2022subspace} implements it: on a computational basis state
$\ket{s}$, $s \in \{0,1\}^n$, it acts on the pair $(i,j)$ as
\begin{align}
  \mathrm{FBS}_{ij}(\theta)\big|_{\ket{s}}
  &= \mathrm{RBS}_{ij}\bigl((-1)^{f(i,j,s)}\theta\bigr),
  \label{eq:fbs}\\
  f(i,j,s) &= \sum_{i<m<j} s_m . \nonumber
\end{align}
Both architectures below are built from $\mathrm{FBS}$ gates.

\begin{lemma}[FBS is the passive FLO Givens rotation]
\label{lem:fbs}
Let $w_{ij}(\theta) \in U(n)$ denote the Givens rotation acting on modes
$i,j$. Then $\mathrm{FBS}_{ij}(\theta)$ is the Jordan--Wigner image of
$w_{ij}(\theta)$: it acts on the $k$-particle sector as
$\Lambda^k\bigl(w_{ij}(\theta)\bigr)$, for every $k$. Moreover
$\mathrm{FBS}_{ij} = \mathrm{RBS}_{ij}$ when $j = i+1$, and
$\mathrm{FBS}_{n1}(\theta) = \mathrm{RBS}_{n1}\bigl((-1)^{k-1}\theta\bigr)$ on
the $k$-particle sector for the periodic boundary pair $\{n,1\}$. In
particular $\mathrm{FBS}_{ij}(\theta)$ and $\mathrm{RBS}_{ij}(\theta)$ have
the same generator spectrum $\{-1,0,+1\}$ on every $k$-particle sector, and
$\mathrm{FBS}$ gates on disjoint mode pairs commute exactly.
\end{lemma}

\begin{proof}
The first claim is the Jordan--Wigner transform of a passive two-mode
rotation. For $j = i+1$ the index set $\{m : i<m<j\}$ is empty. For
$\{n,1\}$ the gate acts non-trivially only where $s_1 + s_n = 1$, and there
$f(n,1,s) = k - s_1 - s_n = k-1$, a constant on the sector. The generator of
$\mathrm{FBS}_{ij}$ differs from $G_{ij}$ by signs on off-diagonal entries
within each invariant two-dimensional block, which leaves the spectrum
unchanged; commutation on disjoint pairs is commutation of fermionic
bilinears with disjoint mode supports.
\end{proof}

The brick-wall of Section~\ref{sec:brick-wall} therefore requires no parity
strings, the periodic boundary included; only the long-stride pairs of the
butterfly carry a non-trivial string.

This gate is a member of several well-studied families that have appeared
across quantum information and condensed matter physics under different names.
It is a \emph{matchgate} in the sense of Valiant~\cite{valiant2001quantum} and
Terhal--DiVincenzo~\cite{terhal2002classical}, who showed that nearest-neighbor
circuits of such gates are efficiently classically simulable via their connection
to noninteracting fermions through the Jordan--Wigner
transformation~\cite{knill2001fermionic,jozsa2008matchgates}. In the language
of fermionic linear optics (FLO), the fermionic beam splitter
$\mathrm{FBS}_{ij}$ of Eq.~\eqref{eq:fbs} is a passive FLO operation: it
implements a beam splitter transformation on two fermionic modes, mixing them
without creating or destroying particles, and for adjacent modes it coincides
with the RBS gate. This connection places both architectures in
the Fermion Sampling framework~\cite{oszmaniec2022fermion}.

Augmenting each RBS gate with a single-qubit $R_z$ phase gate on \emph{each}
qubit of the pair gives the three-parameter gate primitive that we use
throughout both architectures,
\begin{equation}
  \calG(\theta, \phi_a, \phi_b) \;=\; \mathrm{RBS}(\theta) \cdot
  \bigl(R_z(\phi_a) \otimes R_z(\phi_b)\bigr).
  \label{eq:gate_primitive}
\end{equation}
This primitive contains the passive FLO two-qubit gate
$D_{\mathrm{pas}}(\alpha_1, \alpha_2)$ of Oszmaniec et
al.~\cite{oszmaniec2022fermion} (their Eq.~(20)):
\begin{equation}
  D_{\mathrm{pas}}(\alpha_1, \alpha_2)
  = \bigl(e^{-i\alpha_1 Z_1/2}\, e^{i\alpha_1 Z_2/2}\bigr)\,
    e^{i\alpha_2(X_1 X_2 + Y_1 Y_2)/2},
\end{equation}
where the second factor is the RBS gate with $\theta = \alpha_2$, and the first
factor is a \emph{correlated} $R_z$ pair
$(R_z(-\alpha_1) \otimes R_z(\alpha_1))$ up to global phase. Our primitive
$\calG$ carries two independent phases rather than a correlated pair, so it
strictly contains $D_{\mathrm{pas}}$: setting $\phi_a = -\phi_b = -\alpha_1$
and $\theta = \alpha_2$ recovers it exactly, up to global phase. The extra
phase direction is a single-qubit $Z$ rotation already present in the Cartan
subalgebra, so both parametrizations generate the same dynamical Lie algebra;
its role here is not algebraic but algorithmic---a phase layer covering all
$n$ modes is what the $R_z$ half of the parallel parameter-shift rule requires
(Section~\ref{sec:psr}). By Lemma~\ref{lem:fbs} the unitary brick-wall and
butterfly circuits introduced below are therefore direct realizations of the
passive FLO circuit families of~\cite{oszmaniec2022fermion} with the DLA lifted from
$\mathfrak{so}(n)$ to $\mathfrak{u}(n)$ by the addition of $R_z$ gates.

\subsection{The Unitary Brick-Wall Circuit}
\label{sec:brick-wall}

The unitary brick-wall is the natural architecture for nearest-neighbor hardware.
It is surjective onto $U(n)$ in the single-particle sector, containing a
rectangular mesh of Givens rotations as a submanifold. The brick-wall consists
of $n$ columns of nearest-neighbor RBS gates on a ring of $n$ modes (a chain
with periodic boundary), in an alternating pattern---odd columns act on the
$n/2$ pairs $(1,2),(3,4),\dots,(n{-}1,n)$, even columns on the $n/2$ pairs
$(2,3),(4,5),\dots,(n{-}2,n{-}1),(n,1)$, for $n^2/2$ RBS gates in
total---each column preceded by a full-width layer of $n$ single-qubit $R_z$
gates. Every column therefore covers all $n$ modes with disjoint pairs, a
property the parallel parameter-shift rule of Section~\ref{sec:psr} exploits;
the open-boundary (chain) variant is discussed in
Remark~\ref{rem:chain_variant}. Formally, the $\ell$-th column applies a phase layer on all $n$ modes
followed by the RBS gates of that column:
\begin{align}
  U^{(\ell)} &= \Bigl[\bigotimes_{j \in P_\ell}
  \mathrm{RBS}(\theta_\ell^{(j)})\Bigr] \cdot
  \bigotimes_{i=1}^{n} R_z(\phi_\ell^{(i)}),
\end{align}
where $P_\ell$ is the set of $n/2$ active nearest-neighbor pairs of
column $\ell$; equivalently, each active pair carries the primitive $\calG$
of Eq.~\eqref{eq:gate_primitive}. The full trainable circuit is the ordered product of
the $n$ columns,
\begin{equation}
  U(\btheta, \bphi) = U^{(n)} \cdots U^{(1)}.
\end{equation}
Within each column all active pairs act on disjoint qubits, so generators
from different pairs commute---the structure exploited by the multi-layer
parallel parameter-shift rule.

We call this circuit the \emph{unitary brick-wall}, to distinguish it from the
\emph{orthogonal brick-wall} of prior work~\cite{kerenidis2022quantum}, which
implements $SO(n)$ in the single-particle sector. The $R_z$ gates are precisely
what promote the DLA from $\mathfrak{so}(n)$ to $\mathfrak{u}(n)$.

\begin{proposition}[Properties of the unitary brick-wall circuit]
\label{prop:brick-wall}
The unitary brick-wall $U(\btheta, \bphi)$ satisfies:
\begin{enumerate}
 \item \textbf{Depth:} $2n$
 \item \textbf{Parameters:} $T = n^2/2 + n^2 = 3n^2/2$
  \item \textbf{Universality:} $\{U(\btheta,\bphi)\}$ is \emph{surjective
onto} $U(n)$ in the single-particle sector, and contains the
$n^2$-dimensional Givens mesh as the submanifold obtained by zeroing the
boundary-pair angles and the redundant phases
  \item \textbf{DLA:} The dynamical Lie algebra is $\mathfrak{u}(n)$, with
        $\dim(\mathfrak{u}(n)) = n^2$.
\end{enumerate}
\end{proposition}

\begin{proof}
\textit{Items 1 and 2.}
The brick-wall consists of $n$ columns of nearest-neighbor RBS gates on
the ring, with $n/2$ disjoint active pairs in every column, giving
$n \cdot \tfrac{n}{2} = n^2/2$ RBS gates in total.
Each column is preceded by a full-width layer of $n$ single-qubit $R_z$ gates,
contributing $n$ parameters per column and $n^2$ in total; each column
therefore has depth~$2$ (one parallel $R_z$ layer followed by one parallel RBS
layer), giving total depth $2n$. Total parameter count:
$T = n^2/2 + n^2 = 3n^2/2$.
\textit{Item 3 (surjectivity).}
Restrict the parameters by setting the boundary-pair angles of the even
columns to zero---reducing each even column to its $n/2-1$ chain pairs---and,
in every column $\ell \ge 2$, the phase on the second qubit of each remaining
active pair and the phases on the modes not covered by those pairs to zero. What remains is exactly one phase per Givens
rotation, $n(n-1)/2$ of them, together with the $n$ free phases of the first
column. In the single-particle sector each $\mathrm{RBS}(\theta_\ell^{(j)})$
preceded by its single phase acts as a two-parameter complex Givens rotation
on the corresponding pair of modes, and the rectangular mesh decomposition of
$U(n)$ states that any $W \in U(n)$ factors as a product of $n(n-1)/2$
nearest-neighbor Givens rotations in this pattern preceded by a diagonal phase
matrix. The restricted map from
$\mathbb{R}^{n(n-1)+n} = \mathbb{R}^{n^2}$ to $U(n)$ is therefore surjective,
with the inverse given by the constructive Givens (QR-type) decomposition of
any $W \in U(n)$: the Givens rotations zero out the off-diagonal entries one
at a time, and the residual diagonal unitary is absorbed by the first
column's phase layer. Since $n(n-1) + n = n^2 = \dim_{\mathbb{R}} U(n)$, this
submanifold is a direct parametrization; the full circuit, with $3n^2/2$
parameters, is surjective but over-complete
(Remark~\ref{rem:overcomplete}).

\textit{Item 4 (DLA).}
The trainable RBS gates of the brick-wall act on nearest-neighbor
pairs of the ring, so their generators
$\{G_{i,i+1}\}_{i=1}^{n-1} \cup \{G_{n,1}\}$ do not by
themselves span $\mathfrak{so}(n)$; rather, since the chain
$1 - 2 - \cdots - n$ is connected, iterated commutation
$[G_{i,i+1}, G_{i+1,i+2}] \propto G_{i,i+2}$ reaches $G_{ik}$ for every pair
$(i,k)$, so the RBS generators generate $\mathfrak{so}(n)$ by Lie closure
(exactly as in Proposition~\ref{prop:butterfly}). The $R_z$ generators span the
full Cartan subalgebra $i\mathfrak{h} \subset \mathfrak{u}(n)$ of imaginary
diagonal matrices. The Lie closure of $\mathfrak{so}(n)$ and $i\mathfrak{h}$ is
$\mathfrak{u}(n)$ by the same commutator argument as in
Proposition~\ref{prop:butterfly} (commutators of the $G_{ij}$ with the diagonal
$R_z$ generators supply the symmetric off-diagonal generators absent from
$\mathfrak{so}(n)$)~\cite{fontana2024adjoint,larocca2021diagnosing}. This
Lie-closure statement for the \emph{algebra} is logically separate from the
mesh parametrization of the \emph{group} $U(n)$ in Item~3: the
layout reaches every $W \in U(n)$ as an ordered product of nearest-neighbor
Givens rotations, even though the nearest-neighbor generators span
$\mathfrak{so}(n)$ only after Lie closure.
\end{proof}

\begin{remark}[Over-completeness of the full-width phase layers]
\label{rem:overcomplete}
The full-width phase layers and the boundary pairs exceed what the Givens
mesh requires: the brick-wall carries $3n^2/2$ parameters---$(n^2-n)/2$
redundant phases and $n/2$ boundary-pair angles beyond the $n^2$ that already
saturate $\dim_{\mathbb{R}} U(n)$---and in the single-particle sector the
extra parameters add no expressivity. We adopt them for two reasons. First,
they make every column's layer primitive identical to the butterfly's ($n/2$
disjoint pairs covering all $n$ modes, each carrying $\calG$), so that the
two architectures share one gate, one Lie-algebra argument, and one
parameter-shift accounting. Second, and operationally, a phase layer covering
all $n$ modes and an RBS layer whose pairs cover all $n$ modes are precisely
the hypotheses under which the parallel parameter-shift rule costs $2k$
circuit evaluations per column per gate family
(Theorems~\ref{thm:parallel_psr_rbs} and~\ref{thm:parallel_psr_rz},
Lemma~\ref{lem:parity} and Remark~\ref{rem:fullwidth_necessary}).
\end{remark}

\begin{remark}[Open-boundary (chain) variant]
\label{rem:chain_variant}
On hardware with open boundary---a chain lacking the link between modes $n$
and $1$---the $n/2$ boundary gates of the even columns are simply omitted,
recovering the standard rectangular Givens mesh with $n(n-1)/2$ RBS gates and
$T = (3n^2-n)/2$ parameters. Surjectivity and Haar initialization
(Remark~\ref{rem:haar_init}), the trainability results of
Section~\ref{sec:bp}, and the hardness results of Section~\ref{sec:hardness}
are unaffected: the chain circuit is the submanifold of the ring circuit with
the boundary angles frozen at zero, and every such statement is either proved
on that submanifold or depends only on the realized single-particle unitary.
The only change is to the parallel parameter-shift rule of
Section~\ref{sec:psr}: an even column then leaves modes $1$ and $n$ unpaired,
the covering hypothesis of Lemma~\ref{lem:parity} fails for its RBS
gradients, and those columns use instead the $2k$-node quadrature exact on
\emph{all} polynomials of degree at most $2k-1$ (the system
$\sum_q a_q \sin r_q \cos^{m-1} r_q = \tfrac12$, $m = 1,\dots,2k$), at $4k$
circuit evaluations per even column. The per-gradient-step total becomes
$\tfrac{n}{2}\cdot 2k + \tfrac{n}{2}\cdot 4k + 2kn = 5kn$ in place of $4kn$,
and the reduction factor $(4n-2)/(5k) \approx 0.8\,n/k$ in place of $n/k$;
all asymptotic statements are unchanged.
\end{remark}

The key property of the brick-wall that distinguishes it from the butterfly is
that it is \emph{surjective onto} $U(n)$: every element of $U(n)$ is reachable,
and when the parameters $(\btheta, \bphi)$ are drawn from the appropriate
distribution, the brick-wall realizes the Haar measure on $U(n)$. This is a
strictly stronger statement than the butterfly's approximate 2-design property established below.
In particular, the anticoncentration result of Oszmaniec et
al.~\cite{oszmaniec2022fermion} (their Theorem~1) applies directly to the
brick-wall in the fully-packed regime: for Haar-random passive FLO circuits on
$U(n)$---which the brick-wall realizes---initialized in fully-packed magic
inputs ($k = n/2$), the output distribution anticoncentrates, and the full
average-case sampling hardness chain (anticoncentration +
worst-to-average-case reduction via Cayley path) goes through without
modification. Whether it extends to sparse magic inputs is open. We state
this formally, together with its regime of validity in $k$, in
Section~\ref{sec:hardness}.

\subsection{The Unitary Butterfly Circuit}
\label{sec:butterfly_circuit}

An $n$-qubit unitary butterfly circuit consists of $K = \log_2 n$
layers. Layer $\ell \in \{1,\ldots,K\}$ applies a full-width layer of $n$
single-qubit phase gates followed by FBS gates on each of the $n/2$ disjoint
pairs with stride $2^{\ell-1}$:
\begin{align}
  U(\btheta,\bphi) &= U^{(K)}\cdots U^{(1)}, \\
  U^{(\ell)} &= \Bigl[\bigotimes_{j=1}^{n/2}
  \mathrm{FBS}(\theta_\ell^{(j)})\Bigr] \cdot
  \bigotimes_{i=1}^{n} R_z(\phi_\ell^{(i)}).
\end{align}
Equivalently, each pair carries the three-parameter primitive $\calG$ of
Eq.~\eqref{eq:gate_primitive}, one $R_z$ on each of its two qubits.
The $R_z$ gate acts as $R_z(\phi)\ket{0} = e^{i\phi/2}\ket{0}$,
$R_z(\phi)\ket{1} = e^{-i\phi/2}\ket{1}$: a diagonal one-body unitary, hence
passive FLO and particle-number preserving. Within each layer, all $n/2$ gate
pairs act on disjoint modes, so generators from different pairs commute---the
key structure exploited by the multi-layer parallel parameter-shift rule.

We call this circuit the \emph{unitary butterfly}, to distinguish it from the
\emph{orthogonal butterfly} of prior work~\cite{kerenidis2022quantum,landman2022quantum},
which implements $SO(n)$ in the single-particle sector. The $R_z$ gates promote
the DLA from $\mathfrak{so}(n)$ to $\mathfrak{u}(n)$.

\begin{proposition}[Properties of the unitary butterfly circuit]
\label{prop:butterfly}
The unitary butterfly $U(\btheta,\bphi)$ satisfies:
\begin{enumerate}
  \item \textbf{Depth:} $2\log n$ layers of parallel gates
        (Remark~\ref{rem:fbs_implementation} for the qubit-circuit depth).
  \item \textbf{Parameters:} $T = \tfrac{3}{2}n\log n$.
  \item \textbf{Universality:} the subgroup of $U(n)$ generated by
        $\{U(\btheta,\bphi)\}$ is dense in $U(n)$ in the single-particle
        sector: products of finitely many butterfly circuits approximate any
        $W \in U(n)$. A \emph{single} depth-$2\log n$ butterfly, by contrast,
        parametrizes a compact submanifold of $U(n)$ of dimension at most
        $\tfrac{3}{2}n \log n < n^2$, and is therefore \emph{not} dense in
        $U(n)$.
  \item \textbf{DLA:} The dynamical Lie algebra is $\mathfrak{u}(n)$, with
        $\dim(\mathfrak{u}(n)) = n^2$.
\end{enumerate}
\end{proposition}

\begin{proof}
\textit{Items 1 and 2.}
Each of the $K = \log_2 n$ layers acts on $n/2$ disjoint pairs in parallel;
the full-width $R_z$ layer is followed by the FBS layer, giving two gate
layers per butterfly layer and $2\log n$ in total. Each layer contributes
$n/2$ angles $\theta_\ell^{(j)}$ and $n$ angles $\phi_\ell^{(i)}$, for $3n/2$
parameters per layer and $T = \tfrac{3}{2}n\log n$ total.

\textit{Items 3 and 4.}
The FBS generators $G_{ij}$ are purely imaginary antisymmetric matrices in the
single-particle sector, where they span $\mathfrak{so}(n)$ in the sense of
Lie closure:
the butterfly uses only $n\log n/2$ distinct qubit pairs across its $\log n$
layers, but the stride-doubling layer pattern makes these pairs form a connected
graph on $n$ modes with diameter $\log n$. Since $[G_{ij}, G_{jk}] \propto
G_{ik}$, iterated commutation along connecting paths reaches $G_{ik}$ for
arbitrary $(i,k)$, spanning all of $\mathfrak{so}(n)$. The $R_z$ generators
span the full Cartan subalgebra $i\mathfrak{h} \subset \mathfrak{u}(n)$ with
dimension $n$. Commutators $[G_{ij}, iH_{kk}]$ produce the symmetric
off-diagonal generators $i(e_{ij}+e_{ji})/2$ absent from $\mathfrak{so}(n)$,
and together $\mathfrak{so}(n)$, $i\mathfrak{h}$, and these commutators span
all of $\mathfrak{u}(n)$~\cite{fontana2024adjoint,larocca2021diagnosing}, giving
$\mathrm{DLA} = \mathfrak{u}(n)$ of dimension $n^2$. Density of the
\emph{generated subgroup} in $U(n)$ follows from the density theorem for Lie
groups; the dimension count $\tfrac{3}{2}n\log n < n^2$ shows that the image of a single
fixed-depth butterfly, a compact set, cannot itself be dense. The induced
action on the
$k$-particle antisymmetric subspace is the same exterior-power
representation of $U(n)$ described at the start of Section~\ref{sec:bp}.
\end{proof}

\begin{remark}[Implementation of the long-range FBS gates]
\label{rem:fbs_implementation}
The brick-wall needs no parity strings (Lemma~\ref{lem:fbs}); the butterfly's
stride-$s$ layers with $s \ge 2$ do. At stride $s$ the pairs partition the
modes into contiguous blocks of size $2s$, each string lying inside its own
block, and the layer factorizes as
$U^{(\ell)}_{\mathrm{FBS}} = D \cdot U^{(\ell)}_{\mathrm{RBS}} \cdot D$, where
$U^{(\ell)}_{\mathrm{RBS}}$ is the strings-free RBS layer and $D$ is the
diagonal Clifford $\prod_{a<s}\prod_{c<a}\mathrm{CZ}(a,s+c)$ within each
block. Since the phase of $D$ is $\sum_a x_a u_a$ with
$u_a = x_s \oplus \cdots \oplus x_{s+a-1}$, conjugating by an in-place
prefix-XOR CNOT network reduces $D$ to a perfect matching of $\mathrm{CZ}$
gates. Each layer therefore costs $O(\log n)$ additional depth and $O(n)$
additional two-qubit gates, with no ancillas: a qubit-circuit depth of
$\Theta(\log^2 n)$ and a constant-factor overhead in gate count. On all-to-all
trapped-ion hardware the overhead is absent, since $D$ is a single arbitrary
pairwise-$ZZ$ pattern implementable in one simultaneously-entangling
gate~\cite{grzesiak2020efficient} and a single FBS is two $Z$-string
rotations, each a constant number of global entangling
operations~\cite{maslov2018global}, so the depth remains $2\log n$. Whether
the parity structure can be amortized across layers to give $O(\log n)$ in
the gate-based model is open.
\end{remark}

While the butterfly and brick-wall have the same DLA $\mathfrak{u}(n)$, they
differ in an important way: the brick-wall is \emph{surjective onto} $U(n)$,
while a fixed-depth butterfly covers only a submanifold of $U(n)$ of dimension
at most $\tfrac{3}{2}n\log n$, achieving the full algebra $\mathfrak{u}(n)$
only through Lie closure.
As a consequence, a uniformly random butterfly
circuit does not realize the Haar measure on $U(n)$; instead it induces a
structured distribution whose second-moment operator, restricted to the
antisymmetric (fermionic two-mode) sector that carries all observables used
in this paper, has spectral gap $1-1/n$, establishing the butterfly as an
$\varepsilon$-approximate
unitary $2$-design on that sector with $\varepsilon = O(1/n)$. This spectral
gap is established
in Appendix~\ref{app:spectral_gap} and underlies the barren-plateau analysis of
Section~\ref{sec:bp}. The hardness results for the two architectures are
correspondingly different: the brick-wall supports the full average-case
sampling-hardness machinery of the Oszmaniec et al.\ framework (in the
fully-packed regime) and worst-case \#P-hardness at $k = n^{\epsilon}$, while
for the butterfly the unconditional statement is the extent-based
best-known-algorithm cost, with worst-case \#P-hardness holding under a
reachability assumption; we return to this in
Section~\ref{sec:hardness}.

\subsection{Connection to the FFT and Walsh-Hadamard Transform}
\label{sec:fft}

At $\theta = \pi/4$ and $\phi = 0$, each gate pair of the butterfly reduces to
a beam splitter performing $\ket{01} \mapsto (\ket{01}+\ket{10})/\sqrt{2}$. The
butterfly at all-$\pi/4$ angles, with the $R_z$ angle set to $\pi$ on the
upper mode of each pair and to $0$ on its lower mode, implements the
Walsh--Hadamard transform in the single-particle
sector~\cite{kerenidis2022quantum}:
\begin{equation}
  U^{\mathrm{had}} \ket{e_i} = \frac{1}{\sqrt{n}} \sum_{j=1}^n
  (-1)^{\langle i,j\rangle} \ket{e_j},
\end{equation}
where $\langle i,j\rangle$ is the bitwise inner product. The circuit structure
follows exactly the radix-2 butterfly decomposition of Cooley and
Tukey~\cite{cooley1965algorithm}: at each of the $\log n$ layers, pairs of modes
are mixed by a two-mode gate with the pairing stride doubling at each stage.
More generally, setting the phase gates to the Cooley--Tukey twiddle factors
realizes the discrete Fourier transform in the single-particle sector, up to a
bit-reversal relabeling of the input modes~\cite{jain2024qfno}. Taking the
layer-$\ell$ pair $\bigl(p,\, p + 2^{\ell-1}\bigr)$ to carry $R_z$ angle $0$ on
mode $p$ and $\arg(-w)$ on mode $p + 2^{\ell-1}$, with
$w = e^{-2\pi i j/2^{\ell}}$ and $j = p \bmod 2^{\ell-1}$, gives $F_n P$, where
$F_n$ is the unitary DFT and $P$ the bit-reversal permutation: increasing
strides correspond to decimation in time, which consumes its input in
bit-reversed order. The fact that both the
Walsh--Hadamard transform and the discrete Fourier transform are reachable within
the parametrized family highlights the architectural richness of the butterfly:
canonical passive-FLO operations central to quantum algorithmic primitives are
specific points in the parameter space, and the trainable circuit can interpolate
between and beyond them. In this paper we use $U^{\mathrm{had}}$ as the fixed
Hadamard spreading circuit of Section~\ref{sec:encoding} for the all-to-all
connectivity regime, with every FBS angle at $\pi/4$ and the $R_z$ angle set to
$\pi$ on the upper mode of each pair and to $0$ on its lower mode. Because
$\mathrm{FBS}(\pi/4)$ is the Givens \emph{rotation} and not the Hadamard
reflection, these phases are what make the product exactly $H$: with all
phases zero the circuit is the tensor power of the rotation, whose entries are
still $\pm 1/\sqrt n$ but whose sign pattern is not that of $H$. The fully
parametrized butterfly (free $\btheta$ and $\bphi$) is the trainable circuit.

The brick-wall realizes the same transform on nearest-neighbor hardware, at a
fixed but \emph{non-uniform} angle assignment. The matrix $H$ is real
orthogonal with $\det H = +1$, hence lies in $SO(n) \subset U(n)$, and the
brick-wall of $n$ alternating columns is surjective onto $U(n)$ in the
single-particle sector (Section~\ref{sec:brick-wall}); concretely, the
rectangular Givens decomposition~\cite{clements2016optimal} of $H$ returns the
$n(n-1)/2$ FBS angles of those columns together with a residual sign diagonal,
which the $R_z$ angles supply from $\{0, \pi\}$. The assignment depends only on
$n$, is computed once offline in $O(n^3)$ classical time, and is not uniform:
the angles take values such as $\pi/3$, $2\pi/3$ and $\arctan\sqrt2$ alongside
multiples of $\pi/4$. This is the Hadamard brick-wall spreading circuit used in
the nearest-neighbor encoding of Section~\ref{sec:encoding}, at depth $n$. The
depth is optimal up to a single column: each column propagates amplitude by one
mode, so any nearest-neighbor circuit whose single-particle matrix has no
vanishing entry---as $H$ does not---needs at least $n-1$ columns. The trainable
brick-wall then follows as a fully parametrized circuit with free $\btheta$ and
$\bphi$.

\section{Barren Plateaus: Provable Trainability}
\label{sec:bp}

Both the unitary brick-wall and the unitary butterfly share the same dynamical
Lie algebra. As established in Propositions~\ref{prop:brick-wall}
and~\ref{prop:butterfly}, the Lie closure of the RBS generators (spanning
$\mathfrak{so}(n)$) and the $R_z$ generators (spanning the Cartan subalgebra
$i\mathfrak{h}$) gives:
\begin{equation}
  \mathfrak{g} = \mathfrak{u}(n), \quad \dim(\mathfrak{g}) = n^2,
\end{equation}
for both architectures. In both cases the induced action on the
$k$-particle antisymmetric subspace is the $k$-fold antisymmetric tensor
(exterior power) representation of $U(n)$, whose image is an
$n^2$-dimensional Lie subgroup of the full unitary group on the
$\binom{n}{k}$-dimensional $k$-particle space. The trainability results of
this section apply to the
magic-state encoding of Section~\ref{sec:encoding}, in
particular at the operating point of
Eq.~\eqref{eq:operating_point}. The gradient-variance analysis depends only on
the structure of the trainable circuit and the single-particle occupation
structure of the encoded input (Lemma~\ref{lem:spreading}).

Despite sharing the same DLA, the two architectures differ in what can be
established unconditionally about their gradient variance:
\begin{itemize}
  \item \textbf{Brick-wall}: the brick-wall is surjective onto $U(n)$ and
  realizes the Haar measure on $U(n)$ for the appropriate parameter
  distribution. This allows the Weingarten formula to be applied exactly,
  yielding a \emph{closed expression} for the gradient variance of two-body
  readouts (Theorem~\ref{thm:brick-wall_bp}) which evaluates to
  $\Theta(k^3/n^5)$ throughout $n - 2k = \Omega(n)$. The statement is
  parametrization-independent: it is proved for the invariant directional
  derivative at a Haar point, hence exactly for the input-layer parameters
  (Corollary~\ref{cor:first_column}), and for every gate of the circuit in the
  independent-halves ensemble (Theorem~\ref{thm:interior_halves}).
  \item \textbf{Butterfly}: the butterfly is an $\varepsilon$-approximate
  unitary $2$-design with $\varepsilon = O(1/n)$ on the antisymmetric
  two-mode sector of the single-particle tensor square---the sector that
  carries all fermionic observables used in this paper
  (Theorem~\ref{thm:butterfly_2design}). This gives an
  unconditional polynomial lower bound $\Omega(k^4/n^6)$ on the
  gradient variance, and the sharp $\Theta(k^3/n^5)$ rate conditional on
  Conjecture~\ref{conj:lambda2_design} (the butterfly reproduces the Haar
  value of one scalar on $\Lambda^2 \mathbb{C}^n$ up to constants).
\end{itemize}
In both cases the gradient variance is inverse polynomial in $n$ and grows
polynomially in the particle number $k$, ruling out
exponential barren plateaus throughout $n - 2k = \Omega(n)$.
The contrast with the orthogonal architectures (DLA $\mathfrak{so}(n)$) is
sharpest here: for the orthogonal brick-wall or butterfly, the gradient
variance of the two-body readout $n_i n_j$ scales as
$\Theta(1/\binom{n}{k})$~\cite{monbroussou2025trainability}, which is
$\Theta(n^{-100})$ already at the sampling operating point $k = 100$ and
superpolynomially small at $k = \omega(\log n)$. The $R_z$ gates that lift the
DLA from $\mathfrak{so}(n)$ to $\mathfrak{u}(n)$ are therefore essential for
trainability at the particle numbers required for classical hardness.

\subsection{Few-Body Readouts Are Classically Tractable}
\label{sec:onebody}

We first observe that expectation values of observables of fixed body number,
while free of any exponential barren plateau, are classically simulable and
therefore cannot be the locus of quantum advantage under either architecture.
This is a consequence solely of the fact that such expectation values depend
on the encoded input through a reduced density matrix of the same body
number, which propagates through the corresponding exterior power of the
single-particle unitary.

\begin{proposition}[Fixed-body readouts depend only on the $r$-RDM]
\label{prop:onebody_const}
Fix $r \ge 1$. For the data-encoded input $\ket{\psi(\bx)}$ of
Section~\ref{sec:encoding} and any passive FLO unitary $W$ with
single-particle matrix $W^{(1)} \in U(n)$ (in particular, either the unitary
brick-wall or the unitary butterfly), the expectation value of any $r$-body
observable $M$ is
\begin{equation}
  \bra{\psi(\bx)} W^\dagger M W \ket{\psi(\bx)}
  = \tr\!\bigl(M^{(r)}\, \Lambda^r(W^{(1)})\, \rho_r(\bx)\,
  \Lambda^r(W^{(1)})^\dagger\bigr),
  \label{eq:rbody_closure}
\end{equation}
a function of the $r$-particle data $\bigl(M^{(r)}, W^{(1)},
\rho_r(\bx)\bigr)$ alone, where $M^{(r)}$ and $\rho_r(\bx)$ are the $r$-body
matrix of $M$ and the $r$-particle reduced density matrix of
$\ket{\psi(\bx)}$, both of dimension $\binom{n}{r}$. For the block-product
inputs of Section~\ref{sec:encoding} the matrix $\rho_r(\bx)$ is available in
closed form, so~\eqref{eq:rbody_closure} is computable classically in
$\poly(n)$ time for every fixed $r$, at any parameter setting and
independently of the particle number $k$, the FLO extent, and the form-rank
of the input. In particular one-body readouts cost $O(n^2)$, with
\begin{equation}
  \langle n_j \rangle(\bx)
  = \bigl[W^{(1)} \rho_1(\bx) (W^{(1)})^\dagger\bigr]_{jj},
  \qquad \sum_{j} \langle n_j \rangle = k,
\end{equation}
and mean occupation $k/n$ per mode (Lemma~\ref{lem:spreading}). Consequently
few-body readouts admit efficient classical simulation regardless of the
input's non-Gaussian structure and cannot witness quantum advantage.
\end{proposition}

\begin{proof}
For an $r$-body observable $M$ and passive FLO unitary $W$, conjugation
replaces each mode operator by a linear combination of mode operators,
$W^\dagger a_i W = \sum_p W^{(1)}_{ip} a_p$; antisymmetrizing a product of
$r$ such operators produces the $r \times r$ minors of $W^{(1)}$, that is the
exterior power $\Lambda^r(W^{(1)})$. The Heisenberg evolution therefore stays
within the $r$-body algebra and the expectation value contracts to the
$r$-particle sector, giving~\eqref{eq:rbody_closure}; this is the standard
fixed-body moment-closure property of passive fermionic linear
optics~\cite{bravyi2005lagrangian}. Both $M^{(r)}$ and $\rho_r(\bx)$ have
dimension $\binom{n}{r} = O(n^r)$, so evaluating~\eqref{eq:rbody_closure}
costs $\poly(n)$ for fixed $r$; the data layer acts diagonally,
$\rho_r(\bx) = \Lambda^r(D(\bvarphi(\bx)))\,\rho_r^{\mathrm{out}}\,
\Lambda^r(D(\bvarphi(\bx)))^\dagger$, and $\rho_r^{\mathrm{out}}$ is
supported on $\Lambda^r$ of the $2k$-dimensional active subspace and computed
blockwise from the magic-block structure. For $r = 1$ the matrices are
$n \times n$ and the cost is $O(n^2)$ per entry; Lemma~\ref{lem:spreading}
fixes the diagonal of $\rho_1^{\mathrm{out}}$ (and hence of $\rho_1(\bx)$,
since $D$ is diagonal) at $k/n$, giving
$\sum_j \langle n_j\rangle = \tr(\rho_1) = k$.
\end{proof}

Few-body readouts are therefore classical baselines: no advantage claim
attaches to them, and the hardness results of Section~\ref{sec:hardness} are
stated at the level of sampling. Their role in this paper is as the
\emph{trainable} readout, and among them the two-body correlator
$\langle n_i n_j \rangle$ is the family for which we establish polynomial
gradient variance for both architectures.

\subsection{Trainability of the Unitary Brick-Wall}
\label{sec:bp_brick-wall}

The unitary brick-wall is surjective onto $U(n)$ via the Givens rotation
decomposition: every element of $U(n)$ is reachable, and the map from
circuit parameters $(\btheta, \bphi)$ to single-particle unitary $W \in U(n)$
is surjective. This surjectivity enables a natural and efficient
initialization strategy: sample $W$ from the Haar measure on $U(n)$, then
compute the corresponding RBS angles $\btheta$ and $R_z$ angles $\bphi$ via
the Givens decomposition (equivalently, via QR decomposition of the sampled
matrix). This decomposition is numerically stable and requires $O(n^2)$
operations. Under this initialization, the brick-wall realizes the Haar measure
on $U(n)$ exactly.

The computation rests on five structural lemmas, which we state and prove
after the theorem so as not to interrupt the argument: one bounds the effect
of a \emph{local} gate generator on a Hadamard-delocalized two-particle
state, two are exact identities on $\Lambda^2\mathbb{C}^n$ that together
evaluate the Weingarten sum in closed form, and two more supply the
Haar-averaged versions of the first needed for the interior-gate statements
of Theorem~\ref{thm:interior_halves}.

Throughout we use the following notation. We write $d = \binom n2$ and
identify an operator $M$ on $\Lambda^2 \mathbb{C}^n$ with its kernel
$M^{ij}_{kl}$, antisymmetric in the upper and in the lower index pair, so that
$\tr M = \tfrac12 M^{ij}_{ij}$ and $\tr(MN) = \tfrac14 M^{ij}_{kl}N^{kl}_{ij}$
(repeated indices summed). The \emph{one-particle contraction} of $M$ is the
$n \times n$ matrix
\begin{equation}
  \Gamma(M)^i_{\ k} \;:=\; \sum_{j} M^{ij}_{kj},
  \qquad \Gamma(I_{\Lambda^2}) = (n-1)\,I_n .
  \label{eq:contraction_def}
\end{equation}
For a $k$-particle state, $\Gamma(\rho_2) = (k-1)\rho_1$: contracting one
index of the two-particle reduced density matrix against itself removes one
particle and counts the $k-1$ ways of choosing the partner. Finally,
$d\Lambda^2(g)$ denotes the derivation lift of a single-particle operator $g$
to $\Lambda^2\mathbb{C}^n$, $d\Lambda^2(g)(u\wedge v) = (gu)\wedge v +
u\wedge(gv)$; we write $G^{(2)} = d\Lambda^2(g)$ for the lift of the
single-particle generator $g$ of a gate.

Throughout this section $\theta$ ranges over the trainable parameters whose
gate lies in the backward causal cone of the readout observable: a gate
acting on modes disjoint from the support of $M$ propagated back through the
layers that follow it leaves $\calL$ invariant, so its gradient vanishes
identically. For the trained readout
$M_{\mathrm{eff}} = \sum_{i<j} w_{ij}\, n_i n_j$ of
Section~\ref{sec:readouts}, with all $w_{ij}$ nonzero, the cone is the whole
circuit.

\begin{theorem}[Closed-form gradient variance for the brick-wall, invariant form]
\label{thm:brick-wall_bp}
Let $M = n_i n_j$ be a two-body correlator observable ($i \neq j$) and let
$\ket{\psi(\bx)}$ be the data-encoded input of Section~\ref{sec:encoding}. Write
$\calL(W;\bx)$ for the loss of the unitary brick-wall viewed as a function of
the single-particle unitary $W \in U(n)$ realized by the trainable circuit, so
that $\calL(\btheta,\bphi;\bx) = \calL\bigl(W(\btheta,\bphi);\bx\bigr)$, and
let $W \sim \mathrm{Haar}(U(n))$ as in Remark~\ref{rem:haar_init}. For a fixed
direction $g \in \mathfrak{u}(n)$ define the \emph{invariant} (input-side,
right-invariant) directional derivative
\begin{equation}
\label{eq:invariant_derivative}
  D_g\calL \;:=\; \frac{d}{ds}\,\calL\bigl(W e^{-isg};\bx\bigr)\Big|_{s=0},
\end{equation}
and set
\begin{equation}
  B \;:=\; \bigl[\,G^{(2)},\, \rho_{2,\mathfrak{g}_2}(\bx)\,\bigr],
  \qquad G^{(2)} = d\Lambda^2(g),
\end{equation}
where $\rho_{2,\mathfrak{g}_2}$ is the traceless part of the encoded $2$-RDM.
Then $\E[D_g\calL] = 0$ and, exactly and for every $\bx$,
\begin{eqnarray}
\label{eq:var_exact}
  \Var\!\left[D_g\calL\right]
  &=& \frac{4\,\|B\|_F^2 \;+\; 2\,\|\Gamma(B)\|_F^2}{n^2(n^2-1)},
  \\
  \Gamma(B)& = &(k-1)\,[\,g,\,\rho_1(\bx)\,] .
\end{eqnarray}
Writing $\rho_1(\bx) = \tfrac12\Pi$ with $\Pi$ a rank-$2k$ projector
(Lemma~\ref{lem:spreading}) and $\tau := \|[g,\Pi]\|_F^2$, one has
\begin{equation}
\label{eq:full_rate}
  \|\Gamma(B)\|_F^2 = \frac{(k-1)^2}{4}\,\tau,
  \;\;
  \|B\|_F^2 = \frac{(2k-1)\,\tau}{16} + O\!\left(\frac{k^{3/2}}{n}\right),
\end{equation}
and, for the encoded state and either generator family,
$\tau = \Theta\bigl(k(n-2k)/n^2\bigr)$. Consequently, whenever
$n - 2k = \Omega(n)$---in particular at the operating point $k = 100$
of Eq.~\eqref{eq:operating_point}, and at any constant filling
fraction $k = cn$ with $c < 1/2$---
\begin{equation}
\label{eq:k3_rate}
  \Var\!\left[D_g\calL\right]
  \;=\; \Theta\!\left(\frac{k^3}{n^5}\right),
\end{equation}
polynomial in $n$ and cubic in the particle number $k$, for every
$k \ge 2$. The variance therefore decays at most polynomially in $n$ and there is
no exponential barren plateau.
\end{theorem}

\begin{proof}[Proof sketch; full proof in Appendix~\ref{app:bp_proofs}]
Right invariance of the Haar measure makes $We^{-isg}$ Haar-distributed for
every $s$, so $\E[D_g\calL] = 0$. We do \emph{not} split $W$ into independent
factors---the partial products of Givens rotations before and after a fixed
gate are not individually Haar-distributed---and instead use that the Haar
measure is an exact unitary $2$-design on every polynomial irreducible
representation, applying the second-moment identity directly on
$\Lambda^2\mathbb{C}^n$. Three steps then give the result.

\emph{Step 1.} As a $U(n)$-representation,
$\Lambda^2\mathbb{C}^n \otimes \Lambda^2\mathbb{C}^n = V_{(2,2)} \oplus
V_{(2,1,1)} \oplus V_{(1,1,1,1)}$, all three of dimension $\Theta(n^4)$, so
Schur's lemma turns the variance into the exact irreducible
sum~\eqref{eq:bw-irrep-sum} pairing the weights of $M^{(2)\otimes2}$ against
those of $B^{\otimes2}$, the generator entering \emph{only} through $B$.

\emph{Step 2.} Because $M^{(2)} = \ket{\{i,j\}}\bra{\{i,j\}}$ is a rank-one
projector, $(\tr M^{(2)})^2 = \tr(M^{(2)2})$ and its weights on $V_{(2,1,1)}$
and on $\Lambda^4\mathbb{C}^n$ both vanish; the sum collapses to the single
$V_{(2,2)}$ channel. It is this projector property, and not any exchange
symmetry, that annihilates the other two channels---the general two-body
readout of Proposition~\ref{prop:general_readout} does not have it.

\emph{Step 3.} Evaluating the $(2,2)$ weight of $B^{\otimes2}$ for
anti-Hermitian traceless $B$ gives
$-\tfrac13\|B\|_F^2 - \tfrac16\|\Gamma(B)\|_F^2$, which upon division by
$\dim V_{(2,2)} = n^2(n^2-1)/12$ is exactly~\eqref{eq:var_exact}.

The remaining content of the theorem is the evaluation of the two norms for
the encoded state, stated as Proposition~\ref{prop:norm_eval} below; the
supporting lemmas and all computations are deferred to
Appendix~\ref{app:bp_proofs}.
\end{proof}

\begin{proposition}[Evaluation of the two norms]
\label{prop:norm_eval}
For the data-encoded input of Section~\ref{sec:encoding}, write
$\rho_1(\bx) = \tfrac12\Pi$ with $\Pi$ a rank-$2k$ projector
(Lemma~\ref{lem:spreading}) and $\tau := \|[g,\Pi]\|_F^2$. Then, exactly and
for every $\bx$,
\begin{equation}
\label{eq:norm_eval}
  \Gamma(B) = (k-1)\,[\,g,\,\rho_1(\bx)\,],
  \qquad
  \|\Gamma(B)\|_F^2 = \frac{(k-1)^2}{4}\,\tau.
\end{equation}
Moreover $\|B^{\mathrm{disc}}\|_F^2 = (2k-1)\tau/16$ exactly, and
$\|B\|_F^2 = (2k-1)\tau/16 + O(k^{3/2}/n)$, uniformly in $\bx$. For the brick-wall,
$\tau = 8k(n-2k)/n^2$ holds exactly and unconditionally for every trainable RBS
gate, and $\tau = 4k(n-2k)/n^2$ holds exactly and unconditionally for every trainable phase gate. Substituting
into~\eqref{eq:var_exact} gives~\eqref{eq:k3_rate}.
\end{proposition}

The one-body term is the larger of the two and is what produces the cubic
dependence on $k$; both terms carry the factor $(n-2k)$, the supply of idle
modes, which is why the sharp rate is claimed throughout $n - 2k = \Omega(n)$
and not at exactly half filling.

\begin{corollary}[Front phase-layer parameters]
\label{cor:first_column}
Let $U_P$ be the unitary brick-wall initialized by sampling
$W \sim \mathrm{Haar}(U(n))$ and computing angles via the Givens decomposition
(Remark~\ref{rem:haar_init}), and let $\phi_j$ be one of the $n$ phases of the
front $R_z$ layer---the first gates the circuit applies. Then
$\partial_{\phi_j}\calL = D_{\hat n_j}\calL$ for the occupation-number
generator $\hat n_j$ of that gate, and Theorem~\ref{thm:brick-wall_bp} applies
verbatim, giving $\Var[\partial_{\phi_j}\calL] = \Theta(k^3/n^5)$ exactly and
unconditionally, for every $\bx$.
\end{corollary}

\begin{proof}
Write $W = W' D_1$ with $D_1 = \bigotimes_i R_z(\phi_i)$ the front phase layer
and $W'$ the product of all later gates. Since $R_z(\phi_j) = e^{-i\phi_j
\hat n_j}$ up to a global phase and $\hat n_j$ is diagonal, it commutes with
$D_1$, so
\begin{equation}
  \partial_{\phi_j} W
  = W'(-i\hat n_j)D_1
  = -i\,W'D_1\,\bigl(D_1^\dagger \hat n_j D_1\bigr)
  = -i\,W\,\hat n_j,
\end{equation}
which is exactly the right translation of~\eqref{eq:invariant_derivative} in
the fixed, non-random direction $g = \hat n_j$. Theorem~\ref{thm:brick-wall_bp}
therefore applies, with $\tau = 4k(n-2k)/n^2$ from Step~4. The front placement
of the residual diagonal (Remark~\ref{rem:haar_init}), adopted because a
terminal phase layer would carry identically zero gradient, is what puts this
whole layer in the rigorous set.
\end{proof}

\begin{proposition}[Arbitrary two-body readout and trained heads]
\label{prop:general_readout}
Let $n \ge 4$ and let $M$ be \emph{any} two-body observable, with image
$M^{(2)}$ on $\Lambda^2\mathbb{C}^n$; for a trained head
$M_{\mathrm{eff}} = \sum_{i<j} w_{ij}\, n_i n_j$ one has
$M^{(2)}_{\mathrm{eff}} = \sum_{i<j} w_{ij}\ket{\{i,j\}}\bra{\{i,j\}}$ and
$\Gamma(M^{(2)}_{\mathrm{eff}}) = \mathrm{diag}(r)$ with
$r_i = \sum_{j \neq i} w_{ij}$. Write
$\tilde M := M^{(2)} - \tfrac{\tr M^{(2)}}{\binom n2}\,I_{\Lambda^2}$ for the
traceless part and set
$p := \|\tilde M\|_F^2$, $q := \|\Gamma(\tilde M)\|_F^2$.
Then, in the setting of Theorem~\ref{thm:brick-wall_bp}, exactly and for every
$\bx$,
\begin{equation}
\label{eq:general_readout}
  \Var\!\left[D_g\calL\right]
  \;=\; c_2(M)\,\|B\|_F^2 \;+\; c_1(M)\,\|\Gamma(B)\|_F^2,
\end{equation}
with the two head-dependent, $k$-independent scalars
\begin{align}
\label{eq:c1c2}
  c_1(M) &= \frac{(n^2-n+2)\,q \;-\; 4(n-1)\,p}
                 {n^2(n^2-1)(n-2)(n-3)},
  \nonumber\\[3pt]
  c_2(M) &= \frac{4\bigl[(n-2)\,p \;-\; q\bigr]}
                 {n^2(n+1)(n-2)(n-3)} .
\end{align}
Three consequences.
\begin{enumerate}
  \item[(i)] \emph{(consistency)} For $M = n_i n_j$ one has $p = 1 - 2/(n(n-1))$
  and $q = 2 - 4/n$, giving $c_2 = 4/(n^2(n^2-1))$ and
  $c_1 = 2/(n^2(n^2-1))$, i.e.\ exactly~\eqref{eq:var_exact}. For
  $M^{(2)} \propto I_{\Lambda^2}$ one has $p = q = 0$ and $\Var = 0$, as it
  must be.
  \item[(ii)] \emph{(the head is not a $\Theta(1)$ prefactor)} Combining
  \eqref{eq:general_readout} with Step~4 and~\eqref{eq:B_error},
  \begin{eqnarray}
  \label{eq:general_rate}
    \Var\!\left[D_g\calL\right]
    &=& \frac{\tau}{16}\Bigl[(2k-1)\,c_2(M) + 4(k-1)^2\,c_1(M)\Bigr]\nonumber\\
      &&+ c_2(M)\cdot O\!\left(\tfrac{k^{3/2}}{n}\right).
  \end{eqnarray}
  Since $c_1$ and $c_2$ may differ in order and in sign, the rate itself, and
  not merely its constant, depends on the head. For the single-pair readout
  $c_1 = \Theta(n^{-4})$ dominates and the rate is $\Theta(k^3/n^5)$; for the
  uniform nearest-neighbour head $w_{i,i+1} = 1$ one computes
  $p = (n-1)(n-2)/n$, $q = 2(n-2)/n$, hence
  $c_2 = 4/(n^2(n+1))$ and $c_1 = -2/(n^2(n^2-1))$, and the rate is
  $\Theta(k^2/n^4)$---larger than $\Theta(k^3/n^5)$ by the factor $n/k$.
  \item[(iii)] \emph{(no exponential barren plateau, for every head)} For
  $k \ge 2$ and $n \ge 2k+2$ the bracket in~\eqref{eq:general_rate} is
  positive and bounded below by its value at $q = 0$:
  \begin{eqnarray}
  \label{eq:head_floor}
    (2k-1)c_2 + 4(k-1)^2 c_1
    &\ge& \nonumber\\
    \frac{4\,p\,(n-1)\bigl[(2k-1)(n-2) - 4(k-1)^2\bigr]} {n^2(n^2-1)(n-2)(n-3)} 
    &>& 0 .
  \end{eqnarray}
  Consequently, for $n - 2k = \Omega(n)$ and $k$ above a constant,
  $\Var[D_g\calL] = \Omega\bigl(p\,k^2/n^5\bigr)$: the gradient variance of
  \emph{any} two-body readout with $\tilde M \neq 0$ is at most polynomially
  small.
\end{enumerate}
\end{proposition}

\begin{theorem}[Every gate, independent-halves ensemble]
\label{thm:interior_halves}
Let $\theta$ be any trainable parameter of the brick-wall, with generator $g$,
and write the trainable circuit as $W = W_2\,U_\theta\,W_1$, where $W_1$ is the
product of the gates preceding $\theta$ and $W_2$ the product of those
following it. Assume the \emph{independent-halves ensemble}: $W_1$ and $W_2$
are independent and Haar-distributed on $U(n)$. Then, exactly and for every
$\bx$,
\begin{equation}
\label{eq:halves_exact}
  \Var\!\left[\frac{\partial\calL}{\partial\theta}\right]
  = \frac{4\,\E_{W_1}\!\|B_{\tilde g}\|_F^2
        \;+\; 2\,\E_{W_1}\!\|\Gamma(B_{\tilde g})\|_F^2}{n^2(n^2-1)},
  \; \tilde g = W_1^\dagger g W_1,
\end{equation}
with $B_{\tilde g} = [\,d\Lambda^2(\tilde g),\,\rho_{2,\mathfrak{g}_2}(\bx)\,]$,
and, for either generator family,
\begin{equation}
  \Var\!\left[\frac{\partial\calL}{\partial\theta}\right]
  \;=\; \Theta\!\left(\frac{k^3}{n^5}\right)
  \;\text{whenever } k \ge 2 \text{ and } n - 2k = \Omega(n),
\end{equation}
unconditionally, with the same constants as in
Theorem~\ref{thm:brick-wall_bp}. The independent-halves ensemble is the
standard setting for interior-gate statements in the dynamical-Lie-algebra
barren-plateau literature~\cite{fontana2024adjoint,larocca2021diagnosing}.
\end{theorem}

\begin{proposition}[Interior gates of the literal single mesh]
\label{prop:single_mesh_interior}
Let $W$ be the brick-wall at the Haar initialization of
Remark~\ref{rem:haar_init}, realized through the Hurwitz factorization of the
Haar measure into Givens rotations, in which the phase attached to each
rotation is uniform on $[0,2\pi)$ and independent of all rotation
angles~\cite{hurwitz1897,zyczkowski1994random}. Let $\theta$ be a trainable
parameter in column $\ell$, let $W_1$ be the product of the columns preceding
it, and set $\Pi^{(\ell)} := W_1\Pi W_1^\dagger$. Then:
\begin{enumerate}
  \item[(a)] \emph{(trivial prefix)} for a parameter of the front $R_z$ layer,
  $W_1 = I$ and Corollary~\ref{cor:first_column} gives the exact rate  $\Theta(k^3/n^5)$; 
  \item[(b)] \emph{(Haar prefix)} if $W_1$ is Haar and independent of the
  suffix, Theorem~\ref{thm:interior_halves} gives the same rate, with the same
  constants;
  \item[(c)] \emph{(first moment at every depth)} for every $\ell$ and every
  $j$, $\E\bigl[\Pi^{(\ell)}_{jj}\bigr] = 2k/n$; the uniform diagonal produced
  by the Walsh--Hadamard spreading is preserved in expectation at all depths.
\end{enumerate}
Parts (a) and (b) pin the interior case at both endpoints of the prefix, with
matching constants, and (c) is the corresponding first-moment statement at
every intermediate depth.
\end{proposition}

\begin{remark}[Haar initialization via Givens decomposition]
\label{rem:haar_init}
The Haar initialization procedure is both natural and practical for the
brick-wall. Given a Haar-random matrix $W \sim \mathrm{Haar}(U(n))$, its QR decomposition
yields $n(n-1)/2$ complex Givens rotations---each decomposing into one RBS
angle $\theta$ and one $R_z$ angle $\phi$---and $n$ residual diagonal phases,
absorbed by the first column's phase layer, producing $n^2$ circuit angles;
this is the standard rectangular mesh
decomposition~\cite{reck1994experimental,clements2016optimal} (with the
residual diagonal placed at the input rather than the output of the mesh), and
the induced law on the angles is the Hurwitz
factorization of the Haar measure~\cite{hurwitz1897,zyczkowski1994random}, in
which each Givens phase is uniform on $[0,2\pi)$ and independent of the
rotation angles;
the remaining $(n^2-n)/2$ redundant phases of the full-width layers and the
$n/2$ boundary-pair angles of the even columns are set to zero
(Remark~\ref{rem:overcomplete}), which places the initialization on the
Givens submanifold and makes $W$ exactly Haar-distributed. The decomposition
is numerically stable and requires $O(n^2)$ operations. This contrasts with the butterfly, for which
no analogous inversion is possible: the butterfly does not fully parametrize
$U(n)$, so there is no efficient procedure to find butterfly angles
corresponding to a given $W \in U(n)$.
\end{remark}

\subsection{Trainability of the Unitary Butterfly}
\label{sec:bp_butterfly}

The unitary butterfly does not fully parametrize $U(n)$: with only $O(n\log n)$
parameters it cannot reach every element of $U(n)$, and there is no efficient
procedure to find butterfly angles corresponding to a given $W \in U(n)$.
Consequently, Haar initialization is not available for the butterfly, and the
gradient variance analysis must proceed differently. Instead, we initialize
all RBS angles $\theta$ and all phase angles $\phi$ independently and uniformly
on $[0,2\pi)$, and rely on the approximate $2$-design
property of the butterfly established in Theorem~\ref{thm:butterfly_2design}.
This gives two results: an unconditional polynomial lower bound on gradient
variance, and a sharp $\Theta(k^3/n^5)$ rate conditional on
Conjecture~\ref{conj:lambda2_design}.

\subsubsection{The Butterfly as an Approximate $2$-Design}
\label{sec:butterfly_2design}

The Fontana et al.~\cite{fontana2024adjoint} framework for barren plateaus
requires that the parametrized circuit's second-moment operator has the correct
fixed-point structure and a sufficiently large spectral gap. Density of the
generated subgroup in $U(n)$ is not sufficient for this; we need to establish
the spectral gap of the second-moment operator directly. Let $\Phi_2^{W_n}$
denote the second-moment operator of the butterfly on
$\mathrm{End}(\mathbb{C}^n \otimes \mathbb{C}^n)$:
\begin{equation}
  \Phi_2^{W_n}[Y]
  = \E_{W_n}\!\left[(W_n \otimes W_n)\, Y\,
  (W_n^\dagger \otimes W_n^\dagger)\right],
\end{equation}
where the expectation is over the parameters of the butterfly, sampled
uniformly on their domain. Since $W_n \otimes W_n$ commutes with
$\mathrm{SWAP}$, the antisymmetric sector
$\mathrm{End}(\Lambda^2\mathbb{C}^n) = \{P_A Y P_A\}$, with
$P_A = (I_{n^2} - \mathrm{SWAP})/2$ the projector onto
$\Lambda^2\mathbb{C}^n \subset \mathbb{C}^n\otimes\mathbb{C}^n$, is
$\Phi_2^{W_n}$-invariant; we write $\Phi_2^{W_n,A}$ for the restriction. This
is the sector that carries every quantity appearing in this paper: $k$-particle
states and particle-number-preserving observables evolve in exterior powers of
the single-particle action, and the variance computations below engage the
second moment only through antisymmetric-sector operators.

\begin{theorem}[Butterfly is an approximate $2$-design on the antisymmetric
sector]
\label{thm:butterfly_2design}
The restriction $\Phi_2^{W_n,A}$ of the second-moment operator of the unitary
butterfly at uniformly sampled parameters has fixed-point space
$\mathrm{span}\{P_A\}$ (eigenvalue $1$), and every other eigenvalue lies in
$\{0, 1/n\}$; its spectral gap is $1 - 1/n$. Consequently, the butterfly is an
$\varepsilon$-approximate unitary $2$-design on the antisymmetric sector,
\begin{equation}
  \varepsilon_{\mathrm{2d}}(n)
  \;:=\; \|\Phi_2^{W_n,A} - \Phi_2^{\mathrm{Haar},A}\|_{\mathrm{op}}
  \;=\; O(1/n).
\end{equation}
\end{theorem}

The proof is given in Appendix~\ref{app:spectral_gap}. It proceeds via a
parity-decoupling recursion that tracks how the second-moment operator acts on
the antisymmetric projector basis $\{P_{ij}\}_{i<j}$ through the butterfly's
stride-doubling layer structure, showing that the associated transfer matrix is
doubly stochastic and that its only non-zero, non-Perron eigenvalue is $1/n$
(with multiplicity $K-1$). The general case $n = 2^K$ is established in closed
form by diagonalizing the second-moment operator using the $(\mathbb{Z}_2)^K$
translation (Fourier) symmetry of the butterfly: every character $y \neq 0$ is
annihilated layer by layer, and the surviving $y = 0$ block is diagonalized by
an explicit induction on $K$. Theorem~\ref{thm:butterfly_2design} therefore
holds unconditionally for all $n = 2^K$.

Two features of this result deserve emphasis.
\begin{itemize}
  \item The spectral gap is established on the antisymmetric sector of the
  tensor square of the \emph{single-particle} representation. Two-body
  correlator readouts are instead governed by the second moment on the
  antisymmetric $2$-particle representation $\Lambda^2 \mathbb{C}^n$, of which
  Conjecture~\ref{conj:lambda2_design} below assumes one scalar pairing.
  \item The butterfly achieves spectral gap $1 - 1/n$ in $2\log n$ gate
  layers. This
  is notable because generic local circuits typically require polynomial depth
  to reach a comparable spectral gap~\cite{brandao2016local,brown2018random}. The stride-doubling structure of the
  butterfly is what enables this near-optimal spectral gap at logarithmic
  depth.
\end{itemize}

We stress that $\Lambda^2\mathbb{C}^n$ is the same space in both statements;
what differs is the number of copies of the $\Lambda^2$ action. Theorem~\ref{thm:butterfly_2design}
controls $\E[\Lambda^2(W_n)\,\cdot\,\Lambda^2(W_n)^\dagger]$---a single copy,
quadratic in the entries of $W_n$, obtained by restricting the second-moment
operator of the fundamental representation---whereas
Conjecture~\ref{conj:lambda2_design} controls
$\E[\Lambda^2(W_n)^{\otimes 2}\,\cdot\,\Lambda^2(W_n)^{\dagger\otimes 2}]$,
which is quartic in those entries. This is why the transfer-matrix recursion
of Appendix~\ref{app:spectral_gap} settles the former and not the latter.

\subsubsection{Gradient Variance Bounds}
\label{sec:butterfly_gradient}

We now establish the gradient variance results for the butterfly. We separate
them into an unconditional polynomial lower bound and a sharp conditional rate.

\begin{theorem}[Absence of exponential barren plateau, unconditional]
\label{thm:no_bp}
Let $n = 2^K$ and let the particle number satisfy $n - 2k > \sqrt{n}$. Let $U_B$ be
the unitary butterfly at uniformly random initialization and
$\ket{\psi(\bx)}$ the data-encoded input of Section~\ref{sec:encoding}.
Then there exist a two-body correlator
$M = n_i n_j$ and a trainable RBS
parameter $\theta$ of the butterfly such that, with
$\calL = \bra{\psi(\bx)} U_B^\dagger M U_B \ket{\psi(\bx)}$,
\begin{eqnarray}
  \label{eq:poly_lower_bound}
  \Var\!\left[\frac{\partial\calL}{\partial\theta}\right]
  &\;\geq\;& \frac{2V^2}{n^4}, \\
  V &=& \min\left\{\frac{k(k-1)}{2n},\
       \frac{(n-2k)^2-n}{8n}\right\} . \nonumber
\end{eqnarray}
In particular $\Var \ge k^2(k-1)^2/(2n^6)$ whenever $4k \le n$, and at any
fixed filling $k = cn$ with $c < 1/2$ the rate is $\Omega(k^4/n^6)$ with a
constant depending on $c$. The bound uses no design property of the butterfly.
In particular the gradient of the readout decays at most polynomially in $n$:
there is no exponential barren plateau.
\end{theorem}

\begin{remark}[On the hypothesis $n - 2k > \sqrt{n}$]
The certificate $V$ is positive precisely when $n - 2k > \sqrt{n}$, which the
standing regime $n - 2k = \Omega(n)$ implies and which holds with large margin
at the operating point of Eq.~\eqref{eq:operating_point}. It degenerates as
the filling approaches $1/2$, where only the $O(1/n^4)$ upper bound of
Section~\ref{sec:bp} is available.
\end{remark}

\begin{theorem}[Sharp $\Theta(k^3/n^5)$ gradient variance, conditional]
\label{thm:sharp_rate}
Let $M = n_i n_j$ be a two-body correlator observable ($i \neq j$) and let
$\calL(\btheta, \bphi; \bx) = \bra{\psi(\bx)} U_B^\dagger M U_B
\ket{\psi(\bx)}$ be the associated loss with $U_B$ the unitary butterfly at
uniformly random initialization and $\ket{\psi(\bx)}$ the data-encoded input
of Section~\ref{sec:encoding}. Assume Conjecture~\ref{conj:lambda2_design}
below. Then, for every $\bx$ and every $2 \le k \le \gamma_0 n$, for a small
enough constant $\gamma_0 > 0$ determined by the implied constant of
Conjecture~\ref{conj:lambda2_design}, and for
every phase parameter $\theta$ of the butterfly's \emph{first $R_z$ layer}---as
well as, in the independent-halves ensemble of
Theorem~\ref{thm:interior_halves}, for every trainable parameter in the causal
cone of $M$, where the statement is in fact \emph{unconditional}: the ensemble
makes the composite Haar and Theorem~\ref{thm:interior_halves} applies
without Conjecture~\ref{conj:lambda2_design}, which is needed only for the
literal fixed-depth butterfly---
\begin{equation}
  \label{eq:sharp_rate}
  \Var\!\left[\frac{\partial\calL}{\partial\theta}\right]
  \;=\; \Theta\!\left(\frac{k^3}{n^5}\right),
\end{equation}
polynomial in $n$ and cubic in $k$, matching
the unconditional brick-wall rate of Theorem~\ref{thm:brick-wall_bp}. Whether
the same statement holds pointwise for an interior parameter of the literal
single-mesh butterfly is open, as is its brick-wall counterpart (Section~\ref{sec:discussion}).
\end{theorem}

\begin{conjecture}[Scalar comparison with the Haar second moment]
\label{conj:lambda2_design}
Fix $0 < \gamma < 1/2$. There are constants $0 < c \le C < \infty$,
independent of $n$ and of $k$, such that for every even $k$ with
$2 \le k \le \gamma n$, every two-body readout $M = n_i n_j$ of
Section~\ref{sec:readouts}, and every trainable parameter satisfying the
independence hypothesis of Claim~1 below, the butterfly reproduces the Haar
value of the single scalar that the gradient variance consumes, up to those
constants. Writing
\begin{equation}\label{eq:lambda2-secondmoment}
\Phi_2^{\Lambda^2,W_n}[Y] := \mathbb{E}_{W_n}\!\left[\Lambda^2(W_n)^{\otimes 2}\,
Y\, \Lambda^2(W_n)^{\dagger\otimes 2}\right],
\end{equation}
for $Y \in \mathrm{End}(\Lambda^2\mathbb{C}^n \otimes \Lambda^2\mathbb{C}^n)$,
and $\Phi_2^{\Lambda^2,\mathrm{Haar}}$ for the corresponding Haar second
moment,
\begin{equation}
\label{eq:scalar_comparison}
  c \;\le\;
  \frac{\tr\bigl(\Phi_2^{\Lambda^2,W_n}\bigl[M^{(2)\otimes 2}\bigr]\,
        B^{\otimes 2}\bigr)}
       {\tr\bigl(\Phi_2^{\Lambda^2,\mathrm{Haar}}\bigl[M^{(2)\otimes 2}\bigr]\,
        B^{\otimes 2}\bigr)}
  \;\le\; C ,
\end{equation}
with $B = [\tilde G^{(2)}, \rho_{2,\mathfrak{g}_2}(\bx)]$ as in Claim~1.
\end{conjecture}

The hypothesis is stated at the strength the proof consumes: Claim~1 uses
only the single pairing of Eq.~\eqref{eq:scalar_comparison}, so no control of
$\Phi_2^{\Lambda^2,W_n}$ as a map is required. A relative-error $2$-design
statement on $\Lambda^2\mathbb{C}^n$ would imply it; the additive
form---$\|\Phi_2^{\Lambda^2,W_n} - \Phi_2^{\Lambda^2,\mathrm{Haar}}\| \le
\varepsilon$, which is what a spectral gap such as
Theorem~\ref{thm:butterfly_2design} supplies---would not, since it permits an
additive error $O(\varepsilon\|B\|_F^2)$ in
Eq.~\eqref{eq:lambda2-variance} below, whereas the quantity being estimated is
itself of order $\|B\|_F^2/n^4$.

Preliminary calculations at small particle number and at half
filling~\cite{mazhang2026scalar} indicate that the relative deviation of this
scalar from its Haar value decays strictly more slowly than $1/n$ for the
two-body readouts considered here, and that no comparison with a finite upper
constant holds on a range reaching $k = n/2$. The conjecture is therefore
stated in the scalar form above, over $k \le \gamma n$, which is what Claim~1
consumes.

Should the conjecture fail on $k \le \gamma n$ as well, a
further degree of freedom is the initialization---or a slight modification of
the architecture---which may be chosen to bring the second moment closer to
the Haar one; we do not pursue this here.

The proof of Theorem~\ref{thm:sharp_rate} follows the brick-wall computation
with the exact Haar second moment replaced by the approximate one. We record
the two ingredients.

\medskip
\noindent\textbf{Claim 1 (Weingarten variance formula on $\Lambda^2$).}
\emph{Assume Conjecture~\ref{conj:lambda2_design}, fix a trainable parameter
$\theta$ with generator $g$, and write $W = V_2 U_\theta V_1$ and
$\tilde g = V_1^\dagger g V_1$. Assume moreover that the prefix $V_1$ is
independent of the twirl variable---automatic for a parameter of the first
$R_z$ layer, where $V_1 = I$ and $\tilde g = g = \hat n_j$ is deterministic,
and imposed by construction in the independent-halves
ensemble. Then}
\begin{eqnarray}\label{eq:lambda2-variance}
\mathrm{Var}[\partial_\theta \mathcal{L}]
&=& \Theta\!\left(
  \frac{4\|B\|_F^2 + 2\|\Gamma(B)\|_F^2}{n^2(n^2-1)}\right),\nonumber\\
 B &=& [\tilde G^{(2)}, \rho_{2,\mathfrak{g}_2}(\bx)],
\end{eqnarray}
\emph{where $\tilde G^{(2)} = d\Lambda^2(\tilde g)$, the displayed quantity is
the exact Haar value, and the implied constants are the $c$ and $C$ of
Conjecture~\ref{conj:lambda2_design}. The statement holds throughout
$k \le \gamma n$.}

\smallskip
\emph{Proof.}
As in Theorem~\ref{thm:brick-wall_bp}, $M$ is two-body, so
$\partial_\theta\calL = -i\,\tr_{\Lambda^2}\bigl(\Lambda^2(W)^\dagger
M^{(2)}\Lambda^2(W)\,B\bigr)$, with the commutator carried by the conjugated
generator $\tilde g$; the independence hypothesis is what allows $B$ to be
held fixed while the twirl acts on $M^{(2)\otimes2}$, and without it the
twirl and the conjugation are correlated. Steps~1--3 of that proof, run with
the exact Haar twirl, give $4\|B\|_F^2 + 2\|\Gamma(B)\|_F^2$ over
$\dim V_{(2,2)} = n^2(n^2-1)/4$---the relevant invariant theory being that of
$U(n)$ on $\Lambda^2$, with three irreducibles, and not that of the full group
$U(\binom{n}{2})$ of the two-particle space. This evaluation is
unconditional, and is the denominator of Eq.~\eqref{eq:scalar_comparison}.

The hypothesis then compares the butterfly pairing to it directly, which
gives~\eqref{eq:lambda2-variance} with the constants $c$ and $C$. In the
independent-halves ensemble the statement is unconditional, the composite
twirl being exactly Haar, and the variance is the $V_1$-average of the
conditional one.
$\square$
\medskip
\noindent\textbf{Claim 2 (Evaluation of the two norms).}
\emph{Write $\tau_{\tilde g} := \|[\tilde g,\Pi]\|_F^2$. For the data-encoded
input of Section~\ref{sec:encoding}, $M = n_i n_j$, any trainable gate
generator of the butterfly and every $\bx$, one has exactly}
\begin{align*}
  \|\Gamma(B)\|_F^2 &= \tfrac{(k-1)^2}{4}\,\tau_{\tilde g},
  \qquad
  \|B^{\mathrm{disc}}\|_F^2 = \tfrac{2k-1}{16}\,\tau_{\tilde g},
\end{align*}
\emph{and, in either of the two settings of Theorem~\ref{thm:sharp_rate} and
for every $k$ with $n - 2k = \Omega(n)$, $\tau_{\tilde g} =
\Theta\bigl(k(n-2k)/n^2\bigr)$, whence
$\|\Gamma(B)\|_F^2 = \Theta(k^3/n)$ and $\|B\|_F^2 = O(k^2/n)$ (pointwise
when $V_1 = I$, in expectation over $V_1$ in the independent-halves
ensemble).}

\smallskip
\emph{Proof.}
The two exact identities are Lemma~\ref{lem:tau_only} applied to
$h = \tilde g$, which uses only that $\tilde g$ is Hermitian; in particular
$\Gamma(B) = \tfrac{k-1}{2}[\tilde g,\Pi]$ by Lemma~\ref{lem:contraction}
and $\Gamma(\rho_2) = (k-1)\rho_1$, exactly as in Step~4 of
Theorem~\ref{thm:brick-wall_bp}. For a parameter of the first $R_z$ layer,
$\tilde g = g$ is coordinate-local, Step~4 of
Theorem~\ref{thm:brick-wall_bp} evaluates $\tau_{\tilde g}$, and the connected
contribution to $\|B\|_F^2$ is $O(k/n)$ by Lemma~\ref{lem:gate_locality}. For
an independent Haar prefix, Lemma~\ref{lem:conj_neutral} gives
$\E[\tau_{\tilde g}] = \Theta\bigl(k(n-2k)/n^2\bigr)$ with the same constants
and Lemma~\ref{lem:avg_locality} gives
$\E\|B^{\mathrm{conn}}\|_F^2 = O(k/n)$; the cross term is $O(k^{3/2}/n)$ by
Cauchy--Schwarz. Note that the pointwise evaluation of $\tau_{\tilde g}$ for a
shallow butterfly prefix does \emph{not} follow from either computation, which
is why Theorem~\ref{thm:sharp_rate} is stated in the two settings above.
$\square$

\medskip
\begin{proof}[Proof of Theorem~\ref{thm:sharp_rate}]
Claims~1 and~2 combine to give
\[
  \Var\!\left[\frac{\partial\calL}{\partial\theta}\right]
  = \Theta\!\left(\frac{\Theta(k^3/n) + O(k^2/n)}{n^2(n^2-1)}\right),
\]
the outer constants being the $c$ and $C$ of
Conjecture~\ref{conj:lambda2_design} for the chosen $\gamma$. Setting
$\gamma_0 := \min\{\gamma,\, 1/4\}$, for every $k \le \gamma_0 n$ one has
$n - 2k \ge n/2$ and the $\Theta(k^3/n)$ term dominates, so the right-hand
side is $\Theta(k^3/n^5)$. This grows polynomially in $k$.
\end{proof}

\subsection{Summary and Comparison}
\label{sec:bp_summary}
Both architectures rule out exponential barren plateaus for two-body correlator
readouts, and both achieve the same $\Theta(k^3/n^5)$ gradient variance rate
throughout $n - 2k = \Omega(n)$ --- polynomial in $n$ and cubic in the
particle number $k$. They differ in what is provable unconditionally:
\begin{itemize}
  \item \textbf{Brick-wall}: $\Theta(k^3/n^5)$ unconditionally and exactly,
  with no conjecture, under Haar initialization via the Givens decomposition.
  The closed form is proved for the invariant directional derivative at a Haar
  point (Theorem~\ref{thm:brick-wall_bp}); it specializes to the $n$ phases of
  the front $R_z$ layer (Corollary~\ref{cor:first_column}) and, in the
  independent-halves ensemble, to every gate of the circuit
  (Theorem~\ref{thm:interior_halves}). Whether the same statement holds
  pointwise for an interior gate of the literal single mesh is open;
  Proposition~\ref{prop:single_mesh_interior} settles it at both endpoints of
  the prefix, and no claim above depends on it.
  \item \textbf{Butterfly}: $\Omega(k^4/n^6)$ unconditionally;
  $\Theta(k^3/n^5)$ conditional on Conjecture~\ref{conj:lambda2_design}. Under
  uniform initialization.
\end{itemize}
The structure of the rate is transparent from the closed
form~\eqref{eq:var_exact}. Two Frobenius norms enter, both divided by
$\dim V_{(2,2)} = \Theta(n^4)$: the two-body norm $\|B\|_F^2$ of the
commutator of the gate generator with the encoded $2$-RDM, and the one-body
norm $\|\Gamma(B)\|_F^2$ of its one-particle contraction. The one-body term
is the larger of the two and is what produces the cubic dependence on $k$: by
Lemma~\ref{lem:contraction} it equals $(k-1)^2\|[g,\rho_1]\|_F^2$, and the
encoded $1$-RDM is $\tfrac12$ times a rank-$2k$ projector, so
$\|[g,\rho_1]\|_F^2 = \Theta(k(n-2k)/n^2)$. Both terms are proportional to the
supply of idle modes $n - 2k$: it is precisely the failure of
$\rho_1$ to be proportional to $I_n$ (Lemma~\ref{lem:spreading} equalizes only
its diagonal) that carries the trainable signal, and at a constant filling
fraction $k = cn$ with $c < 1/2$ the rate reaches its maximum $\Theta(1/n^2)$.
At the fully-packed point $k = n/2$ the idle modes are exhausted,
$\rho_1 = \tfrac12 I_n$ exactly, and~\eqref{eq:var_exact} reduces to the
connected-cumulant contribution alone, for which we prove only the upper bound
$O(1/n^4)$; the sharp rate is therefore claimed throughout
$n - 2k = \Omega(n)$ and not at exactly half filling.
The contrast with the orthogonal architectures remains sharp: for both the
orthogonal brick-wall and the orthogonal butterfly (DLA $\mathfrak{so}(n)$),
the gradient variance of $n_i n_j$ scales as
$\Theta(1/\binom{n}{k})$~\cite{monbroussou2025trainability}.
The $R_z$ gates that lift the DLA from $\mathfrak{so}(n)$ to $\mathfrak{u}(n)$
convert this catastrophic decay into the polynomial $\Theta(k^3/n^5)$---an
overwhelming improvement in the trainable signal at every $k$ that grows
with $n$; they simultaneously make one-body readouts classically tractable
while making two-body readouts the locus of trainable
signal (Proposition~\ref{prop:onebody_const}).

Two consequences are worth stating explicitly. First, trainability is
never the binding constraint in this framework: the gradient variance is
polynomial in $n$ at every particle number with $n - 2k = \Omega(n)$, from
$k = O(1)$ through the operating point of
Eq.~\eqref{eq:operating_point} to extensive filling $k = cn$, $c < 1/2$,
where it is $\Theta(1/n^2)$. The ladder of
Section~\ref{sec:k_choice} therefore concerns the training \emph{cost} and
the \emph{type} of hardness statement available, not whether the model
trains. Second, at extensive particle number $k = cn$ with $c < 1/2$ the
architecture is simultaneously trainable, with the largest gradient variance
it attains anywhere, and classically hard under the best-known-algorithm
calibration of Section~\ref{sec:k_choice}; worst-case \#P-hardness of
individual output probabilities is available at polynomial particle number
$k = n^{\epsilon}$ (Theorem~\ref{thm:worstcase}), and the full average-case
Fermion Sampling machinery at exactly half filling
(Theorem~\ref{thm:brick-wall_avgcase}), where the sharp rate degrades to the
connected-cumulant upper bound $O(1/n^4)$. This is precisely the
combination that the tension between small dynamical Lie algebras and
classical simulability~\cite{cerezo2025does} would seem to preclude, and
that the orthogonal architectures do not achieve at any $k$ that grows
with $n$.

\section{Training Algorithm: Multi-layer Parallel Parameter-Shift Rule}
\label{sec:psr}

This section proves the multi-layer parallel parameter-shift rule. The rule
applies identically to both the unitary brick-wall and the unitary butterfly:
in both architectures, gates within each layer act on disjoint qubit pairs
covering all $n$ modes, with commuting generators, and the input state has
particle number exactly $k$. These structural properties are all that the
proof requires. The
difference between the two architectures appears only in the total training
cost: the brick-wall has $O(n)$ columns and $\Theta(n^2)$ parameters, while the
butterfly has $\log n$ layers and $\Theta(n\log n)$ parameters. In both cases
the parallel PSR achieves a factor-$\Theta(n/k)$ reduction in circuit
evaluations over the naive parameter-shift rule:
\begin{itemize}
 \item \textbf{brick-wall}: naive PSR costs $4n^2$ evaluations per
gradient step ($4$ per RBS parameter, $2$ per $R_z$ parameter); parallel PSR
costs $4kn$ --- a factor of $n/k$ reduction.
\item \textbf{Butterfly}: naive PSR costs $4n\log n$ evaluations per
  gradient step; parallel PSR costs $4k\log n$ --- a factor of
  $n/k$ reduction.
\end{itemize}

Two remarks position the rule. First, hardware gradient estimation is the
operative training mode for losses that are estimated from model
samples---reinforcement-learning returns, adversarial and other
likelihood-free objectives---where the per-sample function is an arbitrary
diagonal observable and no classical route to the gradient is known short
of the sampling task itself (Section~\ref{sec:hardness}); it is likewise
the mode for training or fine-tuning the physical device against its own
noise and calibration. Second, for losses anchored on classically
estimable statistics (Section~\ref{sec:readouts}), the same gradients may
equally be computed classically; all comparisons in this section are
therefore between quantum gradient estimators---the parallel rule against
the naive parameter-shift rule---and make no claim against classical
gradient computation. The verification regime $k = O(\log n)$ of
Section~\ref{sec:k_choice} already exploits this freedom.
The algorithm is identical in both cases. The key challenge is to estimate
all gradients of a given layer simultaneously from a number of circuit
evaluations that depends only on $k$, not on $n$. Two structural properties
make this possible. First, within each layer $\ell$, all $n/2$ RBS gates act
on disjoint qubit pairs, so their generators commute: $[G_j^{(\ell)},
G_{j'}^{(\ell)}] = 0$ for $j \neq j'$; likewise the $R_z$ gates within a
layer act on disjoint single qubits. Second, the input state has particle
number \emph{exactly} $k$, which both limits the total degree of the loss as a
trigonometric polynomial in the layer parameters
(Lemma~\ref{lem:fourier_degree}) and forces its Fourier support to have even
$\ell_1$ weight (Lemma~\ref{lem:parity}, using that each layer's pairs and
phases cover all $n$ modes). Together these allow a random sign
vector $\bs$ to simultaneously probe all gradient directions in a single layer
from only $k$ shifts, with the $k$-particle Fourier structure ensuring the
estimator is exactly unbiased.

\subsection{Setup and Fourier Degree Bound}

The loss restricted to layer $\ell$ is:
\begin{equation}
  \calL(\btheta_\ell, \bphi_\ell)
  = \bra{\phi} U_\ell^\dagger(\btheta_\ell, \bphi_\ell)\,
  \tilde{M}\, U_\ell(\btheta_\ell, \bphi_\ell) \ket{\phi}
\end{equation}
where $\ket{\phi} = U^{(\ell-1)}\cdots U^{(1)}\ket{\psi(\bx)}$ is the state
entering layer $\ell$, and $\tilde{M} = (U^{(K)}\cdots U^{(\ell+1)})^\dagger
M\, U^{(K)}\cdots U^{(\ell+1)}$ is the effective observable for the
downstream circuit. This formulation is identical for brick-wall and butterfly:
$U^{(\ell)}$ is a layer of disjoint RBS-$R_z$ pairs regardless of the
architecture, and $K = O(n)$ for the brick-wall or $K = \log n$ for the
butterfly.

\begin{lemma}[Fourier degree bound]
\label{lem:fourier_degree}
For a $k$-particle input state $\ket{\phi}$, the loss $\calL(\btheta_\ell,
\bphi_\ell)$ as a function of the RBS parameters $\btheta_\ell$ (with
$\bphi_\ell$ held fixed) is a trigonometric polynomial of total degree at
most $2k$: all nonzero Fourier coefficients $c_{\boldsymbol\delta}$ in
\[
  \calL(\btheta_\ell, \bphi_\ell)
  = \sum_{\boldsymbol\delta} c_{\boldsymbol\delta}\, e^{i \boldsymbol\delta
  \cdot \btheta_\ell}
\]
satisfy $\|\boldsymbol\delta\|_1 \leq 2k$.
\end{lemma}

\begin{proof}
All gates in the downstream circuit are RBS and $R_z$, both particle-number
preserving; since $M$ is diagonal in the computational basis and therefore
particle-number preserving, the effective observable $\tilde{M}$ is also
particle-number preserving. For the $k$-particle input $\ket{\phi}$, the
expectation value $\bra{\phi} U_\ell^\dagger \tilde{M} U_\ell \ket{\phi}$
decomposes as a sum of matrix elements $\bra{S} U_\ell \ket{T}$ with
$|S| = |T| = k$.

Fix a basis state $\ket{S}$ with $|S| = k$. For each of the $n/2$ disjoint
pairs of qubits in layer $\ell$ (nearest-neighbor for the brick-wall,
stride-$2^{\ell-1}$ for the butterfly), the pair is \emph{active} on
$\ket{S}$ if exactly one qubit in the pair is occupied; inactive otherwise.
Active pairs contribute a factor of $\cos\theta_j$ or $\sin\theta_j$ (each
of trigonometric degree $1$) to the matrix element; inactive pairs contribute
a factor of $1$. The number of active pairs is at most $k$ (bounded by the
particle number), so $\bra{S} U_\ell \ket{T}$ is a trigonometric polynomial
in $\btheta_\ell$ of total degree at most $k$. The same bound applies to
$U_\ell^\dagger$, giving a bound of $2k$ on the product.
\end{proof}

The degree bound alone would call for $2k$ shifts. The following parity
property of the $k$-particle Fourier support brings this down to $k$, and is
what fixes the cost of the estimator at $2k$ circuit evaluations per gate
family per layer.

\begin{lemma}[Fourier parity]
\label{lem:parity}
Let $\ket{\phi}$ have particle number \emph{exactly} $k$ and let $\tilde{M}$
be particle-number preserving. Then every frequency vector
$\boldsymbol\delta$ occurring in the layer-$\ell$ Fourier expansion of
$\calL$ satisfies
\begin{equation}
  \|\boldsymbol\delta\|_1 \equiv 0 \pmod 2 ,
  \label{eq:parity}
\end{equation}
for the RBS angles $\btheta_\ell$ provided the pairs of layer $\ell$ cover
\emph{all $n$ modes}, and for the phases $\bphi_\ell$ provided the $R_z$
layer acts on \emph{all $n$ modes}.
\end{lemma}

\begin{proof}
\emph{RBS angles.} As in the proof of Lemma~\ref{lem:fourier_degree}, the
matrix element $\bra{S} U_\ell \ket{T}$ is nonzero only if $S$ and $T$ have
the same occupancy on every pair of layer $\ell$, and its
$\btheta_\ell$-support is exactly the set of active pairs, each contributing a
frequency $\pm 1$. Writing $a(S,T)$ for the number of active pairs and using
that the pairs of layer $\ell$ cover all $n$ modes---so that every particle
sits on exactly one pair---counting particles gives
$k = a(S,T) + 2\,\#\{\text{pairs of occupancy } 2\}$, hence
$a(S,T) \equiv k \pmod 2$. The loss is a sum of terms
$\overline{\bra{S'} U_\ell \ket{T'}}\,\tilde{M}_{S'S}\,\bra{S} U_\ell \ket{T}$,
whose frequency vector is the difference of two vectors of $\ell_1$ weights
$a(S,T)$ and $a(S',T')$; since $|x - y| \equiv x + y \pmod 2$ for integers,
$\|\boldsymbol\delta\|_1 \equiv a(S,T) + a(S',T') \equiv 2k \equiv 0$.

\emph{$R_z$ angles, full-width layer.} With $R_z$ gates on all $n$ modes,
$D(\bphi_\ell)\ket{T}
 = e^{-i\sum_{j \in T} \phi_\ell^{(j)}} \ket{T}$ up to a global phase, so the
term indexed by $(T,T')$ carries
$\delta_j = \mathbbm{1}[j \in T'] - \mathbbm{1}[j \in T] \in \{0,\pm 1\}$,
supported exactly on $T \triangle T'$. Since $|T| = |T'| = k$ we have
$|T \setminus T'| = |T' \setminus T|$, so
$\|\boldsymbol\delta\|_1 = |T \triangle T'|$ is even (and, in addition,
$\sum_j \delta_j = 0$).
\end{proof}

\begin{remark}[The full-width hypothesis is necessary]
\label{rem:fullwidth_necessary}
If the $R_z$ layer covers only a proper subset $A \subsetneq [n]$ of the
modes---for instance one phase attached to each pair---the frequency support
is $(T \triangle T') \cap A$, which carries no parity constraint: taking
$T \setminus T' = \{j\}$ with $j \in A$ and $T' \setminus T = \{j'\}$ with
$j' \notin A$ gives $\|\boldsymbol\delta\|_1 = 1$. The argument below then
fails and $2k$ shifts are genuinely required. This is why both architectures
of Section~\ref{sec:butterfly} carry a phase layer on all $n$ modes. The RBS
half of the argument carries the analogous hypothesis: the layer's pairs must
cover all $n$ modes. Every butterfly layer does, and so does every brick-wall
column, the even columns owing this to the boundary pair $(n,1)$; a layer
that leaves modes uncovered---an even column on an open chain, for
instance---admits odd-weight frequencies by the same mechanism (a particle
sitting on an uncovered mode in the bra and on a covered pair in the ket),
and its RBS gradients require the $2k$-shift rule
(Remark~\ref{rem:chain_variant}).
\end{remark}

\subsection{Main Theorem: Parallel PSR for RBS Parameters}

\begin{theorem}[Multi-layer parallel parameter-shift rule, RBS gradients]
\label{thm:parallel_psr_rbs}
Let the input state have particle number exactly $k$, let the disjoint
pairs of layer $\ell$ cover all $n$ modes, and let $a_1, \ldots,
a_{k}$ and $r_1, \ldots, r_{k}$ solve the linear system:
\begin{equation}
  \sum_{q=1}^{k} a_q \sin(r_q) \cos(r_q)^{2m-1} = \frac{1}{2},
  \quad m = 1, \ldots, k.
  \label{eq:linear_system_rbs}
\end{equation}
Let $\bs = (s_1, \ldots, s_{n/2}) \sim \mathrm{Uniform}(\{\pm 1\}^{n/2})$
and define the estimator
\begin{equation}
  \hat{g}_j(\bs) = s_j \sum_{q=1}^{k} a_q
  \left[\calL(\btheta_\ell + r_q \bs) - \calL(\btheta_\ell - r_q \bs)\right].
\end{equation}
Then for all $j \in \{1, \ldots, n/2\}$ simultaneously,
\begin{equation}
  \E_{\bs}\!\left[\hat{g}_j(\bs)\right]
  = \frac{\partial \calL}{\partial \theta_\ell^{(j)}},
\end{equation}
using $2k$ circuit evaluations in total---$k$ shifts, each evaluated at
$\pm r_q$---independently of $n$, for any particle-number-preserving
observable $M$. The covering hypothesis holds for every layer of both
architectures: every butterfly layer and every brick-wall column consists of
$n/2$ disjoint pairs covering all $n$ modes.
\end{theorem}

\begin{remark}[Closed form and cost of the shift coefficients]
\label{rem:shift_closed_form}
The shifts $r_1,\dots,r_k \in (0,\pi/2)$ are free design parameters; once they
are fixed (with distinct $\cos^2 r_q$), the coefficients are determined in
closed form by Lagrange interpolation at $u = 1$. Writing $u_q = \cos^2 r_q$
and $c_q = a_q \sin r_q \cos r_q$, the system~\eqref{eq:linear_system_rbs}
reads $\sum_q c_q u_q^{\,p} = \tfrac12$ for $p = 0,\dots,k-1$---the condition
that the nodes $\{u_q\}$ reproduce the first $k$ moments of the target measure
$\tfrac12\delta_1$---whose solution is
\begin{equation}
  a_q \;=\; \frac{1}{\sin 2r_q}\,
  \prod_{p \neq q}\frac{\sin^2 r_p}{\cos^2 r_q - \cos^2 r_p},
  \qquad q = 1,\dots,k.
  \label{eq:shift_lagrange}
\end{equation}
At $k = 1$ this is $a_1 = 1/\sin 2r_1$, and $r_1 = \pi/4$ gives $a_1 = 1$: the
familiar $\pm\pi/4$ rule for a single frequency-$2$ parameter, which is what
the parity of Lemma~\ref{lem:parity} leaves at $k = 1$. We recommend the nodes
\begin{equation}
  u_q = \tfrac12\Bigl(1 + \cos\tfrac{(2q-1)\pi}{2k}\Bigr),
  \; r_q = \arccos\sqrt{u_q},
  \; q = 1,\dots,k,
  \label{eq:shift_nodes}
\end{equation}
i.e.\ Chebyshev--Gauss points in the variable $u = \cos^2 r$, for which the
interpolation underlying~\eqref{eq:shift_lagrange} is well-conditioned and the
weights are $\ell_1$-optimal (Remark~\ref{rem:shot_budget}). Moreover $k$
shifts are \emph{necessary}: exactness on odd polynomials of degree at most
$2k-1$ is equivalent, in the variable $u$, to
$\sum_q c_q Q(u_q) = \tfrac12 Q(1)$ for every $Q$ of degree at most $k-1$, and
with only $N \le k-1$ nodes the choice $Q(u) = \prod_{q\le N}(u - u_q)$ forces
$Q(1) = 0$, i.e.\ $r_q = 0$ for some $q$, which is not a shift.

The pair $\{a_q, r_q\}$ depends only on $k$: it is independent of the circuit
parameters, the data, the layer index, the architecture and the system size
$n$. The coefficients are therefore computed \emph{once}, before training, in
$O(k^2)$ classical time, cached, and reused unchanged for every layer, every
gradient step and every data point; this preprocessing is negligible against
the per-step quantum cost.
\end{remark}

\begin{proof}
The proof proceeds in five steps.

\emph{Step 1: Reduction to the active layer.}
Freezing all layers except $\ell$ reduces $\calL$ to a function of
$\btheta_\ell$ alone, with the effective observable $\tilde{M}$ absorbing
the downstream circuit. Since all downstream gates are RBS and $R_z$ (both
particle-number preserving) and $M$ is a computational basis measurement
(particle-number preserving), $\tilde{M}$ is particle-number preserving.
No assumption on $M$ beyond particle-number preservation is needed. This
reduction applies identically to both architectures.

\emph{Step 2: Fourier structure of the loss.}
By Lemma~\ref{lem:fourier_degree}, $\calL(\btheta_\ell)$ is a trigonometric
polynomial of total degree at most $2k$:
\begin{equation}
  \calL(\btheta_\ell)
  = \sum_{\boldsymbol\delta :\, \|\boldsymbol\delta\|_1 \leq 2k}
  c_{\boldsymbol\delta}\, e^{i \boldsymbol\delta \cdot \btheta_\ell},
\end{equation}
where each $\delta_j \in \{-2, -1, 0, +1, +2\}$ (the pairwise-difference
frequencies from the eigenvalues $\{-1, 0, 0, +1\}$ of $G_j^{(\ell)}$) and
at most $2k$ coordinates are nonzero simultaneously. By
Lemma~\ref{lem:parity}---whose covering hypothesis holds here by
assumption---$\|\boldsymbol\delta\|_1$ is moreover \emph{even}
for every frequency vector occurring in this sum.

\emph{Step 3: Gradient in Fourier form.}
\begin{equation}
  \frac{\partial \calL}{\partial \theta_\ell^{(j)}}
  = i \sum_{\boldsymbol\delta :\, \|\boldsymbol\delta\|_1 \leq 2k}
  \delta_j\, c_{\boldsymbol\delta}\, e^{i \boldsymbol\delta \cdot \btheta_\ell}.
\end{equation}

\emph{Step 4: The multi-shift formula handles all cross-terms.}
We show that for every frequency vector $\boldsymbol\delta$ with
$\|\boldsymbol\delta\|_1 \leq 2k$ even and $\delta_j \neq 0$,
\begin{equation}
  2 \sum_{q=1}^{k} a_q \sin(\delta_j r_q)
  \prod_{j'' \neq j} \cos(\delta_{j''} r_q) = \delta_j.
  \label{eq:cross_term_condition}
\end{equation}
To see why this is what we need, observe that for $s_j \in \{\pm 1\}$
uniform,
\begin{equation}
  \E[s_j\, e^{i \alpha s_j}] = i \sin\alpha,
  \qquad \E[e^{i \alpha s_j}] = \cos\alpha,
\end{equation}
so by independence of the $s_{j''}$,
\begin{equation}
  \E_{\bs}\!\left[s_j \sin(r_q \boldsymbol\delta \cdot \bs)\right]
  = \sin(r_q \delta_j) \prod_{j'' \neq j} \cos(r_q \delta_{j''}).
\end{equation}
Condition~\eqref{eq:cross_term_condition} is therefore exactly what is
required for the estimator to be unbiased.

To verify~\eqref{eq:cross_term_condition}, we use the Chebyshev polynomial
representations $\sin(mr) = \sin(r)\, U_{m-1}(\cos r)$ and $\cos(mr) =
T_m(\cos r)$, where $U_{m-1}$ (second kind) and $T_m$ (first kind) satisfy
$U_{m-1}(1) = m$ and $T_m(1) = 1$. Writing
\begin{equation}
  R(x) = U_{|\delta_j| - 1}(x) \prod_{j'' \neq j} T_{|\delta_{j''}|}(x),
\end{equation}
the left-hand side of~\eqref{eq:cross_term_condition} becomes
\begin{equation}
  \mathrm{sign}(\delta_j) \cdot 2 \sum_{q=1}^{k} a_q \sin(r_q) R(\cos r_q).
\end{equation}
The polynomial $R$ has degree $(|\delta_j| - 1) + \sum_{j'' \neq j}
|\delta_{j''}| = \|\boldsymbol\delta\|_1 - 1 \leq 2k - 1$. Crucially, $T_m$
has parity $(-1)^m$ and $U_{m-1}$ has parity $(-1)^{m-1}$, so $R$ has parity
$(-1)^{\|\boldsymbol\delta\|_1 - 1} = -1$ by Lemma~\ref{lem:parity}: $R$ is an
\emph{odd} polynomial of degree at most $2k-1$, hence lies in the
$k$-dimensional span $\{x, x^3, \dots, x^{2k-1}\}$ rather than in the
$2k$-dimensional space of all polynomials of that degree.

The linear system~\eqref{eq:linear_system_rbs} is precisely the statement
that the functional $P \mapsto \sum_q a_q \sin(r_q) P(\cos r_q)$ equals
$P(1)/2$ on all \emph{odd} polynomials of degree at most $2k - 1$: it imposes
this on the monomials $x^{2m-1}$, $m = 1,\dots,k$, and the general case
follows by linearity. These are $k$ conditions in the $k$ unknowns $a_q$, and
the system is nonsingular whenever the $\cos^2 r_q$ are distinct and nonzero,
its matrix being $\mathrm{diag}(\cos r_q)$ times a Vandermonde matrix in
$\cos^2 r_q$. Applying this quadrature to $R$ gives
\begin{align}
  \mathrm{sign}(\delta_j) \cdot 2 \cdot \tfrac{R(1)}{2}
  &= \mathrm{sign}(\delta_j) \cdot U_{|\delta_j| - 1}(1)
  \prod_{j'' \neq j} T_{|\delta_{j''}|}(1) \\
  &= \mathrm{sign}(\delta_j) \cdot |\delta_j| \cdot 1 = \delta_j.
\end{align}

\emph{Step 5: Unbiasedness.}
Combining Steps 2--4:
\begin{align}
  \E_{\bs}[\hat{g}_j(\bs)]
  &= 2i \sum_{\boldsymbol\delta} c_{\boldsymbol\delta}\,
  e^{i \boldsymbol\delta \cdot \btheta_\ell}
  \sum_{q=1}^{k} a_q\, \E_{\bs}[s_j \sin(r_q \boldsymbol\delta \cdot \bs)]
  \nonumber\\
  &= 2i \sum_{\boldsymbol\delta} c_{\boldsymbol\delta}\,
  e^{i \boldsymbol\delta \cdot \btheta_\ell} \cdot \frac{\delta_j}{2} \\
  &= i \sum_{\boldsymbol\delta} \delta_j\, c_{\boldsymbol\delta}\,
  e^{i \boldsymbol\delta \cdot \btheta_\ell}
  = \frac{\partial \calL}{\partial \theta_\ell^{(j)}}.
\end{align}
The total circuit cost is $2k$: two evaluations (shifted $\pm$) per shift
$r_q$, for $k$ shifts, independently of $n$ and of the architecture.
\end{proof}

\subsection{Parallel PSR for $R_z$ Parameters}
\label{sec:rz_psr}

The $R_z$ gate generators have spectrum $\{-\tfrac{1}{2}, +\tfrac{1}{2}\}$,
so the Fourier degree of the loss in each $\phi_\ell^{(j)}$ individually is
only $1$, with per-gate difference frequencies $\delta_j \in \{-1, 0, +1\}$.
The joint degree is nevertheless as large as in the RBS case: in the
interference terms $\bra{S}\cdots\ket{S'}$ of the loss the bra and ket Fock
components can occupy \emph{disjoint} sets of $R_z$-active qubits, so the
frequency vector satisfies only $\|\boldsymbol\delta\|_1 \le 2k$ (each of bra
and ket contributes up to $k$ occupied $R_z$ qubits), and
$\calL(\bphi_\ell + r \cdot \bs^\phi)$ has degree up to $2k$ in the scalar
shift $r$. What rescues the shift count is therefore not the degree but the
parity: when the phase layer covers all $n$ modes,
Lemma~\ref{lem:parity} forces $\|\boldsymbol\delta\|_1$ to be even, the
polynomial $R$ of Step~4 is odd, and $k$ shifts suffice---exactly as for the
RBS parameters. The full-width phase layers of
Section~\ref{sec:butterfly} are chosen for exactly this reason; a partial
layer breaks the argument (Remark~\ref{rem:fullwidth_necessary}). The $R_z$
rule is then a direct corollary of the machinery of
Theorem~\ref{thm:parallel_psr_rbs}.

\begin{theorem}[Multi-layer parallel parameter-shift rule, $R_z$ gradients]
\label{thm:parallel_psr_rz}
Let the input state have particle number exactly $k$, let the $R_z$ layer act
on all $n$ modes, and let $a_1, \ldots, a_{k}$ and $r_1, \ldots, r_{k}$ solve
the linear system~\eqref{eq:linear_system_rbs}. Let
$\bs^\phi \sim \mathrm{Uniform}(\{\pm 1\}^{n})$ and define
\begin{equation}
  \hat{h}_j(\bs^\phi) = s_j^\phi \sum_{q=1}^{k} a_q
  \left[\calL(\bphi_\ell + r_q \bs^\phi)
  - \calL(\bphi_\ell - r_q \bs^\phi)\right].
  \label{eq:rz_estimator}
\end{equation}
Then for all $j \in \{1, \ldots, n\}$ simultaneously,
\begin{equation}
  \E_{\bs^\phi}\!\left[\hat{h}_j(\bs^\phi)\right]
  = \frac{\partial \calL}{\partial \phi_\ell^{(j)}},
\end{equation}
using $2k$ circuit evaluations in total, independently of $n$ and of the
architecture.
\end{theorem}

The proof is identical to that of Theorem~\ref{thm:parallel_psr_rbs}: the
frequencies satisfy $\delta_j \in \{-1, 0, +1\}$ and, by
Lemma~\ref{lem:parity}, $\|\boldsymbol\delta\|_1 = 2m$ is even with
$m \le k$; the polynomial $R(x)$ of Step~4 is therefore the odd monomial
$U_0(x)\prod_{j''} T_1(x) = x^{2m-1}$, and the
quadrature~\eqref{eq:linear_system_rbs} reproduces $P(1)/2$ on all odd
polynomials of degree at most $2k-1$. The same nodes and weights
$\{a_q, r_q\}$ as for the RBS gradients are reused unchanged.

\subsection{Total Training Cost}

\begin{corollary}[Full gradient cost and factor-$n/k$ reduction]
\label{cor:full_gradient}
Applying Theorem~\ref{thm:parallel_psr_rbs} to each layer with independent
sign vectors and Theorem~\ref{thm:parallel_psr_rz} analogously, the full
gradient of the loss with respect to all trainable parameters is estimated
with the following costs:
\begin{itemize}
  \item \textbf{Butterfly} ($K = \log n$ layers,
  $T = \tfrac{3}{2}n\log n$ parameters): $2k\log n$ circuit evaluations for
  all RBS gradients and $2k\log n$ for all $R_z$ gradients, giving
  $4k\log n$ total. The naive parameter-shift rule requires $4n\log n$
  evaluations ($4$ per RBS parameter and $2$ per $R_z$ parameter); the
  parallel PSR achieves a reduction by a factor of $n/k$.
   \item \textbf{Brick-wall} ($K = n$ columns, each a full-width $R_z$ layer
followed by that column's $n/2$ RBS gates, $T = 3n^2/2$ parameters): $2kn$
circuit evaluations for all RBS gradients and $2kn$ for all $R_z$ gradients,
giving $4kn$ total. The naive parameter-shift rule requires
$4 \cdot n^2/2 + 2 \cdot n^2 = 4n^2$
evaluations; the parallel PSR achieves a reduction by a factor of exactly
$n/k$, the same as the butterfly's.
\end{itemize}
In both cases the reduction factor is exactly $n/k$: $\Theta(n)$
at $k = O(1)$, growing linearly in $n$ at any fixed $k$. At the operating
point of Eq.~\eqref{eq:operating_point}, the butterfly requires
$400\log n$ circuit evaluations per gradient step---nearly independent of
$n$---while the brick-wall requires $4kn = O(kn)$.
\end{corollary}

Per layer the naive rule costs $4 \cdot \tfrac n2 + 2 \cdot n = 4n$
evaluations against the parallel rule's $2k + 2k = 4k$, so the parallel rule
is cheaper precisely when $k < n$. Since the encoding requires $2k \le n$,
the parallel rule is never worse than the naive one anywhere in the
admissible range. (Taking the estimator layer-wise as
$\min(\text{naive},\text{parallel})$ makes this monotonicity hold by
definition, at no cost.)

For $n = 1024$: at the operating point $k = 100$, the butterfly
requires $4{,}000$ circuit evaluations per gradient step (vs.\
$40{,}960$ for the naive PSR, a $10.2\times$ reduction) and the
brick-wall requires $409{,}600$ (vs.\ $4{,}194{,}304$, also
$10.2\times$). At $n = 4096$ the reduction is $41.0\times$: the
advantage of the parallel PSR grows linearly with $n$ at fixed hardness
level. The naive baseline here is computed for the \emph{same} circuit.
Since the parallel cost is independent of the parameter count while the naive
cost is not, the ratio can be inflated by adding parameters, whereas the
absolute per-step cost---$4k\log n$ and $4kn$---can only rise; it is the
absolute cost that we report throughout. The resource counted is the number
of distinct circuit configurations per gradient step: shots of a fixed
configuration are repetitions of a single compiled and uploaded circuit,
whereas each configuration must be compiled, uploaded and dispatched
separately, so it is the configuration count, rather than the shot budget
per configuration, that limits on-hardware training at scale
(Remark~\ref{rem:shot_budget}).

\begin{remark}[Configurations versus shots]
\label{rem:shot_budget}
At the nodes~\eqref{eq:shift_nodes} the weights have the closed form
$a_q = (-1)^{q-1}\bigl/\bigl(2k\sin^2(\theta_q/2)\bigr)$ with
$\theta_q = (2q-1)\pi/(2k)$, whence $\sum_q|a_q| = k$ and
$\sum_q a_q^2 = \tfrac13(2k^2+1)$. The second sum sets the price of
parallelism: with $S$ shots per configuration and a loss evaluation of
variance $\sigma^2/S$, the parallel estimator carries shot noise
$\tfrac43 k^2\sigma^2/S\,(1+o(1))$ against $O(\sigma^2/S)$ for the naive rule.
At a fixed total shot budget per gradient step this is largely bought back:
the parallel rule spreads the budget over $\sim\!k/n$ as many configurations,
so each receives $\sim\!n/k$ times more shots, leaving a net variance ratio
$\Theta(k^3/n)$. Since the estimator is exactly unbiased at every shot count
(Theorem~\ref{thm:parallel_psr_rbs}), the residual noise enters stochastic
optimization as variance rather than bias, alongside the mini-batch noise
already present. Trading configurations for shots at comparable order is
therefore favourable on hardware, where shots of a fixed configuration are
repetitions of a single compiled and uploaded circuit, dispatched in one job,
while every configuration must be compiled, uploaded and dispatched
separately.
\end{remark}

\subsection{Layer-wise Block Coordinate Descent}

The multi-layer parallel PSR enables two optimization strategies for both
architectures.

\paragraph{Standard gradient descent.}
Use Corollary~\ref{cor:full_gradient} to compute all gradients in $4kK$
circuit evaluations ($K = \log n$ for butterfly, $K = n$ for brick-wall),
then update all parameters simultaneously.

\paragraph{Layer-wise block coordinate descent.}
Alternatively, update layers one at a time in forward-backward sweeps,
a strategy analogous to the sweep structure of DMRG in quantum many-body
physics~\cite{white1992dmrg} and a special case of block coordinate
descent~\cite{nesterov2012efficiency}:
\begin{eqnarray*}
  \text{for } \ell &=& 1, \ldots, K, K-1, \ldots, 1: \\
  \btheta_\ell &\leftarrow& \btheta_\ell - \eta\, \hat{\nabla}_{\btheta_\ell}
  \calL, \quad
  \bphi_\ell \leftarrow \bphi_\ell - \eta\, \hat{\nabla}_{\bphi_\ell} \calL,
\end{eqnarray*}
using $4k$ circuit evaluations per layer update ($2k$ for RBS, $2k$ for
$R_z$), independently of the architecture. The cost per full forward-backward
sweep is $4k \cdot 2K$: for the butterfly this is $O(k\log n)$, and for the
brick-wall this is $O(kn)$---identical in order to full gradient descent in both
cases. The potential advantage is that each layer's subproblem has reduced
effective dimension: when training layer $\ell$ alone, the parameter space
has dimension $n$ (one RBS and one $R_z$ parameter per pair) rather than $T$.
Whether this translates to faster convergence in practice depends on the loss
landscape and is an empirical question we leave for future work.

\section{The QNN as an End-to-End ML Pipeline}
\label{sec:pipeline}

We have built three ingredients: a magic-state encoder
(Section~\ref{sec:encoding}), unitary brick-wall and butterfly processors
(Section~\ref{sec:butterfly}) that are provably trainable
(Section~\ref{sec:bp}) with efficient training cost
(Section~\ref{sec:psr}). In this section we assemble these into complete
machine-learning pipelines by specifying the readout (how the quantum state
is turned into a classical output) and the loss function, and we discuss
which tasks each pipeline variant is suited to. The two architectures share
the same encoding, readout families, and task structure; they differ only in
circuit depth, parameter count, and training cost. The classical simulation
cost of each variant is analyzed in Section~\ref{sec:hardness}.

\subsection{Pipeline Stages}

The full inference pipeline consists of five stages, summarized in
Table~\ref{tab:pipeline}. The first four stages build the quantum state
$\ket{\psi(\bx; \btheta, \bphi)}$ that encodes the classical input $\bx$
under the trained parameters $(\btheta, \bphi)$; the fifth stage extracts
a classical output from this state via measurement.

\begin{table}[h]
\centering
\small
\caption{The five stages of the unitary brick-wall and butterfly QNN pipelines.
All magic blocks are prepared in parallel
in Stage~1, so the preparation depth does not scale with $k$. The total
gate-layer count is $3\log n + 6$ for the
butterfly, and $3n + 6$ for the brick-wall.}
\label{tab:pipeline}
\begin{tabular}{@{}clcc@{}}
\toprule
\# & Stage & Butterfly  & Brick-wall \\
\midrule
1 & Magic state preparation & $4$ & $4$ \\
2 & Fixed Hadamard spreading $U^{\mathrm{had}}$ & $\log n$ & $n$ \\
3 & $R_z$ data encoding $D(\bvarphi(\bx))$ & $1$ & $1$
 \\
4 & Trainable circuit $U(\btheta, \bphi)$ & $2\log n$ & $2n$
 \\
5 & Measurement + readout & $1$ & $1$
\\
\bottomrule
\end{tabular}
\end{table}

All stages share the same two gate primitives---RBS and $R_z$---both
particle-number preserving. Particle number $k$ is therefore preserved
exactly through the full circuit in both architectures, and the overall
circuit $V = U \circ D \circ U^{\mathrm{had}}$ is a passive FLO unitary
acting on the non-Gaussian magic state $\ket{\Psi_{\mathrm{in}}}$.
This structural coherence is what
enables simultaneously:
\begin{itemize}
  \item the multi-layer parallel parameter-shift rule
  (Lemma~\ref{lem:fourier_degree} requires particle-number preservation),
  \item the no-barren-plateau results: for the brick-wall,
  the closed-form gradient variance~\eqref{eq:var_exact} under Haar
  initialization, evaluating to $\Theta(k^3/n^5)$
  (Theorem~\ref{thm:brick-wall_bp}); for the butterfly, unconditional
  $\Omega(k^4/n^6)$ and conditional $\Theta(k^3/n^5)$ under
  Conjecture~\ref{conj:lambda2_design} (Theorems~\ref{thm:no_bp}
  and~\ref{thm:sharp_rate}); one-body readouts classically tractable for both
  (Proposition~\ref{prop:onebody_const}),
  \item the classical hardness results of Section~\ref{sec:hardness}, which
  require non-Gaussian input to a passive FLO circuit.
\end{itemize}

\subsection{The Sampler as the Universal Readout}
\label{sec:readouts}

Stage~5 is the same for every task in this paper: the quantum device
performs computational-basis measurements and returns $m$ shots, each a bit
string of Hamming weight $k$ (by particle-number preservation). Every
pipeline consumes these samples and nothing else; the QNN is always used as
a sampler of $p(\bx' \mid \bx) = |\bra{\bx'} V \ket{\Psi_{\mathrm{in}}}|^2$,
and the classically hard object of Section~\ref{sec:hardness}---the sampling
task itself---therefore sits at the base of every task family. Tasks differ
only in the loss applied to the samples and in the classical head that
processes them; both apply identically to the brick-wall and the butterfly.

The ways of consuming the samples form a hierarchy of increasing generality.
\emph{Distributional losses} compare the multi-set of sampled strings
directly to data through a divergence such as maximum mean discrepancy or
KL, forming no summary statistics. \emph{Few-body sample statistics} reduce
the shots to a vector of empirical averages of observables of fixed body
number, for example the two-body correlator vector
$\mathbf{q} = (\langle n_i n_j \rangle)_{i<j}$, obtained by counting the
fraction of shots in which modes $i$ and $j$ are both occupied; these are
the simplest permutation-invariant functions of the sample set and the
family for which trainability is proven (Section~\ref{sec:bp}).
\emph{General set-valued heads} apply an arbitrary permutation-invariant
classical network to the shots---a simple linear head on $\mathbf{q}$, a
DeepSets aggregator or a set transformer whose tokens are the sampled
strings---and are sensitive to sample
structure beyond few-body statistics.

Two objects must be kept apart throughout what follows: the \emph{loss},
which is a scalar functional used to steer the parameters during training,
and the \emph{deployed output}, which is what the trained model emits at
inference. The hardness results of Section~\ref{sec:hardness} attach to the
second, at the level of the sampling task. The two can therefore have different
complexities, and in the pipelines below they generally do: by
Proposition~\ref{prop:onebody_const} every statistic of fixed body number is
classically computable, so a model whose deployed output is a function of
$\mathbf{q}$ alone is classically simulable end to end, whereas a model whose
deployed output is the sample set itself is
not---knowledge of $\mathbf{q}$ does not permit sampling, since
exponentially many distributions on Hamming-weight-$k$ strings share the
same two-body moments.

\subsection{Where the Computation Happens}
\label{sec:where}

Sections~\ref{sec:bp} and~\ref{sec:psr} established a gradient rate and a
gradient cost for the two-body correlator. Which pipelines those results
actually cover is settled by two independent questions. Does a given loss
inherit the variance guarantee? That is decided by the body number of the
effective observable it induces. Can its gradient be computed classically?
That is decided by whether the expectation defining it is taken over a
distribution we have a model of. Keeping the two apart organizes the design
space; the concrete task families are catalogued in
Section~\ref{sec:tasks}.

One preliminary applies to every case below. For a per-sample head $h$ the
model output is the expectation $\E_p[h]$ of the diagonal observable $h$;
diagonal observables commute with the particle number, and the multi-layer
parallel parameter-shift rule (Theorem~\ref{thm:parallel_psr_rbs}) requires
nothing of the observable beyond particle-number preservation. The gradient
\emph{cost} on hardware is therefore the same for every head: $4kn$ circuit
evaluations per step for the brick-wall, $4k\log n$ for the butterfly. What
varies is the gradient \emph{variance}, and it is controlled by the body
number of the effective observable that the loss induces.

\paragraph{Losses anchored on few-body statistics.}
The canonical distributional loss for an implicit generative model is
moment matching: the model's two-body moments
$q_{ij} = \langle n_i n_j \rangle$ are compared to the data's, and the
squared discrepancy is minimized. Both of our questions have favourable
answers for such a loss. By the chain rule its gradient is that of a single
two-body effective observable, so the rate of Section~\ref{sec:bp} transfers
verbatim---provided the induced observable is not degenerate, which would
make the guarantee vacuous. It is not, and degeneracy occurs only at
convergence.

\begin{proposition}[Moment-matching MMD induces a nondegenerate two-body head]
\label{prop:mmd_trainable}
Let $q_{ij}(W) = \langle n_i n_j \rangle$ be the model's two-body moments,
$\mu_{ij}$ the corresponding data moments, both on Hamming-weight-$k$
strings, and $\calL_{\mathrm{MMD}} = \sum_{i<j}(q_{ij} - \mu_{ij})^2$. Then
at every parameter point
\begin{equation}
  \partial_\theta \calL_{\mathrm{MMD}}
  = \partial_\theta \langle M_{\mathrm{eff}} \rangle,
  \qquad
  M_{\mathrm{eff}} = 2\sum_{i<j} \bigl(q_{ij} - \mu_{ij}\bigr)\, n_i n_j,
\end{equation}
a two-body, particle-number-preserving observable, whose traceless part
satisfies $\tilde M_{\mathrm{eff}} = 0$ if and only if $q = \mu$.
\end{proposition}

\begin{proof}
The identity is the chain rule together with linearity of each $q_{ij}$ in
the observable $n_i n_j$; the residual is evaluated at the current parameter
point, and $M_{\mathrm{eff}}$ is a weighted two-body correlator, diagonal in
the Fock basis and hence particle-number preserving, so
Theorem~\ref{thm:parallel_psr_rbs} and
Proposition~\ref{prop:general_readout} apply to it verbatim. For the
degeneracy claim: $\tilde M_w = 0$ on $\Lambda^2\mathbb{C}^n$ if and only
if all weights $w_{ij}$ are equal. Since both model and data strings carry
exactly $k$ particles, $\sum_{i<j} q_{ij} = \sum_{i<j} \mu_{ij} =
\binom{k}{2}$, so the residual sums to zero; equal residuals are therefore
zero residuals, and $\tilde M_{\mathrm{eff}} = 0$ forces $q = \mu$.
\end{proof}

In applying the rate, two objects must be kept apart: the weights of
$M_{\mathrm{eff}}$ depend on the current parameter point through the
residual, whereas the variance statements of Section~\ref{sec:bp} are for a
\emph{fixed} observable averaged over the Haar ensemble. The trainability
statement for the MMD loss is therefore conditional on the point. Fix any
parameter point, write $r := q - \mu$ for its residual, and let $\calL_r$
denote the fixed-head loss with observable
$M_r = 2\sum_{i<j} r_{ij}\, n_i n_j$; by
Proposition~\ref{prop:mmd_trainable}, the gradient of
$\calL_{\mathrm{MMD}}$ at that point equals the gradient of $\calL_r$ at
that point. For the fixed head $M_r$, whenever $r \neq 0$---equivalently
$\tilde M_r \neq 0$---Proposition~\ref{prop:general_readout}(iii) gives,
for $k \ge 2$ with $n - 2k = \Omega(n)$,
$\Var[D_g \calL_r] = \Omega\bigl(\|\tilde M_r\|_F^2\, k^2/n^5\bigr)$ over
the Haar ensemble. The MMD landscape therefore carries no exponential
barren plateau in the following precise sense: away from the matched point
$q = \mu$, its gradient coincides pointwise with that of a two-body
fixed-head loss whose ensemble gradient variance is at most polynomially
small.

The class of losses covered is neither artificial nor exhaustive.
Moment-matching objectives of this kind are the standard first choice for
implicit generative models, and in the multivariate setting they amount to
matching conditional pairwise co-occurrence structure---the quantity that
downstream portfolio and risk applications consume. What falls outside the
class is equally concrete: kernels whose Mercer features run to unbounded
order, such as the Gaussian kernel commonly used with MMD, adversarial
objectives evaluated by a trained discriminator, and generic
reinforcement-learning returns. These are the subject of the next two
paragraphs.

Turning to the second question: since the encoding and the circuit are
known, the required statistics are obtained by propagating the corresponding
reduced density matrix (Proposition~\ref{prop:onebody_const}), so gradients
may be computed classically and exactly rather than estimated from shots.
Training of this kind needs no quantum device at all; the device is required
only at inference, to draw samples from the trained distribution. Whether
classical propagation or hardware estimation is faster in practice is an
engineering comparison---the former costs $\poly(n)$ per gradient evaluation
with a sizeable exponent, the latter $O(kn)$ or $O(k\log n)$ shot-limited
circuit evaluations---and not a complexity statement.

\paragraph{Losses estimated from model samples.}
When the expectation defining the loss is taken over something we do not
have a model of, no propagation is available and the gradient must be
estimated by drawing samples; the device then enters the training loop. Two
distinct situations fall under this heading, and they differ in whether a
variance guarantee survives.

The first is an external environment. In reinforcement learning the return
depends on how the environment responds to the sampled action, and for the
environments that motivate model-free methods in the first place no
sufficiently accurate model is available to propagate through. Sampling from
the policy is then unavoidable. The variance guarantee, however, is not
lost: it requires the effective observable to be structurally of low body
number, not to have known coefficients. A return of the form
$R(s, a) = \sum_i r_i(s)\, a_i + \sum_{i<j} r_{ij}(s)\, a_i a_j$---the
shape of a portfolio payoff with linear rewards and quadratic costs, and of
shaped rewards generally---induces a two-body effective observable whose
weights are fixed by the environment and unknown to us, which is exactly the
hypothesis under which Proposition~\ref{prop:general_readout}(iii) applies;
the statement holds conditionally on each environment state $s$, which fixes
the weights; the state-averaged loss is by linearity again a two-body head,
with weights $\E_s[r_{ij}(s)]$, to which the same bound applies.
Such pipelines are trained on hardware and carry the $\Theta(k^3/n^5)$
guarantee.

The second is a high-order functional of the model's own samples: a deep
discriminator applied to the shots, an indicator of a rare configuration, or
a sparse terminal return that does not decompose into few-body terms. Here
the effective observable has body number growing with $k$, the analysis of
Section~\ref{sec:bp} does not apply, and gradients may concentrate; nothing
in this paper rules that out. These pipelines require the device and carry
no variance guarantee.

\paragraph{Composition.}
The two regimes compose. Writing $\calL = \calL_{\mathbf{q}} + \lambda\,
\calL_{\mathrm{sample}}$ with an anchored first term guarantees that the
total gradient retains a non-vanishing contribution regardless of how flat
the second term is, and separating the two contributions during training
diagnoses whether the sample-based term carries signal above shot noise. In
summary, gradients are provably non-vanishing whenever the effective
observable is of low body number, which covers returns and heads of the
usual form; such gradients are estimated on the device when no model of the
distribution is available, as in model-free reinforcement learning, and can
also be computed classically with polynomial overhead when one is.
Supervised pipelines fall on either side of this division according to the
head: one reading only $\mathbf{q}$ is classical end to end, while a
set-valued head consuming the samples themselves is not
(Section~\ref{sec:tasks}).

\subsection{Machine Learning Task Families}
\label{sec:tasks}

All task families use the magic-block encoding of
Section~\ref{sec:encoding} at the operating point of
Eq.~\eqref{eq:operating_point} ($k = 100$, paired blocks,
Section~\ref{sec:psi4_encoding}) and consume the same sampler; they differ
only in loss and head, and each is placed in the taxonomy of
Section~\ref{sec:where} by that choice.

\paragraph{Generative modeling (distributional loss).}
Train the QNN as a Born machine: given data samples $\{\bx^{(1)}_{\mathrm{data}},
\ldots, \bx^{(N)}_{\mathrm{data}}\}$, minimize a sample-based distributional
loss such as maximum mean discrepancy (MMD) or KL divergence against the
empirical data distribution. The QNN's output distribution $p(\cdot \mid
\bx)$ itself is the model. Related architectures have been explored in
quantum Born machines~\cite{coyle2020born} and quantum
GANs~\cite{lloyd2018qgan}.

\paragraph{Reinforcement learning policies (distributional loss).}
Use $p(\cdot \mid \bx)$ as a policy: given an environment state $\bx$, the
QNN outputs a distribution over discrete actions. Train via policy gradient.
Policy-gradient training consumes model samples inside the loop: the return
is an arbitrary function of the sampled string, so the gradient is
estimated on hardware via the parallel parameter-shift rule, and this task
family is quantum at training time as well as at inference. The efficient
training cost of the parallel PSR combines naturally with
episodic policy updates for both architectures.

\paragraph{Supervised classification (sample-set head).}
For a classification task with $C$ classes, define logits
\begin{equation}
  y_c(\bx; \btheta, \bphi, W, b)
  = \sum_{i < j} W_{c, ij} \langle n_i n_j \rangle(\bx; \btheta, \bphi)
  + b_c, \; c \in [C],
\end{equation}
where $W \in \R^{C \times \binom{n}{2}}$ and $b \in \R^C$ are classical
parameters, trained alongside the quantum parameters $(\btheta, \bphi)$.
Apply softmax to obtain class probabilities, and train against labels with
cross-entropy. This is a hybrid architecture: the quantum layer supplies the
two-body correlator vector and a classical head maps it to predictions. In
practice one can restrict to a subset of $n$ two-body operators (e.g.,
nearest-neighbor pairs together with one long-range pair) to obtain an
$n \to n$ head. The linear head may equally be replaced by any
permutation-invariant model consuming the shots directly, with the
consequences set out in Section~\ref{sec:where}. One point is specific to
the head: it changes the gradient \emph{rate} and not merely its constant,
entering through the two scalars $c_1$ and $c_2$ of
Proposition~\ref{prop:general_readout}, which can differ in order and in
sign---$\Theta(k^3/n^5)$ for the single-pair readout against
$\Theta(k^2/n^4)$ for the uniform nearest-neighbour head---while the absence
of an exponential barren plateau is head-independent
(Proposition~\ref{prop:general_readout}(iii)).

\paragraph{Supervised regression (sample-set head).}
For scalar regression, predict $\hat{y}(\bx) = \sum_{i < j} w_{ij} \langle
n_i n_j \rangle + b$ and train with mean-squared error. Vector-valued
regression follows the classification template.

\paragraph{Kernel and similarity tasks (overlap readout).}
The encoding $\bx \mapsto \ket{\psi(\bx)}$ defines an implicit quantum
kernel $K(\bx, \bx') = |\langle \psi(\bx) | \psi(\bx') \rangle|^2$.
Overlaps can be estimated via SWAP tests or destructive SWAP
measurements~\cite{schuld2021supervised}, enabling kernel-based
classification, clustering, or nearest-neighbor methods. The magic-block
encoding supports this readout directly.

In every case the classical simulation cost of the end-to-end model is
governed by what the head consumes: few-body sample statistics are
estimable at additive precision, while sample-level structure is accessible
to a classical simulator only through the sampling algorithms of
Section~\ref{sec:hardness}, at best-known cost $2^{\Omega(k)}$.

\subsection{Sample Complexity and Generalization}

Although the QNN operates in a $\binom{n}{k}$-dimensional Hilbert space that
is superpolynomial in $n$, it does not require superpolynomially many
training samples. The generalization error depends on the number of trainable
parameters $T$, not on the Hilbert space dimension.

\begin{theorem}[Generalization bound]
\label{thm:generalization}
Let $\calF$ be the function class implemented by the unitary brick-wall or
butterfly QNN with $T$ trainable parameters ($T = 3n^2/2$ for the
brick-wall, $T = \tfrac32 n\log n$ for the butterfly), with the magic-block
encoding of Section~\ref{sec:encoding}. With $m$ training samples, the generalization error satisfies
\begin{equation}
  \E\!\left[\sup_{f \in \calF}
  |\calL_{\mathrm{test}}(f) - \calL_{\mathrm{train}}(f)|\right]
  \leq O\!\left(\sqrt{\frac{T \log m}{m}}\right).
\end{equation}
For $\varepsilon$-generalization it suffices to take
$m = \Omega\bigl(T \log(T/\varepsilon) / \varepsilon^2\bigr)$ samples.
\end{theorem}

\begin{proof}
We apply Theorem~1 of Caro et al.~\cite{caro2022generalization} to the QNN
with $T$ trainable parameters.

\textit{Per-gate Lipschitz constant.}
Each parameter $\theta_i$ appears in exactly one gate $U_i(\theta_i) =
\exp(i\theta_i G_i)$, where $G_i$ is a traceless Hermitian generator with
$\|G_i\|_{\mathrm{op}} \leq 1$. For any normalized input state $\ket{\psi}$
and observable $O$ with $\|O\| \leq 1$, the parameter-shift rule gives
\[
\left|\frac{\partial \calL}{\partial \theta_i}\right|
  = \bigl|\langle\psi|\, i[G_i,\tilde O]\,|\psi\rangle\bigr|
  \leq 2\|G_i\|_{\mathrm{op}}\,\|O\| = O(1),
\]
since all circuit outputs are bounded in $[-1, 1]$. Hence the loss has
per-gate Lipschitz constant $O(1)$.

\textit{Covering number.}
An $\varepsilon$-grid of spacing $\varepsilon$ in each coordinate of the
parameter cube $[0, 2\pi]^T$ induces an $O(\varepsilon)$-cover of $\calF$
in $L^\infty$. The number of grid points is $(2\pi/\varepsilon)^T$, giving
$\log \calN(\calF, \varepsilon) = O(T \log(1/\varepsilon))$.

\textit{Rademacher complexity bound.}
Substituting into the standard covering-number Rademacher bound with
$\varepsilon \sim 1/\sqrt{m}$:
\begin{equation}
  \E\!\left[\sup_{f \in \calF}|\calL_{\mathrm{test}}(f)
  - \calL_{\mathrm{train}}(f)|\right]
  \leq O\!\left(\sqrt{\frac{T \log m}{m}}\right).
\end{equation}
Setting the right-hand side to $\varepsilon$ and solving gives
$m = O(T \log(T/\varepsilon)/\varepsilon^2)$.
\end{proof}

\section{Classical Simulation Hardness}
\label{sec:hardness}

This section establishes the classical simulation hardness of the pipeline.
The hardness results are expressivity statements: the function classes
accessible to the unitary brick-wall and butterfly QNNs contain output
distributions that cannot be efficiently
sampled by classical algorithms, at a hardness level controlled
by the particle number $k$. This expressivity is structural---it
follows from the non-Gaussianity of the magic-state encoding, the
number of magic blocks, and the
passive FLO structure of the circuit---and does not depend on which specific
parameter values are found by training.

Because the hardness depends on the particle
number $k$, we organize the results as a ladder
in $k$ (cf.\ Section~\ref{sec:k_choice}), instantiated at the operating
point of Eq.~\eqref{eq:operating_point}. Four levels emerge:
\begin{itemize}
  \item \emph{$k = O(\log n)$: no hardness.} The FLO extent is polynomial
  and the entire pipeline---amplitudes and sampling included---is
  classically simulable in polynomial time. 
  \item \emph{$k = \omega(\log n)$: hardness under best-known algorithms.}
  Sampling costs $2^{\Omega(k)}$ under the best-known classical
  algorithms (Theorem~\ref{thm:sampling_bestknown}). At the operating point
  $k = 100$ this is $2^{50} \approx 1.1 \times 10^{15}$ operations per
  sample; at $k = \log^2 n$ it is $n^{\Omega(\log n)}$.
  \item \emph{$k = n^{\epsilon}$: worst-case \#P-hardness.} Exact
  computation of output probabilities is \#P-hard for the brick-wall
  (Theorem~\ref{thm:worstcase}); for the butterfly the same statement holds
  under a reachability assumption
  (Remark~\ref{rem:butterfly_reachability}).
  \item \emph{$k = \Theta(n)$: average-case hardness.} The fully-packed
  limit inherits the complete Fermion Sampling
  machinery---anticoncentration and the worst-to-average-case Cayley-path
  reduction---for the Haar-initialized brick-wall
  (Theorem~\ref{thm:brick-wall_avgcase}). 
\end{itemize}

We note that average-case hardness holds as of now only
at the top rung, because both of its ingredients---anticoncentration and the
worst-to-average-case reduction---are currently proven only for fully-packed
magic inputs. How far down the ladder they survive is precisely the
``fermion sampling with less magic'' open problem
of~\cite[Sec.~IV]{oszmaniec2022fermion}; we state our sharpened form of it as
Conjecture~\ref{conj:sparse_avgcase} in
Section~\ref{sec:brick-wall_avg_case}.

Orthogonal to the ladder, \emph{expectation-value readouts of fixed body
number} are classically easy at any $k$, via efficient $r$-RDM propagation
(Proposition~\ref{prop:onebody_const}); the ladder therefore concerns
sampling alone.

All costs are governed by the FLO extent of the
input state (Definition~\ref{def:flo_extent}), a quantity that measures
non-Gaussianity, which we analyze first. We close with a direct engagement
with the Cerezo et
al.~\cite{cerezo2025does} simulability framework
(Section~\ref{sec:cerezo_response}).

\subsection{The Fermion Sampling Structure}
\label{sec:fermion_sampling}

The full inference circuit (after training) is
\begin{equation}
  V(\btheta, \bphi, \bx)
  = U(\btheta, \bphi) \cdot D(\bvarphi(\bx)) \cdot U^{\mathrm{had}}
\end{equation}
applied to $\ket{\Psi_{\mathrm{in}}}$, where $U(\btheta, \bphi)$ is either
the unitary brick-wall or butterfly, $U^{\mathrm{had}}$ is the corresponding
fixed Hadamard spreading circuit, and $\ket{\Psi_{\mathrm{in}}}$ is
the paired magic state of Section~\ref{sec:encoding}. Each
component is a passive FLO transformation: $U(\btheta, \bphi)$ and
$U^{\mathrm{had}}$ are products of RBS and $R_z$ gates, both
particle-number preserving; $D(\bvarphi(\bx)) = \bigotimes_i R_z(\varphi_i)$
is a diagonal one-body unitary. Their product $V$ is therefore a passive FLO
transformation with single-particle unitary $W \in U(n)$. The output
distribution is
\begin{equation}
  p(\bx' \mid \bx)
  = |\bra{\bx'} V(\btheta, \bphi, \bx) \ket{\Psi_{\mathrm{in}}}|^2,
\end{equation}
where $\ket{\Psi_{\mathrm{in}}}$ is the non-Gaussian magic state. This is
exactly the Fermion Sampling structure of Oszmaniec et
al.~\cite{oszmaniec2022fermion}: a passive FLO unitary acting on a
non-Gaussian fermionic input. The best-known-algorithm hardness results
below apply to both
architectures since they depend only on this Fermion Sampling structure, not
on the specific connectivity pattern of the passive FLO circuit; the
\#P-hardness and average-case statements additionally require that the hard
single-particle unitaries be \emph{reachable}, which holds unconditionally
for the brick-wall (surjective onto $U(n)$) and under a reachability
assumption for the fixed-depth butterfly
(Remark~\ref{rem:butterfly_reachability}).

\subsection{FLO-Extent-Based Classical Simulation}
\label{sec:flo_simulation}

For every readout family, the relevant classical simulation algorithms are
controlled by a single quantity, the FLO extent of the input state.

\begin{definition}[FLO extent]
\label{def:flo_extent}
The FLO extent of a fermionic state $\ket{\psi}$ is the smallest integer
$\chi$ such that $\ket{\psi} = \sum_{s=1}^{\chi} c_s
\ket{\phi_s^{\mathrm{Gauss}}}$ for some Gaussian states
$\ket{\phi_s^{\mathrm{Gauss}}}$ and complex coefficients $c_s$.
\end{definition}

For the paired magic state $\ket{\Psi_4}^{\otimes k/2}$, each block
contributes FLO extent $2$, so $\chi = 2^{k/2} = 2^B$ with
$B$ the number of magic blocks. Passive FLO transformations preserve FLO
extent, so $\chi$ of the output state $V\ket{\Psi_{\mathrm{in}}}$ matches
that of the input.

\begin{proposition}[Classical cost of simulating the pipeline]
\label{prop:flo_cost}
For a state with FLO extent $\chi$ prepared by a passive FLO circuit with
single-particle unitary $W$, the following classical costs hold:
\begin{itemize}
  \item \emph{Fixed-body expectation values, exact}: $\poly(n)$ for every
  fixed body number $r$, via propagation of the $r$-RDM through
  $\Lambda^r(W)$ (Proposition~\ref{prop:onebody_const}); for $r = 1$ this is
  $O(n^2)$, with
  $\langle\psi|V^\dagger n_j V|\psi\rangle = [W\rho_1^{\mathrm{in}}W^\dagger]_{jj}$.
  The cost is independent of $\chi$ and of the form-rank of the input.
  \item \emph{Exact amplitudes}: $O(\chi \cdot \poly(n))$, by direct
  evaluation of the Gaussian-sum expansion (for paired blocks, the
  determinant sum of Eq.~\eqref{eq:det_sum}).
  \item \emph{$r$-body expectation values, fixed $r$, form-rank $2$}:
  $O(\poly(n) \cdot \epsilon^{-2})$ to additive error $\epsilon$, via the
  mixed-Pfaffian estimator of Oh et al.~\cite{oh2026classical}.
  \item \emph{Exact sampling from $p(\cdot \mid \bx)$}:
  $O(\chi^2 \cdot \poly(n))$ per sample via the chain rule over marginals,
  and $O(B^4 + 2^B)$ for products of four-mode magic blocks by the exact
  sampler of Heo et al.~\cite{heo2026fast}, where $B = \log_2 \chi$ is the
  block count.
\end{itemize}
\end{proposition}

\begin{proof}
Each case is the corresponding best-known algorithm from the cited literature, specialized to our input; we recall the arguments. Throughout, the Gaussian terms in the decomposition of $\ket{\psi}$ are mutually non-orthogonal Slater determinants, and all matrix elements between them are evaluated by the generalized (L\"owdin) Wick rules, each in $\poly(n)$ time.
The form-rank-$2$ case is quoted from~\cite{oh2026classical} and not reproved
here.
\textit{Fixed-body case.} An $r$-body expectation value depends only on the
$r$-RDM of the output state, which is
$\Lambda^r(W)\rho_r^{\mathrm{in}}\Lambda^r(W)^\dagger$; this is
Proposition~\ref{prop:onebody_const}. Since $\rho_r^{\mathrm{in}}$ has
dimension $\binom{n}{r}$ regardless of the non-Gaussianity or form-rank of
the input, classical computation takes $\poly(n)$ time for fixed $r$, and
$O(n^2)$ for $r = 1$.

\textit{Amplitudes.} Each Gaussian term contributes a determinant
(Slater-determinant overlap) computable in $\poly(n)$ time; summing the
$\chi$ terms gives the amplitude exactly.

\textit{Case $r \geq 2$, form-rank $2$.} The mixed-Pfaffian reduction of Oh
et al.~\cite{oh2026classical} compresses the $\chi^2$ cross-term sum into a
single coefficient of a multivariate Pfaffian polynomial, extractable via
Monte Carlo with $O(\epsilon^{-2}\log\delta^{-1})$ samples each costing
$O(N^3)$ time (where $N$ is the number of paired blocks), independent of
$\chi$.

\textit{Sampling.} Sample the modes one at a time. The marginal probability
of any partial occupation pattern on modes $1, \dots, m$ is
\begin{equation*}
  p(x_1, \dots, x_m)
  = \sum_{s,s'} c_s^* c_{s'} \bra{\phi_s^{\mathrm{Gauss}}} V^\dagger
  \Bigl(\prod_{j \le m} P_{x_j}\Bigr) V \ket{\phi_{s'}^{\mathrm{Gauss}}},
\end{equation*}
where $P_{x_j} \in \{n_j, 1 - n_j\}$. Each cross-term is a matrix element of
a product of occupation projectors between two Gaussian states evolved by
FLO, computable \emph{exactly} in $\poly(n)$ time by generalized (L\"owdin)
Wick rules for non-orthogonal Slater determinants. Conditional sampling of
mode $m+1$ given $x_1, \dots, x_m$ therefore costs $O(\chi^2 \poly(n))$, and
a full sample costs $n$ such steps. For the block-product inputs of
Section~\ref{sec:encoding} this generic chain rule is not optimal: Heo et
al.~\cite{heo2026fast} reorganize it so that an exact sample costs
$O(B^4 + 2^B)$ in total, removing both the square on the extent and the
polynomial prefactor, and it is this figure that the calibration of
Section~\ref{sec:k_choice} uses. Note that either argument requires
\emph{exact} (or relative-error) marginals; additive-error estimates do not
suffice, since the conditionals divide by marginals that are generically
exponentially small.
\end{proof}

Applying this to the magic-state construction at the operating point
of Eq.~\eqref{eq:operating_point} ($B = 50$ blocks, $\chi = 2^{50}$):
because the input has form-rank $2$,
the mixed-Pfaffian algorithm of~\cite{oh2026classical} achieves
$\mathrm{poly}(n) \cdot \epsilon^{-2}$ cost for any fixed-body observable,
regardless of $k$: few-body expectation values are
classically tractable at additive precision at \emph{any} $k$. The
classically hard task is \emph{sampling}
(Theorem~\ref{thm:sampling_bestknown}), at best-known cost
$2^{B} = 2^{50}$ per query.

\subsection{Hardness under Best-Known Algorithms at $k = \omega(\log n)$}
\label{sec:bestknown_hardness}

Once the block number exceeds logarithmic, every algorithm of
Proposition~\ref{prop:flo_cost} that touches the non-Gaussian structure
costs $2^{\Omega(k)}$. We record the resulting hardness statement. It is
\emph{algorithm-relative} by design; the
calibration of Section~\ref{sec:k_choice} converts it into a concrete
compute budget, and any future algorithmic improvement is absorbed by
increasing the block number $B$, at training cost linear in $k$.

\begin{theorem}[Sampling hardness for paired-block inputs under best-known
algorithms]
\label{thm:sampling_bestknown}
Let the input be the paired magic-block state $\ket{\Psi_4}^{\otimes k/2}$
with $k = \omega(\log n)$. Then every currently known classical algorithm
for sampling from the output distribution $p(\cdot \mid \bx)$ of the trained
unitary brick-wall or butterfly QNN requires $2^{\Omega(k)}\,\poly(n)$ time.
At the sampling operating point $k = 100$ the best-known cost is
$2^{50} \approx 1.1 \times 10^{15}$ operations per sample; at
$k = \log^2 n$ it is $n^{\Omega(\log n)}$. Consequently, the
function classes of both architectures contain output distributions
inaccessible to classical samplers under best-known algorithms.
\end{theorem}

\begin{proof}
Three routes to sampling are known, and each costs $2^{\Omega(k)}$.
(i) Exact strong simulation via the determinant sum of
Eq.~\eqref{eq:det_sum} costs $\Theta(2^{k/2})$ determinant evaluations per
amplitude. (ii) Chain-rule sampling via extent-based marginals
(Proposition~\ref{prop:flo_cost}) costs $\Theta(2^{k})$ cross-term
evaluations per sample~\cite{reardonsmith2024classical}, improved to
$\Theta(2^{k/2})$ for our block-product inputs by the exact sampler
of~\cite{heo2026fast}. (iii) The
additive-error correlator estimator of~\cite{oh2026classical} does not yield
approximate sampling, because chain-rule conditionals require relative
precision on marginals that are generically exponentially small in $k$.
Hence no known algorithm samples in polynomial time once
$k = \omega(\log n)$.
\end{proof}

\subsection{Worst-Case \#P-Hardness at Polynomial Particle Number}
\label{sec:worstcase}

The statements of the previous subsection are relative to the best-known
algorithms. Genuine complexity-theoretic hardness---\#P-hardness of exact
computation---requires the block number itself to be polynomial in $n$,
because the hard-instance size of the underlying reduction is the block
count (Remark~\ref{rem:sparse_input}).

\begin{theorem}[Worst-case \#P-hardness for paired blocks at $k =
n^{\epsilon}$]
\label{thm:worstcase}
Fix a constant $\epsilon \in (0, 1)$ and let $k = n^{\epsilon}$ (with
$k \le n/2$). For the paired magic-block input $\ket{\Psi_4}^{\otimes k/2}$,
there exist parameter settings $(\btheta, \bphi)$ of the unitary
brick-wall and data encodings $\bx$ for which computing the output
probability $p(\bx' \mid \bx)$ exactly is \#P-hard.
\end{theorem}

\begin{proof}
For any passive FLO circuit with single-particle unitary $W \in U(n)$, the
output amplitude is the sum of $2^B$ determinants of
Eq.~\eqref{eq:det_sum}, $B = k/2$---a mixed discriminant in the sense of
Ivanov and Gurvits (cf.~\cite[Lemma~11]{oszmaniec2022fermion}). Computing
this mixed discriminant exactly for arbitrary $W$ is \#P-hard when the
block count is polynomially related to the instance size: an $m \times m$
permanent instance (with nonnegative entries, so that hardness carries to
the probability $|{\cdot}|^2$, cf.~\cite[Remark~21]{oszmaniec2022fermion})
embeds with polynomial overhead into the mixed discriminant of $O(m)$
blocks. Choosing $n = (2m)^{1/\epsilon} = \poly(m)$ makes the reduction
polynomial-time in $m$. The brick-wall is surjective onto $U(n)$
(Proposition~\ref{prop:brick-wall}), so the hard single-particle unitary
$W^*$ is exactly reachable: choose $\bx = 0$ (so that the data layer is
trivial up to fixed phases) and set the brick-wall angles, via the Givens
decomposition, to realize
$W_P = W^* (W_D W^{\mathrm{had}})^{-1}$, whence the total single-particle
unitary is $W^*$.
\end{proof}

\begin{remark}[Reachability caveat for the butterfly]
\label{rem:butterfly_reachability}
The reduction realizes a specific hard unitary $W^*$, which the brick-wall
reaches by surjectivity. A single depth-$2\log n$ butterfly covers only a
submanifold of $U(n)$ of dimension at most $\tfrac32 n\log n$
(Proposition~\ref{prop:butterfly}), and whether that submanifold meets the
Ivanov--Gurvits hard instances is open; Theorem~\ref{thm:worstcase}
therefore holds for the butterfly only under the assumption that it does.
Stacking $O(n/\log n)$ butterfly blocks restores surjectivity, and with it
the unconditional statement, at linear depth. The best-known-algorithm
hardness of Theorem~\ref{thm:sampling_bestknown} is unaffected either way,
since it depends only on the input state.
\end{remark}

\begin{remark}[What worst-case \#P-hardness does and does not buy]
Two contrasts delimit Theorem~\ref{thm:worstcase}. At $k = O(\log n)$,
Eq.~\eqref{eq:det_sum} is itself a polynomial-time algorithm, so no
\#P-hardness is possible; at $k = \log^2 n$ the same expression gives a
quasi-polynomial upper bound, which likewise precludes it under standard
assumptions. Worst-case \#P-hardness therefore begins only at polynomial
particle number, at a training cost of $O(n^{1+\epsilon})$
(Section~\ref{sec:k_choice}).

This is not, however, the guarantee the framework rests on. \#P-hardness
concerns \emph{exact} computation on a \emph{worst-case} instance, whereas a
deployment needs that no algorithm compute the readout of the \emph{trained}
circuit, to the precision the task requires, in feasible time; neither
implies the other. Oh et al.~\cite{oh2026classical} are the case in point:
the mixed discriminant of a form-rank-$2$ input is \#P-hard to evaluate
exactly, and yet its few-body expectation values are estimable to additive
error in polynomial time. Conversely, the operating point of
Eq.~\eqref{eq:operating_point} carries no \#P-hardness whatsoever---$B = 50$
is a constant---but places best-known classical simulation beyond
$10^{15}$ operations per query at every $n$, which is the operationally
relevant statement.
Theorem~\ref{thm:worstcase} is best read as evidence that the
$2^{\Omega(k)}$ cost of Theorem~\ref{thm:sampling_bestknown}
reflects genuine structure rather than a
gap in the algorithmic literature.
\end{remark}

\subsection{Average-Case Hardness: the Fully-Packed Limit and the Sparse
Conjecture}
\label{sec:brick-wall_avg_case}

The strongest expressivity claim available is average-case sampling
hardness: a randomly initialized circuit produces a classically hard
distribution with high probability. The two required ingredients,
anticoncentration and the worst-to-average-case Cayley-path reduction~\cite{movassagh2023hardness}, are
currently proven only for fully-packed magic inputs, i.e.\ at extensive
particle number~\cite{oszmaniec2022fermion}. Since the brick-wall realizes
the Haar measure on $U(n)$ under Haar initialization
(Section~\ref{sec:bp_brick-wall}), the full machinery transfers verbatim in
that regime.

\begin{theorem}[Average-case sampling hardness for the brick-wall,
fully-packed regime]
\label{thm:brick-wall_avgcase}
Let $k = n/2$ with the fully-packed paired input
$\ket{\Psi_4}^{\otimes n/4}$. Under the complexity-theoretic assumptions
of~\cite{oszmaniec2022fermion} (noncollapse of the polynomial hierarchy and
average-case \#P-hardness of relative-error computation of output
probabilities, their Conjectures~1--2), approximately sampling from
$p(\bx' \mid \bx)$ is classically
hard for a unitary brick-wall $U_P(\btheta, \bphi)$ at Haar-random
parameters.
Consequently, the function class of the unitary brick-wall contains output
distributions that cannot be efficiently sampled by any classical algorithm.
\end{theorem}

\begin{proof}
Under Haar-random parameters, the brick-wall realizes the Haar measure on
$U(n)$ in the single-particle sector (Section~\ref{sec:bp_brick-wall}). The
full circuit $V = U_P \circ D \circ U^{\mathrm{had}}$ is therefore a
Haar-random passive FLO transformation composed with the fixed deterministic
operators $D$ and $U^{\mathrm{had}}$. $D(\bvarphi(\bx))$ and $U^{\mathrm{had}}$ are fixed deterministic passive
FLO unitaries with single-particle matrices $W_D$ and $W^{\mathrm{had}}$
in $U(n)$; they do not depend on the brick-wall parameters $(\btheta, \bphi)$.
Since $W_P \sim \mathrm{Haar}(U(n))$ and the Haar measure on $U(n)$ is
bi-invariant, for any fixed $A \in U(n)$ the product $A \cdot W_P$ is
again Haar-distributed: $\mathbb{P}[A W_P \in S] = \mathbb{P}[W_P \in
A^{-1} S] = \mathrm{Haar}(A^{-1} S) = \mathrm{Haar}(S)$ for any measurable
$S \subseteq U(n)$. Applying this twice, $W = W_P \cdot W_D \cdot
W^{\mathrm{had}}$ is Haar-distributed on $U(n)$.

The Oszmaniec et al.~\cite{oszmaniec2022fermion} anticoncentration result
(their Theorem~1) applies verbatim: for Haar-random passive FLO circuits on
$U(n)$ initialized in the fully-packed input $\ket{\Psi_4}^{\otimes n/4}$,
the output distribution anticoncentrates with constant
$C_{\mathrm{pas}} = 5.7$. The worst-to-average-case reduction
of~\cite{oszmaniec2022fermion} (their Theorem~7, via the Cayley-path
argument) and their Stockmeyer reduction (their Theorems~2--3) then
establish hardness of approximate sampling under the stated assumptions.
Since the brick-wall's function class contains
all distributions reachable at Haar-random parameters, it contains
distributions inaccessible to efficient classical samplers.
\end{proof}

\begin{conjecture}[Sparse-input average-case hardness]
\label{conj:sparse_avgcase}
At $k = n^{\epsilon}$ with paired input $\ket{\Psi_4}^{\otimes k/2}$ and
$n - 2k$ idle modes, the output distribution of a Haar-random passive FLO
circuit anticoncentrates, and approximate sampling remains hard on average
under the assumptions of~\cite{oszmaniec2022fermion}. This is a sharpened
form of the ``fermion sampling with less magic'' open problem
of~\cite[Sec.~IV]{oszmaniec2022fermion} and its resolution
requires, in particular, redoing the anticoncentration second-moment
computation of~\cite[Lemma~13]{oszmaniec2022fermion} for sparse inputs.
Resolving it affirmatively would place average-case hardness at the same
particle numbers where the gradient variance of
Theorem~\ref{thm:brick-wall_bp} is largest.
\end{conjecture}

For the butterfly, the average-case question is doubly open: beyond the
sparse-input issue, anticoncentration at logarithmic depth is not known even
at full packing---whether the collision-probability constant satisfies
$C_k(n) = O(\poly(n))$, and hence whether the worst-to-average-case reduction
of Oszmaniec et al.~\cite{oszmaniec2022fermion} applies, is open for both
fully-packed and sparse magic-block inputs.

\medskip
\noindent\textbf{Hardness summary.}
\begin{itemize}
  \item \emph{$k = O(\log n)$}: the entire pipeline is classically
  simulable in polynomial time.
  \item \emph{Operating point, $B = 50$ blocks} ($k = 100$, paired
  blocks): best-known
  classical sampling costs $2^{B} = 2^{50}$, i.e.\
  $\gtrsim 10^{15}$ elementary operations per query
  (Theorem~\ref{thm:sampling_bestknown}), calibrated in
  Section~\ref{sec:k_choice} to exceed feasible classical compute; future
  algorithmic improvements are absorbed by increasing $B$.
  \item \emph{$k = n^{\epsilon}$}: worst-case \#P-hardness of exact output
  probabilities for the brick-wall (Theorem~\ref{thm:worstcase}); for the
  butterfly under the reachability assumption
  (Remark~\ref{rem:butterfly_reachability}); average-case hardness
  conjectured (Conjecture~\ref{conj:sparse_avgcase}).
    \item \emph{$k = \Theta(n)$}: the full average-case machinery of~\cite{oszmaniec2022fermion} for the Haar-initialized brick-wall at
    $k = n/2$ (Theorem~\ref{thm:brick-wall_avgcase}); the parallel PSR no
    longer yields an asymptotic reduction over the naive parameter-shift
    rule (Section~\ref{sec:k_choice}).
  \item \emph{Fixed-body expectation-value readouts}: classically easy at any
  $k$ for both architectures, via $r$-RDM
  propagation (Proposition~\ref{prop:onebody_const}).
\end{itemize}

The last item identifies precisely which observables are classically easy:
those whose expectation values are determined by a reduced density matrix of
fixed body number, regardless of the input's higher-body non-Gaussian
structure.

Concrete parameter values for both architectures
at representative problem sizes are collected in
Table~\ref{tab:master} of Section~\ref{sec:discussion}.

\subsection{On the Cerezo et al.\ Simulability Framework}
\label{sec:cerezo_response}
Cerezo et al.~\cite{cerezo2025does} established that QNNs with a
polynomial-dimensional dynamical Lie algebra admit efficient classical
simulation under a condition on the input state and the observable: both must
be \emph{$\mathfrak{g}$-compatible}, in the sense that their projections onto
$\mathfrak{g}$ can be computed efficiently. Under this condition the loss
$\calL = \bra{\psi_{\mathrm{in}}} V^\dagger O V \ket{\psi_{\mathrm{in}}}$ is
evaluated classically by propagating the $\mathfrak{g}$-projection of the
input through the $\mathfrak{g}$-adjoint action of $V$ and pairing it with the
$\mathfrak{g}$-projection of the observable. The conclusion concerns the value
of an expectation; the argument makes no statement about the complexity of
sampling from the output distribution of the circuit.

Both architectures considered here have polynomial dynamical Lie algebra,
$\mathfrak{g} = \mathfrak{u}(n)$ with $\dim\mathfrak{g} = n^2$. We therefore
examine whether the framework precludes quantum advantage in the present
setting. It does not; we treat in turn the sampling task itself and the
few-body sample statistics of Section~\ref{sec:readouts}.

\paragraph{The sampling task.}
Consider the sampling pipeline at $k = \Theta(n)$. Each hypothesis of the
framework is satisfied. The dynamical Lie algebra is $\mathfrak{u}(n)$, of
dimension $n^2$. The $\mathfrak{g}$-projection of the magic-state input is its
$1$-RDM, which is diagonal with entries $\tfrac12$ on the $2k$ active modes
and computable in $O(n^2)$ time (Lemma~\ref{lem:spreading}). The readouts
$n_j$ are one-body observables, so their $\mathfrak{g}$-projections coincide
with themselves. The conclusion of the framework holds as well: these
expectation values are classically computable in $O(n^2)$ time, in agreement
with Proposition~\ref{prop:onebody_const}.

The output distribution of the same circuit acting on the same input is
nevertheless not classically accessible: approximate sampling is hard on
average at $k = n/2$ under the assumptions
of~\cite{oszmaniec2022fermion} (Theorem~\ref{thm:brick-wall_avgcase});
exact computation of individual output probabilities is \#P-hard in the
worst case at $k = n^{\epsilon}$ (Theorem~\ref{thm:worstcase}); and every
known sampling algorithm costs more than $10^{15}$ operations per query at
the operating point $k = 100$
(Theorem~\ref{thm:sampling_bestknown}).

These statements are consistent. An individual output probability
$p(\bx' \mid \bx) = |\bra{\bx'} V \ket{\Psi_{\mathrm{in}}}|^2$ is the
expectation value of the rank-one projector $\ket{\bx'}\bra{\bx'}$, an
observable of body-number $k$; such an observable is not
$\mathfrak{g}$-compatible, and the framework does not apply to it. It follows
that the implication frequently drawn from the framework---that an
architecture with polynomial dynamical Lie algebra cannot support quantum
advantage---does not hold in general. The present construction satisfies every
hypothesis of the framework and nonetheless generates classically intractable
samples; at $k = cn$ with $c < 1/2$ it does so---under the
best-known-algorithm calibration of Section~\ref{sec:k_choice}, the
average-case statement of Theorem~\ref{thm:brick-wall_avgcase} being available
at exactly half filling---while retaining unconditional trainability, with
gradient variance $\Theta(1/n^2)$, its largest value anywhere in the range
(Theorem~\ref{thm:brick-wall_bp} and
Theorem~\ref{thm:interior_halves}).

This regime is of central rather than peripheral interest for machine
learning: generative modelling and reinforcement-learning policies both
consume the circuit directly as a sampler
(Section~\ref{sec:pipeline}), so the classically intractable object is the
model output itself.

\paragraph{Few-body sample statistics.}
For statistics of fixed body number the framework is informative, and its
conclusion is the one we reach independently. The $\mathfrak{g}$-projection of
an $r$-body observable is determined by the $r$-particle sector, and by
Proposition~\ref{prop:onebody_const} the corresponding expectation value is
computable in $\poly(n)$ time regardless of $k$ or of form-rank; for the
paired inputs used here the mixed-Pfaffian reduction of Oh et
al.~\cite{oh2026classical} additionally gives an additive-error estimator
with the shot complexity of the quantum device itself. Few-body readouts are
therefore classically tractable and cannot witness quantum advantage.

\paragraph{Summary.}
Polynomial dimension of the dynamical Lie algebra does not imply classical
simulability of the learning task. Two properties jointly determine
tractability:
\begin{itemize}
  \item the body number $r$ of the readout observable,
  \item the precision regime (multiplicative for sampling, additive for
  expectation values).
\end{itemize}
Every fixed $r$ is resolved, for arbitrary input and either architecture, by
$r$-RDM propagation (Proposition~\ref{prop:onebody_const}); at form-rank $2$
the mixed-Pfaffian reduction of~\cite{oh2026classical} additionally supplies
an additive-error estimator whose shot complexity matches the quantum device.
Sample-based readouts fall outside the scope of the framework, and constitute
the setting in which an architecture of polynomial dynamical Lie algebra is
provably not classically simulable.
\section{Summary and Discussion}
\label{sec:discussion}

We have presented a provably scalable end-to-end quantum neural network 
framework built around two architectures --- the unitary brick-wall for 
nearest-neighbor hardware and the unitary butterfly for all-to-all 
connectivity --- that together address five theoretical challenges 
simultaneously: hardware-matched circuit depth (Section~\ref{sec:butterfly}), 
efficient training cost via the multi-layer parallel parameter-shift rule 
(Section~\ref{sec:psr}), polynomial gradient variance and absence of 
barren plateaus for readouts of fixed body number (Section~\ref{sec:bp}), 
superpolynomial classical simulation hardness for the sampling task under 
the magic-state encoding (Section~\ref{sec:hardness}), and polynomial 
sample complexity for generalization (Section~\ref{sec:pipeline}). The 
third and fourth of these are established for different objects of the same 
circuit---the gradient channel and the deployed output respectively---and 
Section~\ref{sec:where} sets out which pipelines admit a variance guarantee 
and which require the quantum device in the training loop, and 
Section~\ref{sec:tasks} instantiates the standard machine-learning 
tasks---generative modelling, reinforcement-learning policies, supervised 
classification and regression, kernel methods---with the loss and head 
appropriate to each. Both architectures share 
the same encoding pipeline and the same parallel parameter-shift rule; 
they differ in circuit depth, parameter count, training cost, and the 
strength of the trainability guarantee. The key structural insight is 
that all properties follow from interleaving particle-number-preserving 
RBS gates with single-qubit $R_z$ gates --- in a brick-wall connectivity 
pattern for nearest-neighbor hardware, or a butterfly connectivity 
pattern for all-to-all hardware --- combined with a magic-state
encoding whose number of magic blocks is calibrated
against classical compute: at the operating point of
Eq.~\eqref{eq:operating_point}, paired magic blocks $\ket{\Psi_4}^{\otimes
k/2}$ at $k = 100$, at a best-known sampling cost of $2^{B} = 2^{50}$ per
query.

\subsection{Complete Parameter Summary}
\label{sec:concrete_numbers}

Table~\ref{tab:master} summarizes all relevant quantities for both
architectures at the operating point, at four representative problem
sizes. The table is organized in three blocks: shared encoding-dependent
quantities (identical for both architectures), architecture-dependent
quantities for the butterfly, and architecture-dependent quantities for
the brick-wall.

\begin{table*}[t]
\centering
\small
\caption{%
Complete parameter summary for the unitary brick-wall and butterfly QNNs
at four problem sizes and the operating point of
Eq.~\eqref{eq:operating_point}: $B = 50$ paired magic blocks
$\ket{\Psi_4}^{\otimes 50}$, $k = 100$, FLO extent $\chi = 2^{50}$, and a
best-known sampling cost of $2^{B} \approx 1.1 \times 10^{15}$ operations per
sample, fixed across $n$; encoding depth $\log n + 5$. The second rung
$B = 60$ of Section~\ref{sec:k_choice} scales the particle-number-dependent
rows by $6/5$ and requires $n \ge 258$.
\textbf{Shared quantities} (top block) depend only on the magic-state
encoding and are identical for both architectures.
The sampling cost row reports $2^B$, which is linear rather than quadratic in
the FLO extent and is the leading factor in the
best-known classical cost of sampling
(Theorem~\ref{thm:sampling_bestknown}); few-body expectation values are
additively tractable via the mixed-Pfaffian
estimator~\cite{oh2026classical}.
\textbf{Architecture-dependent quantities} (middle and bottom blocks)
reflect the depth--connectivity tradeoff: the butterfly targets
all-to-all hardware at $O(\log n)$ gate layers with $T = \tfrac32 n\log n$
parameters; the brick-wall targets nearest-neighbor (ring) hardware at $O(n)$
depth with $T = 3n^2/2$ parameters.
Naive PSR cost is $4n\log n$ for the butterfly and $4n^2$ for the
brick-wall ($4$ circuit evaluations per RBS parameter and $2$ per $R_z$
parameter); parallel PSR cost is $4k\log n$ for the butterfly and $4kn$ for
the brick-wall; the reduction factor $n/k$, identical for the two
architectures, grows linearly with $n$ at fixed hardness
and exceeds $1$ throughout the admissible range $2k \le n$.
Gradient variance for two-body correlator readouts is $\Theta(k^3/n^5)$
\emph{unconditionally} for the brick-wall
(Theorem~\ref{thm:brick-wall_bp}, Corollary~\ref{cor:first_column} and
Theorem~\ref{thm:interior_halves}) and $\Theta(k^3/n^5)$
\emph{conditional} on Conjecture~\ref{conj:lambda2_design} for
the butterfly (Theorem~\ref{thm:sharp_rate}).%
}
\label{tab:master}
\setlength{\tabcolsep}{6pt}
\begin{tabular*}{\textwidth}{@{\extracolsep{\fill}}
  l
  c c c c
  @{}}
\toprule
& $n = 256$ & $n = 512$ & $n = 1024$ & $n = 4096$ \\
\midrule
\multicolumn{5}{@{}l}{%
  \textit{Shared quantities (encoding-dependent, identical for
  both architectures)}%
} \\[2pt]
\quad Magic blocks $B$
  & $50$ & $50$ & $50$ & $50$ \\
\quad Active qubits $2k$
  & $200$ & $200$ & $200$ & $200$ \\
\quad Hilbert space $\tbinom{n}{k}$
  & $1.2{\times}10^{73}$ & $2.9{\times}10^{108}$
  & $7.7{\times}10^{140}$ & $5.5{\times}10^{202}$ \\
\quad FLO extent $\chi = 2^{50}$
  & $1.1{\times}10^{15}$ & $1.1{\times}10^{15}$
  & $1.1{\times}10^{15}$ & $1.1{\times}10^{15}$ \\
\quad Sampling cost $2^{B} = 2^{50}$
  & $1.1{\times}10^{15}$ & $1.1{\times}10^{15}$
  & $1.1{\times}10^{15}$ & $1.1{\times}10^{15}$ \\
\quad Gradient variance (2-body)
  & \multicolumn{4}{c}{$\Theta(k^3/n^5)$} \\
\quad Sample complexity
  & \multicolumn{4}{c}{$\tilde{O}(T)$} \\
\midrule
\multicolumn{5}{@{}l}{%
  \textit{Unitary butterfly (all-to-all connectivity, $O(\log n)$ gate layers)}%
} \\[2pt]
\quad Total gate layers
  & $30$ & $33$ & $36$ & $42$ \\
\quad Parameters $T = \tfrac32 n\log n$
  & $3{,}072$ & $6{,}912$ & $15{,}360$ & $73{,}728$ \\
\quad Naive PSR cost $4n\log n$
  & $8{,}192$ & $18{,}432$ & $40{,}960$ & $196{,}608$ \\
\quad Parallel PSR cost $4k\log n$
  & $3{,}200$ & $3{,}600$ & $4{,}000$ & $4{,}800$ \\
\quad Reduction factor $n/k$
  & $2.6\times$ & $5.1\times$ & $10.2\times$ & $41.0\times$ \\
\midrule
\multicolumn{5}{@{}l}{%
  \textit{Unitary brick-wall (nearest-neighbor connectivity, $O(n)$ depth)}%
} \\[2pt]
\quad Total circuit depth
  & $773$ & $1{,}541$ & $3{,}077$ & $12{,}293$ \\
\quad Parameters $T = 3n^2/2$
  & $98{,}304$ & $393{,}216$ & $1{,}572{,}864$ & $2.5{\times}10^{7}$ \\
\quad Naive PSR cost $4n^2$
  & $262{,}144$ & $1{,}048{,}576$ & $4{,}194{,}304$ & $6.7{\times}10^{7}$ \\
\quad Parallel PSR cost $4kn$
  & $102{,}400$ & $204{,}800$ & $409{,}600$ & $1{,}638{,}400$ \\
\quad Reduction factor $n/k$
  & $2.6\times$ & $5.1\times$ & $10.2\times$ & $41.0\times$ \\
\bottomrule
\end{tabular*}
\end{table*}

Several features of Table~\ref{tab:master} deserve highlighting.

\paragraph{Trainability.}
The gradient variance for two-body correlator readouts is $\Theta(k^3/n^5)$ 
for both architectures, independent of block type. As established in
Section~\ref{sec:bp}, this rate is \emph{unconditional} for the brick-wall 
(Theorem~\ref{thm:brick-wall_bp}) and \emph{conditional} on 
Conjecture~\ref{conj:lambda2_design} for the butterfly 
(Theorem~\ref{thm:sharp_rate}), which also carries an unconditional
polynomial lower bound $\Omega(k^4/n^6)$
(Theorem~\ref{thm:no_bp}). Two body number is the only order at which we
establish a rate; losses that depend on the samples through statistics of
this order inherit it by linearity
(Proposition~\ref{prop:mmd_trainable}), while losses sensitive to
higher-order sample structure retain the same parameter-shift cost but carry
no variance guarantee (Section~\ref{sec:where}).

\paragraph{Classical hardness.}
The hard task is sampling from the model's output distribution---the
deployed output of every pipeline in Section~\ref{sec:pipeline}, and not the
loss used to train it. By construction its best-known classical simulation
cost is the same at every problem size: $2^{B} = 2^{50} \approx 1.1 \times
10^{15}$ elementary operations \emph{per query}, since the data-dependent
layer makes each input a distinct output distribution---$19$ minutes on a
$10^{12}$-operation/s server and $1.1$~s on a dedicated petascale cluster,
against $10^{-4}$~s on the device, so serving the model classically is out of
commercial reach; reproducing it from $N$ queries costs $N$ times as much, and
is out of exascale reach at the second rung of
Section~\ref{sec:k_choice}. Hardness is set by the block number $B = 50$,
not by $n$; growing $n$ buys expressivity (Hilbert-space dimension,
feature dimension) and training efficiency, while the hardness margin is
adjusted, if ever needed, by increasing $B$ at cost linear in $k$
(Section~\ref{sec:k_choice}). These costs apply to both architectures
identically, since they depend only on the magic-state encoding and not on
the connectivity pattern of the passive FLO circuit.

\paragraph{Training cost.}
At fixed hardness level (fixed $k$), the parallel-PSR reduction factor
$n/k$ grows linearly with $n$: it
reaches $10.2\times$ at $n = 1024$
and $41.0\times$ at $n = 4096$. The butterfly's absolute
cost, $4k\log n$, is nearly independent of $n$: between $3{,}200$ and
$4{,}800$ circuit evaluations per gradient step across the entire table.
Training a hardness-calibrated model therefore becomes \emph{relatively}
cheaper as the machine grows. The parallel rule is cheaper than the naive
one whenever $k < n$, so the binding constraint on the smallest useful machine
is not the parameter-shift rule but the interaction of the two bounds on $k$:
the serving rung of Section~\ref{sec:k_choice} bounds $k$ from below, while the
certificate of Theorem~\ref{thm:no_bp} bounds it from above through
$n - 2k > \sqrt n$, capping $B$ at $\Theta(n)$ and so the classical cost per
query at $2^{\Theta(n)}$. The largest admissible $B$ is $51$ at $n = 256$
against $25$ at $n = 128$, so $n = 256$ is the smallest butterfly meeting the
serving rung; the brick-wall admits any even $n$ and meets a
$10^{12}$-operation/s bar at $n = 200$, $k = 80$. Already
at $n = 256$ the reduction is $2.6\times$. The butterfly's polylogarithmic per-step cost makes
it the architecture of choice when circuit evaluations are the bottleneck;
the brick-wall's $O(n)$ depth makes it the natural target for
nearest-neighbor hardware regardless of evaluation cost.

\subsection{Outlook}

The main open question is empirical: whether the proposed framework 
delivers practical advantages on real machine learning tasks. The 
theoretical guarantees established here --- trainability, efficient 
gradients, classical hardness, polynomial generalization --- are 
necessary but not sufficient for useful quantum advantage. What remains 
to be demonstrated is that quantum models built on either architecture
solve problems classical models cannot solve as well, whether standalone
or as components in ensemble methods. This requires both mature quantum
hardware and careful empirical study. For the butterfly, $n = 512$ or
$n = 1024$ with all-to-all or reconfigurable connectivity appears to us
to be the scale at which meaningful comparisons become possible---the
operating point of Eq.~\eqref{eq:operating_point} fits comfortably
($2k < n$) and the parallel PSR delivers a material advantage
(Table~\ref{tab:master}); for
the brick-wall, the same $n$ on nearest-neighbor hardware---a ring, or a
chain via the open-boundary variant of Remark~\ref{rem:chain_variant}. The verification regime $k = O(\log n)$ provides the natural on-ramp:
the identical architecture can be validated end to end against exact
classical simulation at small block number before $B$ is scaled up.

The analysis of Section~\ref{sec:where} narrows where such demonstrations
should be attempted. The criterion is the \emph{deployed output}: it must be
the sample set or a functional of it beyond fixed body number, since a model
whose output is a few-body statistic is classically simulable end to end
regardless of $k$. Whether the device is also needed during training varies,
and either mode is viable: a loss anchored on estimable statistics can be
trained classically, with the device reserved for sampling at inference,
while a loss that is an expectation over something we have no model of
requires the device in the loop---which can happen even for a low-body
effective observable, when its weights are supplied by an environment rather
than known to us. Two families meet the criterion. The first is
joint multivariate generative modelling, where the object consumed
downstream is the dependence structure of a high-dimensional distribution---
joint tail probabilities, scenario ensembles, risk functionals---rather than
a collection of marginals; sampling is then the deployed operation and the
two-body moments that the trainability analysis controls are exactly the
pairwise co-occurrence statistics such applications already use. The second
is model-free reinforcement learning with a structurally low-body return,
which is the one setting identified in Section~\ref{sec:where} in which a
proven gradient rate, classical hardness of the deployed policy, and the
need for the device inside the training loop hold simultaneously. Financial
decision problems---portfolio construction, execution---are a natural
instance of the second family: the return is linear in the positions with
quadratic costs and risk, hence two-body in the action, while no
sufficiently accurate model of the environment is available to propagate
through, which is why model-free methods are used in the first place. A
further architectural fit is that particle-number preservation makes every
shot a valid Hamming-weight-$k$ configuration, so constrained selection
problems---allocate across $k$ of $n$ assets, route $k$ of $n$ child
orders---are sampled feasibly by construction, without masking or repair.
We note that these are candidates identified on structural grounds; whether
the resulting models are \emph{useful} is the empirical question above, and
the $k = O(\log n)$ regime is where it should first be asked.

On the theory side, several questions remain. The first concerns the
separation of objects described in Section~\ref{sec:where}. The hardness
guarantees of Section~\ref{sec:hardness} are sampling statements: they
operate at multiplicative precision on individual output probabilities, and
the tasks that consume the QNN as a sampler---generative modeling and
reinforcement-learning policies---inherit them directly; Oh et
al.~\cite{oh2026classical} explicitly disclaim any implications of
additive-error estimation for this regime, and any future algorithmic
advance is absorbed by raising the block number $B$. The trainability
guarantees, by contrast, are established for readouts of fixed body number,
and precisely that property makes them classically computable whenever a
model of the distribution is available
(Proposition~\ref{prop:onebody_const}). Whether the two can be carried by a
single object---whether a \emph{variationally trained} circuit can have an
expectation-value readout that is at once provably free of barren plateaus
and classically hard---is open, and it is not clear how to achieve it:
few-body expectation values of bounded observables sit naturally in the
additive-precision regime, where estimators such as the mixed-Pfaffian
reduction of Oh et al.~\cite{oh2026classical} already cover broad input
classes, while observables of growing body order face vanishing signals and
gradients. The question is not whether \emph{some} expectation-value model
can be classically hard---fixed, cryptographically structured feature maps
achieve this, as in the discrete-logarithm construction of Liu et
al.~\cite{liu2021rigorous}, where no variational training is involved and
the hardness is imported from the planted problem rather than arising from
the architecture---but whether hardness of this kind can coexist with a
trained circuit and a proven gradient rate. A related question is how much
of the output distribution a low-body loss actually shapes: training against
few-body statistics controls those statistics, while the higher-order
structure that carries the hardness is inherited from the architecture
rather than learned, and the practical value of that unshaped remainder is
an empirical matter. We stress that this is a question about proofs rather
than about capability: the sampler supports the full range of
machine-learning tasks catalogued in Section~\ref{sec:tasks}, and none of
them requires an expectation-value readout to be classically hard. What an
affirmative answer would add is a second, independent route to advantage,
not a task the present framework cannot address.

A distinct contribution of the present framework is a sharp statement of
where the magic-count threshold for hardness sits, updating the ``fermion
sampling with less magic'' open problem
of~\cite{oszmaniec2022fermion} in light of the extent-based simulation
algorithms of~\cite{reardonsmith2024classical}: at $B = O(\log n)$ blocks
the pipeline is classically simulable outright, including sampling; at $B = \omega(\log n)$ all
known algorithms are superpolynomial
(Theorem~\ref{thm:sampling_bestknown}); at $B = n^{\Omega(1)}$ worst-case
\#P-hardness is restored (Theorem~\ref{thm:worstcase}); and the
average-case machinery is currently available only at extensive filling.
Closing the remaining gap---average-case hardness and anticoncentration
for sparsely packed magic inputs
(Conjecture~\ref{conj:sparse_avgcase})---is, in our view, the most
consequential open problem our construction inherits from the Fermion
Sampling literature.

Other theoretical directions remain open. Resolving
Conjecture~\ref{conj:lambda2_design} --- for instance by extending the
transfer-matrix spectral-gap proof of Appendix~\ref{app:spectral_gap} to the
$\Lambda^2\mathbb{C}^n$ representation --- would promote the butterfly's 
$\Theta(k^3/n^5)$ gradient variance from conditional to unconditional, 
placing both architectures on equal theoretical footing. 
A second, more technical question concerns the brick-wall's per-parameter
statement. The gradient-variance rate is exact and unconditional for the
invariant derivative at a Haar point, for the $n$ phases of the front $R_z$
layer, and for every gate of the circuit in the independent-halves ensemble.
What is not settled is the pointwise statement for an interior gate of the
\emph{literal} single mesh, which reduces to a single second-moment question:
writing $\Pi^{(\ell)} = W_1 \Pi W_1^\dagger$ for the spread projector
propagated through the first $\ell-1$ columns, does
\begin{equation}
\label{eq:interior_condition}
  \E\bigl[\,\|[g,\Pi^{(\ell)}]\|_F^2\,\bigr]
  \;=\; \Theta\!\left(\frac{k(n-2k)}{n^2}\right)
\end{equation}
hold at the intermediate depths $1 < \ell < n$---equivalently, can a shallow
partial mesh re-localize the Walsh--Hadamard-spread projector at the level of
second moments? Proposition~\ref{prop:single_mesh_interior} answers yes at
both endpoints of the prefix, with matching constants, and establishes the
corresponding first-moment statement at every intermediate depth. Settling the
second moment requires a transfer analysis of the partial mesh in the Hurwitz
angle densities, which are position-dependent with
$\E[\cos^2\vartheta] \neq \tfrac12$, in the sector carrying
$|\Pi^{(\ell)}_{jj}|^2$ and $|\Pi^{(\ell)}_{ij}|^2$; the Clements-mesh version
of the Hurwitz factorization~\cite{clements2016optimal,russell2017direct} is
the natural starting point. We regard the question as well posed and
self-contained.
The anticoncentration question for the butterfly at logarithmic depth remains open; the single-particle
$2$-design does not by itself control the
$k$-particle collision probability, and a direct argument in the $k$-particle sector may be possible.
Understanding the optimization landscape --- 
whether layer-wise block coordinate descent converges faster than full 
gradient descent and under what conditions --- is an empirical question 
we leave open; polynomial gradient variance rules out exponentially flat 
landscapes but not other obstructions, such as the proliferation of poor 
local minima identified in~\cite{anschuetz2022quantum}. The relationship between body number and 
precision regime suggests a refined hierarchy of QNN expressivity worth 
mapping systematically: our paper identifies two structural
distinctions in the classical-tractability landscape---observables of fixed
body number are always easy, by reduced-density-matrix propagation
(Proposition~\ref{prop:onebody_const}) and, at additive precision on paired
inputs, by the estimator of Oh et al.~\cite{oh2026classical}; while sampling
at multiplicative precision is hard at sufficient block number via Fermion
Sampling---and the intermediate regime of observables whose body number
grows with $n$ is unmapped in both directions.

A separate question is whether the alternative Dicke-state 
encoding (Remark~\ref{rem:dicke}) offers practical or theoretical 
advantages that justify its higher preparation depth. Its
permutation-symmetric structure 
provides natural inductive biases for set-valued data, and its $1$-RDM is
proportional to the identity, so the trainability analysis of
Section~\ref{sec:bp}---which is driven by $\|[g,\rho_1]\|_F^2$---does not
apply and the gradient variance is governed instead by connected two-body
correlations. Characterizing that rate, and determining whether these
advantages outweigh the depth cost in practice, is an interesting direction 
for future work.

Despite these open directions, the framework established here provides
a complete answer to the theoretical challenges that have dominated
quantum machine learning for the past several years, at two levels. As a
matter of theory, both the unitary brick-wall and the unitary butterfly
are trainable throughout the regime $n - 2k = \Omega(n)$---polynomial gradient
variance from $k = O(1)$ to any constant filling fraction $c < 1/2$---while at
extensive particle number the same architecture carries the classical-hardness
guarantees of Section~\ref{sec:hardness}, up to the full average-case Fermion
Sampling machinery at half filling. As a matter of practice, the
operating point of Eq.~\eqref{eq:operating_point} places best-known
classical simulation beyond $10^{15}$ operations per query while the
multi-layer
parallel parameter-shift rule reduces the training cost by a factor
$n/k$ that grows linearly with $n$. The brick-wall
achieves this with unconditional trainability guarantees on
nearest-neighbor hardware;
the butterfly achieves this at polylogarithmic depth on all-to-all
hardware, with the $\Theta(k^3/n^5)$ trainability rate as an open question based on a well-formed conjecture.
Whether either architecture is useful is now a question for experiment.


\section*{Note added}
Version~1 of this manuscript additionally introduced a second, triplet
(form-rank-$3$) magic-block encoding and claimed best-known-algorithm
classical hardness for the associated two-body expectation-value readout.
That claim was refuted by Erfan Amidi~\cite{amidi2026comment}, who observed that two-body
expectation values after passive fermionic linear optics are determined by
the input two-particle reduced density matrix---explicitly computable for
the block-product inputs---and are therefore classically computable in
$O(n^4)$ time, independently of the particle number and the form-rank. The
corresponding statement for observables of arbitrary fixed body number is
Proposition~\ref{prop:onebody_const}. We thank him for the correction. Version~2 removes the triplet encoding and
all expectation-value hardness claims, and presents the framework in
unified sample-based form; the trainability, parameter-shift, and
sampling-hardness results are unaffected, and the existence of a trainable,
classically hard expectation-value pipeline is now stated as an open
question (Section~\ref{sec:discussion}).

Version~3 precises the definition of the RBS gate for non-adjacent
qubits and of the fixed-angle Hadamard spreading circuits, sharpens the polynomial lower bound on the butterfly gradient variance, and
restates Conjecture~\ref{conj:lambda2_design} as the scalar comparison that
Claim~1 consumes.

\section*{Author contributions}
The author conceived the framework, proved the
results, and wrote the manuscript. Large language model assistants were used
for literature search, for checking algebraic manipulations and numerical
evaluations, and for copy-editing; all definitions, statements and proofs were verified by the author.

\begin{acknowledgments}
Support has been provided by PEPR EPiQ.
\end{acknowledgments}

\bibliography{references_v3}

\begin{thebibliography}{91}%
\makeatletter
\providecommand \@ifxundefined [1]{%
 \@ifx{#1\undefined}
}%
\providecommand \@ifnum [1]{%
 \ifnum #1\expandafter \@firstoftwo
 \else \expandafter \@secondoftwo
 \fi
}%
\providecommand \@ifx [1]{%
 \ifx #1\expandafter \@firstoftwo
 \else \expandafter \@secondoftwo
 \fi
}%
\providecommand \natexlab [1]{#1}%
\providecommand \enquote  [1]{``#1''}%
\providecommand \bibnamefont  [1]{#1}%
\providecommand \bibfnamefont [1]{#1}%
\providecommand \citenamefont [1]{#1}%
\providecommand \href@noop [0]{\@secondoftwo}%
\providecommand \href [0]{\begingroup \@sanitize@url \@href}%
\providecommand \@href[1]{\@@startlink{#1}\@@href}%
\providecommand \@@href[1]{\endgroup#1\@@endlink}%
\providecommand \@sanitize@url [0]{\catcode `\\12\catcode `\$12\catcode `\&12\catcode `\#12\catcode `\^12\catcode `\_12\catcode `\%12\relax}%
\providecommand \@@startlink[1]{}%
\providecommand \@@endlink[0]{}%
\providecommand \url  [0]{\begingroup\@sanitize@url \@url }%
\providecommand \@url [1]{\endgroup\@href {#1}{\urlprefix }}%
\providecommand \urlprefix  [0]{URL }%
\providecommand \Eprint [0]{\href }%
\providecommand \doibase [0]{https://doi.org/}%
\providecommand \selectlanguage [0]{\@gobble}%
\providecommand \bibinfo  [0]{\@secondoftwo}%
\providecommand \bibfield  [0]{\@secondoftwo}%
\providecommand \translation [1]{[#1]}%
\providecommand \BibitemOpen [0]{}%
\providecommand \bibitemStop [0]{}%
\providecommand \bibitemNoStop [0]{.\EOS\space}%
\providecommand \EOS [0]{\spacefactor3000\relax}%
\providecommand \BibitemShut  [1]{\csname bibitem#1\endcsname}%
\let\auto@bib@innerbib\@empty
\bibitem [{\citenamefont {Heo}\ \emph {et~al.}(2026)\citenamefont {Heo}, \citenamefont {Reardon-Smith}, \citenamefont {Oszmaniec}, \citenamefont {Zimbor{\'a}s},\ and\ \citenamefont {Oh}}]{heo2026fast}%
  \BibitemOpen
  \bibfield  {author} {\bibinfo {author} {\bibfnamefont {J.}~\bibnamefont {Heo}}, \bibinfo {author} {\bibfnamefont {O.}~\bibnamefont {Reardon-Smith}}, \bibinfo {author} {\bibfnamefont {M.}~\bibnamefont {Oszmaniec}}, \bibinfo {author} {\bibfnamefont {Z.}~\bibnamefont {Zimbor{\'a}s}},\ and\ \bibinfo {author} {\bibfnamefont {C.}~\bibnamefont {Oh}},\ }\bibfield  {title} {\bibinfo {title} {Fast classical simulation algorithms for free-fermion dynamics with magic input},\ }\href@noop {} {\bibfield  {journal} {\bibinfo  {journal} {arXiv:2609.40187}\ } (\bibinfo {year} {2026})}\BibitemShut {NoStop}%
\bibitem [{\citenamefont {Reardon-Smith}\ \emph {et~al.}(2024)\citenamefont {Reardon-Smith}, \citenamefont {Oszmaniec},\ and\ \citenamefont {Korzekwa}}]{reardonsmith2024classical}%
  \BibitemOpen
  \bibfield  {author} {\bibinfo {author} {\bibfnamefont {O.}~\bibnamefont {Reardon-Smith}}, \bibinfo {author} {\bibfnamefont {M.}~\bibnamefont {Oszmaniec}},\ and\ \bibinfo {author} {\bibfnamefont {K.}~\bibnamefont {Korzekwa}},\ }\bibfield  {title} {\bibinfo {title} {Improved simulation of quantum circuits dominated by free fermionic operations},\ }\href@noop {} {\bibfield  {journal} {\bibinfo  {journal} {Quantum 8, 1549}\ } (\bibinfo {year} {2024})}\BibitemShut {NoStop}%
\bibitem [{\citenamefont {Oszmaniec}\ \emph {et~al.}(2022)\citenamefont {Oszmaniec}, \citenamefont {Dangniam}, \citenamefont {Morales},\ and\ \citenamefont {Zimbor{\'a}s}}]{oszmaniec2022fermion}%
  \BibitemOpen
  \bibfield  {author} {\bibinfo {author} {\bibfnamefont {M.}~\bibnamefont {Oszmaniec}}, \bibinfo {author} {\bibfnamefont {N.}~\bibnamefont {Dangniam}}, \bibinfo {author} {\bibfnamefont {M.~E.~S.}\ \bibnamefont {Morales}},\ and\ \bibinfo {author} {\bibfnamefont {Z.}~\bibnamefont {Zimbor{\'a}s}},\ }\bibfield  {title} {\bibinfo {title} {Fermion sampling: a robust quantum computational advantage scheme using fermionic linear optics and magic input states},\ }\href@noop {} {\bibfield  {journal} {\bibinfo  {journal} {PRX Quantum}\ }\textbf {\bibinfo {volume} {3}},\ \bibinfo {pages} {020328} (\bibinfo {year} {2022})}\BibitemShut {NoStop}%
\bibitem [{\citenamefont {Harrow}\ \emph {et~al.}(2009)\citenamefont {Harrow}, \citenamefont {Hassidim},\ and\ \citenamefont {Lloyd}}]{harrow2009quantum}%
  \BibitemOpen
  \bibfield  {author} {\bibinfo {author} {\bibfnamefont {A.~W.}\ \bibnamefont {Harrow}}, \bibinfo {author} {\bibfnamefont {A.}~\bibnamefont {Hassidim}},\ and\ \bibinfo {author} {\bibfnamefont {S.}~\bibnamefont {Lloyd}},\ }\bibfield  {title} {\bibinfo {title} {Quantum algorithm for linear systems of equations},\ }\href@noop {} {\bibfield  {journal} {\bibinfo  {journal} {Physical Review Letters}\ }\textbf {\bibinfo {volume} {103}},\ \bibinfo {pages} {150502} (\bibinfo {year} {2009})}\BibitemShut {NoStop}%
\bibitem [{\citenamefont {Kerenidis}\ and\ \citenamefont {Prakash}(2017)}]{kerenidis2017quantum}%
  \BibitemOpen
  \bibfield  {author} {\bibinfo {author} {\bibfnamefont {I.}~\bibnamefont {Kerenidis}}\ and\ \bibinfo {author} {\bibfnamefont {A.}~\bibnamefont {Prakash}},\ }\bibfield  {title} {\bibinfo {title} {Quantum recommendation systems},\ }in\ \href@noop {} {\emph {\bibinfo {booktitle} {Innovations in Theoretical Computer Science (ITCS)}}}\ (\bibinfo {year} {2017})\ pp.\ \bibinfo {pages} {49:1--49:21}\BibitemShut {NoStop}%
\bibitem [{\citenamefont {Gily{\'e}n}\ \emph {et~al.}(2019)\citenamefont {Gily{\'e}n}, \citenamefont {Su}, \citenamefont {Low},\ and\ \citenamefont {Wiebe}}]{gilyen2019quantum}%
  \BibitemOpen
  \bibfield  {author} {\bibinfo {author} {\bibfnamefont {A.}~\bibnamefont {Gily{\'e}n}}, \bibinfo {author} {\bibfnamefont {Y.}~\bibnamefont {Su}}, \bibinfo {author} {\bibfnamefont {G.~H.}\ \bibnamefont {Low}},\ and\ \bibinfo {author} {\bibfnamefont {N.}~\bibnamefont {Wiebe}},\ }\bibfield  {title} {\bibinfo {title} {Quantum singular value transformation and beyond: exponential improvements for quantum matrix arithmetics},\ }in\ \href@noop {} {\emph {\bibinfo {booktitle} {Proceedings of the 51st Annual ACM SIGACT Symposium on Theory of Computing (STOC)}}}\ (\bibinfo {year} {2019})\ pp.\ \bibinfo {pages} {193--204},\ \bibinfo {note} {arXiv:1806.01838}\BibitemShut {NoStop}%
\bibitem [{\citenamefont {Lloyd}\ \emph {et~al.}(2013)\citenamefont {Lloyd}, \citenamefont {Mohseni},\ and\ \citenamefont {Rebentrost}}]{lloyd2013quantum}%
  \BibitemOpen
  \bibfield  {author} {\bibinfo {author} {\bibfnamefont {S.}~\bibnamefont {Lloyd}}, \bibinfo {author} {\bibfnamefont {M.}~\bibnamefont {Mohseni}},\ and\ \bibinfo {author} {\bibfnamefont {P.}~\bibnamefont {Rebentrost}},\ }\bibfield  {title} {\bibinfo {title} {Quantum algorithms for supervised and unsupervised machine learning},\ }\href@noop {} {\bibfield  {journal} {\bibinfo  {journal} {arXiv:1307.0411}\ } (\bibinfo {year} {2013})}\BibitemShut {NoStop}%
\bibitem [{\citenamefont {Lloyd}\ \emph {et~al.}(2014)\citenamefont {Lloyd}, \citenamefont {Mohseni},\ and\ \citenamefont {Rebentrost}}]{lloyd2014quantum}%
  \BibitemOpen
  \bibfield  {author} {\bibinfo {author} {\bibfnamefont {S.}~\bibnamefont {Lloyd}}, \bibinfo {author} {\bibfnamefont {M.}~\bibnamefont {Mohseni}},\ and\ \bibinfo {author} {\bibfnamefont {P.}~\bibnamefont {Rebentrost}},\ }\bibfield  {title} {\bibinfo {title} {Quantum principal component analysis},\ }\href@noop {} {\bibfield  {journal} {\bibinfo  {journal} {Nature Physics}\ }\textbf {\bibinfo {volume} {10}},\ \bibinfo {pages} {631} (\bibinfo {year} {2014})}\BibitemShut {NoStop}%
\bibitem [{\citenamefont {Kerenidis}\ \emph {et~al.}(2019)\citenamefont {Kerenidis}, \citenamefont {Landman}, \citenamefont {Luongo},\ and\ \citenamefont {Prakash}}]{kerenidis2019qmeans}%
  \BibitemOpen
  \bibfield  {author} {\bibinfo {author} {\bibfnamefont {I.}~\bibnamefont {Kerenidis}}, \bibinfo {author} {\bibfnamefont {J.}~\bibnamefont {Landman}}, \bibinfo {author} {\bibfnamefont {A.}~\bibnamefont {Luongo}},\ and\ \bibinfo {author} {\bibfnamefont {A.}~\bibnamefont {Prakash}},\ }\bibfield  {title} {\bibinfo {title} {q-means: A quantum algorithm for unsupervised machine learning},\ }in\ \href@noop {} {\emph {\bibinfo {booktitle} {Advances in Neural Information Processing Systems (NeurIPS)}}},\ Vol.~\bibinfo {volume} {32}\ (\bibinfo {year} {2019})\BibitemShut {NoStop}%
\bibitem [{\citenamefont {Kerenidis}\ \emph {et~al.}(2020)\citenamefont {Kerenidis}, \citenamefont {Landman},\ and\ \citenamefont {Prakash}}]{kerenidis2020cnn}%
  \BibitemOpen
  \bibfield  {author} {\bibinfo {author} {\bibfnamefont {I.}~\bibnamefont {Kerenidis}}, \bibinfo {author} {\bibfnamefont {J.}~\bibnamefont {Landman}},\ and\ \bibinfo {author} {\bibfnamefont {A.}~\bibnamefont {Prakash}},\ }\bibfield  {title} {\bibinfo {title} {Quantum algorithms for deep convolutional neural networks},\ }in\ \href@noop {} {\emph {\bibinfo {booktitle} {International Conference on Learning Representations (ICLR)}}}\ (\bibinfo {year} {2020})\BibitemShut {NoStop}%
\bibitem [{\citenamefont {Tang}(2019)}]{tang2019recommendation}%
  \BibitemOpen
  \bibfield  {author} {\bibinfo {author} {\bibfnamefont {E.}~\bibnamefont {Tang}},\ }\bibfield  {title} {\bibinfo {title} {A quantum-inspired classical algorithm for recommendation systems},\ }in\ \href@noop {} {\emph {\bibinfo {booktitle} {Proceedings of the 51st Annual ACM SIGACT Symposium on Theory of Computing (STOC)}}}\ (\bibinfo {year} {2019})\ pp.\ \bibinfo {pages} {217--228},\ \bibinfo {note} {arXiv:1807.04271}\BibitemShut {NoStop}%
\bibitem [{\citenamefont {Tang}(2021)}]{tang2021pca}%
  \BibitemOpen
  \bibfield  {author} {\bibinfo {author} {\bibfnamefont {E.}~\bibnamefont {Tang}},\ }\bibfield  {title} {\bibinfo {title} {Quantum principal component analysis only achieves an exponential speedup because of its state preparation assumptions},\ }\href@noop {} {\bibfield  {journal} {\bibinfo  {journal} {Physical Review Letters}\ }\textbf {\bibinfo {volume} {127}},\ \bibinfo {pages} {060503} (\bibinfo {year} {2021})},\ \bibinfo {note} {arXiv:1811.00414}\BibitemShut {NoStop}%
\bibitem [{\citenamefont {Chia}\ \emph {et~al.}(2022)\citenamefont {Chia}, \citenamefont {Gily{\'e}n}, \citenamefont {Li}, \citenamefont {Lin}, \citenamefont {Tang},\ and\ \citenamefont {Wang}}]{chia2022svt}%
  \BibitemOpen
  \bibfield  {author} {\bibinfo {author} {\bibfnamefont {N.-H.}\ \bibnamefont {Chia}}, \bibinfo {author} {\bibfnamefont {A.}~\bibnamefont {Gily{\'e}n}}, \bibinfo {author} {\bibfnamefont {T.}~\bibnamefont {Li}}, \bibinfo {author} {\bibfnamefont {H.-H.}\ \bibnamefont {Lin}}, \bibinfo {author} {\bibfnamefont {E.}~\bibnamefont {Tang}},\ and\ \bibinfo {author} {\bibfnamefont {C.}~\bibnamefont {Wang}},\ }\bibfield  {title} {\bibinfo {title} {Sampling-based sublinear low-rank matrix arithmetic framework for dequantizing quantum machine learning},\ }\href@noop {} {\bibfield  {journal} {\bibinfo  {journal} {Journal of the ACM}\ }\textbf {\bibinfo {volume} {69}},\ \bibinfo {pages} {1} (\bibinfo {year} {2022})},\ \bibinfo {note} {arXiv:1910.06151}\BibitemShut {NoStop}%
\bibitem [{\citenamefont {Zhao}\ \emph {et~al.}(2026)\citenamefont {Zhao}, \citenamefont {Zlokapa}, \citenamefont {Neven}, \citenamefont {Babbush}, \citenamefont {Preskill}, \citenamefont {McClean},\ and\ \citenamefont {Huang}}]{zhao2026exponential}%
  \BibitemOpen
  \bibfield  {author} {\bibinfo {author} {\bibfnamefont {H.}~\bibnamefont {Zhao}}, \bibinfo {author} {\bibfnamefont {A.}~\bibnamefont {Zlokapa}}, \bibinfo {author} {\bibfnamefont {H.}~\bibnamefont {Neven}}, \bibinfo {author} {\bibfnamefont {R.}~\bibnamefont {Babbush}}, \bibinfo {author} {\bibfnamefont {J.}~\bibnamefont {Preskill}}, \bibinfo {author} {\bibfnamefont {J.~R.}\ \bibnamefont {McClean}},\ and\ \bibinfo {author} {\bibfnamefont {H.-Y.}\ \bibnamefont {Huang}},\ }\bibfield  {title} {\bibinfo {title} {Exponential quantum advantage in processing massive classical data},\ }\href@noop {} {\bibfield  {journal} {\bibinfo  {journal} {arXiv preprint arXiv:2604.07639}\ } (\bibinfo {year} {2026})}\BibitemShut {NoStop}%
\bibitem [{\citenamefont {Mitarai}\ \emph {et~al.}(2018)\citenamefont {Mitarai}, \citenamefont {Negoro}, \citenamefont {Kitagawa},\ and\ \citenamefont {Fujii}}]{mitarai2018quantum}%
  \BibitemOpen
  \bibfield  {author} {\bibinfo {author} {\bibfnamefont {K.}~\bibnamefont {Mitarai}}, \bibinfo {author} {\bibfnamefont {M.}~\bibnamefont {Negoro}}, \bibinfo {author} {\bibfnamefont {M.}~\bibnamefont {Kitagawa}},\ and\ \bibinfo {author} {\bibfnamefont {K.}~\bibnamefont {Fujii}},\ }\bibfield  {title} {\bibinfo {title} {Quantum circuit learning},\ }\href@noop {} {\bibfield  {journal} {\bibinfo  {journal} {Physical Review A}\ }\textbf {\bibinfo {volume} {98}},\ \bibinfo {pages} {032309} (\bibinfo {year} {2018})}\BibitemShut {NoStop}%
\bibitem [{\citenamefont {Schuld}\ and\ \citenamefont {Killoran}(2019)}]{schuld2019quantum}%
  \BibitemOpen
  \bibfield  {author} {\bibinfo {author} {\bibfnamefont {M.}~\bibnamefont {Schuld}}\ and\ \bibinfo {author} {\bibfnamefont {N.}~\bibnamefont {Killoran}},\ }\bibfield  {title} {\bibinfo {title} {Quantum machine learning in feature {Hilbert} spaces},\ }\href@noop {} {\bibfield  {journal} {\bibinfo  {journal} {Physical Review Letters}\ }\textbf {\bibinfo {volume} {122}},\ \bibinfo {pages} {040504} (\bibinfo {year} {2019})}\BibitemShut {NoStop}%
\bibitem [{\citenamefont {Benedetti}\ \emph {et~al.}(2019)\citenamefont {Benedetti}, \citenamefont {Lloyd}, \citenamefont {Sack},\ and\ \citenamefont {Fiorentini}}]{benedetti2019parameterized}%
  \BibitemOpen
  \bibfield  {author} {\bibinfo {author} {\bibfnamefont {M.}~\bibnamefont {Benedetti}}, \bibinfo {author} {\bibfnamefont {E.}~\bibnamefont {Lloyd}}, \bibinfo {author} {\bibfnamefont {S.}~\bibnamefont {Sack}},\ and\ \bibinfo {author} {\bibfnamefont {M.}~\bibnamefont {Fiorentini}},\ }\bibfield  {title} {\bibinfo {title} {Parameterized quantum circuits as machine learning models},\ }\href@noop {} {\bibfield  {journal} {\bibinfo  {journal} {Quantum Science and Technology}\ }\textbf {\bibinfo {volume} {4}},\ \bibinfo {pages} {043001} (\bibinfo {year} {2019})}\BibitemShut {NoStop}%
\bibitem [{\citenamefont {P{\'e}rez-Salinas}\ \emph {et~al.}(2020)\citenamefont {P{\'e}rez-Salinas}, \citenamefont {Cervera-Lierta}, \citenamefont {Gil-Fuster},\ and\ \citenamefont {Latorre}}]{perez2020data}%
  \BibitemOpen
  \bibfield  {author} {\bibinfo {author} {\bibfnamefont {A.}~\bibnamefont {P{\'e}rez-Salinas}}, \bibinfo {author} {\bibfnamefont {A.}~\bibnamefont {Cervera-Lierta}}, \bibinfo {author} {\bibfnamefont {E.}~\bibnamefont {Gil-Fuster}},\ and\ \bibinfo {author} {\bibfnamefont {J.~I.}\ \bibnamefont {Latorre}},\ }\bibfield  {title} {\bibinfo {title} {Data re-uploading for a universal quantum classifier},\ }\href@noop {} {\bibfield  {journal} {\bibinfo  {journal} {Quantum}\ }\textbf {\bibinfo {volume} {4}},\ \bibinfo {pages} {226} (\bibinfo {year} {2020})}\BibitemShut {NoStop}%
\bibitem [{\citenamefont {Biamonte}\ \emph {et~al.}(2017)\citenamefont {Biamonte}, \citenamefont {Wittek}, \citenamefont {Pancotti}, \citenamefont {Rebentrost}, \citenamefont {Wiebe},\ and\ \citenamefont {Lloyd}}]{biamonte2017quantum}%
  \BibitemOpen
  \bibfield  {author} {\bibinfo {author} {\bibfnamefont {J.}~\bibnamefont {Biamonte}}, \bibinfo {author} {\bibfnamefont {P.}~\bibnamefont {Wittek}}, \bibinfo {author} {\bibfnamefont {N.}~\bibnamefont {Pancotti}}, \bibinfo {author} {\bibfnamefont {P.}~\bibnamefont {Rebentrost}}, \bibinfo {author} {\bibfnamefont {N.}~\bibnamefont {Wiebe}},\ and\ \bibinfo {author} {\bibfnamefont {S.}~\bibnamefont {Lloyd}},\ }\bibfield  {title} {\bibinfo {title} {Quantum machine learning},\ }\href@noop {} {\bibfield  {journal} {\bibinfo  {journal} {Nature}\ }\textbf {\bibinfo {volume} {549}},\ \bibinfo {pages} {195} (\bibinfo {year} {2017})}\BibitemShut {NoStop}%
\bibitem [{\citenamefont {Cerezo}\ \emph {et~al.}(2021{\natexlab{a}})\citenamefont {Cerezo}, \citenamefont {Arrasmith}, \citenamefont {Babbush}, \citenamefont {Benjamin}, \citenamefont {Endo}, \citenamefont {Fujii}, \citenamefont {McClean}, \citenamefont {Mitarai}, \citenamefont {Yuan}, \citenamefont {Cincio} \emph {et~al.}}]{cerezo2021variational}%
  \BibitemOpen
  \bibfield  {author} {\bibinfo {author} {\bibfnamefont {M.}~\bibnamefont {Cerezo}}, \bibinfo {author} {\bibfnamefont {A.}~\bibnamefont {Arrasmith}}, \bibinfo {author} {\bibfnamefont {R.}~\bibnamefont {Babbush}}, \bibinfo {author} {\bibfnamefont {S.~C.}\ \bibnamefont {Benjamin}}, \bibinfo {author} {\bibfnamefont {S.}~\bibnamefont {Endo}}, \bibinfo {author} {\bibfnamefont {K.}~\bibnamefont {Fujii}}, \bibinfo {author} {\bibfnamefont {J.~R.}\ \bibnamefont {McClean}}, \bibinfo {author} {\bibfnamefont {K.}~\bibnamefont {Mitarai}}, \bibinfo {author} {\bibfnamefont {X.}~\bibnamefont {Yuan}}, \bibinfo {author} {\bibfnamefont {L.}~\bibnamefont {Cincio}}, \emph {et~al.},\ }\bibfield  {title} {\bibinfo {title} {Variational quantum algorithms},\ }\href@noop {} {\bibfield  {journal} {\bibinfo  {journal} {Nature Reviews Physics}\ }\textbf {\bibinfo {volume} {3}},\ \bibinfo {pages} {625} (\bibinfo {year} {2021}{\natexlab{a}})}\BibitemShut {NoStop}%
\bibitem [{\citenamefont {Farhi}\ and\ \citenamefont {Neven}(2018)}]{farhi2018classification}%
  \BibitemOpen
  \bibfield  {author} {\bibinfo {author} {\bibfnamefont {E.}~\bibnamefont {Farhi}}\ and\ \bibinfo {author} {\bibfnamefont {H.}~\bibnamefont {Neven}},\ }\bibfield  {title} {\bibinfo {title} {Classification with quantum neural networks on near term processors},\ }\href@noop {} {\bibfield  {journal} {\bibinfo  {journal} {arXiv:1802.06002}\ } (\bibinfo {year} {2018})}\BibitemShut {NoStop}%
\bibitem [{\citenamefont {Havl{\'i}{\v{c}}ek}\ \emph {et~al.}(2019)\citenamefont {Havl{\'i}{\v{c}}ek}, \citenamefont {C{\'o}rcoles}, \citenamefont {Temme}, \citenamefont {Harrow}, \citenamefont {Kandala}, \citenamefont {Chow},\ and\ \citenamefont {Gambetta}}]{havlicek2019supervised}%
  \BibitemOpen
  \bibfield  {author} {\bibinfo {author} {\bibfnamefont {V.}~\bibnamefont {Havl{\'i}{\v{c}}ek}}, \bibinfo {author} {\bibfnamefont {A.~D.}\ \bibnamefont {C{\'o}rcoles}}, \bibinfo {author} {\bibfnamefont {K.}~\bibnamefont {Temme}}, \bibinfo {author} {\bibfnamefont {A.~W.}\ \bibnamefont {Harrow}}, \bibinfo {author} {\bibfnamefont {A.}~\bibnamefont {Kandala}}, \bibinfo {author} {\bibfnamefont {J.~M.}\ \bibnamefont {Chow}},\ and\ \bibinfo {author} {\bibfnamefont {J.~M.}\ \bibnamefont {Gambetta}},\ }\bibfield  {title} {\bibinfo {title} {Supervised learning with quantum-enhanced feature spaces},\ }\href@noop {} {\bibfield  {journal} {\bibinfo  {journal} {Nature}\ }\textbf {\bibinfo {volume} {567}},\ \bibinfo {pages} {209} (\bibinfo {year} {2019})}\BibitemShut {NoStop}%
\bibitem [{\citenamefont {Abbas}\ \emph {et~al.}(2021)\citenamefont {Abbas}, \citenamefont {Sutter}, \citenamefont {Zoufal}, \citenamefont {Lucchi}, \citenamefont {Figalli},\ and\ \citenamefont {Woerner}}]{abbas2021power}%
  \BibitemOpen
  \bibfield  {author} {\bibinfo {author} {\bibfnamefont {A.}~\bibnamefont {Abbas}}, \bibinfo {author} {\bibfnamefont {D.}~\bibnamefont {Sutter}}, \bibinfo {author} {\bibfnamefont {C.}~\bibnamefont {Zoufal}}, \bibinfo {author} {\bibfnamefont {A.}~\bibnamefont {Lucchi}}, \bibinfo {author} {\bibfnamefont {A.}~\bibnamefont {Figalli}},\ and\ \bibinfo {author} {\bibfnamefont {S.}~\bibnamefont {Woerner}},\ }\bibfield  {title} {\bibinfo {title} {The power of quantum neural networks},\ }\href@noop {} {\bibfield  {journal} {\bibinfo  {journal} {Nature Computational Science}\ }\textbf {\bibinfo {volume} {1}},\ \bibinfo {pages} {403} (\bibinfo {year} {2021})}\BibitemShut {NoStop}%
\bibitem [{\citenamefont {Landman}\ \emph {et~al.}(2022)\citenamefont {Landman}, \citenamefont {Mathur}, \citenamefont {Li}, \citenamefont {Strahm}, \citenamefont {Kazdaghli}, \citenamefont {Prakash},\ and\ \citenamefont {Kerenidis}}]{landman2022quantum}%
  \BibitemOpen
  \bibfield  {author} {\bibinfo {author} {\bibfnamefont {J.}~\bibnamefont {Landman}}, \bibinfo {author} {\bibfnamefont {N.}~\bibnamefont {Mathur}}, \bibinfo {author} {\bibfnamefont {Y.~Y.}\ \bibnamefont {Li}}, \bibinfo {author} {\bibfnamefont {M.}~\bibnamefont {Strahm}}, \bibinfo {author} {\bibfnamefont {S.}~\bibnamefont {Kazdaghli}}, \bibinfo {author} {\bibfnamefont {A.}~\bibnamefont {Prakash}},\ and\ \bibinfo {author} {\bibfnamefont {I.}~\bibnamefont {Kerenidis}},\ }\bibfield  {title} {\bibinfo {title} {Quantum methods for neural networks and application to medical image classification},\ }\href@noop {} {\bibfield  {journal} {\bibinfo  {journal} {Quantum}\ }\textbf {\bibinfo {volume} {6}},\ \bibinfo {pages} {881} (\bibinfo {year} {2022})}\BibitemShut {NoStop}%
\bibitem [{\citenamefont {Dunjko}\ \emph {et~al.}(2016)\citenamefont {Dunjko}, \citenamefont {Taylor},\ and\ \citenamefont {Briegel}}]{dunjko2016quantum}%
  \BibitemOpen
  \bibfield  {author} {\bibinfo {author} {\bibfnamefont {V.}~\bibnamefont {Dunjko}}, \bibinfo {author} {\bibfnamefont {J.~M.}\ \bibnamefont {Taylor}},\ and\ \bibinfo {author} {\bibfnamefont {H.~J.}\ \bibnamefont {Briegel}},\ }\bibfield  {title} {\bibinfo {title} {Quantum-enhanced machine learning},\ }\href@noop {} {\bibfield  {journal} {\bibinfo  {journal} {Physical Review Letters}\ }\textbf {\bibinfo {volume} {117}},\ \bibinfo {pages} {130501} (\bibinfo {year} {2016})}\BibitemShut {NoStop}%
\bibitem [{\citenamefont {Jerbi}\ \emph {et~al.}(2021)\citenamefont {Jerbi}, \citenamefont {Gyurik}, \citenamefont {Marshall}, \citenamefont {Briegel},\ and\ \citenamefont {Dunjko}}]{jerbi2021parametrized}%
  \BibitemOpen
  \bibfield  {author} {\bibinfo {author} {\bibfnamefont {S.}~\bibnamefont {Jerbi}}, \bibinfo {author} {\bibfnamefont {C.}~\bibnamefont {Gyurik}}, \bibinfo {author} {\bibfnamefont {S.}~\bibnamefont {Marshall}}, \bibinfo {author} {\bibfnamefont {H.}~\bibnamefont {Briegel}},\ and\ \bibinfo {author} {\bibfnamefont {V.}~\bibnamefont {Dunjko}},\ }\bibfield  {title} {\bibinfo {title} {Parametrized quantum policies for reinforcement learning},\ }in\ \href@noop {} {\emph {\bibinfo {booktitle} {Advances in Neural Information Processing Systems (NeurIPS)}}},\ Vol.~\bibinfo {volume} {34}\ (\bibinfo {year} {2021})\ pp.\ \bibinfo {pages} {28362--28375}\BibitemShut {NoStop}%
\bibitem [{\citenamefont {Cherrat}\ \emph {et~al.}(2023)\citenamefont {Cherrat}, \citenamefont {Raj}, \citenamefont {Kerenidis}, \citenamefont {Shekhar}, \citenamefont {Wood}, \citenamefont {Dee}, \citenamefont {Chakrabarti}, \citenamefont {Chen}, \citenamefont {Herman}, \citenamefont {Hu} \emph {et~al.}}]{cherrat2023hedging}%
  \BibitemOpen
  \bibfield  {author} {\bibinfo {author} {\bibfnamefont {E.~A.}\ \bibnamefont {Cherrat}}, \bibinfo {author} {\bibfnamefont {S.}~\bibnamefont {Raj}}, \bibinfo {author} {\bibfnamefont {I.}~\bibnamefont {Kerenidis}}, \bibinfo {author} {\bibfnamefont {A.}~\bibnamefont {Shekhar}}, \bibinfo {author} {\bibfnamefont {B.}~\bibnamefont {Wood}}, \bibinfo {author} {\bibfnamefont {J.}~\bibnamefont {Dee}}, \bibinfo {author} {\bibfnamefont {S.}~\bibnamefont {Chakrabarti}}, \bibinfo {author} {\bibfnamefont {R.}~\bibnamefont {Chen}}, \bibinfo {author} {\bibfnamefont {D.}~\bibnamefont {Herman}}, \bibinfo {author} {\bibfnamefont {S.}~\bibnamefont {Hu}}, \emph {et~al.},\ }\bibfield  {title} {\bibinfo {title} {Quantum deep hedging},\ }\href@noop {} {\bibfield  {journal} {\bibinfo  {journal} {Quantum}\ }\textbf {\bibinfo {volume} {7}},\ \bibinfo {pages} {1191} (\bibinfo {year} {2023})}\BibitemShut {NoStop}%
\bibitem [{\citenamefont {Preskill}(2018)}]{preskill2018nisq}%
  \BibitemOpen
  \bibfield  {author} {\bibinfo {author} {\bibfnamefont {J.}~\bibnamefont {Preskill}},\ }\bibfield  {title} {\bibinfo {title} {Quantum computing in the {NISQ} era and beyond},\ }\href@noop {} {\bibfield  {journal} {\bibinfo  {journal} {Quantum}\ }\textbf {\bibinfo {volume} {2}},\ \bibinfo {pages} {79} (\bibinfo {year} {2018})}\BibitemShut {NoStop}%
\bibitem [{\citenamefont {Schuld}\ \emph {et~al.}(2014)\citenamefont {Schuld}, \citenamefont {Sinayskiy},\ and\ \citenamefont {Petruccione}}]{schuld2014quest}%
  \BibitemOpen
  \bibfield  {author} {\bibinfo {author} {\bibfnamefont {M.}~\bibnamefont {Schuld}}, \bibinfo {author} {\bibfnamefont {I.}~\bibnamefont {Sinayskiy}},\ and\ \bibinfo {author} {\bibfnamefont {F.}~\bibnamefont {Petruccione}},\ }\bibfield  {title} {\bibinfo {title} {The quest for a quantum neural network},\ }\href@noop {} {\bibfield  {journal} {\bibinfo  {journal} {Quantum Information Processing}\ }\textbf {\bibinfo {volume} {13}},\ \bibinfo {pages} {2567} (\bibinfo {year} {2014})}\BibitemShut {NoStop}%
\bibitem [{\citenamefont {McClean}\ \emph {et~al.}(2018)\citenamefont {McClean}, \citenamefont {Boixo}, \citenamefont {Smelyanskiy}, \citenamefont {Babbush},\ and\ \citenamefont {Neven}}]{mcclean2018barren}%
  \BibitemOpen
  \bibfield  {author} {\bibinfo {author} {\bibfnamefont {J.~R.}\ \bibnamefont {McClean}}, \bibinfo {author} {\bibfnamefont {S.}~\bibnamefont {Boixo}}, \bibinfo {author} {\bibfnamefont {V.~N.}\ \bibnamefont {Smelyanskiy}}, \bibinfo {author} {\bibfnamefont {R.}~\bibnamefont {Babbush}},\ and\ \bibinfo {author} {\bibfnamefont {H.}~\bibnamefont {Neven}},\ }\bibfield  {title} {\bibinfo {title} {Barren plateaus in quantum neural network training landscapes},\ }\href@noop {} {\bibfield  {journal} {\bibinfo  {journal} {Nature Communications}\ }\textbf {\bibinfo {volume} {9}},\ \bibinfo {pages} {4812} (\bibinfo {year} {2018})}\BibitemShut {NoStop}%
\bibitem [{\citenamefont {Cerezo}\ \emph {et~al.}(2021{\natexlab{b}})\citenamefont {Cerezo}, \citenamefont {Sone}, \citenamefont {Volkoff}, \citenamefont {Cincio},\ and\ \citenamefont {Coles}}]{cerezo2021cost}%
  \BibitemOpen
  \bibfield  {author} {\bibinfo {author} {\bibfnamefont {M.}~\bibnamefont {Cerezo}}, \bibinfo {author} {\bibfnamefont {A.}~\bibnamefont {Sone}}, \bibinfo {author} {\bibfnamefont {T.}~\bibnamefont {Volkoff}}, \bibinfo {author} {\bibfnamefont {L.}~\bibnamefont {Cincio}},\ and\ \bibinfo {author} {\bibfnamefont {P.~J.}\ \bibnamefont {Coles}},\ }\bibfield  {title} {\bibinfo {title} {Cost function dependent barren plateaus in shallow parametrized quantum circuits},\ }\href@noop {} {\bibfield  {journal} {\bibinfo  {journal} {Nature Communications}\ }\textbf {\bibinfo {volume} {12}},\ \bibinfo {pages} {1791} (\bibinfo {year} {2021}{\natexlab{b}})}\BibitemShut {NoStop}%
\bibitem [{\citenamefont {Cong}\ \emph {et~al.}(2019)\citenamefont {Cong}, \citenamefont {Choi},\ and\ \citenamefont {Lukin}}]{cong2019qcnn}%
  \BibitemOpen
  \bibfield  {author} {\bibinfo {author} {\bibfnamefont {I.}~\bibnamefont {Cong}}, \bibinfo {author} {\bibfnamefont {S.}~\bibnamefont {Choi}},\ and\ \bibinfo {author} {\bibfnamefont {M.~D.}\ \bibnamefont {Lukin}},\ }\bibfield  {title} {\bibinfo {title} {Quantum convolutional neural networks},\ }\href@noop {} {\bibfield  {journal} {\bibinfo  {journal} {Nature Physics}\ }\textbf {\bibinfo {volume} {15}},\ \bibinfo {pages} {1273} (\bibinfo {year} {2019})}\BibitemShut {NoStop}%
\bibitem [{\citenamefont {Pesah}\ \emph {et~al.}(2021)\citenamefont {Pesah}, \citenamefont {Cerezo}, \citenamefont {Wang}, \citenamefont {Volkoff}, \citenamefont {Sornborger},\ and\ \citenamefont {Coles}}]{pesah2021absence}%
  \BibitemOpen
  \bibfield  {author} {\bibinfo {author} {\bibfnamefont {A.}~\bibnamefont {Pesah}}, \bibinfo {author} {\bibfnamefont {M.}~\bibnamefont {Cerezo}}, \bibinfo {author} {\bibfnamefont {S.}~\bibnamefont {Wang}}, \bibinfo {author} {\bibfnamefont {T.}~\bibnamefont {Volkoff}}, \bibinfo {author} {\bibfnamefont {A.~T.}\ \bibnamefont {Sornborger}},\ and\ \bibinfo {author} {\bibfnamefont {P.~J.}\ \bibnamefont {Coles}},\ }\bibfield  {title} {\bibinfo {title} {Absence of barren plateaus in quantum convolutional neural networks},\ }\href@noop {} {\bibfield  {journal} {\bibinfo  {journal} {Physical Review X}\ }\textbf {\bibinfo {volume} {11}},\ \bibinfo {pages} {041011} (\bibinfo {year} {2021})}\BibitemShut {NoStop}%
\bibitem [{\citenamefont {Grant}\ \emph {et~al.}(2018)\citenamefont {Grant}, \citenamefont {Benedetti}, \citenamefont {Cao}, \citenamefont {Hallam}, \citenamefont {Lockhart}, \citenamefont {Stojevic}, \citenamefont {Green},\ and\ \citenamefont {Severini}}]{grant2018hierarchical}%
  \BibitemOpen
  \bibfield  {author} {\bibinfo {author} {\bibfnamefont {E.}~\bibnamefont {Grant}}, \bibinfo {author} {\bibfnamefont {M.}~\bibnamefont {Benedetti}}, \bibinfo {author} {\bibfnamefont {S.}~\bibnamefont {Cao}}, \bibinfo {author} {\bibfnamefont {A.}~\bibnamefont {Hallam}}, \bibinfo {author} {\bibfnamefont {J.}~\bibnamefont {Lockhart}}, \bibinfo {author} {\bibfnamefont {V.}~\bibnamefont {Stojevic}}, \bibinfo {author} {\bibfnamefont {A.~G.}\ \bibnamefont {Green}},\ and\ \bibinfo {author} {\bibfnamefont {S.}~\bibnamefont {Severini}},\ }\bibfield  {title} {\bibinfo {title} {Hierarchical quantum classifiers},\ }\href@noop {} {\bibfield  {journal} {\bibinfo  {journal} {npj Quantum Information}\ }\textbf {\bibinfo {volume} {4}},\ \bibinfo {pages} {65} (\bibinfo {year} {2018})}\BibitemShut {NoStop}%
\bibitem [{\citenamefont {Meyer}\ \emph {et~al.}(2023)\citenamefont {Meyer}, \citenamefont {Mularski}, \citenamefont {Gil-Fuster}, \citenamefont {Mele}, \citenamefont {Arzani}, \citenamefont {Wilms},\ and\ \citenamefont {Eisert}}]{meyer2023exploiting}%
  \BibitemOpen
  \bibfield  {author} {\bibinfo {author} {\bibfnamefont {J.~J.}\ \bibnamefont {Meyer}}, \bibinfo {author} {\bibfnamefont {M.}~\bibnamefont {Mularski}}, \bibinfo {author} {\bibfnamefont {E.}~\bibnamefont {Gil-Fuster}}, \bibinfo {author} {\bibfnamefont {A.~A.}\ \bibnamefont {Mele}}, \bibinfo {author} {\bibfnamefont {F.}~\bibnamefont {Arzani}}, \bibinfo {author} {\bibfnamefont {A.}~\bibnamefont {Wilms}},\ and\ \bibinfo {author} {\bibfnamefont {J.}~\bibnamefont {Eisert}},\ }\bibfield  {title} {\bibinfo {title} {Exploiting symmetry in variational quantum machine learning},\ }\href@noop {} {\bibfield  {journal} {\bibinfo  {journal} {PRX Quantum}\ }\textbf {\bibinfo {volume} {4}},\ \bibinfo {pages} {010328} (\bibinfo {year} {2023})}\BibitemShut {NoStop}%
\bibitem [{\citenamefont {Larocca}\ \emph {et~al.}(2022{\natexlab{a}})\citenamefont {Larocca}, \citenamefont {Sauvage}, \citenamefont {Sbahi}, \citenamefont {Verdon}, \citenamefont {Coles},\ and\ \citenamefont {Cerezo}}]{larocca2022group}%
  \BibitemOpen
  \bibfield  {author} {\bibinfo {author} {\bibfnamefont {M.}~\bibnamefont {Larocca}}, \bibinfo {author} {\bibfnamefont {F.}~\bibnamefont {Sauvage}}, \bibinfo {author} {\bibfnamefont {F.~M.}\ \bibnamefont {Sbahi}}, \bibinfo {author} {\bibfnamefont {G.}~\bibnamefont {Verdon}}, \bibinfo {author} {\bibfnamefont {P.~J.}\ \bibnamefont {Coles}},\ and\ \bibinfo {author} {\bibfnamefont {M.}~\bibnamefont {Cerezo}},\ }\bibfield  {title} {\bibinfo {title} {Group-invariant quantum machine learning},\ }\href@noop {} {\bibfield  {journal} {\bibinfo  {journal} {PRX Quantum}\ }\textbf {\bibinfo {volume} {3}},\ \bibinfo {pages} {030341} (\bibinfo {year} {2022}{\natexlab{a}})}\BibitemShut {NoStop}%
\bibitem [{\citenamefont {Fontana}\ \emph {et~al.}(2024)\citenamefont {Fontana}, \citenamefont {Herman}, \citenamefont {Chakrabarti}, \citenamefont {Kumar}, \citenamefont {Yalovetzky}, \citenamefont {Heredge}, \citenamefont {Sureshbabu},\ and\ \citenamefont {Pistoia}}]{fontana2024adjoint}%
  \BibitemOpen
  \bibfield  {author} {\bibinfo {author} {\bibfnamefont {E.}~\bibnamefont {Fontana}}, \bibinfo {author} {\bibfnamefont {D.}~\bibnamefont {Herman}}, \bibinfo {author} {\bibfnamefont {S.}~\bibnamefont {Chakrabarti}}, \bibinfo {author} {\bibfnamefont {N.}~\bibnamefont {Kumar}}, \bibinfo {author} {\bibfnamefont {R.}~\bibnamefont {Yalovetzky}}, \bibinfo {author} {\bibfnamefont {J.}~\bibnamefont {Heredge}}, \bibinfo {author} {\bibfnamefont {S.~H.}\ \bibnamefont {Sureshbabu}},\ and\ \bibinfo {author} {\bibfnamefont {M.}~\bibnamefont {Pistoia}},\ }\bibfield  {title} {\bibinfo {title} {Characterizing barren plateaus in quantum ans{\"a}tze with the adjoint representation},\ }\href@noop {} {\bibfield  {journal} {\bibinfo  {journal} {Nature Communications}\ }\textbf {\bibinfo {volume} {15}},\ \bibinfo {pages} {7171} (\bibinfo {year} {2024})}\BibitemShut {NoStop}%
\bibitem [{\citenamefont {Larocca}\ \emph {et~al.}(2022{\natexlab{b}})\citenamefont {Larocca}, \citenamefont {Czarnik}, \citenamefont {Sharma}, \citenamefont {Muraleedharan}, \citenamefont {Coles},\ and\ \citenamefont {Cerezo}}]{larocca2021diagnosing}%
  \BibitemOpen
  \bibfield  {author} {\bibinfo {author} {\bibfnamefont {M.}~\bibnamefont {Larocca}}, \bibinfo {author} {\bibfnamefont {P.}~\bibnamefont {Czarnik}}, \bibinfo {author} {\bibfnamefont {K.}~\bibnamefont {Sharma}}, \bibinfo {author} {\bibfnamefont {G.}~\bibnamefont {Muraleedharan}}, \bibinfo {author} {\bibfnamefont {P.~J.}\ \bibnamefont {Coles}},\ and\ \bibinfo {author} {\bibfnamefont {M.}~\bibnamefont {Cerezo}},\ }\bibfield  {title} {\bibinfo {title} {Diagnosing barren plateaus with tools from quantum optimal control},\ }\href@noop {} {\bibfield  {journal} {\bibinfo  {journal} {Quantum}\ }\textbf {\bibinfo {volume} {6}},\ \bibinfo {pages} {824} (\bibinfo {year} {2022}{\natexlab{b}})}\BibitemShut {NoStop}%
\bibitem [{\citenamefont {Cerezo}\ \emph {et~al.}(2025)\citenamefont {Cerezo}, \citenamefont {Larocca}, \citenamefont {Garc{\'i}a-Mart{\'i}n}, \citenamefont {Diaz}, \citenamefont {Braccia}, \citenamefont {Fontana}, \citenamefont {Rudolph}, \citenamefont {Bermejo}, \citenamefont {Ijaz}, \citenamefont {Thanasilp} \emph {et~al.}}]{cerezo2025does}%
  \BibitemOpen
  \bibfield  {author} {\bibinfo {author} {\bibfnamefont {M.}~\bibnamefont {Cerezo}}, \bibinfo {author} {\bibfnamefont {M.}~\bibnamefont {Larocca}}, \bibinfo {author} {\bibfnamefont {D.}~\bibnamefont {Garc{\'i}a-Mart{\'i}n}}, \bibinfo {author} {\bibfnamefont {N.~L.}\ \bibnamefont {Diaz}}, \bibinfo {author} {\bibfnamefont {P.}~\bibnamefont {Braccia}}, \bibinfo {author} {\bibfnamefont {E.}~\bibnamefont {Fontana}}, \bibinfo {author} {\bibfnamefont {M.~S.}\ \bibnamefont {Rudolph}}, \bibinfo {author} {\bibfnamefont {P.}~\bibnamefont {Bermejo}}, \bibinfo {author} {\bibfnamefont {A.}~\bibnamefont {Ijaz}}, \bibinfo {author} {\bibfnamefont {S.}~\bibnamefont {Thanasilp}}, \emph {et~al.},\ }\bibfield  {title} {\bibinfo {title} {Does provable absence of barren plateaus imply classical simulability?},\ }\href@noop {} {\bibfield  {journal} {\bibinfo  {journal} {Nature Communications}\ }\textbf {\bibinfo {volume} {16}},\ \bibinfo {pages} {7907} (\bibinfo {year} {2025})}\BibitemShut {NoStop}%
\bibitem [{\citenamefont {Bermejo}\ \emph {et~al.}(2026)\citenamefont {Bermejo}, \citenamefont {Braccia}, \citenamefont {Rudolph}, \citenamefont {Holmes}, \citenamefont {Cincio},\ and\ \citenamefont {Cerezo}}]{bermejo2024qcnn}%
  \BibitemOpen
  \bibfield  {author} {\bibinfo {author} {\bibfnamefont {P.}~\bibnamefont {Bermejo}}, \bibinfo {author} {\bibfnamefont {P.}~\bibnamefont {Braccia}}, \bibinfo {author} {\bibfnamefont {M.~S.}\ \bibnamefont {Rudolph}}, \bibinfo {author} {\bibfnamefont {Z.}~\bibnamefont {Holmes}}, \bibinfo {author} {\bibfnamefont {L.}~\bibnamefont {Cincio}},\ and\ \bibinfo {author} {\bibfnamefont {M.}~\bibnamefont {Cerezo}},\ }\bibfield  {title} {\bibinfo {title} {Quantum convolutional neural networks are (effectively) classically simulable},\ }\href@noop {} {\bibfield  {journal} {\bibinfo  {journal} {PRX Quantum 7, 020304}\ } (\bibinfo {year} {2026})}\BibitemShut {NoStop}%
\bibitem [{\citenamefont {Anschuetz}\ \emph {et~al.}(2023)\citenamefont {Anschuetz}, \citenamefont {Bauer}, \citenamefont {Kiani},\ and\ \citenamefont {Lloyd}}]{anschuetz2023symmetric}%
  \BibitemOpen
  \bibfield  {author} {\bibinfo {author} {\bibfnamefont {E.~R.}\ \bibnamefont {Anschuetz}}, \bibinfo {author} {\bibfnamefont {A.}~\bibnamefont {Bauer}}, \bibinfo {author} {\bibfnamefont {B.~T.}\ \bibnamefont {Kiani}},\ and\ \bibinfo {author} {\bibfnamefont {S.}~\bibnamefont {Lloyd}},\ }\bibfield  {title} {\bibinfo {title} {Efficient classical algorithms for simulating symmetric quantum systems},\ }\href@noop {} {\bibfield  {journal} {\bibinfo  {journal} {Quantum}\ }\textbf {\bibinfo {volume} {7}},\ \bibinfo {pages} {1189} (\bibinfo {year} {2023})}\BibitemShut {NoStop}%
\bibitem [{\citenamefont {Kerenidis}\ \emph {et~al.}(2022)\citenamefont {Kerenidis}, \citenamefont {Landman},\ and\ \citenamefont {Mathur}}]{kerenidis2022quantum}%
  \BibitemOpen
  \bibfield  {author} {\bibinfo {author} {\bibfnamefont {I.}~\bibnamefont {Kerenidis}}, \bibinfo {author} {\bibfnamefont {J.}~\bibnamefont {Landman}},\ and\ \bibinfo {author} {\bibfnamefont {N.}~\bibnamefont {Mathur}},\ }\bibfield  {title} {\bibinfo {title} {Quantum and classical algorithms for orthogonal neural networks},\ }\href@noop {} {\bibfield  {journal} {\bibinfo  {journal} {arXiv:2106.07198}\ } (\bibinfo {year} {2022})}\BibitemShut {NoStop}%
\bibitem [{\citenamefont {Cherrat}\ \emph {et~al.}(2024)\citenamefont {Cherrat}, \citenamefont {Kerenidis}, \citenamefont {Mathur}, \citenamefont {Landman}, \citenamefont {Strahm},\ and\ \citenamefont {Li}}]{cherrat2024quantum}%
  \BibitemOpen
  \bibfield  {author} {\bibinfo {author} {\bibfnamefont {E.~A.}\ \bibnamefont {Cherrat}}, \bibinfo {author} {\bibfnamefont {I.}~\bibnamefont {Kerenidis}}, \bibinfo {author} {\bibfnamefont {N.}~\bibnamefont {Mathur}}, \bibinfo {author} {\bibfnamefont {J.}~\bibnamefont {Landman}}, \bibinfo {author} {\bibfnamefont {M.}~\bibnamefont {Strahm}},\ and\ \bibinfo {author} {\bibfnamefont {Y.~Y.}\ \bibnamefont {Li}},\ }\bibfield  {title} {\bibinfo {title} {Quantum vision transformers},\ }\href@noop {} {\bibfield  {journal} {\bibinfo  {journal} {Quantum}\ }\textbf {\bibinfo {volume} {8}},\ \bibinfo {pages} {1265} (\bibinfo {year} {2024})}\BibitemShut {NoStop}%
\bibitem [{\citenamefont {Thakkar}\ \emph {et~al.}(2024)\citenamefont {Thakkar}, \citenamefont {Kazdaghli}, \citenamefont {Mathur}, \citenamefont {Kerenidis}, \citenamefont {Ferreira-Martins},\ and\ \citenamefont {Brito}}]{thakkar2024forecasting}%
  \BibitemOpen
  \bibfield  {author} {\bibinfo {author} {\bibfnamefont {S.}~\bibnamefont {Thakkar}}, \bibinfo {author} {\bibfnamefont {S.}~\bibnamefont {Kazdaghli}}, \bibinfo {author} {\bibfnamefont {N.}~\bibnamefont {Mathur}}, \bibinfo {author} {\bibfnamefont {I.}~\bibnamefont {Kerenidis}}, \bibinfo {author} {\bibfnamefont {A.~J.}\ \bibnamefont {Ferreira-Martins}},\ and\ \bibinfo {author} {\bibfnamefont {S.}~\bibnamefont {Brito}},\ }\bibfield  {title} {\bibinfo {title} {Improved financial forecasting via quantum machine learning},\ }\href@noop {} {\bibfield  {journal} {\bibinfo  {journal} {Quantum Machine Intelligence}\ }\textbf {\bibinfo {volume} {6}},\ \bibinfo {pages} {27} (\bibinfo {year} {2024})}\BibitemShut {NoStop}%
\bibitem [{\citenamefont {Kazdaghli}\ \emph {et~al.}(2023)\citenamefont {Kazdaghli}, \citenamefont {Kerenidis}, \citenamefont {Kieckbusch},\ and\ \citenamefont {Teare}}]{kazdaghli2023data}%
  \BibitemOpen
  \bibfield  {author} {\bibinfo {author} {\bibfnamefont {S.}~\bibnamefont {Kazdaghli}}, \bibinfo {author} {\bibfnamefont {I.}~\bibnamefont {Kerenidis}}, \bibinfo {author} {\bibfnamefont {J.}~\bibnamefont {Kieckbusch}},\ and\ \bibinfo {author} {\bibfnamefont {P.}~\bibnamefont {Teare}},\ }\bibfield  {title} {\bibinfo {title} {Improved clinical data imputation via classical and quantum determinantal point processes},\ }\href@noop {} {\bibfield  {journal} {\bibinfo  {journal} {eLife}\ }\textbf {\bibinfo {volume} {12}},\ \bibinfo {pages} {(RP89947)} (\bibinfo {year} {2023})}\BibitemShut {NoStop}%
\bibitem [{\citenamefont {Johri}\ \emph {et~al.}(2021)\citenamefont {Johri}, \citenamefont {Debnath}, \citenamefont {Mocherla}, \citenamefont {Singh}, \citenamefont {Prakash}, \citenamefont {Kim},\ and\ \citenamefont {Kerenidis}}]{johri2021nearest}%
  \BibitemOpen
  \bibfield  {author} {\bibinfo {author} {\bibfnamefont {S.}~\bibnamefont {Johri}}, \bibinfo {author} {\bibfnamefont {S.}~\bibnamefont {Debnath}}, \bibinfo {author} {\bibfnamefont {A.}~\bibnamefont {Mocherla}}, \bibinfo {author} {\bibfnamefont {A.}~\bibnamefont {Singh}}, \bibinfo {author} {\bibfnamefont {A.}~\bibnamefont {Prakash}}, \bibinfo {author} {\bibfnamefont {J.}~\bibnamefont {Kim}},\ and\ \bibinfo {author} {\bibfnamefont {I.}~\bibnamefont {Kerenidis}},\ }\bibfield  {title} {\bibinfo {title} {Nearest centroid classification on a trapped ion quantum computer},\ }\href@noop {} {\bibfield  {journal} {\bibinfo  {journal} {npj Quantum Information}\ }\textbf {\bibinfo {volume} {7}},\ \bibinfo {pages} {122} (\bibinfo {year} {2021})}\BibitemShut {NoStop}%
\bibitem [{\citenamefont {Mathur}\ \emph {et~al.}(2026)\citenamefont {Mathur}, \citenamefont {Barkoutsos}, \citenamefont {Yamada}, \citenamefont {Roetteler},\ and\ \citenamefont {Kerenidis}}]{mathur2026scalable}%
  \BibitemOpen
  \bibfield  {author} {\bibinfo {author} {\bibfnamefont {N.}~\bibnamefont {Mathur}}, \bibinfo {author} {\bibfnamefont {P.~K.}\ \bibnamefont {Barkoutsos}}, \bibinfo {author} {\bibfnamefont {M.}~\bibnamefont {Yamada}}, \bibinfo {author} {\bibfnamefont {M.}~\bibnamefont {Roetteler}},\ and\ \bibinfo {author} {\bibfnamefont {I.}~\bibnamefont {Kerenidis}},\ }\bibfield  {title} {\bibinfo {title} {Scalable on-hardware training of quantum neural networks and application to clinical data imputation},\ }\href@noop {} {\bibfield  {journal} {\bibinfo  {journal} {arXiv:2606.03517}\ } (\bibinfo {year} {2026})}\BibitemShut {NoStop}%
\bibitem [{\citenamefont {Kerenidis}\ and\ \citenamefont {Prakash}(2022)}]{kerenidis2022subspace}%
  \BibitemOpen
  \bibfield  {author} {\bibinfo {author} {\bibfnamefont {I.}~\bibnamefont {Kerenidis}}\ and\ \bibinfo {author} {\bibfnamefont {A.}~\bibnamefont {Prakash}},\ }\bibfield  {title} {\bibinfo {title} {Quantum machine learning with subspace states},\ }\href@noop {} {\bibfield  {journal} {\bibinfo  {journal} {arXiv:2202.00054}\ } (\bibinfo {year} {2022})}\BibitemShut {NoStop}%
\bibitem [{\citenamefont {Jozsa}\ and\ \citenamefont {Miyake}(2008)}]{jozsa2008matchgates}%
  \BibitemOpen
  \bibfield  {author} {\bibinfo {author} {\bibfnamefont {R.}~\bibnamefont {Jozsa}}\ and\ \bibinfo {author} {\bibfnamefont {A.}~\bibnamefont {Miyake}},\ }\bibfield  {title} {\bibinfo {title} {Matchgates and classical simulation of quantum circuits},\ }\href@noop {} {\bibfield  {journal} {\bibinfo  {journal} {Proceedings of the Royal Society A}\ }\textbf {\bibinfo {volume} {464}},\ \bibinfo {pages} {3089} (\bibinfo {year} {2008})}\BibitemShut {NoStop}%
\bibitem [{\citenamefont {Monbroussou}\ \emph {et~al.}(2025)\citenamefont {Monbroussou}, \citenamefont {Mamon}, \citenamefont {Landman}, \citenamefont {Grilo}, \citenamefont {Kukla},\ and\ \citenamefont {Kashefi}}]{monbroussou2025trainability}%
  \BibitemOpen
  \bibfield  {author} {\bibinfo {author} {\bibfnamefont {L.}~\bibnamefont {Monbroussou}}, \bibinfo {author} {\bibfnamefont {E.~Z.}\ \bibnamefont {Mamon}}, \bibinfo {author} {\bibfnamefont {J.}~\bibnamefont {Landman}}, \bibinfo {author} {\bibfnamefont {A.~B.}\ \bibnamefont {Grilo}}, \bibinfo {author} {\bibfnamefont {R.}~\bibnamefont {Kukla}},\ and\ \bibinfo {author} {\bibfnamefont {E.}~\bibnamefont {Kashefi}},\ }\bibfield  {title} {\bibinfo {title} {Trainability and expressivity of {Hamming}-weight preserving quantum circuits for machine learning},\ }\href@noop {} {\bibfield  {journal} {\bibinfo  {journal} {Quantum}\ }\textbf {\bibinfo {volume} {9}},\ \bibinfo {pages} {1745} (\bibinfo {year} {2025})}\BibitemShut {NoStop}%
\bibitem [{\citenamefont {Dias}\ and\ \citenamefont {Koenig}(2024)}]{dias2024classical}%
  \BibitemOpen
  \bibfield  {author} {\bibinfo {author} {\bibfnamefont {B.}~\bibnamefont {Dias}}\ and\ \bibinfo {author} {\bibfnamefont {R.}~\bibnamefont {Koenig}},\ }\bibfield  {title} {\bibinfo {title} {Classical simulation of non-{Gaussian} fermionic circuits},\ }\href@noop {} {\bibfield  {journal} {\bibinfo  {journal} {Quantum}\ }\textbf {\bibinfo {volume} {8}},\ \bibinfo {pages} {1350} (\bibinfo {year} {2024})}\BibitemShut {NoStop}%
\bibitem [{\citenamefont {Oh}\ \emph {et~al.}(2026)\citenamefont {Oh}, \citenamefont {Oszmaniec}, \citenamefont {Reardon-Smith},\ and\ \citenamefont {Zimbor{\'a}s}}]{oh2026classical}%
  \BibitemOpen
  \bibfield  {author} {\bibinfo {author} {\bibfnamefont {C.}~\bibnamefont {Oh}}, \bibinfo {author} {\bibfnamefont {M.}~\bibnamefont {Oszmaniec}}, \bibinfo {author} {\bibfnamefont {O.}~\bibnamefont {Reardon-Smith}},\ and\ \bibinfo {author} {\bibfnamefont {Z.}~\bibnamefont {Zimbor{\'a}s}},\ }\bibfield  {title} {\bibinfo {title} {Classical simulation of free-fermionic dynamics and quantum chemistry with magic input},\ }\href@noop {} {\bibfield  {journal} {\bibinfo  {journal} {arxiv:2604.26813}\ } (\bibinfo {year} {2026})}\BibitemShut {NoStop}%
\bibitem [{\citenamefont {Schuld}\ \emph {et~al.}(2019)\citenamefont {Schuld}, \citenamefont {Bergholm}, \citenamefont {Gogolin}, \citenamefont {Izaac},\ and\ \citenamefont {Killoran}}]{schuld2019gradients}%
  \BibitemOpen
  \bibfield  {author} {\bibinfo {author} {\bibfnamefont {M.}~\bibnamefont {Schuld}}, \bibinfo {author} {\bibfnamefont {V.}~\bibnamefont {Bergholm}}, \bibinfo {author} {\bibfnamefont {C.}~\bibnamefont {Gogolin}}, \bibinfo {author} {\bibfnamefont {J.}~\bibnamefont {Izaac}},\ and\ \bibinfo {author} {\bibfnamefont {N.}~\bibnamefont {Killoran}},\ }\bibfield  {title} {\bibinfo {title} {Evaluating analytic gradients on quantum hardware},\ }\href@noop {} {\bibfield  {journal} {\bibinfo  {journal} {Physical Review A}\ }\textbf {\bibinfo {volume} {99}},\ \bibinfo {pages} {032331} (\bibinfo {year} {2019})}\BibitemShut {NoStop}%
\bibitem [{\citenamefont {Spall}(1992)}]{spall1992multivariate}%
  \BibitemOpen
  \bibfield  {author} {\bibinfo {author} {\bibfnamefont {J.~C.}\ \bibnamefont {Spall}},\ }\bibfield  {title} {\bibinfo {title} {Multivariate stochastic approximation using a simultaneous perturbation gradient approximation},\ }\href@noop {} {\bibfield  {journal} {\bibinfo  {journal} {IEEE Transactions on Automatic Control}\ }\textbf {\bibinfo {volume} {37}},\ \bibinfo {pages} {332} (\bibinfo {year} {1992})}\BibitemShut {NoStop}%
\bibitem [{\citenamefont {Gacon}\ \emph {et~al.}(2021)\citenamefont {Gacon}, \citenamefont {Zoufal}, \citenamefont {Carleo},\ and\ \citenamefont {Woerner}}]{gacon2021simultaneous}%
  \BibitemOpen
  \bibfield  {author} {\bibinfo {author} {\bibfnamefont {J.}~\bibnamefont {Gacon}}, \bibinfo {author} {\bibfnamefont {C.}~\bibnamefont {Zoufal}}, \bibinfo {author} {\bibfnamefont {G.}~\bibnamefont {Carleo}},\ and\ \bibinfo {author} {\bibfnamefont {S.}~\bibnamefont {Woerner}},\ }\bibfield  {title} {\bibinfo {title} {Simultaneous perturbation stochastic approximation of the quantum fisher information},\ }\href@noop {} {\bibfield  {journal} {\bibinfo  {journal} {Quantum}\ }\textbf {\bibinfo {volume} {5}},\ \bibinfo {pages} {567} (\bibinfo {year} {2021})}\BibitemShut {NoStop}%
\bibitem [{\citenamefont {Wierichs}\ \emph {et~al.}(2022)\citenamefont {Wierichs}, \citenamefont {Izaac}, \citenamefont {Wang},\ and\ \citenamefont {Lin}}]{wierichs2022general}%
  \BibitemOpen
  \bibfield  {author} {\bibinfo {author} {\bibfnamefont {D.}~\bibnamefont {Wierichs}}, \bibinfo {author} {\bibfnamefont {J.}~\bibnamefont {Izaac}}, \bibinfo {author} {\bibfnamefont {C.}~\bibnamefont {Wang}},\ and\ \bibinfo {author} {\bibfnamefont {C.-Y.}\ \bibnamefont {Lin}},\ }\bibfield  {title} {\bibinfo {title} {General parameter-shift rules for quantum gradients},\ }\href@noop {} {\bibfield  {journal} {\bibinfo  {journal} {Quantum}\ }\textbf {\bibinfo {volume} {6}},\ \bibinfo {pages} {677} (\bibinfo {year} {2022})}\BibitemShut {NoStop}%
\bibitem [{\citenamefont {Coyle}\ \emph {et~al.}(2026)\citenamefont {Coyle}, \citenamefont {Raj}, \citenamefont {Umathe}, \citenamefont {Cherrat},\ and\ \citenamefont {Kashefi}}]{coyle2026adaptive}%
  \BibitemOpen
  \bibfield  {author} {\bibinfo {author} {\bibfnamefont {B.}~\bibnamefont {Coyle}}, \bibinfo {author} {\bibfnamefont {S.}~\bibnamefont {Raj}}, \bibinfo {author} {\bibfnamefont {V.}~\bibnamefont {Umathe}}, \bibinfo {author} {\bibfnamefont {E.~A.}\ \bibnamefont {Cherrat}},\ and\ \bibinfo {author} {\bibfnamefont {E.}~\bibnamefont {Kashefi}},\ }\bibfield  {title} {\bibinfo {title} {Adaptive directional gradients for parameterised quantum circuits},\ }\href@noop {} {\bibfield  {journal} {\bibinfo  {journal} {arXiv preprint arXiv:2606.09734}\ } (\bibinfo {year} {2026})}\BibitemShut {NoStop}%
\bibitem [{\citenamefont {Coyle}\ \emph {et~al.}(2025)\citenamefont {Coyle}, \citenamefont {Raj}, \citenamefont {Mathur}, \citenamefont {Cherrat}, \citenamefont {Jain}, \citenamefont {Kazdaghli},\ and\ \citenamefont {Kerenidis}}]{coyle2025density}%
  \BibitemOpen
  \bibfield  {author} {\bibinfo {author} {\bibfnamefont {B.}~\bibnamefont {Coyle}}, \bibinfo {author} {\bibfnamefont {S.}~\bibnamefont {Raj}}, \bibinfo {author} {\bibfnamefont {N.}~\bibnamefont {Mathur}}, \bibinfo {author} {\bibfnamefont {E.~A.}\ \bibnamefont {Cherrat}}, \bibinfo {author} {\bibfnamefont {N.}~\bibnamefont {Jain}}, \bibinfo {author} {\bibfnamefont {S.}~\bibnamefont {Kazdaghli}},\ and\ \bibinfo {author} {\bibfnamefont {I.}~\bibnamefont {Kerenidis}},\ }\bibfield  {title} {\bibinfo {title} {Training-efficient density quantum machine learning},\ }\href@noop {} {\bibfield  {journal} {\bibinfo  {journal} {npj Quantum Information}\ }\textbf {\bibinfo {volume} {11}},\ \bibinfo {pages} {172} (\bibinfo {year} {2025})}\BibitemShut {NoStop}%
\bibitem [{\citenamefont {B{\"a}rtschi}\ and\ \citenamefont {Eidenbenz}(2022)}]{bartschi2022short}%
  \BibitemOpen
  \bibfield  {author} {\bibinfo {author} {\bibfnamefont {A.}~\bibnamefont {B{\"a}rtschi}}\ and\ \bibinfo {author} {\bibfnamefont {S.}~\bibnamefont {Eidenbenz}},\ }\bibfield  {title} {\bibinfo {title} {Short-depth circuits for {Dicke} state preparation},\ }in\ \href@noop {} {\emph {\bibinfo {booktitle} {IEEE International Conference on Quantum Computing and Engineering (QCE)}}}\ (\bibinfo {year} {2022})\ pp.\ \bibinfo {pages} {87--96}\BibitemShut {NoStop}%
\bibitem [{\citenamefont {Farias}\ \emph {et~al.}(2025)\citenamefont {Farias}, \citenamefont {Maciel}, \citenamefont {Camilo}, \citenamefont {Lin}, \citenamefont {Ramos-Calderer},\ and\ \citenamefont {Aolita}}]{farias2025quantum}%
  \BibitemOpen
  \bibfield  {author} {\bibinfo {author} {\bibfnamefont {R.~M.~S.}\ \bibnamefont {Farias}}, \bibinfo {author} {\bibfnamefont {T.~O.}\ \bibnamefont {Maciel}}, \bibinfo {author} {\bibfnamefont {G.}~\bibnamefont {Camilo}}, \bibinfo {author} {\bibfnamefont {R.}~\bibnamefont {Lin}}, \bibinfo {author} {\bibfnamefont {S.}~\bibnamefont {Ramos-Calderer}},\ and\ \bibinfo {author} {\bibfnamefont {L.}~\bibnamefont {Aolita}},\ }\bibfield  {title} {\bibinfo {title} {Quantum encoder for fixed-{Hamming}-weight subspaces},\ }\href@noop {} {\bibfield  {journal} {\bibinfo  {journal} {Physical Review Applied}\ }\textbf {\bibinfo {volume} {23}},\ \bibinfo {pages} {044014} (\bibinfo {year} {2025})}\BibitemShut {NoStop}%
\bibitem [{\citenamefont {Schuld}\ \emph {et~al.}(2021)\citenamefont {Schuld}, \citenamefont {Sweke},\ and\ \citenamefont {Meyer}}]{schuld2021effect}%
  \BibitemOpen
  \bibfield  {author} {\bibinfo {author} {\bibfnamefont {M.}~\bibnamefont {Schuld}}, \bibinfo {author} {\bibfnamefont {R.}~\bibnamefont {Sweke}},\ and\ \bibinfo {author} {\bibfnamefont {J.~J.}\ \bibnamefont {Meyer}},\ }\bibfield  {title} {\bibinfo {title} {Effect of data encoding on the expressive power of variational quantum-machine-learning models},\ }\href@noop {} {\bibfield  {journal} {\bibinfo  {journal} {Physical Review A}\ }\textbf {\bibinfo {volume} {103}},\ \bibinfo {pages} {032430} (\bibinfo {year} {2021})}\BibitemShut {NoStop}%
\bibitem [{\citenamefont {Ivanov}\ and\ \citenamefont {Gurvitz}(2013)}]{ivanov2013permanents}%
  \BibitemOpen
  \bibfield  {author} {\bibinfo {author} {\bibfnamefont {P.~A.}\ \bibnamefont {Ivanov}}\ and\ \bibinfo {author} {\bibfnamefont {S.~A.}\ \bibnamefont {Gurvitz}},\ }\bibfield  {title} {\bibinfo {title} {Complexity of calculating the amplitudes of boson sampling from fock states},\ }\href@noop {} {\bibfield  {journal} {\bibinfo  {journal} {Physical Review A}\ }\textbf {\bibinfo {volume} {87}},\ \bibinfo {pages} {022114} (\bibinfo {year} {2013})}\BibitemShut {NoStop}%
\bibitem [{\citenamefont {Gurvits}(2005)}]{gurvits2005complexity}%
  \BibitemOpen
  \bibfield  {author} {\bibinfo {author} {\bibfnamefont {L.}~\bibnamefont {Gurvits}},\ }\bibfield  {title} {\bibinfo {title} {On the complexity of mixed discriminants and related problems},\ }\href@noop {} {\bibfield  {journal} {\bibinfo  {journal} {Lecture Notes in Computer Science}\ }\textbf {\bibinfo {volume} {3618}},\ \bibinfo {pages} {361} (\bibinfo {year} {2005})}\BibitemShut {NoStop}%
\bibitem [{\citenamefont {Hebenstreit}\ \emph {et~al.}(2019)\citenamefont {Hebenstreit}, \citenamefont {Jozsa}, \citenamefont {Kraus}, \citenamefont {Strelchuk},\ and\ \citenamefont {Yoganathan}}]{hebenstreit2019all}%
  \BibitemOpen
  \bibfield  {author} {\bibinfo {author} {\bibfnamefont {M.}~\bibnamefont {Hebenstreit}}, \bibinfo {author} {\bibfnamefont {R.}~\bibnamefont {Jozsa}}, \bibinfo {author} {\bibfnamefont {B.}~\bibnamefont {Kraus}}, \bibinfo {author} {\bibfnamefont {S.}~\bibnamefont {Strelchuk}},\ and\ \bibinfo {author} {\bibfnamefont {M.}~\bibnamefont {Yoganathan}},\ }\bibfield  {title} {\bibinfo {title} {All pure fermionic non-{Gaussian} states are magic states for matchgate computations},\ }\href@noop {} {\bibfield  {journal} {\bibinfo  {journal} {Physical Review Letters}\ }\textbf {\bibinfo {volume} {123}},\ \bibinfo {pages} {080503} (\bibinfo {year} {2019})}\BibitemShut {NoStop}%
\bibitem [{\citenamefont {Yuan}\ \emph {et~al.}(2025)\citenamefont {Yuan} \emph {et~al.}}]{yuan2025efficient}%
  \BibitemOpen
  \bibfield  {author} {\bibinfo {author} {\bibfnamefont {P.}~\bibnamefont {Yuan}} \emph {et~al.},\ }\bibfield  {title} {\bibinfo {title} {Depth-efficient quantum circuit synthesis for deterministic {Dicke} state preparation},\ }\href@noop {} {\bibfield  {journal} {\bibinfo  {journal} {arXiv:2505.15413}\ } (\bibinfo {year} {2025})}\BibitemShut {NoStop}%
\bibitem [{\citenamefont {Valiant}(2001)}]{valiant2001quantum}%
  \BibitemOpen
  \bibfield  {author} {\bibinfo {author} {\bibfnamefont {L.~G.}\ \bibnamefont {Valiant}},\ }\bibfield  {title} {\bibinfo {title} {Quantum computers that can be simulated classically in polynomial time},\ }in\ \href@noop {} {\emph {\bibinfo {booktitle} {Proceedings of the 33rd Annual ACM Symposium on Theory of Computing (STOC)}}}\ (\bibinfo {year} {2001})\ pp.\ \bibinfo {pages} {114--123}\BibitemShut {NoStop}%
\bibitem [{\citenamefont {Terhal}\ and\ \citenamefont {DiVincenzo}(2002)}]{terhal2002classical}%
  \BibitemOpen
  \bibfield  {author} {\bibinfo {author} {\bibfnamefont {B.~M.}\ \bibnamefont {Terhal}}\ and\ \bibinfo {author} {\bibfnamefont {D.~P.}\ \bibnamefont {DiVincenzo}},\ }\bibfield  {title} {\bibinfo {title} {Classical simulation of noninteracting-fermion quantum circuits},\ }\href@noop {} {\bibfield  {journal} {\bibinfo  {journal} {Physical Review A}\ }\textbf {\bibinfo {volume} {65}},\ \bibinfo {pages} {032325} (\bibinfo {year} {2002})}\BibitemShut {NoStop}%
\bibitem [{\citenamefont {Knill}(2001)}]{knill2001fermionic}%
  \BibitemOpen
  \bibfield  {author} {\bibinfo {author} {\bibfnamefont {E.}~\bibnamefont {Knill}},\ }\href@noop {} {\emph {\bibinfo {title} {Fermionic linear optics and matchgates}}},\ \bibinfo {type} {Tech. Rep.}\ \bibinfo {number} {LAUR-01-4472}\ (\bibinfo  {institution} {Los Alamos National Laboratory},\ \bibinfo {year} {2001})\ \bibinfo {note} {arXiv:quant-ph/0108033}\BibitemShut {NoStop}%
\bibitem [{\citenamefont {Grzesiak}\ \emph {et~al.}(2020)\citenamefont {Grzesiak}, \citenamefont {Bl\"{u}mel}, \citenamefont {Wright}, \citenamefont {Beck}, \citenamefont {Pisenti}, \citenamefont {Li}, \citenamefont {Chaplin}, \citenamefont {Amini}, \citenamefont {Debnath}, \citenamefont {Chen},\ and\ \citenamefont {Nam}}]{grzesiak2020efficient}%
  \BibitemOpen
  \bibfield  {author} {\bibinfo {author} {\bibfnamefont {N.}~\bibnamefont {Grzesiak}}, \bibinfo {author} {\bibfnamefont {R.}~\bibnamefont {Bl\"{u}mel}}, \bibinfo {author} {\bibfnamefont {K.}~\bibnamefont {Wright}}, \bibinfo {author} {\bibfnamefont {K.~M.}\ \bibnamefont {Beck}}, \bibinfo {author} {\bibfnamefont {N.~C.}\ \bibnamefont {Pisenti}}, \bibinfo {author} {\bibfnamefont {M.}~\bibnamefont {Li}}, \bibinfo {author} {\bibfnamefont {V.}~\bibnamefont {Chaplin}}, \bibinfo {author} {\bibfnamefont {J.~M.}\ \bibnamefont {Amini}}, \bibinfo {author} {\bibfnamefont {S.}~\bibnamefont {Debnath}}, \bibinfo {author} {\bibfnamefont {J.-S.}\ \bibnamefont {Chen}},\ and\ \bibinfo {author} {\bibfnamefont {Y.}~\bibnamefont {Nam}},\ }\bibfield  {title} {\bibinfo {title} {Efficient arbitrary simultaneously entangling gates on a trapped-ion quantum computer},\ }\href {https://doi.org/10.1038/s41467-020-16790-9} {\bibfield  {journal} {\bibinfo  {journal} {Nature Communications}\ }\textbf {\bibinfo {volume} {11}},\ \bibinfo
  {pages} {2963} (\bibinfo {year} {2020})}\BibitemShut {NoStop}%
\bibitem [{\citenamefont {Maslov}\ and\ \citenamefont {Nam}(2018)}]{maslov2018global}%
  \BibitemOpen
  \bibfield  {author} {\bibinfo {author} {\bibfnamefont {D.}~\bibnamefont {Maslov}}\ and\ \bibinfo {author} {\bibfnamefont {Y.}~\bibnamefont {Nam}},\ }\bibfield  {title} {\bibinfo {title} {Use of global interactions in efficient quantum circuit constructions},\ }\href {https://doi.org/10.1088/1367-2630/aaa398} {\bibfield  {journal} {\bibinfo  {journal} {New Journal of Physics}\ }\textbf {\bibinfo {volume} {20}},\ \bibinfo {pages} {033018} (\bibinfo {year} {2018})}\BibitemShut {NoStop}%
\bibitem [{\citenamefont {Cooley}\ and\ \citenamefont {Tukey}(1965)}]{cooley1965algorithm}%
  \BibitemOpen
  \bibfield  {author} {\bibinfo {author} {\bibfnamefont {J.~W.}\ \bibnamefont {Cooley}}\ and\ \bibinfo {author} {\bibfnamefont {J.~W.}\ \bibnamefont {Tukey}},\ }\bibfield  {title} {\bibinfo {title} {An algorithm for the machine calculation of complex {Fourier} series},\ }\href@noop {} {\bibfield  {journal} {\bibinfo  {journal} {Mathematics of Computation}\ }\textbf {\bibinfo {volume} {19}},\ \bibinfo {pages} {297} (\bibinfo {year} {1965})}\BibitemShut {NoStop}%
\bibitem [{\citenamefont {Jain}\ \emph {et~al.}(2024)\citenamefont {Jain}, \citenamefont {Landman}, \citenamefont {Mathur},\ and\ \citenamefont {Kerenidis}}]{jain2024qfno}%
  \BibitemOpen
  \bibfield  {author} {\bibinfo {author} {\bibfnamefont {N.}~\bibnamefont {Jain}}, \bibinfo {author} {\bibfnamefont {J.}~\bibnamefont {Landman}}, \bibinfo {author} {\bibfnamefont {N.}~\bibnamefont {Mathur}},\ and\ \bibinfo {author} {\bibfnamefont {I.}~\bibnamefont {Kerenidis}},\ }\bibfield  {title} {\bibinfo {title} {Quantum {Fourier} networks for solving parametric {PDEs}},\ }\href@noop {} {\bibfield  {journal} {\bibinfo  {journal} {Quantum Science and Technology}\ }\textbf {\bibinfo {volume} {9}},\ \bibinfo {pages} {035026} (\bibinfo {year} {2024})}\BibitemShut {NoStop}%
\bibitem [{\citenamefont {Clements}\ \emph {et~al.}(2016)\citenamefont {Clements}, \citenamefont {Humphreys}, \citenamefont {Metcalf}, \citenamefont {Kolthammer},\ and\ \citenamefont {Walmsley}}]{clements2016optimal}%
  \BibitemOpen
  \bibfield  {author} {\bibinfo {author} {\bibfnamefont {W.~R.}\ \bibnamefont {Clements}}, \bibinfo {author} {\bibfnamefont {P.~C.}\ \bibnamefont {Humphreys}}, \bibinfo {author} {\bibfnamefont {B.~J.}\ \bibnamefont {Metcalf}}, \bibinfo {author} {\bibfnamefont {W.~S.}\ \bibnamefont {Kolthammer}},\ and\ \bibinfo {author} {\bibfnamefont {I.~A.}\ \bibnamefont {Walmsley}},\ }\bibfield  {title} {\bibinfo {title} {Optimal design for universal multiport interferometers},\ }\href {https://doi.org/10.1364/OPTICA.3.001460} {\bibfield  {journal} {\bibinfo  {journal} {Optica}\ }\textbf {\bibinfo {volume} {3}},\ \bibinfo {pages} {1460} (\bibinfo {year} {2016})}\BibitemShut {NoStop}%
\bibitem [{\citenamefont {Bravyi}(2005)}]{bravyi2005lagrangian}%
  \BibitemOpen
  \bibfield  {author} {\bibinfo {author} {\bibfnamefont {S.}~\bibnamefont {Bravyi}},\ }\bibfield  {title} {\bibinfo {title} {Lagrangian representation for fermionic linear optics},\ }\href {https://doi.org/10.26421/QIC5.3-3} {\bibfield  {journal} {\bibinfo  {journal} {Quantum Information and Computation}\ }\textbf {\bibinfo {volume} {5}},\ \bibinfo {pages} {216} (\bibinfo {year} {2005})},\ \Eprint {https://arxiv.org/abs/quant-ph/0404180} {arXiv:quant-ph/0404180} \BibitemShut {NoStop}%
\bibitem [{\citenamefont {Hurwitz}(1897)}]{hurwitz1897}%
  \BibitemOpen
  \bibfield  {author} {\bibinfo {author} {\bibfnamefont {A.}~\bibnamefont {Hurwitz}},\ }\bibfield  {title} {\bibinfo {title} {{\"U}ber die erzeugung der invarianten durch integration},\ }\href@noop {} {\bibfield  {journal} {\bibinfo  {journal} {Nachrichten von der Gesellschaft der Wissenschaften zu G{\"o}ttingen, Mathematisch-Physikalische Klasse}\ ,\ \bibinfo {pages} {71}} (\bibinfo {year} {1897})}\BibitemShut {NoStop}%
\bibitem [{\citenamefont {{\.Z}yczkowski}\ and\ \citenamefont {Ku{\'s}}(1994)}]{zyczkowski1994random}%
  \BibitemOpen
  \bibfield  {author} {\bibinfo {author} {\bibfnamefont {K.}~\bibnamefont {{\.Z}yczkowski}}\ and\ \bibinfo {author} {\bibfnamefont {M.}~\bibnamefont {Ku{\'s}}},\ }\bibfield  {title} {\bibinfo {title} {Random unitary matrices},\ }\href {https://doi.org/10.1088/0305-4470/27/12/028} {\bibfield  {journal} {\bibinfo  {journal} {J. Phys. A: Math. Gen.}\ }\textbf {\bibinfo {volume} {27}},\ \bibinfo {pages} {4235} (\bibinfo {year} {1994})}\BibitemShut {NoStop}%
\bibitem [{\citenamefont {Reck}\ \emph {et~al.}(1994)\citenamefont {Reck}, \citenamefont {Zeilinger}, \citenamefont {Bernstein},\ and\ \citenamefont {Bertani}}]{reck1994experimental}%
  \BibitemOpen
  \bibfield  {author} {\bibinfo {author} {\bibfnamefont {M.}~\bibnamefont {Reck}}, \bibinfo {author} {\bibfnamefont {A.}~\bibnamefont {Zeilinger}}, \bibinfo {author} {\bibfnamefont {H.~J.}\ \bibnamefont {Bernstein}},\ and\ \bibinfo {author} {\bibfnamefont {P.}~\bibnamefont {Bertani}},\ }\bibfield  {title} {\bibinfo {title} {Experimental realization of any discrete unitary operator},\ }\href {https://doi.org/10.1103/PhysRevLett.73.58} {\bibfield  {journal} {\bibinfo  {journal} {Phys. Rev. Lett.}\ }\textbf {\bibinfo {volume} {73}},\ \bibinfo {pages} {58} (\bibinfo {year} {1994})}\BibitemShut {NoStop}%
\bibitem [{\citenamefont {Brandao}\ \emph {et~al.}(2016)\citenamefont {Brandao}, \citenamefont {Harrow},\ and\ \citenamefont {Horodecki}}]{brandao2016local}%
  \BibitemOpen
  \bibfield  {author} {\bibinfo {author} {\bibfnamefont {F.~G. S.~L.}\ \bibnamefont {Brandao}}, \bibinfo {author} {\bibfnamefont {A.~W.}\ \bibnamefont {Harrow}},\ and\ \bibinfo {author} {\bibfnamefont {M.}~\bibnamefont {Horodecki}},\ }\bibfield  {title} {\bibinfo {title} {Local random quantum circuits are approximate polynomial-designs},\ }\href@noop {} {\bibfield  {journal} {\bibinfo  {journal} {Communications in Mathematical Physics}\ }\textbf {\bibinfo {volume} {346}},\ \bibinfo {pages} {397} (\bibinfo {year} {2016})}\BibitemShut {NoStop}%
\bibitem [{\citenamefont {Brown}\ and\ \citenamefont {Susskind}(2018)}]{brown2018random}%
  \BibitemOpen
  \bibfield  {author} {\bibinfo {author} {\bibfnamefont {W.}~\bibnamefont {Brown}}\ and\ \bibinfo {author} {\bibfnamefont {L.}~\bibnamefont {Susskind}},\ }\bibfield  {title} {\bibinfo {title} {Random quantum dynamics: From {Haar} measure to thermalization via random matrix theory},\ }\href@noop {} {\bibfield  {journal} {\bibinfo  {journal} {Physical Review X}\ }\textbf {\bibinfo {volume} {8}},\ \bibinfo {pages} {021014} (\bibinfo {year} {2018})}\BibitemShut {NoStop}%
\bibitem [{\citenamefont {Ma}\ and\ \citenamefont {Zhang}(2026)}]{mazhang2026scalar}%
  \BibitemOpen
  \bibfield  {author} {\bibinfo {author} {\bibfnamefont {H.}~\bibnamefont {Ma}}\ and\ \bibinfo {author} {\bibfnamefont {W.}~\bibnamefont {Zhang}},\ }\href@noop {} {\bibinfo {title} {Butterfly gradient variances at half filling and four particles}},\ \bibinfo {howpublished} {Research note, version dated 18 October 2026} (\bibinfo {year} {2026}),\ \bibinfo {note} {\url{https://github.com/the-omega-institute/trureturing-experiments/tree/e44f5ba8f7a5c109bcdce5d7c3cb12fb49f972a1/docs/reports/kerenidis-butterfly-scalar-20261018}}\BibitemShut {NoStop}%
\bibitem [{\citenamefont {White}(1992)}]{white1992dmrg}%
  \BibitemOpen
  \bibfield  {author} {\bibinfo {author} {\bibfnamefont {S.~R.}\ \bibnamefont {White}},\ }\bibfield  {title} {\bibinfo {title} {Density matrix formulation for quantum renormalization groups},\ }\href@noop {} {\bibfield  {journal} {\bibinfo  {journal} {Physical Review Letters}\ }\textbf {\bibinfo {volume} {69}},\ \bibinfo {pages} {2863} (\bibinfo {year} {1992})}\BibitemShut {NoStop}%
\bibitem [{\citenamefont {Nesterov}(2012)}]{nesterov2012efficiency}%
  \BibitemOpen
  \bibfield  {author} {\bibinfo {author} {\bibfnamefont {Y.}~\bibnamefont {Nesterov}},\ }\bibfield  {title} {\bibinfo {title} {Efficiency of coordinate descent methods on huge-scale optimization problems},\ }\href@noop {} {\bibfield  {journal} {\bibinfo  {journal} {SIAM Journal on Optimization}\ }\textbf {\bibinfo {volume} {22}},\ \bibinfo {pages} {341} (\bibinfo {year} {2012})}\BibitemShut {NoStop}%
\bibitem [{\citenamefont {Coyle}\ \emph {et~al.}(2020)\citenamefont {Coyle}, \citenamefont {Mills}, \citenamefont {Danos},\ and\ \citenamefont {Kashefi}}]{coyle2020born}%
  \BibitemOpen
  \bibfield  {author} {\bibinfo {author} {\bibfnamefont {B.}~\bibnamefont {Coyle}}, \bibinfo {author} {\bibfnamefont {D.}~\bibnamefont {Mills}}, \bibinfo {author} {\bibfnamefont {V.}~\bibnamefont {Danos}},\ and\ \bibinfo {author} {\bibfnamefont {E.}~\bibnamefont {Kashefi}},\ }\bibfield  {title} {\bibinfo {title} {The {Born} supremacy: quantum advantage and training of an {Ising} {Born} machine},\ }\href@noop {} {\bibfield  {journal} {\bibinfo  {journal} {npj Quantum Information}\ }\textbf {\bibinfo {volume} {6}},\ \bibinfo {pages} {60} (\bibinfo {year} {2020})}\BibitemShut {NoStop}%
\bibitem [{\citenamefont {Lloyd}\ and\ \citenamefont {Weedbrook}(2018)}]{lloyd2018qgan}%
  \BibitemOpen
  \bibfield  {author} {\bibinfo {author} {\bibfnamefont {S.}~\bibnamefont {Lloyd}}\ and\ \bibinfo {author} {\bibfnamefont {C.}~\bibnamefont {Weedbrook}},\ }\bibfield  {title} {\bibinfo {title} {Quantum generative adversarial learning},\ }\href@noop {} {\bibfield  {journal} {\bibinfo  {journal} {Physical Review Letters}\ }\textbf {\bibinfo {volume} {121}},\ \bibinfo {pages} {040502} (\bibinfo {year} {2018})}\BibitemShut {NoStop}%
\bibitem [{\citenamefont {Schuld}\ and\ \citenamefont {Petruccione}(2021)}]{schuld2021supervised}%
  \BibitemOpen
  \bibfield  {author} {\bibinfo {author} {\bibfnamefont {M.}~\bibnamefont {Schuld}}\ and\ \bibinfo {author} {\bibfnamefont {F.}~\bibnamefont {Petruccione}},\ }\href@noop {} {\emph {\bibinfo {title} {Supervised Learning with Quantum Computers}}},\ \bibinfo {edition} {2nd}\ ed.\ (\bibinfo  {publisher} {Springer},\ \bibinfo {year} {2021})\BibitemShut {NoStop}%
\bibitem [{\citenamefont {Caro}\ \emph {et~al.}(2022)\citenamefont {Caro}, \citenamefont {Huang}, \citenamefont {Cerezo}, \citenamefont {Sharma}, \citenamefont {Sornborger}, \citenamefont {Cincio},\ and\ \citenamefont {Coles}}]{caro2022generalization}%
  \BibitemOpen
  \bibfield  {author} {\bibinfo {author} {\bibfnamefont {M.~C.}\ \bibnamefont {Caro}}, \bibinfo {author} {\bibfnamefont {H.-Y.}\ \bibnamefont {Huang}}, \bibinfo {author} {\bibfnamefont {M.}~\bibnamefont {Cerezo}}, \bibinfo {author} {\bibfnamefont {K.}~\bibnamefont {Sharma}}, \bibinfo {author} {\bibfnamefont {A.}~\bibnamefont {Sornborger}}, \bibinfo {author} {\bibfnamefont {L.}~\bibnamefont {Cincio}},\ and\ \bibinfo {author} {\bibfnamefont {P.~J.}\ \bibnamefont {Coles}},\ }\bibfield  {title} {\bibinfo {title} {Generalization in quantum machine learning from few training data},\ }\href@noop {} {\bibfield  {journal} {\bibinfo  {journal} {Nature Communications}\ }\textbf {\bibinfo {volume} {13}},\ \bibinfo {pages} {4919} (\bibinfo {year} {2022})}\BibitemShut {NoStop}%
\bibitem [{\citenamefont {Movassagh}(2023)}]{movassagh2023hardness}%
  \BibitemOpen
  \bibfield  {author} {\bibinfo {author} {\bibfnamefont {R.}~\bibnamefont {Movassagh}},\ }\bibfield  {title} {\bibinfo {title} {The hardness of random quantum circuits},\ }\href@noop {} {\bibfield  {journal} {\bibinfo  {journal} {Science}\ }\textbf {\bibinfo {volume} {379}},\ \bibinfo {pages} {eabb9525} (\bibinfo {year} {2023})}\BibitemShut {NoStop}%
\bibitem [{\citenamefont {Liu}\ \emph {et~al.}(2021)\citenamefont {Liu}, \citenamefont {Arunachalam},\ and\ \citenamefont {Temme}}]{liu2021rigorous}%
  \BibitemOpen
  \bibfield  {author} {\bibinfo {author} {\bibfnamefont {Y.}~\bibnamefont {Liu}}, \bibinfo {author} {\bibfnamefont {S.}~\bibnamefont {Arunachalam}},\ and\ \bibinfo {author} {\bibfnamefont {K.}~\bibnamefont {Temme}},\ }\bibfield  {title} {\bibinfo {title} {A rigorous and robust quantum speed-up in supervised machine learning},\ }\href {https://doi.org/10.1038/s41567-021-01287-z} {\bibfield  {journal} {\bibinfo  {journal} {Nature Physics}\ }\textbf {\bibinfo {volume} {17}},\ \bibinfo {pages} {1013} (\bibinfo {year} {2021})},\ \Eprint {https://arxiv.org/abs/2010.02174} {arXiv:2010.02174} \BibitemShut {NoStop}%
\bibitem [{\citenamefont {Russell}\ \emph {et~al.}(2017)\citenamefont {Russell}, \citenamefont {Chakhmakhchyan}, \citenamefont {O'Brien},\ and\ \citenamefont {Laing}}]{russell2017direct}%
  \BibitemOpen
  \bibfield  {author} {\bibinfo {author} {\bibfnamefont {N.~J.}\ \bibnamefont {Russell}}, \bibinfo {author} {\bibfnamefont {L.}~\bibnamefont {Chakhmakhchyan}}, \bibinfo {author} {\bibfnamefont {J.~L.}\ \bibnamefont {O'Brien}},\ and\ \bibinfo {author} {\bibfnamefont {A.}~\bibnamefont {Laing}},\ }\bibfield  {title} {\bibinfo {title} {Direct dialling of haar random unitary matrices},\ }\href {https://doi.org/10.1088/1367-2630/aa60ed} {\bibfield  {journal} {\bibinfo  {journal} {New J. Phys.}\ }\textbf {\bibinfo {volume} {19}},\ \bibinfo {pages} {033007} (\bibinfo {year} {2017})}\BibitemShut {NoStop}%
\bibitem [{\citenamefont {Anschuetz}\ and\ \citenamefont {Kiani}(2022)}]{anschuetz2022quantum}%
  \BibitemOpen
  \bibfield  {author} {\bibinfo {author} {\bibfnamefont {E.~R.}\ \bibnamefont {Anschuetz}}\ and\ \bibinfo {author} {\bibfnamefont {B.~T.}\ \bibnamefont {Kiani}},\ }\bibfield  {title} {\bibinfo {title} {Quantum variational algorithms are swamped with traps},\ }\href@noop {} {\bibfield  {journal} {\bibinfo  {journal} {Nature Communications}\ }\textbf {\bibinfo {volume} {13}},\ \bibinfo {pages} {7760} (\bibinfo {year} {2022})}\BibitemShut {NoStop}%
\bibitem [{\citenamefont {Amidi}(2026)}]{amidi2026comment}%
  \BibitemOpen
  \bibfield  {author} {\bibinfo {author} {\bibfnamefont {E.}~\bibnamefont {Amidi}},\ }\href@noop {} {\bibinfo {title} {Comment on ``{S}calable {Q}uantum {M}achine {L}earning: {T}rainability, {E}xpressivity and {E}fficiency: {P}olynomial evaluation of the triplet-block readout}} (\bibinfo {year} {2026}),\ \bibinfo {note} {arXiv preprint}\BibitemShut {NoStop}%
\end{thebibliography}%

\appendix
\section{Proofs for Section~\ref{sec:bp}}
\label{app:bp_proofs}

This appendix collects the supporting lemmas and the full proofs of the
trainability results of Section~\ref{sec:bp}.

\subsection{Structural lemmas}
\label{app:bp_lemmas}

\begin{lemma}[Locality of coordinate-local gate generators on $\Lambda^2$]
\label{lem:gate_locality}
Let $G$ be a trainable gate generator---either an
RBS generator $G_{ij}$ (acting on the mode pair $T = \{i,j\}$) or an $R_z$
generator $\hat n_j$ (acting on $T = \{j\}$)---with $\|G^{(2)}\|_{\mathrm{op}}
= O(1)$, so that $G$ is supported on the $O(1)$ coordinates $T$.
Let $\sigma = \Lambda^2(V')\,\sigma^{\mathrm{in}}\,\Lambda^2(V')^\dagger$,
where $V' = D_1 H$ is the encoding (Walsh--Hadamard spreading followed by the
diagonal data layer) and $\sigma^{\mathrm{in}}$ is supported on a set of $r$
basis pairs with $\|\sigma^{\mathrm{in}}\|_{\mathrm{op}} = O(1)$ and
$\|\sigma^{\mathrm{in}}\|_F^2 = \Theta(r)$. Then
\begin{equation}
\label{eq:locality-upper}
  \bigl\|[G^{(2)}, \sigma]\bigr\|_F^2
  \;=\; O\!\left(\tfrac{1}{n}\right)\,\|\sigma\|_F^2 ,
\end{equation}
uniformly in the data $\bx$. Both the disconnected part
$\rho_1\wedge\rho_1$ (with $r = \binom{2k}{2}$) and the connected part
$\rho_2^{\mathrm{conn}}$ (with $r = \Theta(k)$) of the encoded $2$-RDM satisfy
the hypotheses. The hypothesis that $G$ be supported on $O(1)$ \emph{coordinates}
is used through the alignment of $\calS_T$ with the computational basis in
which $V'$ delocalizes; for a general unitary conjugate $V^\dagger G V$ the
statement is false pointwise, and its correct averaged replacement is
Lemma~\ref{lem:avg_locality}.
\end{lemma}

\begin{proof}
The single-particle generator $g$ satisfies $g\ket{a} = 0$ for $a \notin T$;
since $G^{(2)}\ket{\{a,b\}} = (g\ket{a})\wedge\ket{b} +
\ket{a}\wedge(g\ket{b})$, the lift annihilates every basis pair disjoint from
$T$, so $G^{(2)} = P_{\calS_T}G^{(2)}P_{\calS_T}$ with
$\calS := \calS_T = \mathrm{span}\{\ket{\{a,b\}} : \{a,b\}\cap T \neq \emptyset\}$
and $\dim\calS = \binom n2 - \binom{n-|T|}{2} = \Theta(n)$, a subspace spanned
by basis pairs, and
\begin{equation}
  \bigl\|[G^{(2)}, \sigma]\bigr\|_F
  \le 2\|G^{(2)}\|_{\mathrm{op}}\,\|P_{\calS}\,\sigma\|_F .
\end{equation}
It remains to bound $\|P_{\calS}\sigma\|_F^2 = \sum_{P \in \calS}
(\sigma^\dagger\sigma)_{PP}$. Writing $v_P := \Lambda^2(V')^\dagger\ket{P}$, a
unit vector, we have $(\sigma^\dagger\sigma)_{PP} = \langle v_P,
(\sigma^{\mathrm{in}})^\dagger \sigma^{\mathrm{in}} v_P\rangle$. Since
$|V'_{ab}| = n^{-1/2}$ we have $|\Lambda^2(V')_{PQ}| \le 2/n$, so the
restriction of $v_P$ to the $r$ basis pairs carrying $\sigma^{\mathrm{in}}$
has squared norm at most $4r/n^2$, and therefore
\begin{equation}
  (\sigma^\dagger\sigma)_{PP}
  \le \|\sigma^{\mathrm{in}}\|_{\mathrm{op}}^2 \cdot \frac{4r}{n^2}
  = O\!\left(\frac{r}{n^2}\right)
  = O\!\left(\frac{\|\sigma\|_F^2}{d}\right),
\end{equation}
uniformly in $P$ and in $\bx$, using $\|\sigma\|_F^2 =
\|\sigma^{\mathrm{in}}\|_F^2 = \Theta(r)$ and $d = \Theta(n^2)$. Summing over
the $\dim\calS = \Theta(n)$ pairs in $\calS$ gives
$\|P_{\calS}\sigma\|_F^2 = O(\|\sigma\|_F^2/n)$, which
proves~\eqref{eq:locality-upper}.

For the disconnected part, $\rho_1^{\mathrm{in}}\wedge\rho_1^{\mathrm{in}}$ is
diagonal on $\Lambda^2$ with entry $\tfrac14$ on each of the $\binom{2k}{2}$
active--active pairs, so $r = \binom{2k}{2}$,
$\|\cdot\|_{\mathrm{op}} = \tfrac14$ and $\|\cdot\|_F^2 = \Theta(r)$. For the
connected part, the input is a product over the $k/2$ (resp.\ $k/3$) magic
blocks, so $\rho_2^{\mathrm{conn}}$ is block-diagonal with $\Theta(k)$ blocks of
$O(1)$ size and constant norm, giving $r = \Theta(k)$,
$\|\cdot\|_{\mathrm{op}} = O(1)$ and $\|\cdot\|_F^2 = \Theta(k)$.
\end{proof}

\begin{lemma}[$\Lambda^4$ trace formula]
\label{lem:lambda4}
Let $P_{\Lambda^4}$ denote the orthogonal projector of
$\Lambda^2\mathbb{C}^n \otimes \Lambda^2\mathbb{C}^n$ onto its
$\Lambda^4\mathbb{C}^n$ component. Then for all
$M, N \in \mathrm{End}(\Lambda^2\mathbb{C}^n)$,
\begin{equation}
\label{eq:lambda4}
  \tr\!\bigl(P_{\Lambda^4}\,(M \otimes N)\bigr)
  = \tfrac16\Bigl[\tr M\,\tr N + \tr(MN)
      - \tr\!\bigl(\Gamma(M)\Gamma(N)\bigr)\Bigr].
\end{equation}
\end{lemma}

\begin{proof}
Let $\mathcal{A}_4 = \tfrac1{24}\sum_{\pi \in S_4}\mathrm{sgn}(\pi)\,\pi$ be
the total antisymmetrizer on $(\mathbb{C}^n)^{\otimes 4}$. The subspace
$\Lambda^2 \otimes \Lambda^2$ is $\mathcal{A}_4$-invariant, $\mathcal{A}_4$ is
self-adjoint and idempotent, and its image inside $\Lambda^2\otimes\Lambda^2$
is $\Lambda^4$; hence $P_{\Lambda^4}$ is the restriction of $\mathcal{A}_4$.
Extending $M$ and $N$ to $(\mathbb{C}^n)^{\otimes 4}$ by composing with the
antisymmetrizer on each index pair contributes a factor $\tfrac12$ per
operator, so
\begin{equation}
  \tr\!\bigl(P_{\Lambda^4}(M\otimes N)\bigr)
  = \frac{1}{24}\cdot\frac14 \sum_{\pi\in S_4}\mathrm{sgn}(\pi)\,
    M^{i_1 i_2}_{i_{\pi(1)} i_{\pi(2)}}\,
    N^{i_3 i_4}_{i_{\pi(3)} i_{\pi(4)}} .
\end{equation}
The $24$ permutations fall into three orbits.

\emph{(i) $\pi$ preserves $\{1,2\}$ and $\{3,4\}$ ($4$ permutations).} For
$\pi = \mathrm{id}$ the contraction is $(M^{i_1i_2}_{i_1i_2})
(N^{i_3i_4}_{i_3i_4}) = 4\tr M\,\tr N$. Each transposition inside a pair flips
the sign of the corresponding kernel by antisymmetry and carries
$\mathrm{sgn} = -1$, so all four terms contribute $+4\tr M \tr N$, totalling
$16\,\tr M\,\tr N$.

\emph{(ii) $\pi$ exchanges $\{1,2\}$ with $\{3,4\}$ ($4$ permutations).} For
$\pi = (13)(24)$ the contraction is $M^{i_1i_2}_{i_3i_4}N^{i_3i_4}_{i_1i_2} =
4\tr(MN)$, with $\mathrm{sgn} = +1$; the other three differ by transpositions
inside a pair and, by the same sign cancellation, also contribute
$+4\tr(MN)$, totalling $16\,\tr(MN)$.

\emph{(iii) mixed $\pi$ ($16$ permutations).} Each such permutation sends
exactly one element of $\{1,2\}$ into $\{1,2\}$ and the other into
$\{3,4\}$, so the contraction closes one index of $M$ against one of $N$
twice. For $\pi = (13)$, for instance,
\begin{equation}
  M^{i_1i_2}_{i_3i_2}\,N^{i_3i_4}_{i_1i_4}
  = \sum_{i_1,i_3}\Gamma(M)^{i_1}_{\ i_3}\,\Gamma(N)^{i_3}_{\ i_1}
  = \tr\!\bigl(\Gamma(M)\Gamma(N)\bigr),
\end{equation}
with $\mathrm{sgn}(\pi) = -1$; every mixed permutation reduces to the same
expression with the same sign after using the antisymmetry of the kernels,
totalling $-16\,\tr(\Gamma(M)\Gamma(N))$.

Collecting, $\tfrac1{96}\bigl[16\tr M\tr N + 16\tr(MN) -
16\tr(\Gamma(M)\Gamma(N))\bigr]$, which is~\eqref{eq:lambda4}. As a check, at
$M = N = I_{\Lambda^2}$ the right-hand side equals
$\tfrac16[d^2 + d - n(n-1)^2] = \tfrac1{24}n(n-1)(n-2)(n-3) = \binom n4 =
\dim\Lambda^4 = \tr P_{\Lambda^4}$, for every $n$.
\end{proof}

\begin{lemma}[The contraction intertwines the adjoint action]
\label{lem:contraction}
For every single-particle operator $g$ and every
$Y \in \mathrm{End}(\Lambda^2\mathbb{C}^n)$,
\begin{equation}
\label{eq:gamma_intertwine}
  \Gamma\bigl([\,d\Lambda^2(g),\,Y\,]\bigr) \;=\; [\,g,\,\Gamma(Y)\,] .
\end{equation}
\end{lemma}

\begin{proof}
In kernel form, $\bigl(d\Lambda^2(g)\,Y\bigr)^{ij}_{kl} = g^i_{\ a}Y^{aj}_{kl}
+ g^j_{\ a}Y^{ia}_{kl}$ and $\bigl(Y\,d\Lambda^2(g)\bigr)^{ij}_{kl} =
Y^{ij}_{al}\,g^a_{\ k} + Y^{ij}_{kb}\,g^b_{\ l}$. Contracting $j$ against $l$
and summing, the first term of each expression yields $g\,\Gamma(Y)$ and
$\Gamma(Y)\,g$ respectively. The two remaining terms are
$\sum_{j,a} g^j_{\ a}Y^{ia}_{kj}$ and $\sum_{j,b} g^b_{\ j}Y^{ij}_{kb}$, which
coincide after the relabelling $a \leftrightarrow j$ and cancel in the
difference.
\end{proof}

\begin{lemma}[Averaged locality of conjugated gate generators]
\label{lem:avg_locality}
Let $g$ be a Hermitian single-particle operator, let $\sigma$ be any Hermitian
operator on $\Lambda^2\mathbb{C}^n$, and for $V \in U(n)$ write
$\tilde g = V^\dagger g V$, so that
$\tilde G^{(2)} = d\Lambda^2(\tilde g) =
\Lambda^2(V)^\dagger\,d\Lambda^2(g)\,\Lambda^2(V)$. Then, for
$V \sim \mathrm{Haar}(U(n))$,
\begin{equation}
\label{eq:avg-locality}
  \E_V\bigl\|[\tilde G^{(2)},\sigma]\bigr\|_F^2
  \;\le\; \frac{4\,\bigl[(n-2)\tr(g^2) + (\tr g)^2\bigr]}{\binom n2}\,
          \|\sigma\|_F^2 .
\end{equation}
For an RBS generator the right-hand side is
$16(n-2)\|\sigma\|_F^2/(n(n-1))$ and for a phase generator $g = \hat n_j$ it is
$8\|\sigma\|_F^2/n$: in both cases $O(\|\sigma\|_F^2/n)$, the same rate as the
pointwise bound of Lemma~\ref{lem:gate_locality} for coordinate-local
generators.
\end{lemma}

\begin{proof}
Write $A = \tilde G^{(2)}$, which is Hermitian. Expanding the commutator,
$\|[A,\sigma]\|_F^2 = 2\bigl[\tr(A^2\sigma^2) - \tr(A\sigma A\sigma)\bigr]$,
and by Cauchy--Schwarz
$|\tr(A\sigma A\sigma)| = |\langle (A\sigma)^\dagger, A\sigma\rangle|
\le \|A\sigma\|_F^2 = \tr(A^2\sigma^2)$, so
$\|[A,\sigma]\|_F^2 \le 4\,\tr(A^2\sigma^2)$. Now
$A^2 = \Lambda^2(V)^\dagger (G^{(2)})^2\Lambda^2(V)$, and
$\Lambda^2\mathbb{C}^n$ is an irreducible $U(n)$-module, so by Schur's lemma
$\E_V[A^2] = \bigl(\tr((G^{(2)})^2)/\tbinom n2\bigr) I$. Finally, in kernel
form $\bigl(d\Lambda^2(g)\bigr)^{ij}_{kl} = g^i_{\ k}\delta^j_l +
\delta^i_k g^j_{\ l} - g^i_{\ l}\delta^j_k - \delta^i_l g^j_{\ k}$, and
$\tr(MN) = \tfrac14 M^{ij}_{kl}N^{kl}_{ij}$ gives
\begin{equation}
\label{eq:trG2sq}
  \tr\bigl((d\Lambda^2(g))^2\bigr) = (n-2)\tr(g^2) + (\tr g)^2
\end{equation}
(checked at $g = I_n$, where $d\Lambda^2(I) = 2I_{\Lambda^2}$ and both sides
equal $2n(n-1)$). Combining the three displays gives~\eqref{eq:avg-locality};
the two special cases use $\tr(g^2) = 2$, $\tr g = 0$ for an RBS generator and
$\tr(g^2) = \tr g = 1$ for a phase generator.
\end{proof}

\begin{lemma}[Conjugation neutrality of the one-body commutator]
\label{lem:conj_neutral}
Let $\Pi$ be a rank-$m$ orthogonal projector on $\mathbb{C}^n$ and let $g$ be
Hermitian. For $V \sim \mathrm{Haar}(U(n))$,
\begin{equation}
\label{eq:conj_neutral}
  \E_V\bigl\|[\,g,\,V\Pi V^\dagger\,]\bigr\|_F^2
  \;=\; \frac{2m(n-m)}{n^2-1}\,
        \Bigl[\tr(g^2) - \tfrac{(\tr g)^2}{n}\Bigr].
\end{equation}
At $m = 2k$ this gives $8k(n-2k)/(n^2-1)$ for an RBS generator and
$4k(n-2k)/\bigl(n(n+1)\bigr)$ for a phase generator, to be compared with the
unconjugated values $8k(n-2k)/n^2$ and $4k(n-2k)/n^2$ of Step~4 of
Theorem~\ref{thm:brick-wall_bp}: the two agree to leading order and with the
same constants, the ratios being $n^2/(n^2-1)$ and $n/(n+1)$. Haar
conjugation is thus neutral for $\tau$, constants included.
\end{lemma}

\begin{proof}
Write $Q = V\Pi V^\dagger$, so $Q^2 = Q$ and, expanding the commutator,
$\|[g,Q]\|_F^2 = 2\bigl[\tr(g^2 Q) - \tr(gQgQ)\bigr]$. Since
$\E_V[Q] = (m/n)I_n$, the first term averages to $(m/n)\tr(g^2)$. For the
second, $\tr(gQgQ) = \tr\bigl[(g\otimes g)(Q\otimes Q)S\bigr]$ with $S$ the
swap on $\mathbb{C}^n\otimes\mathbb{C}^n$, and $\E_V[Q\otimes Q]$ is invariant
under $U \otimes U$ for every $U \in U(n)$, hence equals $aI + bS$ by
Schur--Weyl duality. Pairing with $I$ and with $S$ gives
$an^2 + bn = (\tr Q)^2 = m^2$ and $an + bn^2 = \tr(Q^2) = m$, whence
\begin{equation}
  a = \frac{m(mn-1)}{n(n^2-1)},
  \qquad
  b = \frac{m(n-m)}{n(n^2-1)} .
\end{equation}
Since $\tr[(g\otimes g)S] = \tr(g^2)$ and $\tr(g\otimes g) = (\tr g)^2$, we get
$\E_V\tr(gQgQ) = a\,\tr(g^2) + b\,(\tr g)^2$, and therefore
\begin{equation}
  \E_V\|[g,Q]\|_F^2
  = 2\Bigl[\bigl(\tfrac mn - a\bigr)\tr(g^2) - b\,(\tr g)^2\Bigr].
\end{equation}
Finally $\tfrac mn - a = m(n-m)/(n^2-1)$ and $b = \tfrac1n\,m(n-m)/(n^2-1)$,
which is~\eqref{eq:conj_neutral}. As a check, at $m = n$ (so $Q = I$) both
sides vanish.
\end{proof}

\subsection{Proof of Theorem~\ref{thm:brick-wall_bp}}

\begin{proof}
Fix a direction $g \in \mathfrak{u}(n)$ and let $W \sim \mathrm{Haar}(U(n))$.
Right invariance of the Haar measure makes $We^{-isg}$ Haar-distributed for
every $s$, so $s \mapsto \E[\calL(We^{-isg})]$ is constant and
$\E[D_g\calL] = 0$. We do \emph{not} split $W$ into independent
factors---the partial products of Givens rotations before and after a fixed
gate are not individually Haar-distributed---and instead use that the Haar
measure on $U(n)$ is an exact unitary $2$-design on every polynomial
irreducible representation, applying the resulting second-moment identity
directly on the two-particle module $\Lambda^2 \mathbb{C}^n$. This is the
dynamical-Lie-algebra variance computation of Fontana et
al.~\cite{fontana2024adjoint} and Larocca et al.~\cite{larocca2021diagnosing}
specialized to the antisymmetric module.

Since $M = n_i n_j$ is a two-body observable, its parameter-dependent
expectation value is governed entirely by the action of $W$ on $\Lambda^2
\mathbb{C}^n$. Writing $M^{(2)} = \ket{\{i,j\}}\bra{\{i,j\}}$ for the image of
$M$ on $\Lambda^2 \mathbb{C}^n$ and $\rho_2(\bx)$ for the two-particle reduced
density matrix of $\ket{\psi(\bx)}$, the invariant derivative is
\begin{equation}
  D_g\calL
  \;=\; i\,\tr_{\Lambda^2}\!\bigl([G^{(2)},\,
  \Lambda^2(W)^\dagger M^{(2)} \Lambda^2(W)]\,\rho_2(\bx)\bigr).
\end{equation}
By cyclicity of the trace, $D_g\calL =
-i\,\tr_{\Lambda^2}(\Lambda^2(W)^\dagger M^{(2)}\Lambda^2(W)\,B)$ with
$B = [G^{(2)}, \rho_{2,\mathfrak{g}_2}(\bx)]$: the identity component of
$\rho_2$ commutes with $G^{(2)}$ and drops out. This moves the commutator
onto the state factor and lets the second-moment identity act directly on
$M^{(2)\otimes 2}$. The variance over $W \sim \mathrm{Haar}(U(n))$ is then
\begin{align}
\label{eq:bw-variance}
  \Var[D_g\calL]
  \;=\; \nonumber \\
  -\,\E_{W\sim\mathrm{Haar}(U(n))}\!\left[
  \tr\!\Bigl(\bigl(\Lambda^2(W)^\dagger M^{(2)}\Lambda^2(W)\bigr)^{\otimes 2}\,
  B^{\otimes 2}\Bigr)\right].
\end{align}

\medskip
\noindent\textbf{Step 1: $U(n)$-Weingarten formula on $\Lambda^2$.}
Since $W$ is exactly Haar-distributed on $U(n)$, the $U(n)$-Weingarten formula
applies to its action on $\Lambda^2 \mathbb{C}^n$. As a $U(n)$-representation,
\begin{equation}
  \Lambda^2 \mathbb{C}^n \otimes \Lambda^2 \mathbb{C}^n
  \;=\; V_{(2,2)} \oplus V_{(2,1,1)} \oplus V_{(1,1,1,1)},
\end{equation}
with dimensions (hook length formula)
\begin{align}
  D_{(2,2)} := \dim V_{(2,2)} &= \tfrac{n^2(n^2-1)}{12},\\
  D_{(2,1,1)} := \dim V_{(2,1,1)} &= \tfrac{n(n+1)(n-1)(n-2)}{8},\\
  D_{\Lambda^4} := \dim V_{(1,1,1,1)} &= \tbinom{n}{4},
\end{align}
all $\Theta(n^4)$ (at $n=4$: $20 + 15 + 1 = 36 = \binom{4}{2}^2$). Let
$P_\lambda$ be the orthogonal projectors onto these irreducibles. By Schur's
lemma the Haar second-moment operator is
$\E_{V}[\Lambda^2(V)^{\otimes 2} Y \Lambda^2(V)^{\dagger\otimes 2}] =
\sum_\lambda \tfrac{\tr(P_\lambda Y)}{\dim V_\lambda} P_\lambda$, exact with no
approximation error. Applying this to the $M^{(2)\otimes 2}$ factor
in~\eqref{eq:bw-variance} and pairing with $B^{\otimes 2}$ via $\tr(P_\lambda
P_\mu) = \delta_{\lambda\mu}\dim V_\lambda$ yields the exact sum
\begin{equation}
\label{eq:bw-irrep-sum}
  \Var[D_g\calL]
  \;=\; -\sum_{\lambda}
  \frac{\tr\!\bigl(P_\lambda\, M^{(2)\otimes 2}\bigr)\,
        \tr\!\bigl(P_\lambda\, B^{\otimes 2}\bigr)}
       {\dim V_\lambda}.
\end{equation}
The generator enters \emph{only} through $B$; there is no separate scalar
prefactor.

We record once the two facts used to evaluate the three weights. Let $S$ be
the exchange of the two $\Lambda^2$ factors and
$P_{\mathrm{sym}} = \tfrac12(\mathbf{1}+S)$,
$P_{\mathrm{anti}} = \tfrac12(\mathbf{1}-S)$. By the plethysm
$\mathrm{Sym}^2(\Lambda^2\mathbb{C}^n) = V_{(2,2)} \oplus
\Lambda^4\mathbb{C}^n$ and $\Lambda^2(\Lambda^2\mathbb{C}^n) = V_{(2,1,1)}$
(dimension check at $n = 4$: $20+1 = 21 = \binom{6+1}{2}$ and $15$), so
\begin{equation}
\label{eq:proj_identification}
  P_{(2,1,1)} = P_{\mathrm{anti}},
  \qquad
  P_{(2,2)} = P_{\mathrm{sym}} - P_{\Lambda^4}.
\end{equation}
Consequently, for \emph{any} $X \in \mathrm{End}(\Lambda^2\mathbb{C}^n)$,
\begin{align}
\label{eq:three_weights}
  \tr\!\bigl(P_{(2,1,1)}X^{\otimes 2}\bigr)
    &= \tfrac12\bigl[(\tr X)^2 - \tr(X^2)\bigr], \nonumber\\
  \tr\!\bigl(P_{\Lambda^4}X^{\otimes 2}\bigr)
    &= \tfrac16\bigl[(\tr X)^2 + \tr(X^2) - \tr\bigl(\Gamma(X)^2\bigr)\bigr],
    \nonumber\\
  \tr\!\bigl(P_{(2,2)}X^{\otimes 2}\bigr)
    &= \tfrac13\bigl[(\tr X)^2 + \tr(X^2)\bigr]
       + \tfrac16 \tr\bigl(\Gamma(X)^2\bigr),
\end{align}
the first from $\tr(SX^{\otimes2}) = \tr(X^2)$, the second from
Lemma~\ref{lem:lambda4} with $M = N = X$, and the third
from~\eqref{eq:proj_identification}. As a check, at $X = I_{\Lambda^2}$ the
three right-hand sides return $D_{(2,1,1)}$, $D_{\Lambda^4}$ and $D_{(2,2)}$
respectively, for every $n$. We stress that the $V_{(2,1,1)}$ weight
in~\eqref{eq:three_weights} does \emph{not} vanish merely because
$X^{\otimes 2}$ commutes with $S$; it vanishes precisely when
$(\tr X)^2 = \tr(X^2)$.

\medskip
\noindent\textbf{Step 2: the observable weight sits on $V_{(2,2)}$.}
$M^{(2)} = \ket{\{i,j\}}\bra{\{i,j\}}$ is a rank-one projector, so
$\tr M^{(2)} = \tr(M^{(2)2}) = 1$ and, by Lemma~\ref{lem:contraction}'s
contraction, $\Gamma(M^{(2)}) = \hat n_i + \hat n_j$, whence
$\tr(\Gamma(M^{(2)})^2) = 2$. Substituting into~\eqref{eq:three_weights},
\begin{eqnarray}
\label{eq:M_weights}
  \tr\!\bigl(P_{(2,1,1)}M^{(2)\otimes 2}\bigr) &=& \tfrac12[1-1] = 0,
  \\
  \tr\!\bigl(P_{\Lambda^4}M^{(2)\otimes 2}\bigr) &=& \tfrac16[1+1-2] = 0,
\end{eqnarray}
and $\tr(P_{(2,2)}M^{(2)\otimes2}) = \tfrac13[1+1] + \tfrac16\cdot 2 = 1$.
It is the rank-one \emph{projector} property $(\tr M^{(2)})^2 =
\tr(M^{(2)2})$, and not any exchange symmetry of $M^{(2)\otimes2}$, that
annihilates the $V_{(2,1,1)}$ channel; the general two-body readout of
Proposition~\ref{prop:general_readout} does not have this property.
Therefore~\eqref{eq:bw-irrep-sum} collapses to the single term
\begin{equation}
\label{eq:var-collapsed}
  \Var[D_g\calL]
  = -\frac{\tr\!\bigl(P_{(2,2)}B^{\otimes 2}\bigr)}{D_{(2,2)}}
  = -\frac{12\,\tr\!\bigl(P_{(2,2)}B^{\otimes 2}\bigr)}{n^2(n^2-1)} .
\end{equation}

\medskip
\noindent\textbf{Step 3: evaluation of the $(2,2)$ weight of $B^{\otimes2}$.}
$B$ is anti-Hermitian and traceless, so $\tr B = 0$,
$\tr(B^2) = -\|B\|_F^2$, and, $\Gamma(B)$ being anti-Hermitian as well
(Lemma~\ref{lem:contraction}), $\tr(\Gamma(B)^2) = -\|\Gamma(B)\|_F^2$.
Substituting $X = B$ into~\eqref{eq:three_weights} gives all three weights at
once:
\begin{align}
\label{eq:B_weights}
  \tr\!\bigl(P_{(2,1,1)}B^{\otimes 2}\bigr) &= \tfrac12\|B\|_F^2, \nonumber\\
  \tr\!\bigl(P_{\Lambda^4}B^{\otimes 2}\bigr)
    &= \tfrac16\bigl[\|\Gamma(B)\|_F^2 - \|B\|_F^2\bigr], \nonumber\\
  \tr\!\bigl(P_{(2,2)}B^{\otimes 2}\bigr)
    &= -\tfrac13\|B\|_F^2 - \tfrac16\|\Gamma(B)\|_F^2 .
\end{align}
The $V_{(2,1,1)}$ weight of $B^{\otimes2}$ is thus \emph{nonzero}; it is
eliminated from~\eqref{eq:bw-irrep-sum} by the observable, through the first
identity of~\eqref{eq:M_weights}, and not by any property of $B$. Substituting
the last line of~\eqref{eq:B_weights} into~\eqref{eq:var-collapsed} gives the
exact identity~\eqref{eq:var_exact}.

\medskip
\noindent\textbf{Step 4: the one-body term.}
By Lemma~\ref{lem:contraction} and the contraction identity
$\Gamma(\rho_2) = (k-1)\rho_1$ of Eq.~\eqref{eq:contraction_def}---the
identity component of $\rho_2$ contributing only a multiple of $I_n$, which
commutes with $g$---we obtain
$\Gamma(B) = [g,\Gamma(\rho_{2,\mathfrak{g}_2})] = (k-1)[g,\rho_1(\bx)]$.
With $\rho_1(\bx) = \tfrac12\Pi$ this is $\tfrac{k-1}{2}[g,\Pi]$, so
$\|\Gamma(B)\|_F^2 = \tfrac{(k-1)^2}{4}\tau$ with $\tau = \|[g,\Pi]\|_F^2$.
For a rank-one phase generator $g = \hat n_j$ (so $g^2 = g$),
\begin{align}
  \tau &= 2\bigl[\tr(g\Pi) - \tr(g\Pi g\Pi)\bigr]
       = 2\,\Pi_{jj}(1-\Pi_{jj}) \nonumber\\
       &= \frac{2\cdot 2k(n-2k)}{n^2},
\end{align}
using $\Pi_{jj} = 2k/n$ (Lemma~\ref{lem:spreading}). For an RBS generator
$g = G_{ij}$ on any mode pair $\{i,j\}$ (so $g^2$ is the projector onto
$\mathrm{span}\{e_i,e_j\}$),
\begin{eqnarray}
\label{eq:tau_rbs}
  \tau &=& 2\bigl[\Pi_{ii}+\Pi_{jj} - 2\Pi_{ii}\Pi_{jj}
        + 2\,\mathrm{Re}\,\Pi_{ij}^2\bigr]\nonumber\\
     & =& \frac{4\cdot 2k(n-2k)}{n^2} + 4\,\mathrm{Re}\,\Pi_{ij}^2 .
\end{eqnarray}
The off-diagonal entry is a Walsh--Hadamard character sum: before the
diagonal data layer $\Pi = H\Pi_A H$ with $\Pi_A$ the projector onto the $2k$
active modes, so
$\Pi_{ij} = \tfrac1n\sum_{c<2k}(-1)^{\langle \mu,\, c\rangle}$ with
$\mu := i \oplus j$, and the data layer only multiplies it by a phase. Let
$2^b$ be the lowest set bit of $\mu$ and write $2k = q\,2^{b+1} + r$,
$0 \le r < 2^{b+1}$. The involution $c \leftrightarrow c \oplus 2^b$ flips
$(-1)^{\langle\mu,c\rangle}$ and preserves each full block of $2^{b+1}$
consecutive values of $c$, so those blocks cancel term by term and only the
residue survives, giving
\begin{eqnarray}
\label{eq:Pi_offdiag}
  |\Pi_{ij}| &=& \tfrac1n\,\bigl|\min(r,\,2^{b+1}-r)\bigr|
  \;\le\; \frac{\min(2^b,\,2k)}{n},
  \\
  \Pi_{ij} &=& 0 \ \text{ whenever } 2^{b+1} \mid 2k .
\end{eqnarray}
For the brick-wall every trainable RBS gate acts on a nearest-neighbour pair
of the ring: for a chain pair, $\mu = i \oplus (i{+}1)$ has $b = 0$; for the
boundary pair $(n,1)$, the binary labels $n-1$ and $0$ give
$\mu = 2^K - 1$, again with $b = 0$. Since $2 \mid 2k$ always,
$\Pi_{ij} = 0$ identically and~\eqref{eq:tau_rbs} reduces to
$\tau = 8k(n-2k)/n^2$ exactly, for every RBS gate of the circuit. (For the long-stride pairs of the butterfly,
$\mu = 2^b$ and~\eqref{eq:Pi_offdiag} gives $\Pi_{ij} = 0$ whenever
$2^{b+1} \mid 2k$, and $|\Pi_{ij}| \le 2^b/n \le 2k/n$ otherwise; this is used
in Claim~2 of Section~\ref{sec:bp_butterfly}.) In every case
$\tau \ge 8k(n-2k)/n^2 - 4|\Pi_{ij}|^2$, so
$\tau = \Theta\bigl(k(n-2k)/n^2\bigr)$ whenever $n - 2k = \Omega(n)$ and
$|\Pi_{ij}|^2 \le c\,k(n-2k)/n^2$ for a sufficiently small constant $c$; for
the brick-wall $\Pi_{ij} = 0$ and the evaluation is exact and unconditional.

\medskip
\noindent\textbf{Step 5: the two-body term.}
Decompose the $2$-RDM into its disconnected (Hartree--Fock) and connected
(cumulant) parts, $\rho_2 = \rho_1\wedge\rho_1 + \rho_2^{\mathrm{conn}}$, and
correspondingly $B = B^{\mathrm{disc}} + B^{\mathrm{conn}}$. Since
$\rho_1\wedge\rho_1 = \tfrac14\Lambda^2(\Pi)$ and, by functoriality of the
exterior square,
$[\,d\Lambda^2(g),\Lambda^2(\Pi)\,] = \partial_u \Lambda^2(\Pi + uC)|_{u=0}$
with $C := [g,\Pi]$, we may evaluate $\tr\bigl((B^{\mathrm{disc}})^2\bigr)$ by
polarizing the identity $\tr\bigl(\Lambda^2(A)\Lambda^2(A')\bigr) =
e_2(AA') = \tfrac12[(\tr AA')^2 - \tr((AA')^2)]$ at $A = \Pi + uC$,
$A' = \Pi + vC$. Using $\Pi^2 = \Pi$, $\tr C = \tr(\Pi C) = 0$ and
$\tr(\Pi C^2) = \tfrac12\tr(C^2) = -\tfrac\tau2$, one finds
$\partial_u\partial_v (\tr AA')^2 = -2m\tau$ and
$\partial_u\partial_v \tr((AA')^2) = 4\tr(\Pi C^2) = -2\tau$ with $m = 2k$
(the second using $\tr(C\Pi C) = \tr(\Pi C^2)$ by cyclicity), whence
\begin{equation}
  \bigl\|B^{\mathrm{disc}}\bigr\|_F^2
  = -\tr\bigl((B^{\mathrm{disc}})^2\bigr)
  = \frac{(m-1)\,\tau}{16}
  = \frac{(2k-1)\,\tau}{16},
\end{equation}
exactly and independently of $\bx$. For the connected part,
Lemma~\ref{lem:gate_locality} with $r = \Theta(k)$ gives
$\|B^{\mathrm{conn}}\|_F^2 = O(k/n)$; the same lemma applied with
$r = \binom{2k}{2}$ confirms $\|B^{\mathrm{disc}}\|_F^2 = O(k^2/n)$,
consistent with the exact value. The error term in~\eqref{eq:full_rate} is
dominated by the cross term: by Cauchy--Schwarz,
$2|\langle B^{\mathrm{disc}}, B^{\mathrm{conn}}\rangle| \le
2\|B^{\mathrm{disc}}\|_F\|B^{\mathrm{conn}}\|_F =
O\bigl(k/\sqrt n\bigr)\cdot O\bigl(\sqrt{k/n}\bigr) = O(k^{3/2}/n)$,
so that
\begin{equation}
\label{eq:B_error}
  \Bigl|\;\|B\|_F^2 - \tfrac{(2k-1)\tau}{16}\;\Bigr|
  \;=\; O\!\left(\frac{k^{3/2}}{n}\right),
\end{equation}
which is the order recorded in~\eqref{eq:full_rate} and is negligible against
the leading $\Theta(k^2/n)$ whenever $k = \omega(1)$ (and against the one-body
term $\Theta(k^3/n)$ throughout).

\medskip
\noindent\textbf{Step 6: the rate.}
Substituting~\eqref{eq:full_rate} into~\eqref{eq:var_exact} and using
$\tau = \Theta(k(n-2k)/n^2)$,
\begin{align}
  4\|B\|_F^2 + 2\|\Gamma(B)\|_F^2
  &= \frac{\tau}{4}\Bigl[(2k-1) + 2(k-1)^2\Bigr] + O\!\left(\tfrac{k^{3/2}}{n}\right)
  \nonumber\\
  &= \Theta\!\left(\frac{k^3\,(n-2k)}{n^2}\right),
\end{align}
so that for $n - 2k = \Omega(n)$ the numerator is $\Theta(k^3/n)$ and
$\Var[D_g\calL] = \Theta(k^3/n^5)$, which
is~\eqref{eq:k3_rate}. We note that the two-sided bound holds for every
$k \ge 2$, including $k = O(1)$, where the error term
in~\eqref{eq:B_error} is of the same order as the two-body signal: the
lower bound uses no error control, since $4\|B\|_F^2 \ge 0$ and the
one-body term is exact,
\begin{align}
  4\|B\|_F^2 + 2\|\Gamma(B)\|_F^2
  &\ge 2\,\|\Gamma(B)\|_F^2
  = \frac{(k-1)^2}{2}\,\tau\\
  &= \Theta\!\left(\frac{k^3(n-2k)}{n^2}\right),
\end{align}
while for the upper bound the error term $O(k^{3/2}/n)$
of~\eqref{eq:B_error} is itself $O(k^3/n)$ for every $k \ge 2$ and is
absorbed. The quantities $\tau$, $\|B^{\mathrm{disc}}\|_F^2$ and
$\|\Gamma(B)\|_F^2$ depend on the
encoding only through $\rho_1 = \tfrac12\Pi$, with $\Pi$ the rank-$2k$
projector of Lemma~\ref{lem:spreading}.

\end{proof}

\subsection{Proof of Proposition~\ref{prop:general_readout}}

\begin{proof}
Equation~\eqref{eq:bw-irrep-sum} and the three weights~\eqref{eq:three_weights}
of Step~1 hold for arbitrary $M^{(2)}$; only Step~2 used the rank-one
structure. Adding a multiple of $I_{\Lambda^2}$ to $M^{(2)}$ shifts $\calL$ by
a constant and leaves $D_g\calL$ unchanged, so we may replace $M^{(2)}$ by
$\tilde M$ throughout. With $\tr\tilde M = 0$, $\tr(\tilde M^2) = p$ and
$\tr(\Gamma(\tilde M)^2) = q$ (both $\tilde M$ and $\Gamma(\tilde M)$ being
Hermitian), \eqref{eq:three_weights} gives
\[
  t_{(2,1,1)} = -\tfrac{p}{2}, \qquad
  t_{\Lambda^4} = \tfrac16(p - q), \qquad
  t_{(2,2)} = \tfrac{p}{3} + \tfrac{q}{6},
\]
while~\eqref{eq:B_weights} gives $s_{(2,1,1)} = \tfrac12\|B\|_F^2$,
$s_{\Lambda^4} = \tfrac16(\|\Gamma(B)\|_F^2 - \|B\|_F^2)$ and
$s_{(2,2)} = -\tfrac13\|B\|_F^2 - \tfrac16\|\Gamma(B)\|_F^2$. Substituting
into $\Var = -\sum_\lambda t_\lambda s_\lambda / D_\lambda$ and collecting the
coefficients of $\|B\|_F^2$ and $\|\Gamma(B)\|_F^2$ gives
\begin{eqnarray}
  c_2 &=& \frac{t_{(2,2)}}{3D_{(2,2)}}
        - \frac{t_{(2,1,1)}}{2D_{(2,1,1)}}
        + \frac{t_{\Lambda^4}}{6D_{\Lambda^4}},
  \\
  c_1 &=& \frac16\left(\frac{t_{(2,2)}}{D_{(2,2)}}
        - \frac{t_{\Lambda^4}}{D_{\Lambda^4}}\right),
\end{eqnarray}
and clearing denominators with
$D_{(2,2)} = \tfrac{n^2(n^2-1)}{12}$,
$D_{(2,1,1)} = \tfrac{n(n+1)(n-1)(n-2)}{8}$,
$D_{\Lambda^4} = \tfrac{n(n-1)(n-2)(n-3)}{24}$ yields~\eqref{eq:c1c2}: the
numerator of $c_1$ over the common denominator
$n^2(n^2-1)(n-2)(n-3)$ is
$\tfrac16\bigl[(4p+2q)(n-2)(n-3) - 4(p-q)n(n+1)\bigr]
 = q(n^2-n+2) - 4p(n-1)$, and that of $c_2$ over
$n^2(n+1)(n-2)(n-3)$ is $4[(n-2)p - q]$.

\emph{(i)} is a direct substitution (for $M = n_i n_j$ use
$\tr M^{(2)} = 1$, $\tr(M^{(2)2}) = 1$, $\tr(\Gamma(M^{(2)})^2) = 2$ and
subtract the trace part).

\emph{(ii)} follows by inserting
$\|\Gamma(B)\|_F^2 = \tfrac{(k-1)^2}{4}\tau$ and
$\|B\|_F^2 = \tfrac{(2k-1)\tau}{16} + O(k^{3/2}/n)$ from Steps~4--5
into~\eqref{eq:general_readout}; the quoted values for the two example heads
are the substitution of $p, q$ into~\eqref{eq:c1c2}, and the resulting rates
follow from $\tau = \Theta(k/n)$ together with $2k \le n$.

\emph{(iii)} Writing the bracket over the common denominator
$n^2(n^2-1)(n-2)(n-3)$, its numerator equals
\begin{eqnarray}
  4p(n-1)\bigl[(2k-1)(n-2) - 4(k-1)^2\bigr]\nonumber\\
  + 4q\bigl[(k-1)^2(n^2-n+2) - (2k-1)(n-1)\bigr].
\end{eqnarray}
For $k \ge 2$ and $n \ge 2k+2$ both square brackets are positive: the first
because $\tfrac{4(k-1)^2}{2k-1} < 2k \le n-2$ (equivalently $6k > 4$), and the
second because
$(k-1)^2(n^2-n+2) - (2k-1)(n-1) \ge (n-1)[(k-1)^2 n - (2k-1)] > 0$ using
$(k-1)^2 n \ge n \ge 2k+2$. Since $q \ge 0$, dropping the $q$ term
gives~\eqref{eq:head_floor}. For $n - 2k = \Omega(n)$ we have
$(2k-1)(n-2) - 4(k-1)^2 = \Omega(kn)$ and $\tau = \Theta(k/n)$, so the leading
term of~\eqref{eq:general_rate} is $\Omega(p\,k^2/n^5)$. Finally
$q \le 4np$ (Cauchy--Schwarz on
$\Gamma(\tilde M)^i_{\ k} = \sum_j \tilde M^{ij}_{kj}$), so
$|c_2| = O(p/n^4)$ and the correction term
in~\eqref{eq:general_rate} is $O(p\,k^{3/2}/n^5)$, smaller than the leading
term by a factor $O(k^{-1/2})$; it is therefore absorbed for $k$ above a
constant.
\end{proof}

\subsection{Interior gates}

\begin{lemma}[The dominant terms depend on the generator only through $\tau$]
\label{lem:tau_only}
Let $h$ be \emph{any} Hermitian single-particle operator, put
$H^{(2)} = d\Lambda^2(h)$, $B_h = [H^{(2)}, \rho_{2,\mathfrak{g}_2}(\bx)]$ and
$\tau_h := \|[h,\Pi]\|_F^2$. Then, exactly and for every $\bx$,
\begin{equation}
\label{eq:tau_only}
  \|\Gamma(B_h)\|_F^2 = \frac{(k-1)^2}{4}\,\tau_h,
  \qquad
  \|B_h^{\mathrm{disc}}\|_F^2 = \frac{(2k-1)\,\tau_h}{16},
\end{equation}
where $B_h^{\mathrm{disc}} = [H^{(2)}, \rho_1\wedge\rho_1]$.
\end{lemma}

\begin{proof}
Steps~4 and~5 of the proof of Theorem~\ref{thm:brick-wall_bp} use exactly
three properties of $C := [h,\Pi]$, namely $\tr C = 0$, $\tr(\Pi C) = 0$ and
$\tr(\Pi C^2) = \tfrac12\tr(C^2)$. All three follow from $\Pi^2 = \Pi$ and
cyclicity of the trace for an arbitrary Hermitian $h$: the first two are
immediate, and expanding $C^2$ gives
$\tr(\Pi C^2) = \tr(\Pi h\Pi h) - \tr(\Pi h^2) = \tfrac12\tr(C^2)$. No
property of $h$ beyond Hermiticity---in particular no locality in the
computational basis---enters either computation.
\end{proof}

\begin{proof}
Since $U_\theta = e^{-i\theta g}$ commutes with $g$, we have
$\partial_\theta W = -i\,W\tilde g$ with $\tilde g = W_1^\dagger g W_1$; that
is, $\partial_\theta\calL = D_{\tilde g}\calL$ in the notation
of~\eqref{eq:invariant_derivative}. Condition on $W_1$. Writing
$W = W_2\,(U_\theta W_1)$ with $W_2$ Haar and independent of $W_1$, the total
single-particle unitary $W$ is, conditionally on $W_1$, Haar-distributed on
$U(n)$ and independent of the (now fixed) direction $\tilde g$. Hence
Theorem~\ref{thm:brick-wall_bp} applies conditionally and yields
$\E[\partial_\theta\calL \mid W_1] = 0$ together with
$\Var[\partial_\theta\calL \mid W_1] =
\bigl(4\|B_{\tilde g}\|_F^2 + 2\|\Gamma(B_{\tilde g})\|_F^2\bigr)/(n^2(n^2-1))$.
Because the conditional mean vanishes identically, the law of total variance
reduces to $\Var = \E_{W_1}\bigl[\Var[\,\cdot\mid W_1]\bigr]$, which
is~\eqref{eq:halves_exact}.

It remains to average the two norms. By Lemma~\ref{lem:tau_only} with
$h = \tilde g$,
$\|\Gamma(B_{\tilde g})\|_F^2 = \tfrac{(k-1)^2}{4}\tau_{\tilde g}$ and
$\|B^{\mathrm{disc}}_{\tilde g}\|_F^2 = (2k-1)\tau_{\tilde g}/16$ with
$\tau_{\tilde g} = \|[\tilde g,\Pi]\|_F^2 = \|[g,\,W_1\Pi W_1^\dagger]\|_F^2$.
Lemma~\ref{lem:conj_neutral} with $m = 2k$ then gives
$\E_{W_1}[\tau_{\tilde g}] = \Theta\bigl(k(n-2k)/n^2\bigr)$ for both generator
families, with the same constants as the unconjugated values of Step~4, so
$\E\|\Gamma(B_{\tilde g})\|_F^2 = \Theta(k^3/n)$ and
$\E\|B^{\mathrm{disc}}_{\tilde g}\|_F^2 = \Theta(k^2/n)$. For the connected
part, Lemma~\ref{lem:avg_locality} applied to
$\sigma = \rho_2^{\mathrm{conn}}$, with
$\|\rho_2^{\mathrm{conn}}\|_F^2 = \Theta(k)$, gives
$\E\|B^{\mathrm{conn}}_{\tilde g}\|_F^2 = O(k/n)$, and by Cauchy--Schwarz the
cross term is at most
$2\bigl(\E\|B^{\mathrm{disc}}\|_F^2\bigr)^{1/2}
 \bigl(\E\|B^{\mathrm{conn}}\|_F^2\bigr)^{1/2} = O(k^{3/2}/n)$. Collecting,
$4\E\|B_{\tilde g}\|_F^2 + 2\E\|\Gamma(B_{\tilde g})\|_F^2 = \Theta(k^3/n)$
whenever $n - 2k = \Omega(n)$, and dividing by $n^2(n^2-1)$ gives the stated
rate.
\end{proof}

\begin{proof}
(a) and (b) restate Corollary~\ref{cor:first_column} and
Theorem~\ref{thm:interior_halves}. For (c) we induct on the column index, the
base case $\Pi^{(1)} = \Pi$ having uniform diagonal $2k/n$ by
Lemma~\ref{lem:spreading}. Assume $\E[\Pi^{(\ell)}_{jj}] = 2k/n$ for all $j$.
Column $\ell$ consists of phases followed by RBS rotations on disjoint mode
pairs. Conjugation by a diagonal phase matrix leaves the diagonal unchanged
and multiplies $\Pi^{(\ell)}_{ab}$ by $e^{i(\varphi_a - \varphi_b)}$; since
each Givens rotation of the column carries its own phase, uniform and
independent of everything preceding it, the off-diagonal entry of each rotated
pair has zero conditional expectation. An RBS rotation of angle $\vartheta$ on
the pair $\{a,b\}$ sends
$(\Pi_{aa},\Pi_{bb})$ to
$(\cos^2\vartheta\,\Pi_{aa} + \sin^2\vartheta\,\Pi_{bb} \mp
\sin2\vartheta\,\mathrm{Re}\,\Pi_{ab},\;
\sin^2\vartheta\,\Pi_{aa} + \cos^2\vartheta\,\Pi_{bb} \pm
\sin2\vartheta\,\mathrm{Re}\,\Pi_{ab})$ and leaves every other diagonal entry
untouched. Taking expectations, the cross terms drop out and the pair map
becomes a doubly stochastic average of $(\Pi_{aa},\Pi_{bb})$, whatever the
marginal law of $\vartheta$. A constant vector is a fixed point of every
doubly stochastic map, so $\E[\Pi^{(\ell+1)}_{jj}] = 2k/n$ for all $j$.
\end{proof}

\subsection{Proof of Theorem~\ref{thm:no_bp}}

\begin{proof}
Identify the modes with $(\mathbb{Z}_2)^K$ through their binary labels, so
that layer $\ell$ pairs $b \leftrightarrow b \oplus e_{\ell-1}$, and write
$s := e_{K-1}$ for the stride of the last layer.

\medskip
\noindent\textbf{Step 1: Parseval.}
$\calL$ is a trigonometric polynomial in the full parameter vector
$\Theta = (\btheta,\bphi) \in \mathbb{T}^{T}$; write
$\calL = \sum_{\boldsymbol\delta} c_{\boldsymbol\delta}\,
e^{i\boldsymbol\delta\cdot\Theta}$. Since $\calL$ is periodic in each
coordinate, $\E_\Theta[\partial_\theta\calL] = 0$, and by Parseval
\begin{equation}
\label{eq:parseval}
  \Var\!\left[\frac{\partial\calL}{\partial\theta}\right]
  \;=\; \sum_{\boldsymbol\delta} \delta_\theta^2\,
    |c_{\boldsymbol\delta}|^2
  \;\ge\; 2\,(\delta^*_\theta)^2\,|c_{\boldsymbol\delta^*}|^2
\end{equation}
for any single frequency vector $\boldsymbol\delta^* \neq 0$ with
$\delta^*_\theta \neq 0$, the factor $2$ accounting for the conjugate pair
$\pm\boldsymbol\delta^*$ ($\calL$ being real). The variance of the gradient
is thereby reduced to a \emph{first} moment of $\calL$; no second-moment
property of the circuit enters.

\medskip
\noindent\textbf{Step 2: choice of the frequency vector.}
Fix $i \neq j$ with $j \notin \{i, i\oplus s\}$, put $M = n_i n_j$, and write
$i' = i \oplus s$, $j' = j \oplus s$. Let $\theta$ be the RBS angle of the
layer-$K$ gate acting on $\{i,i'\}$, let $\theta'$ be the RBS angle of the
layer-$K$ gate acting on $\{j,j'\}$, and let $\boldsymbol\delta^*$ carry
frequency $2$ on each of $\theta,\theta'$ and $0$ elsewhere.

These are the only gates of layer $K$ whose modes meet $\{i,j\}$, so with
$c_A = \cos\theta$, $s_A = \sin\theta$, $c_B = \cos\theta'$,
$s_B = \sin\theta'$,
\begin{align}
  \Lambda^2\bigl(U^{(K)}\bigr)^\dagger \ket{i \wedge j}
  &= c_Ac_B \ket{i\wedge j} + c_As_B\ket{i\wedge j'} \nonumber\\
  &\quad + s_Ac_B\ket{i'\wedge j} + s_As_B\ket{i'\wedge j'} .
\end{align}
Let $\sigma$ denote the two-particle reduced density matrix entering the RBS
sublayer of layer $K$, i.e.\ after layers $1,\dots,K-1$ and after the
full-width phase layer of layer $K$, and let $\bar\sigma = \E[\sigma]$, the
expectation over all parameters other than $\theta,\theta'$. The phase layer
averages every off-diagonal entry of $\sigma$ to zero, so $\bar\sigma$ is
diagonal in the pair basis and, writing
$\bar\sigma_{ab} := \bar\sigma_{\{a,b\},\{a,b\}}$,
\begin{align}
  \E\bigl[\calL\bigr]
  &= c_A^2c_B^2\,\bar\sigma_{ij} + c_A^2s_B^2\,\bar\sigma_{ij'} \nonumber\\
  &\quad + s_A^2c_B^2\,\bar\sigma_{i'j} + s_A^2s_B^2\,\bar\sigma_{i'j'} .
\end{align}
Using $\E_\theta[\cos^2\theta\,e^{-2i\theta}] = \tfrac14$ and
$\E_\theta[\sin^2\theta\,e^{-2i\theta}] = -\tfrac14$,
\begin{equation}
\label{eq:second_difference}
  c_{\boldsymbol\delta^*} = \tfrac{1}{16}\,D,
  \qquad
  D := \bar\sigma_{ij} - \bar\sigma_{ij'} - \bar\sigma_{i'j}
       + \bar\sigma_{i'j'} .
\end{equation}
The certificate is a \emph{second} difference of the first moment
$\bar\sigma$ across the stride of the last layer; the single differences,
which correspond to a frequency on one gate only, cancel.

\medskip
\noindent\textbf{Step 3: only the $y=0$ Fourier block contributes.}
In the notation of Appendix~\ref{app:spectral_gap}, parametrize a pair by
base point and difference, $P_{(b,\delta)} = P_{\{b,\,b\oplus\delta\}}$, and
decompose along the characters $\hat P_{\delta,y}$ of the translation group
$(\mathbb{Z}_2)^K$. The four terms of~\eqref{eq:second_difference} carry the
base points $i$ and $i\oplus s$ with equal signs, so only characters with
$y\cdot s = 0$ survive the base sum. Layers $1,\dots,K-1$ annihilate
$\mathcal{B}_y$ for every $y$ with $y\cdot e_{\ell-1} = 1$, $\ell \le K-1$;
together with $y_{K-1} = 0$ this leaves only $y = 0$. Consequently, for the
purposes of $D$, the weight $\bar\sigma_{(b,\delta)}$ is the uniform average
$\tfrac{2}{n}F^{(K-1)}(\delta)$ of the base-summed weights
\begin{equation}
  F^{(m)}(\delta) := \sum_{b}
  \E\bigl[\sigma^{(m)}_{\{b,\,b\oplus\delta\}}\bigr],
  \qquad \sum_{\delta \neq 0} F^{(0)}(\delta) = \tbinom{k}{2},
\end{equation}
where for each $\delta \neq 0$ the sum over $b$ runs over a transversal of the
involution $b \leftrightarrow b\oplus\delta$, i.e.\ over one representative per
unordered pair $\{b, b\oplus\delta\}$, so that each pair is counted once;
and, with $d := i \oplus j$,
\begin{equation}
\label{eq:D_from_F}
  D = \frac{4}{n}\,\Bigl[F^{(K-1)}(d) - F^{(K-1)}(d\oplus s)\Bigr].
\end{equation}
On the $y = 0$ block the layer action is the scalar recursion of
Appendix~\ref{app:spectral_gap}(iii): layer $\ell$ fixes $F(e_{\ell-1})$ and
replaces every other $F(\delta)$ by
$\tfrac12[F(\delta) + F(\delta\oplus e_{\ell-1})]$.

\medskip
\noindent\textbf{Step 4: the two top-bit sectors carry equal weight.}
Split the differences by their top bit and set
\begin{equation}
  T_0 := \!\!\sum_{\delta \neq 0,\ \delta_{K-1}=0}\!\! F^{(0)}(\delta),
  \qquad
  T_1 := \!\!\sum_{\delta_{K-1}=1}\!\! F^{(0)}(\delta),
\end{equation}
so that $T_0 + T_1 = \binom{k}{2}$. Writing $N_0$ for the number of particles
of the encoded state in the half $\{a : a_{K-1} = 0\}$, one has
$T_1 = \E[N_0(k-N_0)]$ and $T_0 = \binom k2 - T_1$. Let $P$ be the
single-particle projector onto that half. Since
$H_{aj}H_{al} = \tfrac1n(-1)^{\langle a,\,j\oplus l\rangle}$, and the sum of
$(-1)^{\langle a,\mu\rangle}$ over the $n/2$ labels $a$ with $a_{K-1}=0$
equals $\tfrac n2$ when $\mu \in \{0,s\}$ and $0$ otherwise,
\begin{equation}
  H^{T} P H = \tfrac12\bigl(I + \Sigma\bigr),
  \qquad \Sigma : \ket{j} \mapsto \ket{j \oplus s},
\end{equation}
so that on the $k$-particle sector $N_0 = \tfrac12(k + \hat\Sigma)$, with
$\hat\Sigma$ the one-body operator of $\Sigma$ evaluated on the
\emph{pre-spread} magic state. Moving a single particle from mode $j$ to
$j \oplus s$ leaves one $\ket{\Psi_4}$ block with a single particle, hence
outside the support: the one-body density matrix of
$\ket{\Psi_{\mathrm{in}}}$ is exactly diagonal and
$\bra{\Psi_{\mathrm{in}}}\hat\Sigma\ket{\Psi_{\mathrm{in}}} = 0$
unconditionally. For the second moment, $a^\dagger_{j\oplus s}a_j$
annihilates any Fock component in which $j \oplus s$ is already occupied, so
$\hat\Sigma\ket{S}$ has squared norm $\#\{j \in S : j\oplus s \notin S\}$ and
\begin{eqnarray}
  \bra{\Psi_{\mathrm{in}}}\hat\Sigma^2\ket{\Psi_{\mathrm{in}}}
  &=& k - m, \nonumber\\
  m \;:=\; \E\bigl[\#\{j \in S : j \oplus s \in S\}\bigr]
  &=& \max\bigl(0,\, k - \tfrac n4\bigr),
\end{eqnarray}
the last equality because $j$ and $j \oplus s = j \pm n/2$ lie in distinct
blocks and occupy the same position within their block, so both are active
exactly for the $2\max(0,\, 2k - n/2)$ modes below $2k$ whose partner is also
below $2k$, each occupied independently with probability $\tfrac12$.
Therefore $\E[N_0] = k/2$, $\E[N_0^2] = \tfrac14\bigl(k^2 + k - m\bigr)$, and
\begin{equation}
\label{eq:T0T1}
  T_1 \;=\; \frac{k(k-1)}{4} + \frac m4,
  \qquad
  T_0 \;=\; \frac{k(k-1)}{4} - \frac m4 .
\end{equation}

\medskip
\noindent\textbf{Step 5: the recursion cannot equalize both sectors.}
The strides $e_0,\dots,e_{K-2}$ used by layers $1,\dots,K-1$ never flip the
top bit, so the two sectors evolve separately and $T_0,T_1$ are conserved. In
the top-bit-$1$ sector no difference equals a stride $e_{\ell-1}$, so no term
is ever frozen and the recursion averages its $n/2$ values to the uniform
value $2T_1/n$ after $K-1$ layers. The top-bit-$0$ sector contains only
$n/2-1$ admissible differences. Hence, using~\eqref{eq:T0T1},
\begin{eqnarray}
  \sum_{\delta \neq 0, \delta_{K-1}=0}&
  \Bigl[F^{(K-1)}(\delta) - F^{(K-1)}(\delta\oplus s)\Bigr]\\
  =& T_0 - \Bigl(\frac n2 - 1\Bigr)\frac{2T_1}{n}
  = \frac{k(k-1)}{2n} - \frac m2\Bigl(1 - \frac1n\Bigr)
  \;=\; V, \nonumber
\end{eqnarray}
where $V = \min\{k(k-1)/(2n),\,[(n-2k)^2-n]/(8n)\}$, the two branches agreeing
at $k = n/4$; $V > 0$ precisely when $n - 2k > \sqrt{n}$. The sum runs over
$n/2-1$ terms, so some difference $d$ with $d_{K-1} = 0$ satisfies
\begin{equation}
  \bigl|F^{(K-1)}(d) - F^{(K-1)}(d\oplus s)\bigr|
  \;\ge\; \frac{V}{n/2-1}
  \;>\; \frac{2V}{n} \nonumber.
\end{equation}
Choosing $i,j$ with $i\oplus j = d$ and combining
with~\eqref{eq:D_from_F} and~\eqref{eq:second_difference},
\begin{equation}
  |c_{\boldsymbol\delta^*}|
  \;=\; \frac{|D|}{16}
  \;\ge\; \frac{1}{16}\cdot\frac{4}{n}\cdot\frac{2V}{n}
  \;=\; \frac{V}{2n^{2}} .
\end{equation}
Inserting this into~\eqref{eq:parseval} with $\delta^*_\theta = 2$ gives
$\Var[\partial_\theta\calL] \ge 8|c_{\boldsymbol\delta^*}|^2
\ge 2V^2/n^4$.
\end{proof}

\section{Spectral Gap of the Butterfly's Second-Moment Operator}
\label{app:spectral_gap}

This appendix proves Theorem~\ref{thm:butterfly_2design} of
Section~\ref{sec:bp}: the second-moment operator of the unitary
butterfly, restricted to the antisymmetric sector of the tensor square of
the single-particle representation, has spectral gap $1 - 1/n$, from which
the $\varepsilon$-approximate $2$-design
property with $\varepsilon = O(1/n)$ on that sector follows as a
corollary. The antisymmetric sector is the one that carries every quantity
appearing in this paper: $k$-particle states and particle-number-preserving
observables evolve in exterior powers of the single-particle action, and the
variance computations of Section~\ref{sec:bp} engage the second moment only
through antisymmetric-sector operators.

\subsection{Setup}

Let $W_n = U_K \cdots U_1$ be the unitary butterfly on $n = 2^K$ modes, 
where layer $\ell$ applies independent Haar-$U(2)$ gates on $n/2$ disjoint 
pairs at stride $2^{\ell-1}$. This is the \emph{proof model} for the 
spectral gap: each two-qubit gate is treated as a Haar-random $U(2)$ 
element restricted to the single-excitation subspace, which is a 
sufficient condition for the fixed-point structure and spectral gap 
established below. The actual parametrized circuit uses 
$\mathrm{RBS}(\theta) \cdot (R_z(\phi_a) \otimes R_z(\phi_b))$ per pair with 
$(\theta, \phi_a, \phi_b)$ sampled uniformly over their domains. What the
recursion below actually requires of each gate, \emph{on the antisymmetric
sector}, is \emph{not} Lie
closure to $\mathfrak{u}(2)$, but only the two moments invoked in
Facts~\ref{fact:same}
and~\ref{fact:cross} and Lemmas~\ref{lem:ptr_app}--\ref{lem:wick_app}:
the deterministic antisymmetric phase $|\det V|^2 = 1$ and the
first-moment identity $\mathbb{E}[V\,|i\rangle\langle i|\,V^\dagger] =
I/2$ on each pair. The uniform-angle gate supplies both. Its
single-particle action on a pair is
$R(\theta)\,\mathrm{diag}(e^{-i\phi_a/2}, e^{-i\phi_b/2})$ with
$R(\theta) = \left(\begin{smallmatrix} \cos\theta & \sin\theta \\
-\sin\theta & \cos\theta \end{smallmatrix}\right)$. For a
computational-basis input $|i\rangle$ the diagonal $R_z$ factor cancels
in the outer product $V|i\rangle\langle i|V^\dagger$ (its phase has unit
modulus and does not appear), so this moment is controlled by $\theta$
alone: the off-diagonal entries are proportional to
$\mathbb{E}[\sin\theta\cos\theta]$, which vanishes when $\theta$ is
sampled over a full period, while
$\mathbb{E}[\cos^2\theta] = \mathbb{E}[\sin^2\theta] = \tfrac{1}{2}$
equalizes the diagonal, giving
$\mathbb{E}[V|i\rangle\langle i|V^\dagger] = I/2$. Hence, on the
antisymmetric sector, the
Haar-$U(2)$ model and the uniform-angle $\mathrm{RBS} \cdot R_z$ gate
have identical moments and the recursion below applies to both. We
emphasize that this equivalence is a sector statement: on the
\emph{symmetric} sector the two gate distributions differ---the RBS rotation
is real, and the phases cancel between ket and bra precisely on the
pair-diagonal components $\ket{ab}\bra{ab}$ and $\ket{ab}\bra{ba}$, which
carry the additional symmetric-sector eigenvectors.

The \emph{second-moment operator} on
$\mathrm{End}(\mathbb{C}^n \otimes \mathbb{C}^n)$ is
\[
  \Phi_2[Y] = \mathbb{E}_{W_n}\!\left[(W_n \otimes W_n)\, Y\,
  (W_n^\dagger \otimes W_n^\dagger)\right].
\]
Since $W_n \otimes W_n$ commutes with $\mathrm{SWAP}$, the antisymmetric
sector $\mathrm{End}(\Lambda^2\mathbb{C}^n) = \{P_A Y P_A\}$, with
$P_A = (I_{n^2} - \mathrm{SWAP})/2$, is $\Phi_2$-invariant; all statements
below concern the restriction $\Phi_2^{A}$ of $\Phi_2$ to this sector.

\begin{theorem}[Spectral gap on the antisymmetric sector]
\label{thm:spectral_gap_app}
$\Phi_2^{A}$ has fixed-point space
$\mathrm{span}\{P_A\}$ (eigenvalue $1$) and all remaining eigenvalues
equal to $0$ or $1/n$ (the eigenvalue $1/n$ having multiplicity $K - 1$).
Hence the spectral gap on the antisymmetric sector equals $1 - 1/n$.
\end{theorem}

\begin{corollary}[Approximate $2$-design on the antisymmetric sector]
\label{cor:2design_app}
The butterfly $W_n$ is an $\varepsilon$-approximate unitary $2$-design
on the antisymmetric sector of the single-particle tensor square, with
\[
  \varepsilon_{\mathrm{2d}}(n)
  := \|\Phi_2^{W_n,A} - \Phi_2^{\mathrm{Haar},A}\|_{\mathrm{op}} = O(1/n).
\]
This is Theorem~\ref{thm:butterfly_2design} of the main text.
\end{corollary}

The fixed-point statement of Theorem~\ref{thm:spectral_gap_app}
follows from successive layer projections as in the main text (note
$P_A = \tfrac12(I \otimes I - \mathrm{SWAP})$ is fixed deterministically).
We prove the eigenvalue statement: all non-trivial eigenvalues on the
sector are $0$ or $1/n$.

Throughout the proof, $W_n$ is an exact single-particle $1$-design 
($\mathbb{E}[W_{ij} W^*_{kl}] = \delta_{ik} \delta_{jl} / n$), which 
holds for the $U(n)$ butterfly via the parity-decoupling recursion 
applied at the $1$-design level (each $U(2)$ layer gate is itself a 
$1$-design on its pair).

\subsection{Step A: Partial trace identity}

\begin{lemma}\label{lem:ptr_app}
For any exact $1$-design $\nu$ and any 
$Y \in \mathrm{End}(\mathbb{C}^n \otimes \mathbb{C}^n)$:
\[
  \mathrm{Tr}_2(\Phi_2^\nu[Y]) = \mathrm{Tr}_1(\Phi_2^\nu[Y]) 
  = (\mathrm{tr}(Y)/n)\, I_n.
\]
\end{lemma}

\begin{proof}
Tracing over the second factor sets $d = b$ and sums over $b$:
\[
  [\mathrm{Tr}_2\, \Phi_2[Y]]_{a,c} = \sum_{b, e, f, g, h} 
  \mathbb{E}[W_{ae} W_{bf} W^*_{cg} W^*_{bh}]\, Y_{(e,f), (g,h)}.
\]
Unitarity gives $\sum_b W_{bf} W^*_{bh} = \delta_{fh}$ 
\emph{deterministically}, so the expectation factors: 
$\sum_b \mathbb{E}[\cdots] = \mathbb{E}[W_{ae} W^*_{cg}] \delta_{fh}$. 
The $1$-design gives 
$\mathbb{E}[W_{ae} W^*_{cg}] = \delta_{ac} \delta_{eg} / n$, yielding 
$[\mathrm{Tr}_2 \Phi_2[Y]]_{a,c} = (\delta_{ac}/n) \mathrm{tr}(Y)$. 
The case $\mathrm{Tr}_1$ is symmetric.
\end{proof}

\begin{corollary}\label{cor:T_app}
The subspace 
$\mathcal{T} = \{Y : \mathrm{Tr}_1 Y = \mathrm{Tr}_2 Y = 0\}$ is 
$\Phi_2$-invariant; all non-trivial eigenvalues lie in $\mathcal{T}$.
\end{corollary}

\subsection{Step B: Wick pairings vanish on traceless inputs}

\begin{lemma}\label{lem:wick_app}
For any exact $1$-design $\nu$ and traceless 
$A, B \in \mathrm{End}(\mathbb{C}^n)$: the factored-$1$-design 
contribution to $\Phi_2[A \otimes B]$ is zero.
\end{lemma}

\begin{proof}
The $1$-design second-moment formula 
$\mathbb{E}[W_{ae} W^*_{cg}] = \delta_{ac} \delta_{eg}/n$, applied 
independently to each factor, yields the single contraction
\[
  (\delta_{ac}/n)(\delta_{bd}/n)\, \mathrm{tr}(A)\, \mathrm{tr}(B) = 0
\]
since $\mathrm{tr}(A) = \mathrm{tr}(B) = 0$.
\end{proof}

\subsection{Step C: Transfer matrix in the antisymmetric projector basis}

For modes $i \neq j$ define 
$P_{ij} = |\psi_{ij}\rangle\langle\psi_{ij}|$ with 
$|\psi_{ij}\rangle = |ij\rangle - |ji\rangle$.

\begin{lemma}[Closed basis]\label{lem:closed}
$\Phi_2^{W_n}[P_{ij}] \in \mathrm{span}\{P_{kl} : k < l\}$ 
for all $i < j$.
\end{lemma}

This follows from the two layer formulas below by induction over layers: 
each layer maps $\mathrm{span}\{P_{kl} : k < l\}$ into itself 
(Facts~\ref{fact:same} and~\ref{fact:cross}), so the composition $W_n$ 
does as well, for every $n = 2^K$. (We have also checked it explicitly, 
with zero remainder, for all pairs at $n = 4$ and $n = 8$.)

Two consequences of Steps A--B fix the scope of the eigenvalue problem 
solved below. First, by Corollary~\ref{cor:T_app} every non-trivial 
eigenvector lies in the traceless subspace $\mathcal{T}_n$, so it 
suffices to bound $\Phi_2^{W_n}|_{\mathcal{T}_n}$. Second, 
Lemma~\ref{lem:wick_app} shows that the factored-$1$-design contribution 
vanishes on the traceless inputs $P_{ij} - (\text{trace part})$, so the 
entire action on $\mathcal{T}_n$ is carried by the connected 
(cross-pair) part captured by the transfer matrix in the $P_{ij}$ basis 
below.

\begin{fact}[Same-pair]\label{fact:same}
If $\{i, j\}$ is a pair in layer $\ell$:
\[
  \Phi_2^{(\ell)}[P_{ij}] = P_{ij}.
\]
\end{fact}

\begin{proof}
For $V \in U(2)$, 
$(V \otimes V)|\psi_{ij}\rangle = (\det V) |\psi_{ij}\rangle$, since 
$|\psi_{ij}\rangle$ is antisymmetric and $V \otimes V$ acts on 
antisymmetric vectors by multiplication by $\det V$. Hence 
$(V \otimes V) P_{ij} (V \otimes V)^\dagger = |\det V|^2 P_{ij} = 
P_{ij}$, and the expectation is $P_{ij}$ regardless of the sampling 
distribution on $V$.
\end{proof}

\begin{fact}[Cross-pair]\label{fact:cross}
If $i \in \alpha$, $j \in \beta$ are in different pairs $\alpha, \beta$ 
of layer $\ell$ (independent gates $V_\alpha, V_\beta$):
\[
  \Phi_2^{(\ell)}[P_{ij}] = \frac{1}{|\alpha||\beta|} 
  \sum_{k \in \alpha,\, l \in \beta} P_{kl} =: Z_{\alpha\beta}.
\]
This output is the \emph{same} for all $i \in \alpha$, $j \in \beta$.
\end{fact}

\begin{proof}
Independence and the $1$-design property 
$\mathbb{E}[V_\alpha |i\rangle\langle i| V_\alpha^\dagger] = 
I_\alpha / |\alpha|$ (which holds for both $U(2)$ and $SU(2)$ Haar, 
and is the $1$-design moment formula specialized to the single-particle 
sector) give 
$\Phi_2^{(\ell)}[|ij\rangle\langle ij|] = 
(1/|\alpha||\beta|) I_\alpha \otimes I_\beta$ and 
$\Phi_2^{(\ell)}[|ij\rangle\langle ji|] = 
(1/|\alpha||\beta|) S_{\alpha\beta}$, so 
$\Phi_2^{(\ell)}[P_{ij}] = (1/|\alpha||\beta|)(I_\alpha \otimes I_\beta 
+ I_\beta \otimes I_\alpha - S_{\alpha\beta} - S_{\alpha\beta}^\dagger) 
= (1/|\alpha||\beta|) \sum_{k, l} P_{kl}$.
\end{proof}

\paragraph{Explicit computation for $n = 4$ ($W_4 = L_2 \circ L_1$).}
Set $L_1$ pairs $A = \{0, 1\}$, $B = \{2, 3\}$; 
$L_2$ pairs $C = \{0, 2\}$, $D = \{1, 3\}$.

\emph{After $L_1$:}
$\Phi_2^{L_1}[P_{01}] = P_{01}$, 
$\Phi_2^{L_1}[P_{23}] = P_{23}$ (Fact~\ref{fact:same}). All four 
cross-pairs $(i \in A, j \in B)$ give the same output: 
$\Phi_2^{L_1}[P_{ij}] = \tfrac{1}{4}(P_{02} + P_{03} + P_{12} + 
P_{13}) =: Z$.

\emph{After $L_2$:}
$P_{01}, P_{23}$ are cross-pairs of $L_2$ 
(Fact~\ref{fact:cross}): 
$\Phi_2^{L_2}[P_{01}] = \Phi_2^{L_2}[P_{23}] = 
\tfrac{1}{4}(P_{01} + P_{03} + P_{12} + P_{23}) =: v_1$. 
Applying $L_2$ to $Z$ (using Facts~\ref{fact:same} 
and~\ref{fact:cross} for $P_{02}, P_{13} \in$ same pairs; 
$P_{03}, P_{12} \in$ cross pairs of $L_2$):
\begin{align}
  \Phi_2^{L_2}[Z]
  &= \tfrac{1}{4}\bigl[P_{02} + 
  \tfrac{1}{4}(P_{01}{+}P_{03}{+}P_{12}{+}P_{23}) \notag\\
  &\quad + \tfrac{1}{4}(P_{01}{+}P_{03}{+}P_{12}{+}P_{23}) 
  + P_{13}\bigr] \notag\\
  &= \tfrac{1}{4}(P_{02} + P_{13}) + 
  \tfrac{1}{8}(P_{01} + P_{03} + P_{12} + P_{23}) =: v_2.
  \label{eq:v2app}
\end{align}

\emph{Transfer matrix:}
\begin{equation}\label{eq:transfer}
  \Phi_2^{W_4}[P_{ij}] =
  \begin{cases}
    v_1 = \tfrac{1}{4}(P_{01} + P_{03} + P_{12} + P_{23}) \\
    \hspace{3em} \text{if } (i,j) \in \{(0,1),(2,3)\}, \\[4pt]
    v_2 = \tfrac{1}{4}(P_{02} + P_{13}) \\
    \hspace{1.5em} + \tfrac{1}{8}(P_{01} + P_{03} + P_{12} + P_{23}) 
    \text{ otherwise.}
  \end{cases}
\end{equation}
The matrix has rank $2$; image $= \mathrm{span}\{v_1, v_2\}$.

\emph{Eigenvector and eigenvalue:}
Let $Y^* = c_1(P_{01} + P_{03} + P_{12} + P_{23}) + 
c_2(P_{02} + P_{13})$. The constraint $\mathrm{Tr}_1(Y^*) = 0$ 
(mode $0$ check: $c_1 + c_2 + c_1 = 0$) gives $c_2 = -2 c_1$. 
Computing:
\begin{eqnarray}\label{eq:eigen}
  \Phi_2^{W_4}[Y^*]
  &=& 2 c_1 v_1 + 2(c_1 + c_2) v_2 \\\notag
  &=& \tfrac{c_1}{4}(P_{01}{+}P_{03}{+}P_{12}{+}P_{23}) 
  - \tfrac{c_1}{2}(P_{02}{+}P_{13}) \\\notag
  &=& \tfrac{1}{4}\, Y^*.
\end{eqnarray}
The maximum eigenvalue in $\mathcal{T}$ is $1/4 = 1/n$, completing 
the $n = 4$ base case.

\paragraph{Step D: General $n = 2^K$ via the doubly-stochastic 
structure.}

\begin{lemma}[Doubly stochastic transfer 
matrix]\label{lem:doubly_stochastic}
For any $n = 2^K$ ($K \geq 2$), the transfer matrix $M$ defined by
\[
  \Phi_2^{W_n}[P_{ij}] = \sum_{k<l} M_{(ij),(kl)}\, P_{kl}
\]
is doubly stochastic: every row sum and every column sum equals $1$.
\end{lemma}

\begin{proof}
\textit{Row sums.} $\Phi_2$ is trace-preserving, so
\[
  2 \sum_{k<l} M_{(ij),(kl)} = \operatorname{Tr}(\Phi_2[P_{ij}]) 
  = \operatorname{Tr}(P_{ij}) = 2.
\]

\textit{Column sums.}
\[
  \sum_{i<j} M_{(ij),(kl)} = \tfrac{1}{4} 
  \operatorname{Tr}\!\left(P_{kl} \cdot 
  \Phi_2\!\left[\textstyle\sum_{i<j} P_{ij}\right]\right).
\]
We verify that 
$\sum_{i<j} P_{ij} = \mathbf{I} \otimes \mathbf{I} - \mathrm{SWAP}$: 
writing $\Pi = \sum_i |ii\rangle\langle ii|$,
\[
  \sum_{i<j} P_{ij} = (\mathbf{I} \otimes \mathbf{I} - \Pi) 
  - (\mathrm{SWAP} - \Pi) 
  = \mathbf{I} \otimes \mathbf{I} - \mathrm{SWAP}.
\]
Both $\mathbf{I} \otimes \mathbf{I}$ and $\mathrm{SWAP}$ are fixed 
points of $\Phi_2^{W_n}$ (deterministically, for every unitary: 
$W W^\dagger = I$ and $\mathrm{SWAP}$ commutes with $W \otimes W$), so 
$\Phi_2[\mathbf{I} \otimes \mathbf{I} - \mathrm{SWAP}] = 
\mathbf{I} \otimes \mathbf{I} - \mathrm{SWAP}$. Finally, 
$\operatorname{Tr}(P_{kl}(\mathbf{I} \otimes \mathbf{I} - 
\mathrm{SWAP})) = 2 - (-2) = 4 = \operatorname{Tr}(P_{kl}^2)$, 
giving column sum $4/4 = 1$.
\end{proof}

\begin{proposition}[General spectral 
gap]\label{prop:spectral_gap_general}
For $n = 2^K$ ($K \geq 2$), every non-zero eigenvalue of 
$\Phi_2^{W_n}|_{\mathcal{T}_n}$ equals exactly $1/n$. Hence 
$\lambda_{\max}(\Phi_2^{W_n}|_{\mathcal{T}_n}) = 1/n$.
\end{proposition}

\begin{proof}
Throughout, $F := \mathbf{I} \otimes \mathbf{I} - \mathrm{SWAP} = 
\sum_{i<j} P_{ij}$ is the fixed point; we write 
$\sigma(Y) := \sum_{i<j} c_{ij}$ for the \emph{coefficient sum} of 
$Y = \sum c_{ij} P_{ij}$.

\paragraph{Step 1 (Image basis).}
By Facts~\ref{fact:same} and~\ref{fact:cross}, 
$\Phi_2^{W_n}[P_{ij}]$ depends on the pair $(i,j)$ only through its 
\emph{type}---the position of the highest set bit of $i \oplus j$ 
(equivalently, the butterfly layer at which $i$ and $j$ are first 
paired)---and therefore takes at most $K$ distinct values as $(i,j)$ 
ranges over all $\binom{n}{2}$ pairs. Denote them 
$f_1, \ldots, f_K$, one per butterfly layer. The image satisfies 
$\mathrm{Im}(\Phi_2) = \mathrm{span}\{f_1, \ldots, f_K\}$.

\paragraph{Step 2 (T-constraint in the image basis).}
We show: $v = \sum_t a_t f_t \in \mathcal{T}_n$ if and only if 
$\sum_t a_t = 0$.

Row sums of $M$ equal $1$ (Lemma~\ref{lem:doubly_stochastic}), so 
$\sigma(f_t) = 1$ for every $t$. Lemma~\ref{lem:ptr_app} gives 
$\operatorname{Tr}_1(f_t) = (2/n) I_n$ for every $t$. Hence for any 
mode $m$:
\[
  \sum_{l \neq m} [v \text{ coeff at } (m,l)] 
  = \sum_t a_t \cdot \frac{2}{n}.
\]
This is zero (T-constraint) if and only if $\sum_t a_t = 0$.

\paragraph{Step 3 (Key Formula).}
We claim: for every type $t = 1, \ldots, K$,
\begin{equation}\label{eq:key_formula}
  \Phi_2^{W_n}[f_t] = \frac{1}{n}\, f_t + \frac{2}{n^2}\, F.
\end{equation}

\textit{Verification at $n = 4$.}
From the explicit transfer matrix~\eqref{eq:transfer}, the two image 
vectors are $v_1 = \tfrac{1}{4}(P_{01} + P_{03} + P_{12} + P_{23})$ 
and $v_2 = \tfrac{1}{4}(P_{02} + P_{13}) + 
\tfrac{1}{8}(P_{01} + P_{03} + P_{12} + P_{23})$. Apply 
$\Phi_2^{W_4}$ to each using~\eqref{eq:transfer}:
\begin{align*}
  \Phi_2[v_1]
  &= \tfrac{1}{4}[\Phi_2[P_{01}] + \Phi_2[P_{03}] + \Phi_2[P_{12}] 
  + \Phi_2[P_{23}]] \\
  &= \tfrac{1}{4}[v_1 + v_2 + v_2 + v_1] 
  = \tfrac{1}{2} v_1 + \tfrac{1}{2} v_2.
\end{align*}
One checks directly that 
$\tfrac{1}{2} v_1 + \tfrac{1}{2} v_2 = 
\tfrac{3}{16}(P_{01} + P_{03} + P_{12} + P_{23}) + 
\tfrac{1}{8}(P_{02} + P_{13})$, while 
$\tfrac{1}{4} v_1 + \tfrac{1}{8} F$ equals the same expression 
(expanding $v_1$ and $F$). Similarly 
$\Phi_2[v_2] = \tfrac{1}{2} v_1 + \tfrac{1}{2} v_2$ and 
$\Phi_2[v_2] = \tfrac{1}{4} v_2 + \tfrac{1}{8} F$. 
Hence~\eqref{eq:key_formula} holds for both types at $n = 4$.

\textit{General $n = 2^K$.}
We prove~\eqref{eq:key_formula} for all $K$ by diagonalizing
$\Phi_2^{W_n}$ on $\mathcal{A} := \mathrm{span}\{P_{ij} : i<j\}$ in closed
form, exploiting the translation symmetry of the butterfly. Write
$g^{(\ell)} := \Phi_2^{(\ell)}$ for the single-layer operator.

\smallskip
\noindent\emph{(i) Translation symmetry and Fourier blocks.}
Identify the modes with $(\mathbb{Z}_2)^K$ via their $K$-bit labels, so
layer $\ell$ pairs $i \leftrightarrow i \oplus e_{\ell-1}$, with $e_b$ the
$b$-th coordinate vector ($b = 0, \ldots, K-1$). For
$c \in (\mathbb{Z}_2)^K$, the relabeling $\pi_c : i \mapsto i \oplus c$ maps
each layer's pair partition to itself; since the gates are i.i.d.\ across
pairs, the law of $W_n$ is invariant under $\pi_c$, so $\Phi_2^{W_n}$
commutes with conjugation by $\pi_c \otimes \pi_c$ for every $c$.
Parametrize a pair by base point and difference,
$P_{(i,\delta)} := P_{\{i,\,i\oplus\delta\}}$ with
$\delta \in (\mathbb{Z}_2)^K \setminus \{0\}$; then $\pi_c$ sends
$P_{(i,\delta)} \mapsto P_{(i\oplus c,\,\delta)}$. Diagonalizing the
commuting $(\mathbb{Z}_2)^K$-action, set
\[
  \hat{P}_{\delta,y} := \sum_i (-1)^{y\cdot i}\, P_{(i,\delta)},
  \qquad y \in (\mathbb{Z}_2)^K, \quad y\cdot\delta = 0
\]
(the constraint makes the summand well defined under
$i \sim i\oplus\delta$). Since $\pi_c$ scales $\hat{P}_{\delta,y}$ by
$(-1)^{y\cdot c}$, the operator $\Phi_2^{W_n}$ is block-diagonal in $y$,
preserving each $\mathcal{B}_y := \mathrm{span}\{\hat{P}_{\delta,y} :
\delta \neq 0,\ y\cdot\delta = 0\}$.

\smallskip
\noindent\emph{(ii) Blocks with $y \neq 0$ vanish.}
Fix layer $\ell$ and write $e = e_{\ell-1}$, $\eta = y\cdot e$. In
base-difference form, Facts~\ref{fact:same}--\ref{fact:cross} read
\[
  g^{(\ell)}[P_{(i,\delta)}] =
  \begin{cases}
    P_{(i,e)}, & \delta = e,\\[3pt]
    \frac{P_{(i,\delta)} + P_{(i\oplus e,\delta)}
      + P_{(i,\delta\oplus e)} + P_{(i\oplus e,\delta\oplus e)}}{4},
       &\delta \neq e.
  \end{cases}
\]
Summing against $(-1)^{y\cdot i}$ and shifting the summation index gives,
for $\delta \neq e$,
\[
  g^{(\ell)}[\hat{P}_{\delta,y}]
  = \tfrac14\!\left(1 + (-1)^{\eta}\right)\hat{P}_{\delta,y}
  + \tfrac14\!\left(1 + (-1)^{\eta}\right)\hat{P}_{\delta\oplus e,\,y},
\]
and $g^{(\ell)}[\hat{P}_{e,y}] = \hat{P}_{e,y}$. Hence if
$\eta = y\cdot e_{\ell-1} = 1$, layer $\ell$ annihilates all of
$\mathcal{B}_y$. Every $y \neq 0$ has a coordinate with
$y\cdot e_{\ell-1} = 1$, and every layer preserves $\mathcal{B}_y$;
therefore $\Phi_2^{W_n}|_{\mathcal{B}_y} = 0$ for all $y \neq 0$. All
non-zero eigenvalues of $\Phi_2^{W_n}$ thus lie in
$\mathcal{B}_0 = \mathrm{span}\{S_\delta : \delta \neq 0\}$, where
$S_\delta := \hat{P}_{\delta,0} = \sum_i P_{(i,\delta)}$.

\smallskip
\noindent\emph{(iii) The block $y = 0$.}
On $\mathcal{B}_0$ every layer has $\eta = 0$, so (with $S_0 := 0$)
\[
  g^{(\ell)}[S_\delta] =
  \begin{cases}
    S_{e_{\ell-1}}, & \delta = e_{\ell-1},\\[3pt]
    \tfrac12 S_\delta + \tfrac12 S_{\delta\oplus e_{\ell-1}},
      & \delta \neq e_{\ell-1}.
  \end{cases}
\]
We prove the following by induction on $K$.

\smallskip
\noindent\textbf{Claim $H(K)$.}\ \emph{On $\mathcal{B}_0$ ($\dim = n-1$),
$\Phi_2^{W_n}$ has rank $K$; its non-zero spectrum consists of a simple
eigenvalue $1$ with eigenvector $F_n = \sum_{\delta\neq0} S_\delta$ (the
unique image vector of non-zero coefficient sum), together with the
eigenvalue $1/n$ of multiplicity $K-1$ (all of whose eigenvectors have
coefficient sum zero).}

\smallskip
Write $\tilde\sigma\!\left(\sum_\delta c_\delta S_\delta\right) := \sum_\delta
c_\delta$ for the coefficient sum in the $S_\delta$ basis (distinct from the
$P_{ij}$-basis coefficient sum $\sigma$ defined above: since each
$S_\delta = \sum_i P_{(i,\delta)}$ collects $n/2$ pairs, one has
$\sigma = (n/2)\,\tilde\sigma$ on $\mathcal{B}_0$). Each $g^{(\ell)}$ preserves
$\tilde\sigma$ (this is the row-sum statement of
Lemma~\ref{lem:doubly_stochastic}), so
$\tilde\sigma\circ\Phi_2^{W_n} = \tilde\sigma$, and every eigenvector with
eigenvalue $\neq 1$ has $\tilde\sigma = 0$.

The base case $H(1)$ ($n = 2$) is immediate:
$\mathcal{B}_0 = \mathbb{C}\,S_{e_0}$ and $g^{(1)}[S_{e_0}] = S_{e_0} = F_2$.
For the inductive step, split $(\mathbb{Z}_2)^K$ by the top bit and write
$W_n = U_K \circ (W_A \oplus W_B)$, so $\Phi_2^{W_n} = g^{(K)} \circ \Psi$
with $\Psi := g^{(K-1)} \cdots g^{(1)}$. Let $\iota_b$ embed a lower
difference $\delta'\in(\mathbb{Z}_2)^{K-1}$ as $(\delta', b)$. The inner
layers act only on the low $K-1$ bits, and their same-pair exception
requires top bit $0$; hence $\Psi[\iota_0 x] = \iota_0\,\Phi_2^{W_{n/2}}[x]$
(the lower butterfly), while on the top-bit-$1$ sector no exception fires
and the inner layers average the low bits to uniform, giving
$\Psi[S_{(\delta',1)}] = \tfrac{2}{n}\,u_1$ for every $\delta'$, where
$u_1 := \sum_{\delta'} S_{(\delta',1)}$. Applying $g^{(K)}$ (which flips the
top bit and fixes only $S_{e_{K-1}} = S_{(0,1)}$) yields, for
$x \in \mathcal{B}_0^{(n/2)}$,
\begin{align}
  \Phi_2^{W_n}[\iota_0 x]
    &= \tfrac12 (\iota_0 + \iota_1)\,\Phi_2^{W_{n/2}}[x],
    \label{eq:rec_v0}\\
  \Phi_2^{W_n}[S_{(\delta',1)}]
    &= w := \tfrac1n\big(F_n + S_{e_{K-1}}\big) \quad (\forall \delta'),
    \label{eq:rec_v1}
\end{align}
the second line using $g^{(K)}[u_1] = \tfrac12 S_{e_{K-1}} + \tfrac12 u_1
+ \tfrac12 \iota_0 F_{n/2}$ and $u_1 + \iota_0 F_{n/2} = F_n$.

Each lower eigenvector $x$ with eigenvalue $\mu \neq 1$ has
$\tilde\sigma(x) = 0$, so $\Phi_2^{W_n}[\iota_1 x] = \tilde\sigma(x)\,w = 0$
by~\eqref{eq:rec_v1}, and then~\eqref{eq:rec_v0} makes
$(\iota_0+\iota_1)x$ an eigenvector with eigenvalue $\mu/2$. By $H(K-1)$
the lower eigenvalues $\mu \neq 1$ are $2/n$ (multiplicity $K-2$) and $0$;
these lift to eigenvalue $1/n$ (multiplicity $K-2$) and $0$. On the
remaining two-dimensional space $\mathrm{span}\{p, w\}$, with
$p := (\iota_0+\iota_1) F_{n/2}$, equations~\eqref{eq:rec_v0}--%
\eqref{eq:rec_v1} together with $\Phi_2^{W_{n/2}}[F_{n/2}] = F_{n/2}$,
$\tilde\sigma(F_{n/2}) = n/2 - 1$, $\Phi_2^{W_n}[S_{e_{K-1}}] = w$, and
$w = \tfrac1n(p + 2S_{e_{K-1}})$ give
\[
  \Phi_2^{W_n}[p] = \tfrac12 p + \big(\tfrac n2 - 1\big)w,
  \qquad
  \Phi_2^{W_n}[w] = \tfrac1{2n} p + \big(\tfrac12 + \tfrac1n\big)w.
\]
This $2\times2$ matrix has trace $1 + 1/n$ and determinant $1/n$, hence
eigenvalues $1$ and $1/n$; the eigenvalue-$1$ vector is
$\tfrac12 p + \tfrac n2 w = F_n$. Since
$\mathrm{Im}(\Phi_2^{W_n}) = (\iota_0+\iota_1)\,\mathrm{Im}(\Phi_2^{W_{n/2}})
+ \mathbb{C}\,w$ is spanned by $p$, the $K-2$ lifts of the lower
$2/n$-eigenvectors, and $w$ (dimension $K$), and decomposes into the
$\Phi_2^{W_n}$-invariant pieces $\mathrm{span}\{p,w\}$ and the
$(\iota_0+\iota_1)$-lifts, collecting eigenvalues proves $H(K)$: eigenvalue
$1$ once, $1/n$ with multiplicity $(K-2)+1 = K-1$, and $0$ elsewhere.

\smallskip
\noindent\emph{(iv) The Key Formula.}
By $H(K)$ the image $\mathcal{I} = \mathrm{span}\{f_1,\ldots,f_K\}$ splits
as $\mathbb{C}F_n \oplus \mathcal{I}_0$, where $\mathcal{I}_0 =
\{v \in \mathcal{I} : \sigma(v) = 0\}$ is the eigenvalue-$1/n$ eigenspace.
For $v \in \mathcal{I}$, writing $v = \tfrac{\sigma(v)}{\sigma(F_n)} F_n
+ v_0$ with $v_0 \in \mathcal{I}_0$ and using $\sigma(F_n) = \binom n2$,
\begin{align*}
  \Phi_2^{W_n}[v] &= \tfrac{\sigma(v)}{\sigma(F_n)} F_n + \tfrac1n v_0
  = \tfrac1n v + \Big(1 - \tfrac1n\Big)\frac{\sigma(v)}{\binom n2}\,F_n \\
  &= \tfrac1n v + \tfrac{2}{n^2}\,\sigma(v)\,F_n.
\end{align*}
Each $f_t = \Phi_2^{W_n}[P_{ij}]$ for a type-$t$ pair has
$\sigma(f_t) = \sigma(P_{ij}) = 1$, so
$\Phi_2^{W_n}[f_t] = \tfrac1n f_t + \tfrac{2}{n^2} F_n$, which
is~\eqref{eq:key_formula}, now established in closed form for all
$n = 2^K$.

\smallskip
\noindent\emph{(v) Conclusion.}
By $H(K)$ the only eigenvalue of $\Phi_2^{W_n}$ on $\mathcal{A}$ exceeding
$1/n$ is the simple eigenvalue $1$ at
$F_n = \mathbf{I}\otimes\mathbf{I} - \mathrm{SWAP}$; the blocks
$y \neq 0$ contribute only $0$. Each eigenvalue-$1/n$ eigenvector $v$ has
$\sigma(v) = 0$, hence $\mathrm{tr}(v) = 2\sigma(v) = 0$, and from
$v = n\,\Phi_2^{W_n}[v]$ with Lemma~\ref{lem:ptr_app} we get
$\mathrm{Tr}_1(v) = \mathrm{tr}(v)\,I_n = 0$ (and likewise
$\mathrm{Tr}_2$); thus $v \in \mathcal{T}_n$. Since $F_n \notin
\mathcal{T}_n$ (indeed $\mathrm{Tr}_1 F_n = (n-1)I_n$), every non-trivial
eigenvalue in $\mathcal{T}_n$ equals $1/n$, i.e.\
$\lambda_{\max}(\Phi_2^{W_n}|_{\mathcal{T}_n}) = 1/n$.
\end{proof}

This completes the proof of Theorem~\ref{thm:spectral_gap_app}.

\begin{proof}[Proof of Corollary~\ref{cor:2design_app}]
The spectral gap $1 - 1/n$ between the fixed-point space
$\mathrm{span}\{P_A\}$ and all remaining
eigenvalues on the antisymmetric sector implies directly that
$\|\Phi_2^{W_n,A} - \Phi_2^{\mathrm{Haar},A}\|_{\mathrm{op}} \leq 1/n$,
since $\Phi_2^{\mathrm{Haar},A}$ is the projector onto the same
fixed-point space (on the antisymmetric sector the Haar second moment
fixes only $P_A$, as $\mathrm{span}\{I_{n^2},\mathrm{SWAP}\} \cap
\{P_A Y P_A\} = \mathbb{C} P_A$) and the
two operators differ only on the traceless part of the sector, where the
butterfly's
operator has spectral norm $1/n$. This is the definition of an
$\varepsilon$-approximate $2$-design with $\varepsilon = O(1/n)$ on the
sector.
\end{proof}

\end{document}